\documentclass[12pt,reqno]{amsart}

\usepackage[margin = 1.25in]{geometry}
\usepackage{amssymb,amsmath,amsthm,amsfonts}
\usepackage{float}
\usepackage{afterpage}
\usepackage{times}
\usepackage[flushleft]{threeparttable}
\usepackage[title]{appendix}
\newtheorem{definition}{Definition}
\newtheorem{assumption}{Assumption}
\newtheorem{proposition}{Proposition}
\newtheorem{example}{Example}
\newtheorem{remark}{Remark}
\usepackage[]{pdfpages}
\newtheorem{theorem}{Theorem}
\newtheorem*{theorem*}{Theorem}
\newtheorem*{proposition*}{Proposition}
\newtheorem{lemma}{Lemma}
\usepackage[ruled,vlined]{algorithm2e}
\usepackage{mathrsfs}
\newtheorem{corollary}{Corollary}

\LinesNumbered
\DeclareMathOperator*{\argmin}{arg\,min}

\usepackage{graphicx}
\usepackage[inline]{enumitem}
\setlist[enumerate]{label=\arabic*.}
\setlist[enumerate,2]{label=(\alph*)}
\usepackage[colorlinks,
linkcolor = black,
            urlcolor  = blue,
            citecolor = black]{hyperref}
\usepackage[round]{natbib}
\usepackage{booktabs}
\usepackage{multirow}
\usepackage[all=normal, paragraphs=tight, bibbreaks=tight, floats=tight,
mathdisplays=tight,bibnotes=tight]{savetrees}
\makeatletter
\renewcommand\paragraph{\@startsection{paragraph}{4}%
  \z@{.5em}{-1em}%
  {\normalfont\bfseries}}
  \renewcommand\@addpunct[1]{}
\makeatother
\title{Semiparametric Bayesian Inference for a Conditional Moment Equality Model}
\author{Christopher D. Walker}
\date{\today. Duke University, Department of Economics: \href{mailto:christopher.walker@duke.edu}{christopher.walker@duke.edu}. This paper builds on the first chapter of my Ph.D. dissertation at Harvard University that had the same title. I am grateful to my Ph.D. advisors Elie Tamer, Isaiah Andrews, Anna Mikusheva, and Neil Shephard for their dedication and support. I also thank Aslihan Asil, Harvey Barnhard, Nick Cao, Jake Carlson, Jiafeng Chen, Meagan Currie, Edward Glaeser, Jeff Gortmaker, Andrew Kao, Isaac Meza, Ariel Pakes, Bas Sanders, Diego Santa Maria, Adam Rosen, Rahul Singh, James Stock, James Stratton, Davide Taglialatela, Wilbur Townsend, Jaume Vives-i-Bastida, Davide Viviano, and Jennifer Walsh for helpful comments. Finally, I thank numerous seminar and conference participants at Harvard University, Monash University, Notre Dame Econometrics Junior Conference, Brown University, Oxford Bayesian Reading Group, Georgetown University, University of Michigan, University of Toronto, University of Western Ontario, Columbia University, University of California Berkeley, Rice University, University of Pennsylvania, Duke University Department of Economics, University of Southern California, The University of Sydney Business Analytics, 2025 NBER/NSF SBIES Conference, Duke University Department of Statistical Science, Microeconometrics Class of 2025 Conference at Northwestern University, University of Virginia, The Pennsylvania State University, Stanford University, and Princeton University. All errors are mine.}
\begin{document}
\begin{abstract}
I propose a semiparametric Bayesian inference framework for conditional moment equalities. The core idea is that these models deterministically map a conditional distribution of data to a structural parameter via the restriction that a conditional expectation equals zero. Consequently, a posterior for the conditional distribution leads to a posterior for the structural parameter by minimizing the distance of the conditional moments to zero. The method has similar flexibility to frequentist semiparametric estimators and does not require converting the conditional moments into unconditional moments. I also establish frequentist asymptotic optimality of my proposal via a semiparametric Bernstein-von Mises theorem (BvM), which establishes that the posterior for the structural parameter is asymptotically normal and matches the \cite{chamberlain1987asymptotic} semiparametric efficiency bound. The BvM conditions are verified for Gaussian process priors and complement the numerical aspects of the paper in which these priors are used to estimate welfare effects. 
\end{abstract}
\maketitle
\newpage
\section{Introduction}\label{sec:intro}
\subsection{Overview}
Many economic parameters are identified via conditional moment equalities. These restrictions state that a finite-dimensional parameter of interest $\gamma $ uniquely solves $E_{P_{W|Z}}[g(W,Z,\gamma)|Z] = 0$, where $W $ and $Z $ are observed random vectors, $P_{W|Z}$ is the conditional distribution of $W $ given $Z$, and $g$ is a known vector-valued function. A canonical example is linear instrumental variables (IV) in which the regression model $Y=\gamma_{0}+\gamma_{1}D+\gamma_{2}X+U$ with $E[U|X,A] = 0$ leads to $E_{P_{W|Z}}[g(W,Z,\gamma)|Z]=0$ with $W=(Y,D)'$, $Z = (X,A)' $, and $g(W,Z,\gamma) = Y-\gamma_{0}-\gamma_{1}D-\gamma_{2}X$. Other examples include Euler equations \citep{hansen1982generalized}, demand and supply models \citep{berry1995automobile}, and peer effects models \citep{graham2008identifying}, to name a few. Conditional moment equalities are also theoretically interesting because they impose overidentifying restrictions on $P_{W|Z}$. This makes efficient estimation nontrivial and has motivated important work on estimation using the \cite{chamberlain1987asymptotic} optimal instruments \citep{newey1990efficient,NEWEY1993419,donald2001choosing,donald2003empirical,kitamura2004empirical,donald2009choosing,chen2021harmless,chib2022bayesian}.

This paper makes three contributions. The first is introducing a semiparametric Bayesian inference framework for conditional moment equalities. Consistent with some approaches to conditional moment equalities (e.g., \cite{ai2003efficient}), I leave $P_{W|Z}$ unrestricted and redefine $\gamma$ as a minimum distance estimand $\argmin_{\gamma \in \Gamma}||E_{P_{W|Z}}[g(W,Z,\gamma)|Z]||^{2}$, where $||\cdot||$ is a norm. Observing that this defines a deterministic relationship between $P_{W|Z}$ and $\gamma$, a Bayesian framework arises in which a posterior for $P_{W|Z}$ implies a posterior for $\gamma$. The prior for $P_{W|Z}$ is supported over a nonparametric family, a feature ensuring that this Bayesian approach has similar flexibility to frequentist semiparametric estimators that do not rely on a tightly parametrized likelihood model. Moreover, by targeting $P_{W|Z}$, the posterior for $\gamma$ does not require converting the conditional moments into unconditional moments. This is appealing because ad hoc selections of unconditional moments can result in efficiency loss, and, more severely, some (efficient) choices can result in $\gamma$ not being identifiable.\footnote{See Example 2 of \cite{dominguez2004consistent} for an instance where the unconditional moment equalities based on the optimal IVs have multiple roots, despite $\gamma$ uniquely minimizing the distance of the conditional moments to zero.} These benefits are compounded by the fact that, as a fully Bayesian procedure, it inherits all the basic benefits of Bayesian analysis (e.g., likelihood principle, expected utility decisions, simultaneous inference).

The second contribution is a semiparametric Bernstein-von Mises theorem (BvM) for the posterior for $\gamma$. BvMs provide conditions under which a posterior has the same limiting distribution as an efficient frequentist estimator, and are often invoked as a claim of large-sample robustness of the posterior to the choice of prior. Although BvMs hold under mild conditions in correctly specified regular parametric models, several papers find that these theorems can be delicate in semiparametric applications where the parameter of interest is a functional of a smooth nonparametric object, such as a probability density function, a regression function, etc. (see \cite{bickel2012semiparametric,castillo2012gaussian,castillo2012semiparametric,rivoirard2012bernstein,castillo2015bernstein,ray2020semiparametric}). Since $\gamma$ is a functional of a nonparametric $P_{W|Z}$, these findings raise the possibility that the posterior for $\gamma$ may be sensitive to the choice of prior for $P_{W|Z}$ in large samples. To resolve these concerns, I prove a general semiparametric BvM for the posterior for $\gamma$ in Section \ref{sec:BVM}, providing conditions under which it converges to a normal distribution, centered at an efficient estimator, and with variance equal to the \cite{chamberlain1987asymptotic} efficiency bound. As far as I am aware, this is the first fully Bayesian semiparametric BvM for the conditional moment equality model that does not require explicit conversion into unconditional moments.\footnote{See \cite{kato2013quasi} and \cite{kankanala2025generalized} for quasi-Bayesian results of a similar nature.} A byproduct of the BvM is that some Bayes estimators are asymptotically efficient frequentist estimators.

The third contribution is providing implementation details and verifying the BvM conditions for a flexible class of priors based on Gaussian processes (GP). GPs induce priors for conditional densities supported on the smoothness classes typically imposed on first-stage nuisance parameters, making concrete the flexibility of this Bayesian approach. Section \ref{sec:empirical} demonstrates implementation for these priors with an empirical illustration focused on estimating welfare effects, an exercise capturing a common use of conditional moments (i.e., estimating model primitives and evaluating counterfactuals).\footnote{A complementary data-calibrated simulation comparing my method with alternatives is in Appendix \ref{ap:simulation}.} Section \ref{sec:example} complements these numerical aspects by formally verifying the high-level BvM conditions for these GP priors. The key requirement is that the \cite{chamberlain1987asymptotic} optimal residual is well-approximated by a reproducing kernel Hilbert space (RKHS), a `no-bias' condition that restricts the smoothness of the optimal residual relative to the GP.
\subsection{Literature}
This paper connects to several literatures in econometrics and statistics. The first is the Bayesian analysis of moment conditions. \cite{chamberlain2003nonparametric} perform inference for structural parameters defined by \textit{unconditional} moment equalities by solving the moment restrictions using draws from a Dirichlet process posterior for the data distribution. The key conceptual distinction between their approach and mine is that the former cannot be applied to conditional moment equalities with continuous $Z$ unless the conditional moments are converted into unconditional moments.\footnote{In an unpublished manuscript, \cite{chamberlain1995semiparametric} consider conditional moment equalities, however their Dirichlet priors explicitly require the support of $Z$ to be small relative to the sample size, ruling out continuous $Z$.} \cite{chib2022bayesian} is also relevant because they propose a Bayesian inference procedure for conditional moment equalities in which the conditional moments are converted into a sieve of unconditional moments, and inference is performed using the Bayesian exponentially tilted empirical likelihood (ETEL) framework of \cite{schennach2005bayesian}. This conversion (and the use of ETEL) makes \cite{chib2022bayesian} fundamentally different to my proposal. Other papers on Bayesian analysis of moments include \cite{lancaster2010bayesian}, \cite{kitamura2011bayesian}, \cite{pelenis2014bayesian}, \cite{shin2015bayesian}, \cite{bornn2018moment}, and \cite{chib2018bayesian}.

Another related literature is on semiparametric BvMs. \cite{castillo2015bernstein} prove a general BvM for smooth functionals of nonparametric models, with substantial analysis devoted to those of \textit{unconditional} probability density functions (building on \cite{rivoirard2012bernstein}). I characterize $\gamma$ as a smooth functional of a \textit{conditional} probability density function, and, via conditional asymptotics, I am able to modify their proof to show that the $\gamma$ posterior achieves the \cite{chamberlain1987asymptotic} limit. \cite{chib2022bayesian} also prove a BvM compatible with \cite{chamberlain1987asymptotic}, however their proof differs from mine because, under appropriate conditions, the sieve of unconditional moment equalities makes parametric BvM arguments (cf. Theorem 10.1 of \cite{vaart_1998}) applicable to the ETEL. Other contributions to the semiparametric BvM literature include \cite{shen2002asymptotic}, \cite{bickel2012semiparametric}, \cite{castillo2012gaussian,castillo2012semiparametric}, \cite{norets2015bayesian}, \cite{ray2020semiparametric}, \cite{monard2021bernstein}, \cite{breunig2025double,breunig2025semiparametricbayesiandifferenceindifferences}, \cite{yiu2025semiparametric}, and \cite{walker2024parametrization}.

A third literature is efficient frequentist estimation of conditional moment equalities. The posterior for $\gamma$ is formalized as that of a conditional estimand that solves an efficiently weighted minimum distance problem. For this reason, my approach offers a Bayesian analog to the \cite{ai2003efficient} framework (except I focus on finite-dimensional parameters). Moreover, the first-order conditions of the minimum distance estimand relates my framework to method of moments estimators that plug in estimators of the \cite{chamberlain1987asymptotic} optimal instruments (e.g., those studied in \cite{robinson1987heteroskedasticity}, \cite{newey1990efficient}, \cite{NEWEY1993419}, and \cite{chen2021harmless}). Other frequentist approaches to optimal instruments estimation include Generalized Method of Moments (GMM) and Generalized Empirical Likelihood (EL) estimators based on sieves of unconditional moments \citep{donald2001choosing,donald2003empirical,donald2009choosing}, and localized EL estimators \citep{,kitamura2004empirical}.

A final related literature is quasi-Bayes. Popularized in \cite{chernozhukov2003mcmc}, quasi-Bayes is a sampling-based frequentist estimation framework in which an extremum objective forms a log quasi-likelihood, a prior is assumed for the parameter, and Bayes rule is used to obtain a quasi-posterior. \cite{kato2013quasi} and \cite{kankanala2025generalized} develop quasi-Bayesian frameworks for conditional moment equalities, and prove quasi-Bayesian BvMs for a structural function defined by conditional moment equalities and linear functionals of it, respectively. A key conceptual distinction between my approach and quasi-Bayes is that mine has a finite-sample Bayesian interpretation. Other papers on quasi-Bayesian estimation include \cite{liao2011posterior}, \cite{florens2012nonparametric,florens2021gaussian}, \cite{chen2018monte}, \cite{andrews2022optimal}, and \cite{kankanala2023quasi}.
\subsection{Outline}
The paper is organized as follows. Section \ref{sec:overview} sets up the general Bayesian inference framework and provides a flexible class of priors, Section \ref{sec:empirical} presents an empirical illustration, Section \ref{sec:BVM} states the semiparametric BvM and related results, Section \ref{sec:example} verifies the BvM conditions for the GP priors, and Section \ref{sec:conclusion} concludes with some extensions. The Appendix contains proofs of theorems and corollaries, technical details about implementation, a simulation study calibrated to the empirical application, and further discussion. Proofs of propositions, and lemmas, as well as some additional discussion are in the Online Appendix \citep{walker2026supplement}. Notation is introduced when appropriate.
\section{Bayesian Inference Framework}\label{sec:overview}
This section formalizes semiparametric Bayesian inference and describes implementation for a class of priors based on Gaussian processes.
\subsection{Semiparametric Bayesian Inference}\label{sec:framework}
I start with Bayesian inference for the conditional distribution $P_{W|Z}$ of $W$ given $Z$. Let $\{(W_{i}',Z_{i}')'\}_{i \geq 1}$, $W_{i} \in \mathcal{W} \subseteq \mathbb{R}^{d_{w}}$ and $Z_{i} \in \mathcal{Z} \subseteq \mathbb{R}^{d_{z}}$, be a sequence of random vectors for which the first $n \geq 1$ elements $\{(W_{i}',Z_{i}')'\}_{i=1}^{n}$ forms the observed data. The sampling model conditions on the realization $\{z_{i}\}_{i \geq 1}$ of $\{Z_{i}\}_{i \geq 1}$ and assumes that the elements of $W^{(n)}= (W_{1}',...,W_{n}')'$ satisfy $W_{i}|p_{W|Z} \overset{ind}{\sim} P_{W|z_{i}}$ for $i=1,...,n$, where $p_{W|Z}$ is a conditional probability density function, and $P_{W|z_{i}}$ is the probability distribution attached to $p_{W|Z}(\cdot|z_{i})$. Since the model parameters are conditional densities $p_{W|Z}$, the \textit{prior} $\Pi$ is a (conditional) probability measure defined over the space $\mathcal{P}_{W|Z}$ of conditional densities (equipped with a $\sigma$-algebra $\mathscr{P}_{W|Z}$).\footnote{The qualifier `conditional' is used because implicitly $\Pi$ also defined conditional on $\{Z_{i}\}_{i \geq 1}$.} There are no parametric restrictions on $\mathcal{P}_{W|Z}$, so $\Pi$ should be thought of as a probability measure over an infinite-dimensional space of conditional densities.
\begin{assumption}\label{as:measurability}
For each $n \geq 1$ and almost every fixed realization $\{z_{i}\}_{i \geq 1}$ of $\{Z_{i}\}_{i \geq 1}$, the mapping $(w^{(n)},p_{W|Z}) \mapsto \prod_{i=1}^{n}p_{W|Z}(w_{i}|z_{i})$ is measurable on $(\mathcal{W}^{n} \times \mathcal{P}_{W|Z}, \mathscr{W}^{\otimes n} \otimes \mathscr{P}_{W|Z})$, where $\mathscr{W}$ is a $\sigma$-algebra over $\mathcal{W}$, and $\mathscr{W}^{\otimes n}$ is the corresponding product $\sigma$-algebra.
\end{assumption}
\begin{assumption}\label{as:evidence}
Let $p_{0,W|Z}$ be the true conditional density. For each $ n \geq 1$ and almost every fixed realization $\{z_{i}\}_{i \geq 1}$ of $\{Z_{i}\}_{i \geq 1}$, $P_{0,W|Z}^{(n)}(\int \prod_{i=1}^{n}p_{W|Z}(W_{i}|z_{i})d\Pi(p_{W|Z}) \in (0,\infty)) = 1$, where $P_{0,W|Z}^{(n)} = \bigotimes_{i=1}^{n}P_{0,W|z_{i}}$ (and $P_{0,W|z_{i}}$ is the distribution attached to $p_{0,W|Z}(\cdot|z_{i})$).
\end{assumption}
Assumptions \ref{as:measurability} and \ref{as:evidence} are standard regularity conditions that imply that the conditional distribution $\Pi(p_{W|Z}\in \cdot |W^{(n)})$ of $p_{W|Z}$ given $W^{(n)}$ (known as the \textit{posterior}) has a version satisfying Bayes theorem,\footnote{See, for example, Section 1.3 of \cite{ghosal2017fundamentals} for similar regularity conditions.}
\begin{align*}
    \Pi\left(p_{W|Z} \in A \middle |W^{(n)}\right) = \frac{\int_{A}L_{n}(p_{W|Z})d\Pi(p_{W|Z})}{\int L_{n}(p_{W|Z})d\Pi(p_{W|Z})}, \ \ A \in \mathscr{P}_{W|Z}, \ L_{n}(p_{W|Z}) = \prod_{i=1}^{n}p_{W|Z}(W_{i}|z_{i}).
\end{align*}
Since $\Pi(p_{W|Z} \in \cdot | W^{(n)})$ contains an inference theory for (transformations of) $p_{W|Z}$, inference for $\gamma$ is obtained upon formalizing the estimand $\argmin_{\gamma \in \Gamma}||E_{P_{W|Z}}[g(W,Z,\gamma)|Z]||^{2}$.

The estimand is an iterated minimum distance estimand. For a given conditional density $p_{W|Z}$, let $m(z_{i},\gamma) = E_{P_{W|Z}}[g(W,Z,\gamma)|Z =z_{i}]$ be the $d_{g}\times 1$ vector of conditional moments at $z_{i}$, let $\Sigma(z_{i},\gamma) = E_{P_{W|Z}}[g(W,Z,\gamma)g(W,Z,\gamma)'|Z=z_{i}]$ be the $d_{g}\times d_{g}$ matrix of conditional second moments at $z_{i}$, and, for $\tilde{\gamma} \in \Gamma$ fixed, define a conditional estimand $$q_{n}(\tilde{\gamma},p_{W|Z}) = \argmin_{\gamma \in \Gamma}Q_{n}(\gamma,\tilde{\gamma},p_{W|Z}), $$ where\footnote{The terminology `conditional estimand' follows \cite{abadie2014inference}.}
\begin{align}
Q_{n}(\gamma,\tilde{\gamma},p_{W|Z}) = \frac{1}{n}\sum_{i=1}^{n}m(z_{i},\gamma)'\Sigma^{-1}(z_{i},\tilde{\gamma})m(z_{i},\gamma) .   \label{eq:criterionfunctionsec2}
\end{align} 
To eliminate dependence on $\tilde{\gamma}$, I adopt a similar technique to \cite{hansen2021inference} and consider an iterated estimand defined by the recursive relation
\begin{align}\label{eq:limitfp}
    \gamma_{n} = \lim_{r\rightarrow \infty}\gamma_{n,r}, \quad \gamma_{n,r} = q(\gamma_{n,r-1},p_{W|Z}), \quad r \geq 1,
\end{align}
where $\gamma_{n,0}$ is some element of $\Gamma$ (e.g., the identity-weighted estimand). Notice that $\gamma_{n}$ does not require converting the conditional moments into unconditional moments. Further, since $\gamma_{n}$ satisfies the first-order conditions $n^{-1}\sum_{i=1}^{n}M(z_{i},\gamma_{n})'\Sigma^{-1}(z_{i},\gamma_{n})m(z_{i},\gamma_{n})=0$, it implicitly plugs in an estimate $M(z_{i},\gamma_{n})' \Sigma^{-1}(z_{i},\gamma_{n})$ of the \cite{chamberlain1987asymptotic} optimal instruments, where $M(z_{i},\gamma)$ is the $d_{g} \times d_{\gamma}$ Jacobian of $m(z_{i},\gamma)$ in $\gamma$. This is important for the efficiency guarantees in Section \ref{sec:BVM}.
\begin{assumption}\label{as:uemfp}
For almost every fixed realization $\{z_{i}\}_{i \geq 1}$ of $\{Z_{i}\}_{i \geq 1}$ and for each $n \geq 1$, the limit $\gamma_{n}$ exists for each $p_{W|Z} \in \mathcal{P}_{W|Z}$ and $p_{W|Z} \mapsto \gamma_{n}$ is  $\mathscr{P}_{W|Z}$-measurable.
\end{assumption}
Assumption \ref{as:uemfp} guarantees that $\Pi(p_{W|Z} \in \cdot |W^{(n)})$ leads to a well-defined marginal posterior for $\gamma_{n}$ given by the pushforward measure $\Pi(\gamma_{n} \in \cdot |W^{(n)}) = \Pi(p_{W|Z} \in \cdot|W^{(n)}) \circ \gamma_{n}^{-1}$.\footnote{Two points. First, the pushforward satisfies $(\Pi(p_{W|Z} \in \cdot|W^{(n)}) \circ \gamma_{n}^{-1})(A) = \Pi(p_{W|Z}: \gamma_{n} \in A|W^{(n)})$ for events $A$. Second, Appendix \ref{ap:measurability} contains sufficient conditions for Assumption \ref{as:uemfp} that may be of independent interest.} This formalizes the posterior for $\gamma$. Like $\Pi(p_{W|Z} \in \cdot | W^{(n)})$, it leads to inference about $\gamma_{n}$ and any function $f(\gamma_{n})$. For example, in the leading case where $f(\gamma)$ is scalar, a point estimate can be obtained using the posterior median $c_{n,f}(0.5)$, while uncertainty can be quantified using a $(1-\alpha)$-equitailed probability interval $CS_{n}(1-\alpha) = [c_{n,f}(\alpha/2),c_{n,f}(1-\alpha/2)]$ (a type of credible set), where $\alpha \in (0,1/2)$, and $c_{n,f}(q)$, $q \in (0,1)$, is the $q$-quantile of $\Pi(f(\gamma_{n}) \in \cdot | W^{(n)})$. Section \ref{sec:BVM} establishes the asymptotic equivalence of these Bayes estimators and frequentist estimators based on the \cite{chamberlain1987asymptotic} optimal instruments.

\begin{example}\label{ex:linearIV}
Recall that linear IV sets $W = (Y,D')'$, $Z = (X',A')'$, and $g(W,Z,\gamma) = Y-\gamma_{0}-\gamma_{1}'D-\gamma_{2}'X$. In such case, $\gamma = (\gamma_{0},\gamma_{1}',\gamma_{2}')'$, $m(z,\gamma) = m_{y}(z)-M(z)'\gamma$, where $m_{y}(z) = E_{P_{W|Z}}[Y|Z=z]$ and $M(z) = (1,m_{d}(z),x)'$ with $m_{d}(z) = E_{P_{W|Z}}[D|Z=z]$, and $\Sigma(z,\gamma) = E_{P_{W|Z}}[g^{2}(W,Z,\gamma)|Z=z]$. Substituting $m(z,\gamma)$ and $\Sigma(z,\tilde{\gamma})$ into (\ref{eq:criterionfunctionsec2}),
\begin{align}\label{eq:IVminimizer}
q_{n}(\tilde{\gamma},p_{W|Z}) = \left(\frac{1}{n}\sum_{i=1}^{n}\Sigma^{-1}(z_{i},\tilde{\gamma})M(z_{i})M(z_{i})' \right)^{-1}\frac{1}{n}\sum_{i=1}^{n}\Sigma^{-1}(z_{i},\tilde{\gamma})M(z_{i})m_{y}(z_{i})   
\end{align}
Since $p_{W|Z}$ entirely determines $\{q_{n}(\tilde{\gamma},p_{W|Z}): \tilde{\gamma} \in \Gamma\}$, a posterior for $p_{W|Z}$ implies a marginal posterior for $\{q_{n}(\tilde{\gamma},p_{W|Z}):\tilde{\gamma} \in \Gamma\}$, and, under Assumption \ref{as:uemfp}, leads to a posterior for $\gamma_{n}$. An example of $f$ is $f(\gamma_{n}) = \gamma_{n,1}$, the subvector of coefficients on $D$.
\end{example}
\begin{remark}[Other Minimum Distance Estimands]
Although this paper focuses on the iterated estimand $\gamma_{n}$, the `two-step' estimand $\gamma_{n,1} = q_{n}(\gamma_{n,0},p_{W|Z})$ and continuous updating estimand $\gamma_{n,CU} = \argmin_{\gamma \in \Gamma} Q_{n}(\gamma,\gamma,p_{W|Z})$ can also be viewed as deterministic transformations of $p_{W|Z}$. Consequently, $\Pi(p_{W|Z} \in \cdot | W^{(n)})$ leads to a joint posterior for $(\gamma_{n},\gamma_{n,1},\gamma_{n,CU})$. I leave analysis of this joint posterior as future work.
\end{remark}
\begin{remark}[Prior on the Structural Parameter]
The marginal prior $\Pi_{\Gamma}$ for $\gamma$, given by $\Pi_{\Gamma} = \Pi \circ \gamma_{n}^{-1}$, may not agree with a researcher's prior $\tilde{\Pi}_{\Gamma}$ about $\gamma$. Under regularity conditions, a marginal prior $\tilde{\Pi}_{\Gamma}$ can be incorporated using techniques from \cite{dunson2015marginally}. 
\end{remark}
\begin{remark}[Extensions]
The posterior $\Pi(p_{W|Z} \in \cdot | W^{(n)})$ can be used for other inference problems related to the conditional moment equalities (i.e., beyond inference for $\gamma_{n}$). See Section \ref{sec:conclusion} for several examples that could each form the basis for future research.
\end{remark}
\subsection{A Flexible Class of Priors}\label{sec:implementation}
I present a flexible class of priors for $p_{W|Z}$. The priors are based on logistic transformations of Gaussian processes (GPs), and are a leading class of priors in the Bayesian nonparametrics literature \citep{tokdar2007towards,TOKDAR200734,vaart2008rates,tokdar2010bayesian,vehtari2014laplace}. Specifically, the conditional density of $W$ given $Z$ is modeled as an infinite-dimensional exponential family 
\begin{align*}
p_{\theta,B}(w|z)= f_{\theta}(w|z)\frac{\exp(B(F_{\theta,z}(w),F_{0}(z)))}{\int_{[0,1]^{d_{w}}}\exp(B(u,F_{0}(z)))du},    
\end{align*}
where $f_{\theta}(\cdot|z)$, $\theta \in \Theta \subseteq \mathbb{R}^{d_{\theta}}$ with $d_{\theta} < \infty$, is a conditional density, $F_{\theta,z}(\cdot)$ is the corresponding $d_{w}\times 1$ vector of conditional cumulative distribution functions, $F_{0}(\cdot)$ is a known invertible mapping taking values in $[0,1]^{d_{z}}$, and $B: [0,1]^{d} \rightarrow \mathbb{R}$ is a function with $d = d_{w}+ d_{z}$.\footnote{Let $f_{\theta,j}(w_{j}|w_{1},...,w_{j-1},z)$ denote the density of $W_{j}$ given $\{W_{j'}\}_{j'=1}^{j-1}$ and $Z$ under $f_{\theta}(w|z)$. The elements of $F_{\theta,z}(w)=(F_{\theta,1,z}(w),...,F_{\theta,d_{w},z}(w))$ satisfy $F_{\theta,j,z}(w) = \int_{-\infty}^{w_{j}} f_{\theta,j}(t|w_{1},...,w_{j-1},z)dt$. Using this and the change of variables $u=F_{\theta,z}(w)$, one can show $\int_{\mathcal{W}}p_{\theta,B}(w|z)dw = 1$ for each $z$.} Since $(\theta,B)$ determines $p_{\theta,B}$, a prior for $p_{\theta,B}$ is implied by a prior $\Pi$ for $(\theta,B)$. I set $\Pi = \Pi_{\Theta} \otimes \Pi_{\mathcal{B}}$, where $\Pi_{\Theta}$ is a probability distribution over $\Theta$, and $\Pi_{\mathcal{B}}$ is the probability law of a centered GP. To define the latter, let $(C([0,1]^{d},\mathbb{R}),||\cdot||_{\infty})$ be the space of continuous $h: [0,1]^{d} \rightarrow \mathbb{R}$ with $||h||_{\infty} = \sup_{t \in [0,1]^{d}}|h(t)|$.
\begin{definition}\label{def:GP}
A centered GP in $(C([0,1]^{d},\mathbb{R}),||\cdot||_{\infty})$ is a continuous stochastic process $\{B(t):t \in [0,1]^{d}\}$ such that, for every $t_{1},...,t_{k} \in [0,1]^{d}$ and $k \geq 1$, $(B(t_{1}),...,B(t_{k}))'$ follows a mean-zero Gaussian distribution with covariance matrix $(\kappa(t_{j},t_{l}))_{j,l=1}^{k}$. The symmetric, nonnegative-definite $\kappa: [0,1]^{d} \times [0,1]^{d} \rightarrow \mathbb{R}$ is called the covariance function.\footnote{The function $\kappa$ being symmetric and nonnegative-definite means that, for any $t_{1},...,t_{k} \in [0,1]^{d}$, the matrix $(\kappa(t_{j},t_{l}))_{j,l=1}^{k}$ is symmetric and positive-semidefinite, thereby making it suitable for covariance matrices.}
\end{definition}

Logistically transformed GPs are flexible. Fixing $\theta$ and imposing $z \in [0,1]^{d_{z}}$ for simplicity, $p_{0,W|Z}$ is in the support of the prior if $\Pi_{\mathcal{B}}(||B-b_{0,\theta}||_{\infty} < \delta) > 0$ for each $\delta > 0$, where $b_{0,\theta}(t) = \log p_{0,W|Z}(F_{\theta,z}^{-1}(u)|z) - \log f_{\theta}(F_{\theta,z}^{-1}(u)|z)$ for $t = (u',z')' \in [0,1]^{d}$ is $\log(p_{0,W|Z}/f_{\theta})$ with $w$ transformed to be defined on $[0,1]^{d_{w}}$.\footnote{Two comments. First, the support claim holds because the Kullback-Leibler support of the prior for $p_{\theta,B}$, which is the standard definition of the support of a nonparametric prior, is determined by the uniform support of $\Pi_{\mathcal{B}}$ (Lemma 3.1 in \cite{vaart2008rates}). Second, to allow for $z \in \mathcal{Z}\nsubseteq [0,1]^{d_{z}}$, replace $z$ with $F_{0}^{-1}(v)$ where $v \in [0,1]^{d_{z}}$.} For common GPs, this condition is typically satisfied if $b_{0,\theta}$ is continuous; however, to ensure certain frequentist properties for statistical functionals (i.e., $\gamma_{n}$), continuity of $b_{0,\theta}$ is often strengthened to H\"{o}lder or Sobolev-type smoothness restriction on $b_{0,\theta}$.\footnote{The H\"{o}lder space $C^{\alpha}([0,1]^{d},\mathbb{R})$, $\alpha > 0$, is the set of functions $f: [0,1]^{d} \rightarrow \mathbb{R}$ that are $\lfloor \alpha \rfloor$-times differentiable with bounded derivatives, and, additionally, the $\lfloor \alpha \rfloor$th derivative is $(\alpha-\lfloor \alpha \rfloor)$-H\"{o}lder continuous. It is equipped with norm $||f||_{\alpha} = \max_{|k| \leq \lfloor \alpha \rfloor}||D^{k}f||_{\infty}+ \max_{|k| = \lfloor \alpha \rfloor}\sup_{t,s \in [0,1]^{d}: t \neq s}\frac{|(D^{k}f)(t)-(D^{k}f)(s)|}{||t-s||_{2}^{\alpha-\lfloor\alpha \rfloor}}$, where $k=(k_{1},...,k_{d})'$ is a vector of nonnegative integers, $|k| = \sum_{j=1}^{d}k_{j}$, $D^{k}$ is the differential operator, and $||\cdot||_{2}$ is the Euclidean norm. The Sobolev space $S^{\alpha}([0,1]^{d},\mathbb{R})$, $\alpha > 0$, comprises functions $f:[0,1]^{d} \rightarrow \mathbb{R}$ that are restrictions of functions $f: \mathbb{R}^{d} \rightarrow \mathbb{R}$ with Fourier transforms $\hat{f}$ such that $||f||_{2,2,\alpha}^{2} = \int_{\mathbb{R}^{d}} (1+||\lambda||_{2}^{2})^{\alpha}|\hat{f}(\lambda)|^{2} < \infty$. The Sobolev space norm is $||\cdot||_{2,2,\alpha}$. \label{fn:functionspaces}} Similar smoothness conditions are often imposed on first-stage nuisance parameters in frequentist estimation of conditional moment equalities (e.g., \cite{newey1990efficient,NEWEY1993419}, \cite{ai2003efficient}, \cite{kankanala2025generalized}).
\begin{example}\label{ex:matern}
A centered GP belongs to the Mat\'{e}rn class if the covariance function satisfies $\kappa(t,s) = \int_{\mathbb{R}^{d}}\exp(i\lambda'(t-s))\mu(\lambda)d\lambda$ for $t,s \in [0,1]^{d}$, where $\mu(\lambda) = (1+||\lambda||_{2}^{2})^{-\alpha-d/2}$, $\alpha > 0$, and, here, $i$ is the imaginary unit. The hyperparameter $\alpha$ indexes smoothness because the sample paths $B$ take values in the H\"{o}lder spaces $ C^{a}([0,1]^{d},\mathbb{R})$ for any $a< \alpha$ \citep{van2011information}. In Section \ref{sec:example}, I show that if $b_{0,\theta}$ is in the intersection of a H\"{o}lder space $C^{\alpha_{0}}([0,1]^{d},\mathbb{R})$ and Sobolev space $S^{\alpha_{0}}([0,1]^{d},\mathbb{R})$ for some $\alpha_{0} > 0$ (and other conditions hold), then the general BvM for $\gamma_{n}$ from Section \ref{sec:BVM} can be verified for Mat\'{e}rn GPs. In practice, special functions are used to compute $\kappa$ \citep{williams2006gaussian}.
\end{example}

Posterior computation for logistic GP priors proceeds as follows. The MCMC algorithms proposed in \cite{tokdar2007towards} and \cite{tokdar2010bayesian} can be used to obtain $S \geq 1$ draws $\{(\theta^{[s]},B^{[s]})\}_{s=1}^{S}$ from the posterior for $(\theta,B)$ (see Appendix \ref{ap:empirical}). Given a posterior draw $(\theta^{[s]},B^{[s]})$, the conditional moments $m^{[s]}(z_{i},\gamma)$ and $\Sigma^{[s]}(z_{i},\gamma)$ at $z_{i}$, $i=1,...,n$, can be estimated using importance sampling: for each $i$, 1. generate $\{W_{i,j}^{[s]}\}_{j=1}^{J}\overset{iid}{\sim} f_{\theta^{[s]}}(\cdot|z_{i})$, 2. compute importance weights $\{\omega_{i,j}^{[s]}\}_{j=1}^{J}$, 
\begin{align*}
\omega_{i,j}^{[s]} = \frac{\exp(B^{[s]}(F_{\theta^{[s]},z_{i}}(W_{i,j}^{[s]}),F_{0}(z_{i})))}{\sum_{j'=1}^{J}\exp(B^{[s]}(F_{\theta^{[s]},z_{i}}(W_{i,j'}^{[s]}),F_{0}(z_{i})))}, \quad j=1,...,J  
\end{align*}
and 3. compute importance-weighted averages 
\begin{align*}
m_{J}^{[s]}(z_{i},\gamma) = \sum_{j=1}^{J}\omega_{i,j}^{[s]}g(W_{i,j}^{[s]},z_{i},\gamma), \quad \Sigma_{J}^{[s]}(z_{i},\gamma) = \sum_{j=1}^{J}\omega_{i,j}^{[s]}g(W_{i,j}^{[s]},z_{i},\gamma)g(W_{i,j}^{[s]},z_{i},\gamma)'.  
\end{align*}
A draw $\gamma_{n}^{[s]}$ from $\Pi(\gamma_{n} \in \cdot | W^{(n)})$ is then obtained by solving the iterated minimum distance problem, replacing $m(z_{i},\gamma)$ and $\Sigma(z_{i},\tilde{\gamma})$ in (\ref{eq:criterionfunctionsec2}) with $m_{J}^{[s]}(z_{i},\gamma)$ and $\Sigma_{J}^{[s]}(z_{i},\tilde{\gamma})$, respectively. Performing these steps over $\{(\theta^{[s]},B^{[s]})\}_{s=1}^{S}$ leads to a sample $\{\gamma_{n}^{[s]}\}_{s=1}^{S}$ from $\Pi(\gamma_{n} \in \cdot | W^{(n)})$. A sample from $\Pi(f(\gamma_{n}) \in \cdot | W^{(n)})$ is obtained by computing $\{f(\gamma_{n}^{[s]})\}_{s=1}^{S}$ using $\{\gamma_{n}^{[s]}\}_{s=1}^{S}$. Bayes estimators (e.g., posterior medians, credible sets) are computed using the empirical distribution of $\{f(\gamma_{n}^{[s]})\}_{s=1}^{S}$.
\begin{remark}[Other Priors]
There are many other priors for conditional densities, such as covariate-dependent mixtures of normals \citep{norets2010approximation,norets2014posterior,norets2017adaptive}, and covariate-dependent stick-breaking priors \citep{dunson2008kernel,ren2011logistic,rodriguez2011nonparametric}. The modularity of my approach (due to the unrestricted $p_{W|Z}$) means that, in principle, it could be applied to all of these priors (with precise implementation details depending on the class), and verifying the high-level BvM conditions in Section \ref{sec:BVM} may offer some interesting new theoretical insights. I leave this for future research.
\end{remark}
\section{Empirical Illustration: Estimating Welfare Effects}\label{sec:empirical}
This section uses my proposal to estimate welfare effects of price changes. A data-calibrated simulation that compares my approach with alternatives is in Appendix \ref{ap:simulation}.
\subsection{Dataset and Target Parameter}
I use the 2001 National Household Travel Survey gasoline demand dataset from \cite{blundell2012measuring,blundell2017nonparametric}.\footnote{The dataset was obtained from the publicly available replication files of \cite{chen2018optimal}. They can be found using the link \url{https://github.com/timothymchristensen/NPIV}.} This dataset contains household gasoline consumption $Q$ (in gallons), average price $P$ of gasoline (in dollars per gallon) in the household's county, household income $Y$ (in dollars), and distance $A$ (in 1,000 kilometers) of a household's state capital to a major oil platform in the Gulf of Mexico. There are $n=4,812$ households in the dataset.

The parameter of interest is the welfare effect of a gasoline price change (e.g., due to taxation). Household gasoline demand is modeled using a constant elasticity specification 
\begin{align*}
\log Q = \gamma_{0}+ \gamma_{1} \log P + \gamma_{2} \log Y + U,    
\end{align*}
where the error term $U$ satisfies $E[U|\log Y, \log A] = 0$.\footnote{Distance as an excluded IV is proposed in \cite{blundell2012measuring} and is implemented in \cite{chen2018optimal} too. A justification is that distance is a cost-shifter in that it affects prices only through its impact on firm transportation costs.} The welfare effect of  a price change from $p^{0}$ to $p^{1}$ at income $y$ is measured using deadweight loss (DWL), 
\begin{align*}
DWL(p^{0},p^{1},y,\gamma) = S(p^{0},y,\gamma)-(p^{1}-p^{0})q(p^{1},y,\gamma),    
\end{align*}
where $S(p^{0},y,\gamma)$ is consumer surplus \citep{hausman1981exact,hausman1995nonparametric}, and $q(p,y,\gamma) = \exp(\gamma_{0}+\gamma_{1} \log p + \gamma_{2} \log y)$. \cite{walker2026supplement} details DWL calculation.
\subsection{Semiparametric Bayesian Inference}
Let $W = (\log Q ,\log P)'$, $Z = (\log Y, \log A)'$, and $g(W,Z,\gamma) = \log Q - \gamma_{0} - \gamma_{1} \log P -\gamma_{2} \log Y$. Since $E[U|\log Y, \log A] = 0$ implies $E_{P_{W|Z}}[g(W,Z,\gamma)|Z] = 0$, Bayesian inference for $\gamma$ can be obtained via the approach in Section \ref{sec:overview} (in fact, it is an instance of Example \ref{ex:linearIV}). Furthermore, since $DWL(p^{0},p^{1},y,\gamma)$ is a deterministic function of $\gamma$ (as the researcher sets $(p^{0},p^{1},y)$), one automatically obtains a posterior for $DWL(p^{0},p^{1},\gamma)$. This highlights that my framework offers simultaneous inference for model primitives ($\gamma$) \textit{and} counterfactuals ($DWL(p^{0},p^{1},y,\gamma)$).

Figure \ref{fig:dwl} and Table \ref{tab:empirical} report the posterior densities and Bayes estimates, respectively, of the DWL for $p^{0} = \$1.22$, $p^{1} = \$1.44$, and $y \in \{\$42500,\$57500,\$72500\}$.\footnote{These price-income combinations are the same as those reported in \cite{blundell2012measuring}.} The posteriors are based on a logistic GP prior for $p_{W|Z}$ (from Section \ref{sec:implementation}), with a homoskedastic Gaussian linear regression of $W$ on $Z$ for $f_{\theta}$, diffuse priors for the location-scale parameters $\theta$, and a Mat\'{e}rn GP for $B$ with $\alpha = 5/2$.\footnote{Two comments. First, Appendix \ref{ap:empirical} contains more details about the prior. Second, setting $\alpha = 5/2$ follows recommendations in \cite{williams2006gaussian}; \cite{walker2026supplement} shows changing $\alpha$ has a modest effect on the estimates.} Figure \ref{fig:dwl} indicates that the DWL posteriors are approximately symmetric and bell-shaped across all income groups. The estimates in Table \ref{tab:empirical} reveal that DWL as a percentage of tax paid is almost identical across the income groups, however, as a proportion of income, deadweight loss decreases monotonically with income level. These findings are qualitatively similar to the constant elasticity results in \cite{blundell2012measuring}, though my results relax price exogeneity.
\begin{figure}
    \centering
        \caption{Deadweight Loss Posterior Densities}
    \includegraphics[width=\textwidth]{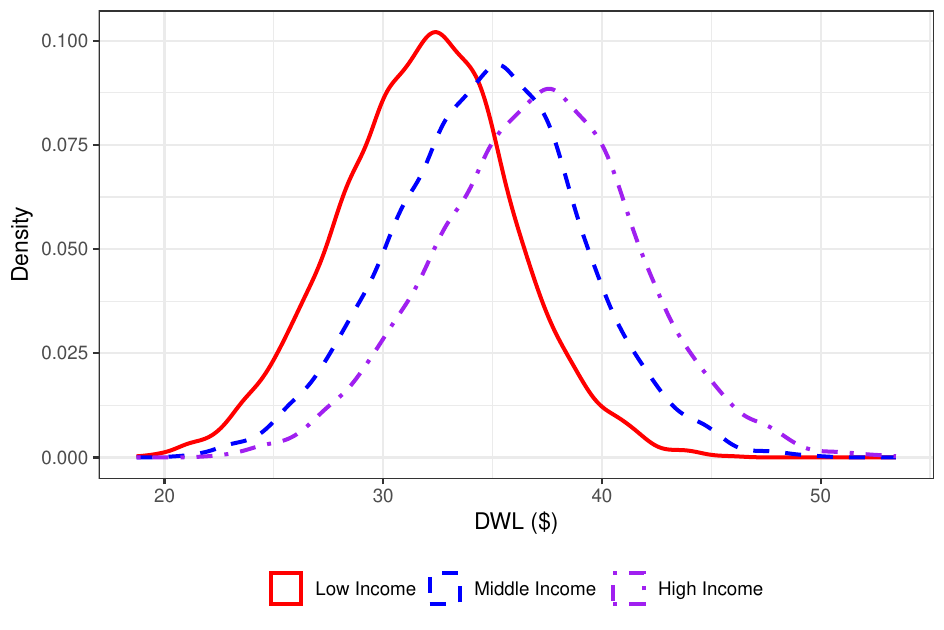}
    \label{fig:dwl}
\end{figure}
\begin{table}
\centering
\caption{Deadweight Loss Estimates}\label{tab:empirical}
\resizebox{\textwidth}{!}{ 
\begin{threeparttable}
\begin{tabular}{l cc cc cc}
\toprule
  & \multicolumn{2}{c}{Low Income} 
  & \multicolumn{2}{c}{Middle Income} 
  & \multicolumn{2}{c}{High Income} \\
\cmidrule(lr){2-3} \cmidrule(lr){4-5} \cmidrule(lr){6-7}
  & Point Est. & Cred. Int. 
  & Point Est. & Cred. Int. 
  & Point Est. & Cred. Int.  \\
\midrule
DWL (\$)        & $32.02$ & $(23.86,39.82)$  & $34.83 $ & $(25.98,43.25)$ & $ 37.15 $ & $(27.75,45.12)$ \\
As a \% of Income    & $0.075$ & $(0.056,0.093)$ & $0.061$ & $(0.045,0.075)$ & $0.051$ & $(0.038,0.064)$ \\
As a \% of Tax Paid  & $13.39$ & $(9.69,17.09)$  & $13.40$ & $(9.71,17.11)$  & $13.42$ & $(9.72,17.13)$  \\
\bottomrule
\end{tabular}
\begin{tablenotes}
\small
\item \textit{Notes:} Results are based on $100{,}000$ iterations of the Metropolis-within-Gibbs outlined in Appendix \ref{ap:postsample}, discarding the first half as burn-in and saving every fifth draw thereafter. Point estimates are posterior medians and interval estimates are 95\% equitailed probability intervals.
\end{tablenotes}
\end{threeparttable}
}
\end{table}
\section{Semiparametric Bernstein-von Mises Theorem}\label{sec:BVM}
This section proves a general BvM, proving that $\Pi(\gamma_{n} \in \cdot|W^{(n)})$ is asymptotically Gaussian and achieves the \cite{chamberlain1987asymptotic} semiparametric efficiency bound.
\subsection{Assumptions about the Data-Generating Process (DGP)}\label{sec:DGPHL}
I start with assumptions about the DGP. For notation, let $||\cdot||_{2}$ be the Euclidean norm, let $||\cdot||_{op}$ be the operator norm (i.e., the maximum singular value), let $\lambda_{min}(A)$ and $\lambda_{max}(A)$ be the minimum and maximum eigenvalues, respectively, of matrix $A$, let $\text{int}(B)$ be the interior of a set $B$, and let $vec(A)$ be the vectorized matrix $A$. A function class is Glivenko-Cantelli if it obeys a uniform strong law of large numbers (USLLN).
\begin{assumption}\label{as:dgp2}
1. $\{(W_{i}',Z_{i}')'\}_{i \geq 1}$ is an i.i.d sequence with common distribution $P_{0,WZ} = P_{0,W|Z} \otimes P_{0,Z}$, where $P_{0,W|Z}$ is the conditional distribution of $W$ given $Z$, and $P_{0,Z}$ is the marginal distribution of $Z$; 2. $P_{0,W|Z}$ has a density $p_{0,W|Z}$.
\end{assumption}
\begin{assumption}\label{as:correctspec} 1. $\Gamma$ is a compact subset of $\mathbb{R}^{d_{\gamma}}$, $d_{\gamma} < \infty$; 2. there is a unique $\gamma_{0} \in \text{int}(\Gamma)$ such that $m_{0}(z,\gamma_{0})= 0$ for every $z \in \mathcal{Z}$, where $m_{0}(z,\gamma)$ denotes $m(z,\gamma)$ at $P_{0,W|Z}$.
\end{assumption}
\begin{assumption}\label{as:moments}
1. $\sup_{(z,\gamma) \in \mathcal{Z}\times \Gamma}||m_{0}(z,\gamma)||_{2} < \infty$; 2. $m_{0}(z,\gamma)$ is twice continuously differentiable in $\gamma$ in a neighborhood $\Gamma_{0}$ of $\gamma_{0}$ for each $z \in \mathcal{Z}$, $\sup_{(z,\gamma) \in \mathcal{Z}\times \Gamma_{0}}||M_{0}(z,\gamma)||_{op}< \infty $, and $\sup_{(z,\gamma) \in \mathcal{Z}\times \Gamma_{0}}||\partial vec(M_{0}(z,\gamma))/\partial \gamma'||_{op} < \infty$, where $M_{0}(z,\gamma)$ is the Jacobian of $m_{0}(z,\gamma)$ in $\gamma$.
\end{assumption}
\begin{assumption}\label{as:covariance}
Let $\Sigma_{0}(z,\gamma)$ denote $\Sigma(z,\gamma)$ at $P_{0,W|Z}$. 1. there exists $0 < \underline{\lambda} < \overline{\lambda}< \infty$ such that $\underline{\lambda} \leq \inf_{(z,\gamma) \in \mathcal{Z}\times \Gamma}\lambda_{min}(\Sigma_{0}(z,\gamma))\leq \sup_{(z,\gamma) \in \mathcal{Z} \times \Gamma} \lambda_{max}(\Sigma_{0}(z,\gamma)) \leq \overline{\lambda}$; 2. $\Sigma_{0}(z,\gamma)$ is continuously differentiable in $\gamma$ in a neighborhood $\Gamma_{0}$ of $\gamma_{0}$ for each $z \in \mathcal{Z}$ with Jacobian satisfying $\sup_{(z,\gamma) \in \mathcal{Z}\times \Gamma_{0}}||\partial vec(\Sigma_{0}(z,\gamma))/\partial \gamma'||_{op} < \infty$.
\end{assumption}
\begin{assumption}\label{as:criterion}
1. the class $\{m_{0}(\cdot,\gamma)'\Sigma_{0}^{-1}(\cdot,\tilde{\gamma})m_{0}(\cdot,\gamma):(\gamma,\tilde{\gamma}) \in \Gamma \times \Gamma\}$ is $P_{0,Z}$-Glivenko-Cantelli; 2. for every $\xi>0$, $\inf_{\gamma \in \mathbb{R}^{d_{\gamma}}:||\gamma-\gamma_{0}||_{2} \geq \xi} E_{P_{0,Z}}[||m_{0}(Z,\gamma)||_{2}^{2}]> 0$.
\end{assumption}
\begin{assumption}\label{as:localident}
Let $V_{0}(Z,\tilde{\gamma}) = M_{0}(Z,\gamma_{0})'\Sigma_{0}^{-1}(Z,\tilde{\gamma})M_{0}(Z,\gamma_{0})$. 1. $\{V_{0}(\cdot,\tilde{\gamma}): \tilde{\gamma} \in \Gamma\}$ is $P_{0,Z}$-Glivenko-Cantelli; 2. $E_{P_{0,Z}}[M_{0}(Z,\gamma_{0})'M_{0}(Z,\gamma_{0})]$ is positive definite.
\end{assumption}
Assumptions \ref{as:dgp2}--\ref{as:localident} are similar restrictions to those encountered in the optimal instruments literature. Assumption \ref{as:dgp2} imposes the same sampling restriction as \cite{chamberlain1987asymptotic}, and, relative to Section \ref{sec:overview}, is a restriction on $\{Z_{i}\}_{i \geq 1}$.
Assumption \ref{as:correctspec} imposes correct specification, requires that $\gamma_{0}$ be point identified, and that $\gamma_{0}$ belongs to the interior of a compact $\Gamma \subseteq \mathbb{R}^{d_{\gamma}}$. Misspecification is discussed in Section \ref{sec:conclusion}. Assumption \ref{as:moments} imposes some smoothness restrictions on $\gamma \mapsto m_{0}(z,\gamma)$. It accommodates the differentiability conditions on $g(W,Z,\gamma)$ that are routinely imposed in the optimal instruments literature (e.g., \cite{chamberlain1987asymptotic}, \cite{newey1990efficient,NEWEY1993419}, \cite{donald2003empirical}, and \cite{chib2022bayesian}), while allowing settings where $g(W,Z,\gamma)$ is nondifferentiable but $m_{0}(Z,\gamma)$ is differentiable (e.g., quantile IV \citep{chernozhukov2005iv,chernozhukov2006instrumental}). Assumption \ref{as:covariance} imposes a compactness restriction on $\{\Sigma_{0}(z,\gamma): (z,\gamma) \in \mathcal{Z} \times \Gamma\}$ and a mild smoothness requirement for $\gamma \mapsto \Sigma_{0}(z,\gamma)$. Assumption \ref{as:criterion}.1 requires $Q_{n}(\gamma,\tilde{\gamma},p_{0,W|Z})$ obeys a USLLN over $\Gamma \times \Gamma$, while Assumption \ref{as:criterion}.2, combined with Assumption \ref{as:covariance}.1, implies that the criterion $Q_{P_{0,Z}}(\gamma,\tilde{\gamma},p_{0,W|Z}) := E_{P_{0,Z}}[m_{0}(Z,\gamma)'\Sigma_{0}^{-1}(Z,\tilde{\gamma})m_{0}(Z,\gamma)]$ has a uniformly well-separated minimum at $\gamma_{0}$, an important condition for consistent estimation of $\gamma_{0}$. Assumption \ref{as:localident}.1 is a USLLN condition for the class $\{V_{0}(\cdot,\tilde{\gamma}): \tilde{\gamma} \in \Gamma\}$, and Assumption \ref{as:localident}.2 is a strong identification condition that, combined with Assumption \ref{as:covariance}.1, implies $V_{0} :=E_{P_{0,Z}}[V_{0}(Z,\gamma_{0})]$ is positive definite, so the \cite{chamberlain1987asymptotic} efficiency bound, $V_{0}^{-1}$, is finite.
\subsection{Assumptions about the Prior}
The next assumption concerns $\Pi(p_{W|Z} \in \cdot |W^{(n)})$. For notation, let $||f||_{n,2}^{2} = n^{-1}\sum_{i=1}^{n}||f(z_{i})||^{2}$ and $||f||_{n,\infty} = \max_{1 \leq i \leq n }||f(z_{i})||$, where $||\cdot||$ is $||\cdot||_{2}$ if $f(z)$ is a vector and is $||\cdot||_{op}$ if $f(z)$ is a matrix, let $$\tilde{\chi}_{0}(W,Z) = -V_{0}^{-1}M_{0}(Z,\gamma_{0})'\Sigma_{0}^{-1}(Z,\gamma_{0})g(W,Z,\gamma_{0}) $$ be the \cite{chamberlain1987asymptotic} efficient influence function (EIF) and note that $V_{0}^{-1}$ also satisfies $V_{0}^{-1} = E_{P_{0,WZ}}[\tilde{\chi}_{0}(W,Z)\tilde{\chi}_{0}(W,Z)']$, let $\overset{P_{0,W|Z}^{(n)}}{\longrightarrow}$ denote convergence in probability under $P_{0,W|Z}^{(n)}=\bigotimes_{i=1}^{n}P_{0,W|z_{i}}$, let $P_{0,Z}^{\infty}$ be the probability law of the i.i.d sequence $\{Z_{i}\}_{i \geq 1}$ (i.e., $P_{0,Z}^{\infty} = \bigotimes_{i \geq 1}P_{0,Z}$), and sequences $\{x_{n}\}_{n \geq 1}$ and $\{y_{n}\}_{n \geq 1}$ satisfy $x_{n} = o(y_{n})$ if $x_{n}/y_{n} \rightarrow 0$ as $n\rightarrow \infty$ and $x_{n} = O(y_{n})$ if there is a constant $C > 0$ such that $|x_{n}/y_{n}| \leq C$ for large $n$.
\begin{assumption}\label{as:concentration}
For $P_{0,Z}^{\infty}$-almost every fixed realization $\{z_{i}\}_{i \geq 1}$ of $\{Z_{i}\}_{i \geq 1}$, there exists a sequence of sets $\{\tilde{\mathcal{P}}_{n,W|Z}\}_{n\geq 1}$ such that
\begin{enumerate}
    \item $\Pi(p_{W|Z} \in \tilde{\mathcal{P}}_{n,W|Z} |W^{(n)}) \longrightarrow 1$ in $P_{0,W|Z}^{(n)}$-probability as $n\rightarrow \infty$
    \item For each $p_{W|Z} \in \tilde{\mathcal{P}}_{n,W|Z}$,
    \begin{enumerate}
     \item $m(z_{i},\gamma)$, $i=1,...,n$, is twice continuously differentiable in $\gamma$ over $\Gamma_{0}$
     \item $\Sigma(z_{i},\gamma)$, $i=1,...,n$, is once continuously differentiable in $\gamma$ over $\Gamma_{0}$,
    \end{enumerate}
    where $\Gamma_{0}$ is the neighborhood of $\gamma_{0}$ from Assumption \ref{as:moments}.
    \item The following holds:
    \begin{enumerate}
\item $\sup_{(\gamma,p_{W|Z}) \in \Gamma \times \tilde{\mathcal{P}}_{n,W|Z}}||m(\cdot,\gamma)-m_{0}(\cdot,\gamma)||_{n,2} = o(n^{-\frac{1}{4}})$ \item $\sup_{(\gamma,p_{W|Z}) \in \Gamma_{0} \times \tilde{\mathcal{P}}_{n,W|Z}}||M(\cdot,\gamma)-M_{0}(\cdot,\gamma)||_{n,2} = o(n^{-\frac{1}{4}})$ \item $\sup_{(\gamma,p_{W|Z}) \in \Gamma \times \tilde{\mathcal{P}}_{n,W|Z}}||\Sigma(\cdot,\gamma)-\Sigma_{0}(\cdot,\gamma)||_{n,\infty} = o(n^{-\frac{1}{4}})$ \item $\sup_{(\gamma,p_{W|Z}) \in \Gamma_{0} \times \tilde{\mathcal{P}}_{n,W|Z}}||\partial \text{vec}(M(\cdot,\gamma)')/\partial \gamma'||_{n,2} = O(1)$ \item $\sup_{(\gamma,p_{W|Z}) \in \Gamma_{0} \times \tilde{\mathcal{P}}_{n,W|Z}}||\partial\text{vec}(\Sigma(\cdot,\gamma))/\partial \gamma'||_{n,\infty}= O(1)$ \item $\sup_{p_{W|Z} \in  \tilde{\mathcal{P}}_{n,W|Z}}\max_{1 \leq i \leq n}E_{P_{W|Z}}[\exp(|t'\tilde{\chi}_{0}(W,Z)|)|Z=z_{i}] = O(1)$
\end{enumerate}
for any $t$ in a neighborhood of zero.\footnote{Use of $o(\cdot)$ and $O(\cdot)$ is appropriate because, conditional on $\{Z_{i}\}_{i \geq 1}$, these objects are nonrandom.}
\end{enumerate}
\end{assumption}
Assumption \ref{as:concentration} states that, conditional on $\{Z_{i}\}_{i \geq 1}$, there are sets $\{\tilde{\mathcal{P}}_{n,W|Z}\}_{n \geq 1}$ on which the posterior concentrates as $n\rightarrow \infty$, and these sets are structured enough to ensure certain smoothness (in $\gamma$) and convergence guarantees for $m$, $M$, and $\Sigma$. The assumption is general in that it is stated for an arbitrary prior $\Pi$ for $p_{W|Z}$; I derive examples of $\{\tilde{\mathcal{P}}_{n,W|Z}\}_{n \geq 1}$ for certain GPs in Section \ref{sec:example}. As for the contents, Parts 3(a) and 3(b) state that $m(\cdot,\gamma)$ and $M(\cdot,\gamma)$ converge to their true counterparts $m_{0}(\cdot,\gamma)$ and $M_{0}(\cdot,\gamma)$, respectively, at a rate faster than $n^{-1/4}$ under the empirical $L^{2}$ norm. Part 3(c) states that $\Sigma(\cdot,\gamma)$ converges to $\Sigma_{0}(\cdot,\gamma)$ at a rate faster than $n^{-1/4}$ under the empirical supremum norm, a condition analogous to Assumption 3.4(iii) in \cite{ai2003efficient}. Requiring $o(n^{-1/4})$ nuisance convergence rates is generally considered a mild condition in semiparametric applications. Part 3(d) and 3(e) are boundedness conditions that help show that $\tilde{\gamma} \mapsto q_{n}(\tilde{\gamma},p_{W|Z})$ is a uniform contraction mapping in large samples, a property that helps establish that $\gamma_{n}$ uniformly converges to $\gamma_{0}$ along $\{\tilde{\mathcal{P}}_{n,W|Z}\}_{n \geq 1}$. Part 3(f) requires that the EIF has sufficiently thin tails over $\tilde{\mathcal{P}}_{n,W|Z}$, a condition sufficient to control a remainder in an expansion of $\log L_{n}(p_{W|Z})$.
\subsection{Bernstein-von Mises Theorem}
\subsubsection{Main Result}
Assumptions \ref{as:measurability}--\ref{as:concentration} lead to a Bernstein-von Mises theorem (BvM). Informally, the BvM states that $\gamma_{n}|W^{(n)}\overset{a}{\sim} \mathcal{N}(\hat{\gamma}_{n},\frac{1}{n}V_{0}^{-1})$ for large $n$, where $\hat{\gamma}_{n} = \gamma_{0}+n^{-1}\sum_{i=1}^{n}\tilde{\chi}_{0}(W_{i},z_{i})$ is a hypothetical efficient estimator of $\gamma_{0}$. Consequently, the posterior for $\gamma_{n}$ behaves like a best regular estimator $\hat{\gamma}_{n}$ for large $n$, thereby establishing its frequentist asymptotic optimality.  

The key to the BvM is an asymptotically linear representation of $\gamma_{n}$. Stated as a theorem below, it shows that $\gamma_{n}$ is asymptotically linear in $p_{W|Z}$ with Riesz representer given by $\tilde{\chi}_{0}$. It can be thought of as a Bayesian analogue to the asymptotically linear representation of an efficient estimator, and validates that $\Pi(\gamma_{n} \in \cdot | W^{(n)})$ accounts for the optimal instruments as $n\rightarrow \infty$.
\begin{theorem}\label{thm:representation}
If Assumptions \ref{as:uemfp}, \ref{as:dgp2}, \ref{as:correctspec}, \ref{as:moments}, \ref{as:covariance}, \ref{as:criterion}, \ref{as:localident}, and \ref{as:concentration}.1--\ref{as:concentration}.5 hold, then, for $P_{0,Z}^{\infty}$-almost every fixed realization $\{z_{i}\}_{i \geq 1}$ of $\{Z_{i}\}_{i \geq 1}$,
\begin{align*}
    \sup_{p_{W|Z} \in \tilde{\mathcal{P}}_{n,W|Z}}\left | \left | \sqrt{n}(\gamma_{n}-\gamma_{0}) - \frac{1}{\sqrt{n}}\sum_{i=1}^{n}E_{P_{W|Z}}[\tilde{\chi}_{0}(W,Z)|Z=z_{i}] \right | \right |_{2} =o(1),
\end{align*}
as $n\rightarrow \infty$, where $\{\tilde{\mathcal{P}}_{n,W|Z}\}_{n \geq 1}$ are the sets defined in Assumption \ref{as:concentration}.
\end{theorem}
Theorem \ref{thm:representation} establishes that, asymptotically, $\gamma_{n}$ is an (efficient) linear functional of $p_{W|Z}$. Consequently, a semiparametric BvM for $\gamma_{n}$ can be established by building on \cite{castillo2015bernstein}, who prove BvMs for \textit{unconditional} density functions with i.i.d data.

\begin{theorem}\label{thm:bvm}
Suppose that Assumptions \ref{as:measurability}--\ref{as:concentration} hold and that, for $P_{0,Z}^{\infty}$-almost every fixed realization $\{z_{i}\}_{i \geq 1}$ of $\{Z_{i}\}_{i \geq 1}$, the prior invariance condition
\begin{align}\label{eq:invariancegeneral}
\frac{\int_{\tilde{\mathcal{P}}_{n,W|Z}}L_{n}(p_{t,n,W|Z})d\Pi(p_{W|Z})}{\int_{\tilde{\mathcal{P}}_{n,W|Z}}L_{n}(p_{W|Z}) d\Pi(p_{W|Z})} = 1 + o_{P_{0,W|Z}^{(n)}}(1) 
\end{align}
holds for any $t$ in a neighborhood of zero, where 
\begin{align*}
p_{t,n,W|Z}(w|z)= \frac{p_{W|Z}(w|z)\exp\left(-\frac{1}{\sqrt{n}}t'\tilde{\chi}_{0}(w,z)\right)}{\int_{\mathcal{W}}p_{W|Z}(\tilde{w}|z)\exp\left(-\frac{1}{\sqrt{n}}t'\tilde{\chi}_{0}(\tilde{w},z)\right)d\tilde{w}}.    
\end{align*} Then, for almost every fixed realization $\{z_{i}\}_{i \geq 1}$ of $\{Z_{i}\}_{i \geq 1}$, the structural posterior satisfies
\begin{align*}
     d_{BL}\left(\Pi\left(\sqrt{n}(\gamma_{n}-\hat{\gamma}_{n}) \in \cdot\middle |W^{(n)}\right), \mathcal{N}(0,V_{0}^{-1})\right) \overset{P_{0,W|Z}^{(n)}}{\longrightarrow} 0
\end{align*}
as $n\rightarrow \infty$, where $d_{BL}$ is the bounded Lipschitz metric.\footnote{For any two probability measures $P$ and $Q$ on a measurable space $(\mathcal{X},\mathscr{X})$, the bounded Lipschitz distance is given by $d_{BL}(P,Q) = \sup_{f \in Lip_{1}}|\int f d P - \int f dQ|$, where $Lip_{1}$ is the set of functions $f: \mathcal{X} \rightarrow [-1,1]$ with $\sup_{x \in \mathcal{X}}|f(x)| \leq 1$ and $|f(x)-f(y)| \leq d(x,y)$ for all $x,y \in \mathcal{X}$ with $x \neq y$, where $d$ is a metric over $\mathcal{X}$. Convergence in $d_{BL}$ is equivalent to convergence in distribution.}
\end{theorem}
\subsubsection{Comments about the BvM}
Several aspects of Theorem \ref{thm:bvm} should be highlighted. First, Theorem \ref{thm:bvm} implies that some Bayesian point estimators and credible sets are asymptotically efficient. Let $f: \mathbb{R}^{d_{\gamma}} \rightarrow \mathbb{R}$ be continuously differentiable at $\gamma_{0}$ and recall that $c_{n,f}(q)$, $q \in (0,1)$, is the $q$-quantile of $\Pi(f(\gamma_{n}) \in \cdot|W^{(n)})$. Corollary \ref{cor:quantile} states that, conditional on $\{Z_{i}\}_{i \geq 1}$, $c_{n,f}(q)=f(\hat{\gamma}_{n}) + \Phi^{-1}(q)n^{-1/2}\Omega_{0,f}^{1/2}+o_{P_{0,W|Z}^{(n)}}(n^{-1/2})$, where $\Phi(\cdot)$ is the $\mathcal{N}(0,1)$ cumulative distribution function and $\Omega_{0,f} =\frac{\partial f(\gamma_{0})}{\partial \gamma'}V_{0}^{-1}\frac{\partial f(\gamma_{0})'}{\partial \gamma}$ is the asymptotic variance of $f(\hat{\gamma}_{n})$. Consequently, by setting $q= 0.5$, the posterior median $c_{n,f}(0.5)$ is first-order asymptotically equivalent to the efficient estimator $f(\hat{\gamma}_{n})$.\footnote{The restrictions on $f$ imply the delta method preserves efficiency (see Section 25.7 of \cite{vaart_1998}).} Moreover, the equitailed probability interval $CS_{n,f}(1-\alpha)= [c_{n,f}(\alpha/2), c_{n,f}(1-\alpha/2)]$, where $\alpha \in (0,1/2)$, is first-order asymptotically equivalent to an efficient Wald confidence interval $[f(\hat{\gamma}_{n})+\Phi^{-1}(\alpha/2)\sqrt{\Omega_{0,f}/n},f(\hat{\gamma}_{n})-\Phi^{-1}(\alpha/2)\sqrt{\Omega_{0,f}/n}]$, thereby leading to best asymptotic frequentist uncertainty quantification.
\begin{corollary}\label{cor:quantile}
If the assumptions of Theorem \ref{thm:bvm} hold and $f:\mathbb{R}^{d_{\gamma}}\rightarrow \mathbb{R}$ is continuously differentiable at $\gamma_{0}$, then, for any $q \in (0,1)$ and $P_{0,Z}^{\infty}$-almost every $\{z_{i}\}_{i \geq 1}$, $c_{n,f}(q)$ satisfies
\begin{align*}
    c_{n,f}(q) = f(\hat{\gamma}_{n}) + \Phi^{-1}(q)\sqrt{\frac{\Omega_{0,f}}{n}} + o_{P_{0,W|Z}^{(n)}}\left(\frac{1}{\sqrt{n}}\right), \quad \Omega_{0,f} = \frac{\partial f(\gamma_{0})}{\partial \gamma'}V_{0}^{-1}\frac{\partial f(\gamma_{0})'}{\partial \gamma}.
\end{align*}
\end{corollary}
A second aspect is that Theorem \ref{thm:bvm} also holds unconditionally, where stochastic convergence is defined with respect to $P_{0,WZ}$. Corollary \ref{cor:unconditional} formalizes this, and, by extension, implies an unconditional version of Corollary \ref{cor:quantile} (see Remark \ref{rem:uncon}).
\begin{corollary}\label{cor:unconditional}
If the assumptions of Theorem \ref{thm:bvm} holds, then
\begin{align*}
    d_{BL}\left(\Pi\left(\sqrt{n}(\gamma_{n}-\hat{\gamma}_{n}) \in \cdot \middle |W^{(n)}\right), \mathcal{N}\left(0,V_{0}^{-1}\right)\right)\overset{P_{0,WZ}^{n}}{\longrightarrow} 0   
\end{align*}
as $n\rightarrow \infty$, where $P_{0,WZ}^{n}$ is $n$-fold product of $P_{0,WZ}$.
\end{corollary}
A final aspect worth discussing is condition (\ref{eq:invariancegeneral}). These conditions are common in semiparametric BvMs (see \cite{castillo2012semiparametric,castillo2012gaussian}, \cite{rivoirard2012bernstein}, \cite{castillo2015bernstein}, \cite{ray2020semiparametric}, and \cite{monard2021bernstein}, to name a few), and arise from expanding the log-likelihood ($\log L_{n}(p_{W|Z})$) around the least favorable parametric submodel ($p_{t,n,W|Z}$). The name `prior invariance' reflects that its verification typically requires the prior be roughly unchanged under shifts in the direction of the EIF. Indeed, Section \ref{sec:example} verifies (\ref{eq:invariancegeneral}) for GPs by finding a sequence $\{\tilde{g}_{n}\}_{n \geq 1}$ that uniformly approximates the \cite{chamberlain1987asymptotic} optimal residual $\tilde{g}_{0} = -V_{0}\tilde{\chi}_{0}$, with $o(n^{-1/4})$ approximation error, and for which the law $\Pi_{n,t}$ of $p_{n,t,W|Z}\propto p_{W|Z}\exp(-t'V_{0}^{-1}\tilde{g}_{n}/\sqrt{n})$ under $\Pi$ satisfies $\Pi_{n,t} \ll \Pi$ and $d \Pi_{n,t}/d\Pi \rightarrow 1$ uniformly along $\{\tilde{\mathcal{P}}_{n,W|Z}\}_{n \geq 1}$ as $n\rightarrow \infty$. For the Mat\'{e}rn process (i.e., Example \ref{ex:matern}), this restricts the GP sample path smoothness relative to $\tilde{g}_{0}$.
\section{Verifying Assumption \ref{as:concentration} and Prior Invariance with GPs}\label{sec:example}
I provide sufficient conditions for Assumption \ref{as:concentration} and (\ref{eq:invariancegeneral}) for the priors from Section \ref{sec:implementation}.
\subsection{Setup and Assumptions}
Recall $p_{\theta,B}(w|z) \propto f_{\theta}(w|z)\exp(B(F_{\theta,z}(w),F_{0}(z)))$, where $f_{\theta}(\cdot|z)$, $\theta \in \Theta \subseteq \mathbb{R}^{d_{\theta}}$ with $d_{\theta} < \infty$, is a conditional density, $F_{\theta,z}(\cdot)$ is the $d_{w}\times 1$ vector of conditional CDFs attached to $f_{\theta}(\cdot|z)$, $F_{0}: \mathcal{Z} \rightarrow [0,1]^{d_{z}}$ is a known invertible transformation, and $B: T \rightarrow \mathbb{R}$ is a function. A prior for $p_{\theta,B}$ is obtained from a prior $\Pi$ for $(\theta,B)$. Also, recall that $(C^{\alpha}([0,1]^{d},\mathbb{R}),||\cdot||_{\alpha})$ and $(S^{\alpha}([0,1]^{d},\mathbb{R}),||\cdot||_{2,2,\alpha})$ denote H\"{o}lder and Sobolev spaces, respectively, of order $\alpha >0$ (see Footnote \ref{fn:functionspaces} for precise details).
\begin{assumption}\label{as:boundedinternal}
1. $\mathcal{W}\subseteq \mathbb{R}^{d_{w}}$ is compact; 2. $\mathcal{Z} = [0,1]^{d_{z}}$; 3. $p_{0,W|Z}$ is continuous and positive on $\mathcal{W} \times \mathcal{Z}$; 4. $P_{0,Z}$ admits a Lebesgue density $p_{0,Z}$ that is bounded away from zero and infinity.
\end{assumption}
\begin{assumption}\label{as:differentiablemoments}
1. $(w,z,\gamma) \mapsto g_{j}(w,z,\gamma)$ is continuous in all arguments over $\mathcal{W}\times \mathcal{Z} \times \Gamma$ for $j=1,...,d_{g}$; 2. $g_{j}(w,z,\gamma)$, $j=1,...,d_{g}$, is differentiable in $\gamma$ over a neighborhood $\Gamma_{0}$ of $\gamma_{0}$ with first and second derivatives continuous in all arguments over $\mathcal{W} \times \mathcal{Z} \times \Gamma_{0}$.
\end{assumption}
\begin{assumption}\label{as:parametricfamily}
1. $\Theta \subseteq \mathbb{R}^{d_{\theta}}$ is compact; 2. $(w,z,\theta) \mapsto f_{\theta,j}(w_{j}|w_{1},...,w_{j-1},z)$ is continuous in all arguments for all $j=1,...,d_{w}$; 3. $f_{\theta,j}(w_{j}|w_{1},...,w_{j-1},z) > 0$ for all $(w,z,\theta) \in \mathcal{W} \times \mathcal{Z} \times \Theta$ and all $j=1,...,d_{w}$; 4. there is a constant $L>0$ such that $\sup_{(w,z) \in \mathcal{W} \times \mathcal{Z}}|\log f_{\theta_{1}}(w|z) - \log f_{\theta_{2}}(w|z)| \leq L ||\theta_{1}-\theta_{2}||_{2}$ for all $\theta_{1},\theta_{2} \in \Theta$.\footnote{Recall that $f_{\theta,j}(w_{j}|w_{1},...,w_{j-1},z)$ is the density of $W_{j}$ given $\{W_{j'}\}_{j'=1}^{j-1}$ and $Z$ under $f_{\theta}(w|z)$.}
\end{assumption}
\begin{assumption}\label{as:exprior}
$\Pi = \Pi_{\Theta} \otimes \Pi_{\mathcal{B}}$, where 1. $\Pi_{\Theta}$ is a probability distribution over $\Theta$ with a continuous Lebesgue density $\pi_{\Theta}$, and 2. $\Pi_{\mathcal{B}}$ is the law $\tilde{\Pi}_{\mathcal{B}}$ of a centered Gaussian process $B$ in $(C([0,1]^{d},\mathbb{R}),||\cdot||_{\infty})$ conditioned on $\{||B||_{\infty}\leq \bar{B}\}$ for some large $\bar{B} \in (0,\infty)$ with $B \sim \tilde{\Pi}_{\mathcal{B}}$ taking values in $(C^{a}([0,1]^{d},\mathbb{R}),||\cdot||_{a})$ for every $0<a<\alpha$ and $\alpha > d_{z}/2$.
\end{assumption}
Assumptions \ref{as:boundedinternal}--\ref{as:exprior} comprise \textit{sufficient} conditions on the DGP and prior that will be used to verify Assumption \ref{as:concentration}. The main limitation of Assumption \ref{as:boundedinternal} is compactness of $\mathcal{W}$, however it is not overly restrictive relative to the literature on semiparametric BvMs for density functionals \citep{rivoirard2012bernstein,castillo2015bernstein}, and, more broadly, posterior consistency with infinite-dimensional exponential families \citep{scricciolo2006convergence,rivoirard2012posterior}. I expect compact $\mathcal{W}$ can be relaxed (see Footnote \ref{fn:boundedness}). Compactness of $\mathcal{Z}$ is a standard assumption in nonparametric conditional density estimation, and, given such an assumption, setting $\mathcal{Z} = [0,1]^{d_{z}}$ is without loss of generality (and enables setting $F_{0}(z)=z$). Assumption \ref{as:differentiablemoments} is also stronger than the high-level smoothness conditions imposed in Section \ref{sec:BVM}, as it rules out models with nondifferentiable $g(W,Z,\gamma)$. Assumption \ref{as:parametricfamily} imposes regularity conditions on $\{f_{\theta}: \theta \in \Theta\}$ that are satisfied, for example, when $f_{\theta}$ is a Gaussian linear regression (truncated to $\mathcal{W}$), $\mathcal{Z}$ is compact (i.e. Assumption \ref{as:boundedinternal}.2), and $\Theta$ constrains the slope and variance to compact sets, with the variance bounded away from zero. Assumption \ref{as:exprior} is a sample path smoothness restriction that many Gaussian processes satisfy (e.g., the Mat\'{e}rn process from Example \ref{ex:matern}). Uniformly bounding $B$ is a technical device used to control random approximation errors, and, since the magnitude of $\bar{B}$ has no role, it can be viewed as an arbitrarily large but finite number.\footnote{Lemma 5.1 in \cite{van2008reproducing} shows $\tilde{\Pi}_{\mathcal{B}}(||B||_{\infty} \leq \bar{B}) > 0$ for any $\bar{B} > 0$.} 
\subsection{Verifying Assumption \ref{as:concentration}}
Under Assumptions \ref{as:boundedinternal} and \ref{as:differentiablemoments}, only Assumptions \ref{as:concentration}.3(a)--10.3(c) need verification. Let $h(p_{\theta,B}(\cdot|z),p_{0,W|Z}(\cdot|z))$ be the Hellinger distance between $p_{\theta,B}$ and $p_{0,W|Z}$ at $z$, and let $h_{n,2}(p_{\theta,B},p_{0,W|Z}) = (n^{-1}\sum_{i=1}^{n}h^{2}(p_{\theta,B}(\cdot|z_{i}),p_{0,W|Z}(\cdot|z_{i})))^{1/2}$ be the root mean square Hellinger (RMSH) distance.\footnote{$h(p_{\theta,B}(\cdot|z),p_{0,W|Z}(\cdot|z))$ satisfies $h^{2}(p_{\theta,B}(\cdot|z),p_{0,W|Z}(\cdot|z)) = \int_{\mathcal{W}} (p_{\theta,B}^{1/2}(w|z)-p_{0,W|Z}^{1/2}(w|z))^{2}dw$.} Under Assumptions \ref{as:correctspec}.1, \ref{as:boundedinternal}, \ref{as:differentiablemoments}, the uniform Lipschitz conditions hold: $\sup_{\gamma \in \Gamma}||m(\cdot,\gamma)-m_{0}(\cdot,\gamma)||_{n,2} \leq C_{1} h_{n,2}(p_{\theta,B},p_{0,W|Z})$ and $\sup_{\gamma \in \Gamma_{0}}||M(\cdot,\gamma)-M_{0}(\cdot,\gamma)||_{n,2} \leq C_{2} h_{n,2}(p_{\theta,B},p_{0,W|Z})$ for constants $C_{1},C_{2}>0$. Consequently, a RMSH convergence rate is an upper bound on the rate of convergence for $m(\cdot,\gamma)$ and $M(\cdot,\gamma)$ to $m_{0}(\cdot,\gamma)$ and $M_{0}(\cdot,\gamma)$, respectively, in the empirical $L^{2}$ norm.\footnote{This Lipschitz property (and the resulting convergence implications) seems to be where compactness of $\mathcal{W}$ is most crucial. Subject to modified restrictions on $\{f_{\theta}: \theta \in \Theta\}$, I conjecture that compactness could be relaxed the following ways: (i) boundedness of $g_{j}(W,Z,\gamma)$ and its derivatives in $\gamma$, (ii) uniform integrability of $\{g_{j}(W,z,\gamma): (z,\gamma) \in \mathcal{Z} \times \Gamma\}$ (and similarly the derivatives in $\gamma$, except replacing $\Gamma$ with $\Gamma_{0}$) with respect to $p_{0,W|Z}$ and $\{p_{\theta,B}: \theta \in \Theta, \ ||B||_{\infty} \leq \bar{B}\}$ (or a subset on which the posterior concentrates as $n\rightarrow \infty$), or (iii) convergence around $p_{0,W|Z}$ under a stronger norm (e.g., $L^{2}$). The intuition is that (i) and (ii) preserves RMSH convergence implying $||\cdot||_{n,2}$-convergence of conditional expectations, while (iii) could relax boundedness via, in the case of $L^{2}$, the Cauchy-Schwarz inequality. \label{fn:boundedness}}  This means that Assumptions \ref{as:concentration}.3(a) and \ref{as:concentration}.3(b) are verified if the RMSH rate is $o(n^{-1/4})$ (by restricting $\tilde{\mathcal{P}}_{n,W|Z}$ to be contained in an appropriately shrinking RMSH ball around $p_{0,W|Z}$).

Let $\mathcal{H}$ be the Reproducing Kernel Hilbert Space (RKHS) of the unrestricted centered GP (i.e., corresponding to $\tilde{\Pi}_{\mathcal{B}}$), and let $||\cdot||_{\mathcal{H}}$ denote the RKHS norm (i.e., $||\cdot||_{\mathcal{H}}^{2} = \langle \cdot,\cdot \rangle_{\mathcal{H}}$, where $\langle \cdot,\cdot \rangle_{\mathcal{H}}$ is the RKHS inner product).\footnote{See \cite{van2008reproducing} for a formal definition of the RKHS of a centered GP.} The RMSH rate is determined by the RKHS. To formalize this, recall that $b_{0,\theta}(u,z) = \log p_{0,W|Z}(F_{\theta,z}^{-1}(u)|z)-\log f_{\theta}(F_{\theta,z}^{-1}(u)|z)$ for each $(u',z')' \in [0,1]^{d}$, and define the \textit{supremum norm concentration function} at $b_{0,\theta}$ as 
\begin{align*}
\varphi_{b_{0,\theta}}(\delta) = \inf_{h \in \mathcal{H}: ||h-b_{0,\theta}||_{\infty}< \delta}\frac{1}{2}||h||_{\mathcal{H}}^{2}-\log \tilde{\Pi}_{\mathcal{B}}(||B||_{\infty}<\delta)    
\end{align*} 
for $\delta > 0$. The concentration function uses the RKHS to measure the prior mass around $b_{0,\theta}$ (see Lemma 5.3 in \cite{van2008reproducing}). Building on \cite{ghosal2007convergence} and \cite{vaart2008rates}, Proposition \ref{prop:RMSHrate} states that the RMSH rate $\delta_{n}$ satisfies $\sup_{\theta \in \Theta}\varphi_{b_{0,\theta}}(\delta_{n}) \leq n \delta_{n}^{2}$. For notation, let $\overline{\mathcal{H}}$ denote the $||\cdot||_{\infty}$-closure of $\mathcal{H}$.
\begin{proposition}\label{prop:RMSHrate}
Suppose Assumptions \ref{as:dgp2}, \ref{as:boundedinternal}, \ref{as:parametricfamily}, and \ref{as:exprior} holds. If 1. $\sup_{\theta \in \Theta}||b_{0,\theta}||_{\infty} \leq \bar{B}$, 2. $\{b_{0,\theta}: \theta \in \Theta\} \subseteq \overline{\mathcal{H}}$, and 3. $\sup_{\theta \in \Theta}\varphi_{b_{0,\theta}}(\delta_{n}) \leq n \delta_{n}^{2}$ for some $\{\delta_{n}\}_{n \geq 1}$ satisfying $\delta_{n} \rightarrow 0$ and $n\delta_{n}^{2} \rightarrow \infty$ as $n\rightarrow \infty$, then, for $P_{0,Z}^{\infty}$-almost every fixed realization $\{z_{i}\}_{i \geq 1}$ of $\{Z_{i}\}_{i \geq 1}$, $\Pi(h_{n,2}(p_{\theta,B},p_{0,W|Z}) < D \delta_{n}| W^{(n)}) \rightarrow 1$ in $P_{0,W|Z}^{(n)}$-probability as $n\rightarrow \infty$ for some $D>0$.
\end{proposition}
Proposition \ref{prop:RMSHrate} is not sufficient for Assumption \ref{as:concentration}.3(c). Proposition \ref{prop:interpolation} below proves that, under smoothness conditions, an RMSH rate leads to a convergence rate in the empirical supremum Hellinger distance $h_{n,\infty}(p_{\theta,B},p_{0,W|Z})=\max_{1 \leq i \leq n} h(p_{\theta,B}(\cdot|z_{i}),p_{0,W|Z}(\cdot|z_{i}))$.\footnote{The smoothness conditions on $f_{\theta}$ hold, for example, if $f_{\theta}(w|z)$ is a homoskedastic Gaussian linear regression (truncated to $\mathcal{W}$), with bounded $z$, and with compactly supported slopes and variances (and variances bounded away from zero).} By Assumptions \ref{as:correctspec}, \ref{as:boundedinternal}, and \ref{as:differentiablemoments}, this leads to conditions for Assumption \ref{as:concentration}.3(c) by a similar Lipschitz property (i.e., $\sup_{\gamma \in \Gamma}||\Sigma(\cdot,\gamma)-\Sigma_{0}(\cdot,\gamma)||_{n,\infty} \leq C_{3} h_{n,\infty}(p_{\theta,B},p_{0,W|Z})$ for some $C_{3} > 0$).\footnote{A similar conjecture to Footnote \ref{fn:boundedness} applies.} It is important to note the $h_{n,\infty}$ rate is \textit{only} an upper bound and sharper convergence rates may be possible; I leave full treatment of optimal $h_{n,\infty}$ rates for future research because, to the best of my knowledge, posterior contraction rates in the supremum distance for conditional densities is an open question (and is of general statistical interest).
\begin{proposition}\label{prop:interpolation}
Suppose Assumptions \ref{as:dgp2}, \ref{as:boundedinternal}, \ref{as:parametricfamily}, and \ref{as:exprior} hold, the conditions of Proposition \ref{prop:RMSHrate} hold, $p_{0,W|Z}$ satisfies $\int_{\mathcal{W}} ||\log p_{0,W|Z}(w|\cdot)||_{\alpha_{0}}dw < \infty$ for some $\alpha_{0} > d_{z}/2$, and $\{f_{\theta}: \theta \in \Theta\}$ satisfies $\sup_{\theta \in \Theta} \int_{\mathcal{W}}|| \log f_{\theta}(w|\cdot) ||_{a}dw < \infty$ and $\sup_{\theta \in \Theta}\int_{\mathcal{W}}||F_{\theta,\cdot}(w)||_{a}^{\lfloor a \rfloor}dw < \infty$  for each $a \in (0,\alpha)$. Then, for $P_{0,Z}^{\infty}$-almost every fixed realization $\{z_{i}\}_{i \geq 1}$ of $\{Z_{i}\}_{i \geq 1}$, there exists $\{\tilde{\mathcal{P}}_{1,n,W|Z}\}_{n \geq 1}$ and $\{\tilde{\delta}_{n}\}_{n \geq 1}$ such that $\Pi(\tilde{\mathcal{P}}_{n,1,W|Z}|W^{(n)}) \rightarrow 1$ in $P_{0,W|Z}^{(n)}$-probability as $n\rightarrow \infty$ and $\sup_{p_{\theta,B} \in \tilde{\mathcal{P}}_{1,n,W|Z}}h_{n,\infty}(p_{\theta,B},p_{0,W|Z})\leq \tilde{D}\tilde{\delta}_{n}$ for a constant $\tilde{D}>0$, where $\tilde{\delta}_{n}$ depends on $\alpha,\alpha_{0}$, $d_{z}$, and the RMSH rate $\delta_{n}$.
\end{proposition}
\addtocounter{example}{-1}
\begin{example}[Continued]
Suppose $\{b_{0,\theta}: \theta \in \Theta\} \subseteq C^{\alpha_{0}}([0,1]^{d},\mathbb{R}) \cap S^{\alpha_{0}}([0,1]^{d},\mathbb{R})$ for some $\alpha_{0} > 0$, $\sup_{\theta \in \Theta}||b_{0,\theta}||_{\alpha_{0}} < \infty$, and $\sup_{\theta \in \Theta}||b_{0,\theta}||_{2,2,\alpha_{0}} < \infty$. Then, Lemma 3 of \cite{van2011information} and Lemma 13 in \cite{walker2026supplement} imply that the $h_{n,2}$ rate is $\delta_{n} = n^{-\min\{\alpha,\alpha_{0}\}/(2 \alpha + d)}$. Consequently, $\delta_{n} = o(n^{-1/4})$ iff $\alpha > d/2$ and $\alpha_{0} > \alpha/2 + d/4$. The second inequality restricts oversmoothing because, if $\alpha \leq \alpha_{0}$ (i.e., undersmoothing), the contraction rate is $\delta_{n} = n^{-\alpha/(2\alpha+d)}$ and $\delta_{n} = o(n^{-1/4})$ if $\alpha > d/2$. If $\alpha = \alpha_{0}$, then $\delta_{n}$ is the minimax rate $n^{-\alpha_{0}/(2\alpha_{0}+d)}$ and $\delta_{n} = o(n^{-1/4})$ if $\alpha_{0} > d/2$. Substituting $\delta_{n} = n^{-\min\{\alpha,\alpha_{0}\}/(2\alpha + d)}$ into the $h_{n,\infty}$ rate $\tilde{\delta}_{n}$, one can show $\tilde{\delta}_{n} = o(n^{-1/4})$ if $\alpha/2 + 3d_{z}/4<\alpha_{0}$ and $1/4 + (d_{z}+(\alpha-a))/(2\alpha+d_{z}) < \min\{\alpha,\alpha_{0}\}/(2\alpha + d)$.\footnote{See the discussion following Lemma 13 in \cite{walker2026supplement} for more details on the rate calculations.}
\end{example}
\subsection{Verifying Prior Invariance} The next proposition states that (\ref{eq:invariancegeneral}) holds if the \cite{chamberlain1987asymptotic} optimal residual, $\tilde{g}_{0}(w,z,\gamma_{0}) = M_{0}(z,\gamma_{0})'\Sigma_{0}^{-1}(z,\gamma_{0})g(w,z,\gamma_{0})$, is well-approximated by sequences in $\mathcal{H}$. Let $\tilde{g}_{0,\theta}(u,z,\gamma_{0}) = \tilde{g}_{0}(F_{\theta,z}^{-1}(u),z,\gamma_{0}),z,\gamma_{0})$ for $(u,z,\theta) \in [0,1]^{d_{w}} \times [0,1]^{d_{z}} \times \Theta$ be the optimal residual transformed so that is defined on $[0,1]^{d}$.
\begin{proposition}\label{prop:invariance}
Suppose Assumptions \ref{as:dgp2}, \ref{as:correctspec}, \ref{as:boundedinternal}, \ref{as:differentiablemoments}, \ref{as:parametricfamily}, \ref{as:exprior}, and the conditions of Proposition \ref{prop:RMSHrate} hold. Further, suppose there are sequences $\{\tilde{g}_{n,\theta,j}^{*}: \theta \in \Theta\}_{n \geq 1}$, $j=1,...,d_{\gamma},$ in $ \mathcal{H}$ and $\{\zeta_{n}\}_{n\geq 1}$ in $\mathbb{R}_{++}$ with $\zeta_{n}\rightarrow 0$ such that 1. $\sup_{\theta \in \Theta}\max_{1 \leq j \leq d_{\gamma}}||\tilde{g}_{n,\theta,j}^{*}-\tilde{g}_{0,\theta,j}||_{\infty} \lesssim \zeta_{n}$, 2. $\sup_{\theta \in \Theta}\max_{1 \leq j \leq d_{\gamma}}||\tilde{g}_{n,\theta,j}^{*}||_{\mathcal{H}} \lesssim \sqrt{n}\zeta_{n}$, 3. $\sqrt{n}\zeta_{n}\delta_{n}\rightarrow 0$ as $n\rightarrow \infty$, where $\delta_{n}$ is the RMSH rate, and 4. $\{\Delta_{n,\theta}^{*}: \theta \in \Theta\}$ with $\Delta_{n,\theta}^{*}(w,z)=\tilde{g}_{n,\theta}^{*}(F_{\theta,z}(w),z)-\tilde{g}_{0,\theta}(F_{\theta,z}(w),z)$ for $(w,z) \in \mathcal{W} \times [0,1]^{d_{z}}$ satisfies, for $P_{0,Z}^{\infty}$-almost every fixed realization $\{z_{i}\}_{i\geq 1}$ of $\{Z_{i}\}_{i \geq 1}$,
\begin{align*}
 \sup_{\theta \in \Theta}\left|\left|\frac{1}{\sqrt{n}}\sum_{i=1}^{n}\{\Delta_{n,\theta}^{*}(W_{i},z_{i})-E_{P_{0,W|Z}}[\Delta_{n,\theta}^{*}(W,Z)|Z=z_{i}]\}\right|\right|_{2} \overset{P_{0,W|Z}^{(n)}}{\longrightarrow} 0   
\end{align*}
as $n\rightarrow \infty$. Then, for $P_{0,Z}^{\infty}$-almost every fixed realization $\{z_{i}\}_{i\geq 1}$ of $\{Z_{i}\}_{i \geq 1}$, there is $\{\tilde{\mathcal{P}}_{n,W|Z}\}_{n \geq 1}$ such that $\Pi(p_{W|Z} \in \tilde{\mathcal{P}}_{n,W|Z}|W^{(n)}) \rightarrow 1$ in $P_{0,W|Z}^{(n)}$-probability and (\ref{eq:invariancegeneral}) holds.
\end{proposition}
Some comments on Proposition \ref{prop:invariance}. Sequences satisfying 1. and 2. can be found by solving $\sup_{\theta \in \Theta}\inf_{f \in \mathcal{H}: ||f-\tilde{g}_{0,\theta,j}||_{\infty} \leq \zeta_{n,j}}\frac{1}{2}||f||_{\mathcal{H}}^{2} \leq n \zeta_{n,j}^{2}$ for $j=1,...,d_{\gamma}$ and setting $\zeta_{n} = \max_{1 \leq j \leq d_{\gamma}}\zeta_{n,j}$. Since LHS is the component of the concentration function that describes the approximation of $\tilde{g}_{0,\theta,j}$ by $\mathcal{H}$, this can be checked for many GPs \citep{vaart2008rates,van2008reproducing,van2011information}. The condition $\sqrt{n}\zeta_{n}\delta_{n}\rightarrow 0$ as $n\rightarrow \infty$ in 3. is known as a `no-bias' condition, and, since $\delta_{n} = o(n^{-1/4})$, it holds if $\zeta_{n} = o(n^{-1/4})$.\footnote{This label follows \cite{castillo2012gaussian,castillo2012semiparametric}, \cite{rivoirard2012bernstein}, and \cite{castillo2015bernstein}, where similar conditions are likened to frequentist `no-bias' conditions (e.g., Section 25.8 of \cite{vaart_1998}).} Example \ref{ex:matern} below demonstrates that this restricts the smoothness of $B$ relative to $\tilde{g}_{0,\theta}$. Part 4. amounts to a complexity constraint on $\{\Delta_{n,\theta}^{*}: \theta \in \Theta\}$ because $\sup_{\theta \in \Theta}||\Delta_{n,\theta}^{*}||_{\infty} = o(1)$ implies that the convergence holds pointwise. Finally, if $\{g_{0,\theta,j}: \theta \in \Theta\} \subseteq \mathcal{H}$ for each $j$, then Proposition \ref{prop:invariance} holds trivially (i.e., $g_{n,\theta,j}^{*} = \tilde{g}_{0,\theta,j}$ and $\zeta_{n,j} = 0$ for each $j$), however, for many GPs, the RKHS comprises a relatively small class of functions.
\addtocounter{example}{-1}
\begin{example}[Continued]
Suppose that $\{\tilde{g}_{0,\theta,j}: \theta \in \Theta\} \subseteq C^{\beta_{0,j}}(T,\mathbb{R}) \cap H^{\beta_{0,j}}([0,1]^{d},\mathbb{R})$ for some $\beta_{0,1},...,\beta_{0,d_{\gamma}} >0$, and $\sup_{\theta \in \Theta}||\tilde{g}_{0,\theta,j}||_{\beta_{0,j}} < \infty$ and $\sup_{\theta \in \Theta}||\tilde{g}_{0,\theta,j}||_{2,2,\beta_{0,j}} < \infty$ for $j=1,...,d_{\gamma}$. If $\beta_{0,j} < \alpha$ for some $j$, then Lemma 13 in \cite{walker2026supplement} yields that $\zeta_{n} = n^{-\underline{\beta}/(2\alpha + d)}$, where $\underline{\beta} = \min_{1 \leq j \leq d_{\gamma}}\beta_{0,j}$, and $ \zeta_{n} = o(n^{-1/4})$ if $\underline{\beta} > \alpha/2+ d/4$. If $\alpha + d/2 >\underline{\beta} > \alpha$, then $\zeta_{n} = o(n^{-1/4})$ holds as long as $\alpha > d/2$. Finally, if $\underline{\beta} > \alpha + d/2$, then $\{\tilde{g}_{0,\theta,j}: \theta \in \Theta\} \subseteq \mathcal{H}$ for all $j$ and Proposition \ref{prop:invariance} holds trivially.
\end{example}
\section{Conclusion and Extensions}\label{sec:conclusion}
This paper proposes semiparametric Bayesian inference for conditional moment equalities. The central idea is that a posterior for a conditional distribution of data implies a posterior for a minimum distance estimand based on the conditional moments. The framework has similar flexibility to frequentist semiparametric estimators, and does not require converting the conditional moments to unconditional moments. I also establish the method's frequentist optimality via a BvM, providing conditions under which the posterior is asymptotically equivalent to a \cite{chamberlain1987asymptotic} efficient estimator.

My paper offers several directions for future research. First, there are important settings in which $\gamma$ is infinite-dimensional. An example is nonparametric IV in which $W=(Y,D)'$, $Z = (X,A)'$, and $g(W,Z,\gamma) = Y-\gamma(D,X)$, where $\gamma(\cdot,\cdot)$ is an unknown function \citep{ai2003efficient,newey2003instrumental}. Conceptually, the same ideas apply: a posterior for $P_{W|Z}$ implies a posterior for a \textit{function-valued} $\gamma$. However, theoretical issues, such as non-compact function spaces, ill-posedness of identifying conditions, etc., may warrant different minimum distance estimands (e.g., penalized estimands like in \cite{chen2012estimation}). 

Second, model misspecification is an important concern for conditional moment equality models (i.e., when $E_{P_{W|Z}}[g(W,Z,\gamma)|Z] \neq 0$ for all $\gamma \in \Gamma$, with positive probability). The posterior for the value function $Q_{n}(\gamma_{n},\gamma_{n},p_{W|Z})$, a `$J$-statistic'-like estimand, contains information about the compatibility of the conditional moments with the data. Developing a Bayesian specification assessment framework using this posterior and comparing it with classical overidentifying restrictions tests would be interesting. Moreover, reporting the posterior distribution of $Q_{n}(\gamma_{n},\gamma_{n},p_{W|Z})$ may also connect with the recommendations in \cite{andrews2025purpose} of reporting $J$-statistics (not $J$-tests) in overidentified models.

Finally, competing structural models often lead to nonnested conditional moment equalities $E_{P_{W|Z}}[g_{1}(W,Z,\gamma_{1})|Z] = 0$ and $E_{P_{W|Z}}[g_{2}(W,Z,\gamma_{2})|Z] = 0$. A posterior for $p_{W|Z}$ leads to a joint posterior for $(\gamma_{n,1},Q_{n,1},\gamma_{n,2},Q_{n,2})$, where $Q_{n,j}$, $j \in \{1,2\}$, is shorthand for $Q_{n,j}(\gamma_{n,j},\gamma_{n,j},p_{W|Z})$, with $Q_{n,j}(\gamma,\tilde{\gamma},p_{W|Z})$ being (\ref{eq:criterionfunctionsec2}) for $E_{P_{W|Z}}[g_{j}(W,Z,\gamma_{j})|Z]$. The posterior for $Q_{n,1}-Q_{n,2}$ can be used to assess which conditional moments are most plausible.

\bibliographystyle{ecca}
\bibliography{conditional}

@article{chamberlain2003nonparametric,
  title={Nonparametric applications of Bayesian inference},
  author={Chamberlain, Gary and Imbens, Guido W},
  journal={Journal of Business \& Economic Statistics},
  volume={21},
  number={1},
  pages={12--18},
  year={2003},
  publisher={Taylor \& Francis}
}

@article{chamberlain1987asymptotic,
  title={Asymptotic efficiency in estimation with conditional moment restrictions},
  author={Chamberlain, Gary},
  journal={Journal of econometrics},
  volume={34},
  number={3},
  pages={305--334},
  year={1987},
  publisher={Elsevier}
}

@article{berry1995automobile,
 author = {Steven Berry and James Levinsohn and Ariel Pakes},
 journal = {Econometrica},
 number = {4},
 pages = {841--890},
 publisher = {[Wiley, Econometric Society]},
 title = {Automobile Prices in Market Equilibrium},
 date = {2023-11-17},
 volume = {63},
 year = {1995}
}

@book{ghosal2017fundamentals,
  title={Fundamentals of nonparametric Bayesian inference},
  author={Ghosal, Subhashis and Van der Vaart, Aad},
  volume={44},
  year={2017},
  publisher={Cambridge University Press}
}

@article{bornn2018moment,
    author = {Bornn, Luke and Shephard, Neil and Solgi, Reza},
    title = {Moment Conditions and Bayesian Non-Parametrics},
    journal = {Journal of the Royal Statistical Society Series B: Statistical Methodology},
    volume = {81},
    number = {1},
    pages = {5-43},
    year = {2018},
}

@article{ray2020semiparametric,
author = {Kolyan Ray and Aad W van der Vaart},
title = {{Semiparametric Bayesian causal inference}},
volume = {48},
journal = {The Annals of Statistics},
number = {5},
publisher = {Institute of Mathematical Statistics},
pages = {2999 -- 3020},
year = {2020},
}

@article{castillo2015bernstein,
author = {Isma{\"e}l Castillo and Judith Rousseau},
title = {{A Bernstein–von Mises theorem for smooth functionals in semiparametric models}},
volume = {43},
journal = {The Annals of Statistics},
number = {6},
publisher = {Institute of Mathematical Statistics},
pages = {2353 -- 2383},
year = {2015},
}

@article{ai2003efficient,
author = {Ai, Chunrong and Chen, Xiaohong},
title = {Efficient Estimation of Models with Conditional Moment Restrictions Containing Unknown Functions},
journal = {Econometrica},
volume = {71},
number = {6},
pages = {1795-1843},
year = {2003}
}

@article{rivoirard2012bernstein,
author = {Vincent Rivoirard and Judith Rousseau},
title = {{Bernstein–von Mises theorem for linear functionals of the density}},
volume = {40},
journal = {The Annals of Statistics},
number = {3},
publisher = {Institute of Mathematical Statistics},
pages = {1489 -- 1523},
year = {2012}
}

@article{chib2018bayesian,
  title={Bayesian estimation and comparison of moment condition models},
  author={Chib, Siddhartha and Shin, Minchul and Simoni, Anna},
  journal={Journal of the American Statistical Association},
  volume={113},
  number={524},
  pages={1656--1668},
  year={2018},
  publisher={Taylor \& Francis}
}

@article{chib2022bayesian,
  title={Bayesian estimation and comparison of conditional moment models},
  author={Chib, Siddhartha and Shin, Minchul and Simoni, Anna},
  journal={Journal of the Royal Statistical Society Series B: Statistical Methodology},
  volume={84},
  number={3},
  pages={740--764},
  year={2022},
  publisher={Oxford University Press}
}

@article{newey1990efficient,
  title={Efficient instrumental variables estimation of nonlinear models},
  author={Newey, Whitney K},
  journal={Econometrica: Journal of the Econometric Society},
  pages={809--837},
  year={1990},
  publisher={JSTOR}
}

@article{donald2009choosing,
  title={Choosing instrumental variables in conditional moment restriction models},
  author={Donald, Stephen G and Imbens, Guido W and Newey, Whitney K},
  journal={Journal of Econometrics},
  volume={152},
  number={1},
  pages={28--36},
  year={2009},
  publisher={Elsevier}
}

@article{donald2003empirical,
  title={Empirical likelihood estimation and consistent tests with conditional moment restrictions},
  author={Donald, Stephen G and Imbens, Guido W and Newey, Whitney K},
  journal={Journal of Econometrics},
  volume={117},
  number={1},
  pages={55--93},
  year={2003},
  publisher={Elsevier}
}

@article{donald2001choosing,
  title={Choosing the number of instruments},
  author={Donald, Stephen G and Newey, Whitney K},
  journal={Econometrica},
  volume={69},
  number={5},
  pages={1161--1191},
  year={2001},
  publisher={Wiley Online Library}
}

@article{dominguez2004consistent,
  title={Consistent estimation of models defined by conditional moment restrictions},
  author={Dom{\'\i}nguez, Manuel A and Lobato, Ignacio N},
  journal={Econometrica},
  volume={72},
  number={5},
  pages={1601--1615},
  year={2004},
  publisher={Wiley Online Library}
}

@article{castillo2012semiparametric,
  title={A semiparametric Bernstein--von Mises theorem for Gaussian process priors},
  author={Castillo, Isma{\"e}l},
  journal={Probability Theory and Related Fields},
  volume={152},
  pages={53--99},
  year={2012},
  publisher={Springer}
}

@article{ghosal2007convergence,
author = {Subhashis Ghosal and Aad W van der Vaart},
title = {{Convergence rates of posterior distributions for noniid observations}},
volume = {35},
journal = {The Annals of Statistics},
number = {1},
publisher = {Institute of Mathematical Statistics},
pages = {192 -- 223},
year = {2007}
}

@article{tokdar2007towards,
  title={Towards a faster implementation of density estimation with logistic Gaussian process priors},
  author={Tokdar, Surya T},
  journal={Journal of Computational and Graphical Statistics},
  volume={16},
  number={3},
  pages={633--655},
  year={2007},
  publisher={Taylor \& Francis}
}

@article{tokdar2010bayesian,
author = {Surya T. Tokdar and Yu M. Zhu and Jayanta K. Ghosh},
title = {{Bayesian density regression with logistic Gaussian process and subspace projection}},
volume = {5},
journal = {Bayesian Analysis},
number = {2},
publisher = {International Society for Bayesian Analysis},
pages = {319 -- 344},
year = {2010}
}

@book{vaart_1998, place={Cambridge}, series={Cambridge Series in Statistical and Probabilistic Mathematics}, title={Asymptotic Statistics}, publisher={Cambridge University Press}, author={van der Vaart, A. W.}, year={1998}, collection={Cambridge Series in Statistical and Probabilistic Mathematics}}

@article{van2008reproducing,
  title={Reproducing kernel Hilbert spaces of Gaussian priors},
  author={van der Vaart, Aad W and van Zanten, J Harry},
  journal={IMS Collections},
  volume={3},
  pages={200--222},
  year={2008}
}

@article{vaart2008rates,
author = {A. W. van der Vaart and J. H. van Zanten},
title = {{Rates of contraction of posterior distributions based on Gaussian process priors}},
volume = {36},
journal = {The Annals of Statistics},
number = {3},
publisher = {Institute of Mathematical Statistics},
pages = {1435 -- 1463},
year = {2008},
}

@article{chen2018monte,
  title={Monte Carlo confidence sets for identified sets},
  author={Chen, Xiaohong and Christensen, Timothy M and Tamer, Elie},
  journal={Econometrica},
  volume={86},
  number={6},
  pages={1965--2018},
  year={2018},
  publisher={Wiley Online Library}
}

@book{williams2006gaussian,
  title={Gaussian processes for machine learning},
  author={Rasmussen, Carl Edward and Williams, Christopher KI},
  volume={2},
  number={3},
  year={2006},
  publisher={MIT press Cambridge, MA}
}

@article{chernozhukov2005iv,
 author = {Victor Chernozhukov and Christian Hansen},
 journal = {Econometrica},
 number = {1},
 pages = {245--261},
 publisher = {[Wiley, Econometric Society]},
 title = {An IV Model of Quantile Treatment Effects},
 date = {2024-02-25},
 volume = {73},
 year = {2005}
}

@book{vaart2023weak,
  title={Weak Convergence and Empirical Processes: With Applications to Statistics},
  author={van der Vaart, Aad W and Wellner, Jon A},
  volume={44},
  year={2023},
  publisher={Springer}
}

@article{jennrich1969asymptotic,
author = {Robert I. Jennrich},
title = {{Asymptotic Properties of Non-Linear Least Squares Estimators}},
volume = {40},
journal = {The Annals of Mathematical Statistics},
number = {2},
publisher = {Institute of Mathematical Statistics},
pages = {633 -- 643},
year = {1969},
}

@article{bickel2012semiparametric,
author = {P. J. Bickel and B. J. K. Kleijn},
title = {{The semiparametric Bernstein–von Mises theorem}},
volume = {40},
journal = {The Annals of Statistics},
number = {1},
publisher = {Institute of Mathematical Statistics},
pages = {206 -- 237},
year = {2012}
}

@article{chernozhukov2006instrumental,
  title={Instrumental quantile regression inference for structural and treatment effect models},
  author={Chernozhukov, Victor and Hansen, Christian},
  journal={Journal of Econometrics},
  volume={132},
  number={2},
  pages={491--525},
  year={2006},
  publisher={Elsevier}
}

@article{hansen1982generalized,
  title={Generalized instrumental variables estimation of nonlinear rational expectations models},
  author={Hansen, Lars Peter and Singleton, Kenneth J},
  journal={Econometrica: Journal of the Econometric Society},
  pages={1269--1286},
  year={1982},
  publisher={JSTOR}
}

@article{kitamura2004empirical,
  title={Empirical likelihood-based inference in conditional moment restriction models},
  author={Kitamura, Yuichi and Tripathi, Gautam and Ahn, Hyungtaik},
  journal={Econometrica},
  volume={72},
  number={6},
  pages={1667--1714},
  year={2004},
  publisher={Wiley Online Library}
}

@article{monard2021bernstein,
author = {Fran{\c{c}}ois Monard and Richard Nickl and Gabriel P. Paternain},
title = {{Statistical guarantees for Bayesian uncertainty quantification in nonlinear inverse problems with Gaussian process priors}},
volume = {49},
journal = {The Annals of Statistics},
number = {6},
publisher = {Institute of Mathematical Statistics},
pages = {3255 -- 3298},
year = {2021},
doi = {10.1214/21-AOS2082}
}

@article{castillo2012gaussian,
 author = {Ismaël Castillo},
 journal = {Sankhyā: The Indian Journal of Statistics, Series A (2008-)},
 number = {2},
 pages = {194--221},
 publisher = {Springer},
 title = {Semiparametric Bernstein-von Mises theorem and bias, illustrated with Gaussian process priors},
 date = {2024-03-12},
 volume = {74},
 year = {2012}
}

@article{hansen2021inference,
  title={Inference for iterated GMM under misspecification},
  author={Hansen, Bruce E and Lee, Seojeong},
  journal={Econometrica},
  volume={89},
  number={3},
  pages={1419--1447},
  year={2021},
  publisher={Wiley Online Library}
}

@article{robinson1987heteroskedasticity,
 author = {P. M. Robinson},
 journal = {Econometrica},
 number = {4},
 pages = {875--891},
 publisher = {[Wiley, Econometric Society]},
 title = {Asymptotically Efficient Estimation in the Presence of Heteroskedasticity of Unknown Form},
 date = {2024-03-18},
 volume = {55},
 year = {1987}
}

@incollection{NEWEY1993419,
title = {16 Efficient estimation of models with conditional moment restrictions},
series = {Handbook of Statistics},
publisher = {Elsevier},
volume = {11},
pages = {419-454},
year = {1993},
booktitle = {Econometrics},
issn = {0169-7161},
doi = {https://doi.org/10.1016/S0169-7161(05)80051-3},
author = {Whitney K. Newey}
}

@article{COHEN199354,
title = {Wavelets on the Interval and Fast Wavelet Transforms},
journal = {Applied and Computational Harmonic Analysis},
volume = {1},
number = {1},
pages = {54-81},
year = {1993},
author = {Albert Cohen and Ingrid Daubechies and Pierre Vial}
}

@misc{chen2021harmless,
      title={Mostly Harmless Machine Learning: Learning Optimal Instruments in Linear IV Models}, 
      author={Jiafeng Chen and Daniel L. Chen and Greg Lewis},
      year={2021},
      eprint={2011.06158},
      archivePrefix={arXiv},
      primaryClass={econ.EM}
}

@article{graham2008identifying,
 ISSN = {00129682, 14680262},
 author = {Bryan S. Graham},
 journal = {Econometrica},
 number = {3},
 pages = {643--660},
 publisher = {[Wiley, Econometric Society]},
 title = {Identifying Social Interactions through Conditional Variance Restrictions},
 date = {2024-04-09},
 volume = {76},
 year = {2008}
}

@article{norets2015bayesian,
  title={Bayesian regression with nonparametric heteroskedasticity},
  author={Norets, Andriy},
  journal={Journal of econometrics},
  volume={185},
  number={2},
  pages={409--419},
  year={2015},
  publisher={Elsevier}
}

@ARTICLE{walker2024parametrization,
    AUTHOR = {Christopher D. Walker},
     TITLE = {Parametrization, prior independence, and the semiparametric Bernstein-von Mises theorem for the partially linear model},
   JOURNAL = {Bernoulli},
  FJOURNAL = {Bernoulli},
      YEAR = {2026},
    VOLUME = {32},
    NUMBER = {2},
     PAGES = {1503-1522},
      ISSN = {1350-7265},
       DOI = {10.3150/25-BEJ1917},
      SICI = {1350-7265(2026)32:2<1503:PPIATS>2.0.CO;2-0},
}

@article{castillo2014bayesian,
author = {Isma{\"e}l Castillo},
title = {{On Bayesian supremum norm contraction rates}},
volume = {42},
journal = {The Annals of Statistics},
number = {5},
publisher = {Institute of Mathematical Statistics},
pages = {2058 -- 2091},
year = {2014},
doi = {10.1214/14-AOS1253}
}

@article{schennach2005bayesian,
    author = {Schennach, Susanne M.},
    title = "{Bayesian exponentially tilted empirical likelihood}",
    journal = {Biometrika},
    volume = {92},
    number = {1},
    pages = {31-46},
    year = {2005},
    month = {03},
    issn = {0006-3444},
    doi = {10.1093/biomet/92.1.31}
}

@article{florens2021gaussian,
author = {Jean-Pierre Florens and Anna Simoni},
title = {Gaussian Processes and Bayesian Moment Estimation},
journal = {Journal of Business \& Economic Statistics},
volume = {39},
number = {2},
pages = {482--492},
year = {2021},
publisher = {Taylor \& Francis},
doi = {10.1080/07350015.2019.1668799}
}

@article{kitamura2011bayesian,
  title={Bayesian analysis of moment condition models using nonparametric priors},
  author={Kitamura, Yuichi and Otsu, Taisuke},
  journal={Unpublished Manuscript, Yale University},
  year={2011}
}

@book{shin2015bayesian,
  title={Bayesian Gmm},
  author={Shin, Min Chul},
  year={2015},
  publisher={University of Pennsylvania}
}

@incollection{CHEN20075549,
title = {Chapter 76 Large Sample Sieve Estimation of Semi-Nonparametric Models},
editor = {James J. Heckman and Edward E. Leamer},
series = {Handbook of Econometrics},
publisher = {Elsevier},
volume = {6},
pages = {5549-5632},
year = {2007},
issn = {1573-4412},
doi = {https://doi.org/10.1016/S1573-4412(07)06076-X},
author = {Xiaohong Chen}
}

@article{kato2013quasi,
author = {Kengo Kato},
title = {{Quasi-Bayesian analysis of nonparametric instrumental variables models}},
volume = {41},
journal = {The Annals of Statistics},
number = {5},
publisher = {Institute of Mathematical Statistics},
pages = {2359 -- 2390},
year = {2013},
doi = {10.1214/13-AOS1150}
}

@article{abadie2014inference,
author = {Alberto Abadie and Guido W. Imbens and Fanyin Zheng},
title = {Inference for Misspecified Models With Fixed Regressors},
journal = {Journal of the American Statistical Association},
volume = {109},
number = {508},
pages = {1601--1614},
year = {2014},
publisher = {Taylor \& Francis},
}

@article{pelenis2014bayesian,
  title={Bayesian regression with heteroscedastic error density and parametric mean function},
  author={Pelenis, Justinas},
  journal={Journal of econometrics},
  volume={178},
  pages={624--638},
  year={2014},
  publisher={Elsevier}
}

@article{chernozhukov2003mcmc,
  title={An MCMC approach to classical estimation},
  author={Chernozhukov, Victor and Hong, Han},
  journal={Journal of econometrics},
  volume={115},
  number={2},
  pages={293--346},
  year={2003},
  publisher={Elsevier}
}

@article{liao2011posterior,
author = {Yuan Liao and Wenxin Jiang},
title = {{Posterior consistency of nonparametric conditional moment restricted models}},
volume = {39},
journal = {The Annals of Statistics},
number = {6},
publisher = {Institute of Mathematical Statistics},
pages = {3003 -- 3031},
year = {2011},
doi = {10.1214/11-AOS930}
}

@article{chen2012estimation,
  title={Estimation of nonparametric conditional moment models with possibly nonsmooth generalized residuals},
  author={Chen, Xiaohong and Pouzo, Demian},
  journal={Econometrica},
  volume={80},
  number={1},
  pages={277--321},
  year={2012},
  publisher={Wiley Online Library}
}

@article{andrews2022optimal,
  title={Optimal decision rules for weak GMM},
  author={Andrews, Isaiah and Mikusheva, Anna},
  journal={Econometrica},
  volume={90},
  number={2},
  pages={715--748},
  year={2022},
  publisher={Wiley Online Library}
}

@article{bharucha1976fixed,
  title={Fixed point theorems in probabilistic analysis},
  author={Bharucha-Reid, Albert Turner},
  journal={Bulletin of the American Mathematical Society},
  volume={82},
  number={5},
  pages={641--657},
  year={1976}
}

@article{blundell2012measuring,
  title={Measuring the price responsiveness of gasoline demand: Economic shape restrictions and nonparametric demand estimation},
  author={Blundell, Richard and Horowitz, Joel L and Parey, Matthias},
  journal={Quantitative Economics},
  volume={3},
  number={1},
  pages={29--51},
  year={2012},
  publisher={Wiley Online Library}
}

@article{chen2018optimal,
  title={Optimal sup-norm rates and uniform inference on nonlinear functionals of nonparametric IV regression},
  author={Chen, Xiaohong and Christensen, Timothy M},
  journal={Quantitative Economics},
  volume={9},
  number={1},
  pages={39--84},
  year={2018},
  publisher={Wiley Online Library}
}

@book{gine2016mathematical,
  title={Mathematical foundations of infinite-dimensional statistical models},
  author={Gin{\'e}, Evarist and Nickl, Richard},
  year={2016},
  publisher={Cambridge university press}
}

@article{van2011information,
  author  = {Aad W van der Vaart and Harry van Zanten},
  title   = {Information Rates of Nonparametric Gaussian Process Methods},
  journal = {Journal of Machine Learning Research},
  year    = {2011},
  volume  = {12},
  number  = {60},
  pages   = {2095--2119}
}

@article{hausman1981exact,
  title={Exact consumer's surplus and deadweight loss},
  author={Hausman, Jerry A},
  journal={The American Economic Review},
  volume={71},
  number={4},
  pages={662--676},
  year={1981},
  publisher={JSTOR}
}

@article{hausman1995nonparametric,
  title={Nonparametric estimation of exact consumers surplus and deadweight loss},
  author={Hausman, Jerry A and Newey, Whitney K},
  journal={Econometrica: Journal of the Econometric Society},
  pages={1445--1476},
  year={1995},
  publisher={JSTOR}
}

@article{blundell2017nonparametric,
  title={Nonparametric estimation of a nonseparable demand function under the Slutsky inequality restriction},
  author={Blundell, Richard and Horowitz, Joel L and Parey, Matthias},
  journal={Review of Economics and Statistics},
  volume={99},
  number={2},
  pages={291--304},
  year={2017},
  publisher={MIT Press One Rogers Street, Cambridge, MA 02142-1209, USA journals-info~…}
}

@article{shen2002asymptotic,
  title={Asymptotic normality of semiparametric and nonparametric posterior distributions},
  author={Shen, Xiaotong},
  journal={Journal of the American Statistical Association},
  volume={97},
  number={457},
  pages={222--235},
  year={2002},
  publisher={Taylor \& Francis}
}

@article{norets2014posterior,
  title={Posterior consistency in conditional density estimation by covariate dependent mixtures},
  author={Norets, Andriy and Pelenis, Justinas},
  journal={Econometric Theory},
  volume={30},
  number={3},
  pages={606--646},
  year={2014},
  publisher={Cambridge University Press}
}

@article{norets2010approximation,
author = {Andriy Norets},
title = {{Approximation of conditional densities by smooth mixtures of regressions}},
volume = {38},
journal = {The Annals of Statistics},
number = {3},
publisher = {Institute of Mathematical Statistics},
pages = {1733 -- 1766},
year = {2010},
doi = {10.1214/09-AOS765}
}

@article{norets2017adaptive,
  title={Adaptive Bayesian estimation of conditional densities},
  author={Norets, Andriy and Pati, Debdeep},
  journal={Econometric Theory},
  volume={33},
  number={4},
  pages={980--1012},
  year={2017},
  publisher={Cambridge University Press}
}

@article{vehtari2014laplace,
author = {Jaakko Riihim{\"a}ki and Aki Vehtari},
title = {{Laplace Approximation for Logistic Gaussian Process Density Estimation and Regression}},
volume = {9},
journal = {Bayesian Analysis},
number = {2},
publisher = {International Society for Bayesian Analysis},
pages = {425 -- 448},
year = {2014},
doi = {10.1214/14-BA872}
}

@article{dunson2015marginally,
 ISSN = {13697412, 14679868},
 author = {David C. Kessler and Peter D. Hoff and David B. Dunson},
 journal = {Journal of the Royal Statistical Society. Series B (Statistical Methodology)},
 number = {1},
 pages = {35--58},
 publisher = {[Royal Statistical Society, Wiley]},
 title = {Marginally specified priors for non-parametric Bayesian estimation},
 date = {2024-07-13},
 volume = {77},
 year = {2015}
}

@article{newey2003instrumental,
 ISSN = {00129682, 14680262},
 author = {Whitney K. Newey and James L. Powell},
 journal = {Econometrica},
 number = {5},
 pages = {1565--1578},
 publisher = {[Wiley, Econometric Society]},
 title = {Instrumental Variable Estimation of Nonparametric Models},
 date = {2024-07-28},
 volume = {71},
 year = {2003}
}

@article{lancaster2010bayesian,
author = {Lancaster, Tony and Jun, Sung Jae},
title = {Bayesian quantile regression methods},
journal = {Journal of Applied Econometrics},
volume = {25},
number = {2},
pages = {287-307},
doi = {https://doi.org/10.1002/jae.1069},
year = {2010}
}

@TechReport{chamberlain1995semiparametric,
  author={Gary Chamberlain and Guido W. Imbens},
  title={{Semiparametric Applications of Bayesian Influence}},
  year=1995,
  institution={Harvard - Institute of Economic Research},
  type={Harvard Institute of Economic Research Working Papers},
  number={1716},
  doi={},
}

@article{CHEN2015447,
title = {Optimal uniform convergence rates and asymptotic normality for series estimators under weak dependence and weak conditions},
journal = {Journal of Econometrics},
volume = {188},
number = {2},
pages = {447-465},
year = {2015},
author = {Xiaohong Chen and Timothy M. Christensen}
}

@article{TOKDAR200734,
title = {Posterior consistency of logistic Gaussian process priors in density estimation},
journal = {Journal of Statistical Planning and Inference},
volume = {137},
number = {1},
pages = {34-42},
year = {2007},
author = {Surya T. Tokdar and Jayanta K. Ghosh}
}

@book{goldberger1991course,
publisher = {Harvard University Press},
isbn = {0674175441},
year = {1991},
title = {A Course in Econometrics},
language = {eng},
address = {Cambridge, Mass.},
author = {Goldberger, Arthur Stanley},
}

@article{breunig2025double,
  title={Double robust bayesian inference on average treatment effects},
  author={Breunig, Christoph and Liu, Ruixuan and Yu, Zhengfei},
  journal={Econometrica},
  volume={93},
  number={2},
  pages={539--568},
  year={2025},
  publisher={Wiley Online Library}
}

@article{yiu2025semiparametric,
    author = {Yiu, Andrew and Fong, Edwin and Holmes, Chris and Rousseau, Judith},
    title = {Semiparametric posterior corrections},
    journal = {Journal of the Royal Statistical Society Series B: Statistical Methodology},
    volume = {87},
    number = {4},
    pages = {1025-1054},
    year = {2025},
    month = {02},
    issn = {1369-7412},
    doi = {10.1093/jrsssb/qkaf005},
    url = {https://doi.org/10.1093/jrsssb/qkaf005},
    eprint = {https://academic.oup.com/jrsssb/article-pdf/87/4/1025/62045089/qkaf005.pdf},
}

@article{dunson2008kernel,
  title={Kernel stick-breaking processes},
  author={Dunson, David B and Park, Ju-Hyun},
  journal={Biometrika},
  volume={95},
  number={2},
  pages={307--323},
  year={2008},
  publisher={Oxford University Press}
}

@article{rodriguez2011nonparametric,
  title={Nonparametric Bayesian models through probit stick-breaking processes},
  author={Rodriguez, Abel and Dunson, David B},
  journal={Bayesian Analysis},
  volume={6},
  number={1},
  pages={145--178},
  year={2011}
}

@article{ren2011logistic,
  title={Logistic stick-breaking process.},
  author={Ren, Lu and Du, Lan and Dunson, David B and others},
  journal={Journal of Machine Learning Research},
  volume={12},
  number={1},
  year={2011}
}

@article{kankanala2023quasi,
  title={Quasi-Bayes in Latent Variable Models},
  author={Kankanala, Sid},
  journal={arXiv preprint arXiv:2311.06831},
  year={2023}
}

@article{kankanala2025generalized,
  title={Generalized Bayes in Conditional Moment Restriction Models},
  author={Kankanala, Sid},
  journal={arXiv preprint arXiv:2510.01036},
  year={2025}
}

@article{walker2026supplement,
    author ={Walker, Christopher D.} ,
    title = {Online Appendix to `Semiparametric Bayesian Inference for Conditional Moment Equalities'},
    year ={2026}
}

@article{andrews2025purpose,
  title={The purpose of an estimator is what it does: Misspecification, estimands, and over-identification},
  author={Andrews, Isaiah and Chen, Jiafeng and Tecchio, Otavio},
  journal={arXiv preprint arXiv:2508.13076},
  year={2025}
}

@article{rivoirard2012posterior,
author = {Vincent Rivoirard and Judith Rousseau},
title = {{Posterior Concentration Rates for Infinite Dimensional Exponential Families}},
volume = {7},
journal = {Bayesian Analysis},
number = {2},
publisher = {International Society for Bayesian Analysis},
pages = {311 -- 334},
year = {2012},
doi = {10.1214/12-BA710},
URL = {https://doi.org/10.1214/12-BA710}
}

@article{scricciolo2006convergence,
author = {Catia Scricciolo},
title = {{Convergence rates for Bayesian density estimation of infinite-dimensional exponential families}},
volume = {34},
journal = {The Annals of Statistics},
number = {6},
publisher = {Institute of Mathematical Statistics},
pages = {2897 -- 2920},
year = {2006},
doi = {10.1214/009053606000000911},
URL = {https://doi.org/10.1214/009053606000000911}
}

@article{florens2012nonparametric,
  title={Nonparametric estimation of an instrumental regression: A quasi-Bayesian approach based on regularized posterior},
  author={Florens, Jean-Pierre and Simoni, Anna},
  journal={Journal of Econometrics},
  volume={170},
  number={2},
  pages={458--475},
  year={2012},
  publisher={Elsevier}
}

@article{Chib01111995,
author = {Siddhartha Chib and Edward Greenberg},
title = {Understanding the Metropolis-Hastings Algorithm},
journal = {The American Statistician},
volume = {49},
number = {4},
pages = {327--335},
year = {1995},
publisher = {Taylor \& Francis}
}

@book{Vershynin_2018, place={Cambridge}, series={Cambridge Series in Statistical and Probabilistic Mathematics}, title={High-Dimensional Probability: An Introduction with Applications in Data Science}, publisher={Cambridge University Press}, author={Vershynin, Roman}, year={2018}, collection={Cambridge Series in Statistical and Probabilistic Mathematics}}

@misc{breunig2025semiparametricbayesiandifferenceindifferences,
      title={Semiparametric Bayesian Difference-in-Differences}, 
      author={Christoph Breunig and Ruixuan Liu and Zhengfei Yu},
      year={2025},
      eprint={2412.04605},
      archivePrefix={arXiv},
      primaryClass={econ.EM},
      url={https://arxiv.org/abs/2412.04605}, 
}
\newpage
\appendix
\section{Existence and Measurability of $\gamma_{n}$}\label{ap:measurability}
I state sufficient conditions for Assumption \ref{as:uemfp}. A discussion follows.
\begin{assumption}\label{as:param}
The following holds for almost every realization $\{z_{i}\}_{i \geq 1}$ of $\{Z_{i}\}_{i \geq 1}$ and every $n \geq 1$, 1. $\Gamma$ is a compact subset of $(\mathbb{R}^{d_{\gamma}},||\cdot||_{2})$, 2. $p_{W|Z} \mapsto Q_{n}(\gamma,\tilde{\gamma},p_{W|Z})$ is $\mathscr{P}_{W|Z}$-measurable for each $(\gamma,\tilde{\gamma}) \in \Gamma \times \Gamma$, 3. $\gamma \mapsto Q_{n}(\gamma,\tilde{\gamma},p_{W|Z})$ is continuous for every $\tilde{\gamma} \in \Gamma$ and every $p_{W|Z} \in \mathcal{P}_{W|Z}$, and 4. $q_{n}(\tilde{\gamma},p_{W|Z})= \argmin_{\gamma \in \Gamma} Q_{n}(\gamma,\tilde{\gamma},p_{W|Z})$ is unique and satisfies $q_{n}(\tilde{\gamma},p_{W|Z}) \in \text{int}(\Gamma)$ for every $\tilde{\gamma} \in \Gamma$ and every $p_{W|Z} \in \mathcal{P}_{W|Z}$.
\end{assumption}
\begin{assumption}\label{as:fp}
The following holds for almost every realization $\{z_{i}\}_{ i\geq 1}$ of $\{Z_{i}\}_{ i\geq 1}$ and every $n \geq 1$: 1. $\gamma \mapsto m(z_{i},\gamma)$, $i=1,...,n$, is twice continuously differentiable for each $p_{W|Z} \in \mathcal{P}_{W|Z}$, 2. $\gamma \mapsto \Sigma(z_{i},\gamma)$,$i=1,...,n$, is continuously differentiable for each $p_{W|Z} \in \mathcal{P}_{W|Z}$, 3. $\inf_{\gamma\in \Gamma}\min_{1 \leq i \leq n}\lambda_{min}(\Sigma(z_{i},\gamma)) > 0$ for each $p_{W|Z} \in \mathcal{P}_{W|Z}$, and 4. $(x_{1},x_{2}) \mapsto Q_{n}(x_{1},x_{2},p_{W|Z})$ satisfies
\begin{align*}
\sup_{\tilde{\gamma} \in \Gamma}\left | \left |  \frac{\partial^{2}}{\partial x_{1} \partial x_{2}'}Q_{n}(q_{n}(\tilde{\gamma},p_{W|Z}),\tilde{\gamma},p_{W|Z})\right| \right|_{op} < \inf_{\tilde{\gamma}\in \Gamma} \lambda_{min}\left(\frac{\partial^{2}}{\partial x_{1} \partial x_{1}'}Q_{n}(q_{n}(\tilde{\gamma},p_{W|Z}),\tilde{\gamma},p_{W|Z})\right)  
\end{align*}
for every $p_{W|Z} \in \mathcal{P}_{W|Z}$.
\end{assumption}
Assumptions \ref{as:param} and \ref{as:fp} are similar to Assumption 1 in \cite{hansen2021inference}, except that they apply to $q_{n}(\tilde{\gamma},p_{W|Z})$ rather than a GMM estimand. The key condition is Assumption \ref{as:fp}.4, which requires that $n^{-1}\sum_{i=1}^{n}M(z_{i},q_{n}(\tilde{\gamma},p_{W|Z}))'\Sigma^{-1}(z_{i},\tilde{\gamma})m(z_{i},q_{n}(\tilde{\gamma},p_{W|Z}))$ not deviate too far from zero under small changes to $\Sigma^{-1}(z_{i},\tilde{\gamma})$ (with deviations constrained by how strongly identified $q_{n}(\tilde{\gamma},p_{W|Z})$ is). Combined with the other assumptions, the condition implies $\tilde{\gamma} \mapsto q_{n}(\tilde{\gamma},p_{W|Z})$ is a measurable contraction mapping. Consequently, a probabilistic analogue to the Banach Fixed Point Theorem \citep{bharucha1976fixed} reveals that $\tilde{\gamma} \mapsto q_{n}(\tilde{\gamma},p_{W|Z})$ has a unique fixed point $\gamma_{n}$ for each $p_{W|Z}$ that can be characterized as the $r\rightarrow \infty$ limit of the recursive relation $\gamma_{n,r} = q_{n}(\gamma_{n,r-1},p_{W|Z})$ (and $p_{W|Z} \mapsto \gamma_{n}$ is measurable). This discussion is summarized in the next proposition, and, since the result concerns a conditional minimum distance estimand, it may be of independent interest.
\begin{proposition}\label{prop:fixedpoint}
If Assumptions \ref{as:param} and \ref{as:fp} hold, then Assumption \ref{as:uemfp} holds.
\end{proposition}
\section{Implementation Details and Data-Calibrated Simulation}\label{ap:empirical}
\subsection{Sampling from the Posterior for the Conditional Density}\label{ap:postsample}
\subsubsection{Covariance Function and Base Density} The Mat\'{e}rn covariance function is computed using the formula $\kappa(t,s) = \Gamma^{-1}(\alpha)2^{1-\alpha}(\sqrt{2\alpha}||t-s||_{2})^{\alpha}K_{\alpha}(\sqrt{2\alpha}||t-s||_{2})$ for $t,s \in [0,1]^{d}$, where $\Gamma(\cdot)$ is the gamma function, $K_{\alpha}$ is the modified Bessel function of the second kind of order $\alpha$, and $\alpha = 5/2$ (following \cite{williams2006gaussian}).  I specify $f_{\theta}(w|z)$ to be the density of $\mathcal{N}(\mu_{q}'\tilde{Z},\xi_{q}^{2})\otimes \mathcal{N}(\mu_{p}'Z,\xi_{p}^{2})$, where $\tilde{Z} = (Z',\log P)'$ and, for this case only, $Z = (1,\log Y, \log A)$. This means $\theta = (\mu_{q}',\mu_{p}',\xi_{p}^{2},\xi_{q}^{2})'$. I set $\xi_{q}^{2},\xi_{p}^{2} \overset{iid}{\sim}InvGamma(0.001,0.001)$ and $\mu_{j}|\xi_{p}^{2},\xi_{q}^{2} \overset{ind}{\sim} \mathcal{N}(0,10^{3}\cdot \xi_{j}^{2})$ for each $j \in \{p,q\}$. Finally, $F_{0}(z_{i})=(F_{0,1}(z_{i,1}),F_{0,2}(z_{i,2}))'$, where $F_{0,j}(z_{i,j}) = (z_{i,j}-\underline{z}_{n,j})/(\overline{z}_{n,j}-\underline{z}_{n,j})$, $z_{i,1} = \log y_{i}$, $z_{i,2} = \log a_{i}$, $\overline{z}_{n,j} = \max_{1 \leq i \leq n}z_{i,j}$, and $\underline{z}_{n,j} = \min_{1 \leq i \leq n}z_{i,j}$.
\subsubsection{Sampling from the Posterior for $p_{\theta,B}$} I follow \cite{tokdar2007towards} and \cite{tokdar2010bayesian}, replacing $B$ with the \textit{kriging} $B^{*}$, where $B^{*}(t) = E[B(t)|B(t_{1}),...,B(t_{K})]$, $t \in [0,1]^{d}$, for a prespecified set of nodes $t_{1},...,t_{K} \in [0,1]^{d}$.\footnote{The reason for kriging is that the normalizing integral in $p_{\theta,B}(\cdot|z)$ depends on the entire function $B(\cdot,F_{0}(z)): [0,1]^{d_{w}} \rightarrow \mathbb{R}$, which complicates likelihood evaluation. It is an approximation because if the nodes grow dense in $[0,1]^{d}$ as $K\rightarrow \infty$ then $B^{*}$ converges to $B$ \citep{ghosal2017fundamentals}} $B^{*}$ is also a centered GP satisfying $B^{*}(t) = \lambda'\mathcal{K}^{-1/2}\kappa_{1:K}(t)$, where $\lambda \sim \mathcal{N}(0,I)$, $\mathcal{K} = (\kappa(t_{k},t_{j}))_{k,j=1}^{K}$, and $\kappa_{1:K}(t) = (\kappa(t_{1},t),...,\kappa(t_{K},t))'$, and leads to a Metropolis-within-Gibbs algorithm to sample from
\begin{align*}
    \pi(\lambda, \theta|W^{(n)}) \propto \prod_{i=1}^{n}\left\{f_{\theta}(W_{i}|z_{i})\frac{\exp(B^{*}(F_{\theta,z_{i}}(W_{i}),F_{0}(z_{i})))}{\int_{[0,1]^{d_{w}}}\exp(B^{*}(u,F_{0}(z_{i})))du}\right\}p_{\mathcal{N}(0,I)}(\lambda)\pi_{\Theta}(\theta),
\end{align*}
where $\pi_{\Theta}$ is the density of $\Pi_{\Theta}$. Specifically, given an initial parameter value $(\theta_{old}',\lambda_{old}')'$, a sweep of the Gibbs sampling algorithm applies the following steps:
\begin{enumerate}
 \item Sample $\lambda_{new} \sim \Pi(\lambda \in \cdot | W^{(n)},\theta_{old})$ using Metropolis-Hastings (MH) with proposal $\mathcal{N}(\lambda_{old},c^{2} I)$.
 \item Sample $\theta_{new} \sim \Pi(\theta \in \cdot | W^{(n)},\lambda_{new}) $ using MH with proposal density $q_{n}(\theta|\theta_{old}) \propto \prod_{i=1}^{n}f_{\theta}(W_{i}|z_{i})\pi_{\Theta}(\theta)$.
\end{enumerate}
Iterating over these steps $S$ times leads to a sample $\{(\lambda^{[s]},\theta^{[s]})\}_{s=1}^{S}$ from $\pi(\lambda,\theta|W^{(n)})$. In terms of specific implementation: 1. I adapt $c>0$ during the burn-in to target an acceptance probability of 35\%, 2. sampling from $q_{n}(\cdot|\theta_{old})$ exploits conjugacy of the normal-inverse gamma prior in Gaussian linear regression and conforms to a recommendation in \cite{Chib01111995}, and 3. I set $\{t_{k}\}_{k=1}^{K} = (\{j/2\}_{j=0}^{2})^{2} \times \{\hat{F}_{0,n,j}\}_{j=1}^{3}$, where $\hat{F}_{0,n,j}$ are $k$-means cluster centers of $\{F_{0}(z_{i})\}_{j=1}^{n}$, setting the number of clusters equal to $3$.\footnote{The choice $(\{j/2\}_{j=0}^{2})^{2}$ for $W$ corresponds to a node configuration in \cite{tokdar2007towards}, while $\{\hat{z}_{n,j}\}_{j=1}^{3}$ reflects that $\prod_{i=1}^{n}p_{\theta,B}(W_{i}|z_{i})$ only depends on $\{B(\cdot,z_{i})\}_{i=1}^{n}$ thereby serving as parsimonious approximation to using $\{z_{i}\}_{i=1}^{n}$ as nodes for $z$, an alternative that eliminates node selection for $z$ but increases the number of parameters to $9n$.} There are data-driven node selection schemes \citep{tokdar2007towards,tokdar2010bayesian}, however these can be computationally intensive. Given $\{(B^{*,[s]},\theta^{[s]})\}_{s=1}^{S}$ (with $B^{*,[s]}(t) =\lambda^{[s]'}\mathcal{K}^{-1/2}\kappa_{1:K}(t)$), I obtain $\{\gamma_{n}^{[s]}\}_{s=1}^{S}$ using the importance sampling algorithm described in Section \ref{sec:implementation}.
\begin{remark}
Kriging is introduced purely for computation. The asymptotic analysis in Section \ref{sec:example} is based on the underlying $B$, which is consistent with standard theoretical treatment of GP
density estimation \citep{TOKDAR200734,vaart2008rates,tokdar2010bayesian,castillo2015bernstein}. This distinction is important because if advances in Bayesian computation make it possible bypass kriging (and use $B$), then the asymptotic theory still applies.\footnote{Note that, since $B^{*}$ is a centered GP, conditions for Propositions \ref{prop:RMSHrate}--\ref{prop:invariance} could be derived for this class (i.e., an exercise similar to Example \ref{ex:matern}, except with covariance function $\kappa(t,s) = \kappa_{1:K}(t)'\mathcal{K}^{-1}\kappa_{1:K}(s)$).}
\end{remark}
\subsection{Data-Calibrated Simulation}\label{ap:simulation}
\subsubsection{Simulation DGP and Target Parameter}\label{ap:DGPprioradditional} 
I calibrate the simulation DGP to the 2001 National Household Travel Survey gasoline demand dataset. Specifically, I simulate:
\begin{align}
    \log Q &= \gamma_{0} + \gamma_{1} \log P + \gamma_{2} \log Y + U_{1} ,\quad E[U_{1}|\log Y, \log A] = 0\label{eq:struc} \\
    \log P &= \lambda_{0} + \lambda_{1}\log Y + \lambda_{2} \log A + U_{2}, \quad E[U_{2}|\log Y, \log A] = 0 \label{eq:rf},
\end{align}
where $(\gamma_{0},\gamma_{1},\gamma_{2}) = (4.413,-1.631,0.293)$ and $(\lambda_{0},\lambda_{1},\lambda_{2}) = (0.178,-0.001, 0.066)$. The values of $\gamma$ and $\lambda$ are the Wald estimates based on $E[U_{1}(1,\log Y,\log A)] =0$ and the OLS estimates of (\ref{eq:rf}), respectively, using the full dataset. The target parameter is deadweight loss (DWL) from a price change from $p^{0} = \$1.22$ and $p^{1} = \$1.44$ for $y \in \{\$42500,\$72500\}$ (resulting in DWL of $\$33.05$ and $\$38.74$ for the low- and high-income groups, respectively). Samples $\{(A_{i_{j}}^{*},Y_{i_{j}}^{*}, Q_{i_{j}}^{*}, P_{i_{j}}^{*})\}_{j=1}^{n}$ are generated as follows:
\begin{enumerate}
    \item Sample indices $\{i_{1},...,i_{n}\}$ by sampling from $\{1,...,4812\}$ with replacement.
    \item Sample $(\varepsilon_{i_{1},1}, \varepsilon_{i_{1},2}, \varepsilon_{i_{1},3},\varepsilon_{i_{1},4}),...,(\varepsilon_{i_{n},1}, \varepsilon_{i_{n},2}, \varepsilon_{i_{n},3},\varepsilon_{i_{n},4})$ independently from $\mathcal{N}(0, \widehat{\Omega})$, where $\widehat{\Omega} =  \frac{1}{100}\text{diag}(\widehat{Var}(\log Y), \widehat{Var}(\log A), \widehat{Var}(\log P),\widehat{Var}(\log Q))$ with $\widehat{Var}$ denoting the sample variance.
    \item Generate $\log Y_{i_{j}}^{*}= \log Y_{i_{j}}+\varepsilon_{i_{j},1}$ and $\log A_{i_{j}}^{*} = \log A_{i_{j}}+\varepsilon_{i_{j},2}$ for $j=1,...,n$.
    \item Sample $R_{i_{1}},...,R_{i_{n}}$ independently from a Rademacher distribution.
    \item Generate $\log P_{i_{j}}^{*} = \lambda_{0}+ \lambda_{1}\log Y_{i_{j}}^{*}+ \lambda_{2}\log A_{i_{j}}^{*}+ R_{i_{j}}U_{i_{j},2}+ \varepsilon_{i_{j},3}$ for $j=1,...,n$, where $U_{i,2}$ are the OLS residuals from the full dataset.
    \item Generate $\log Q_{i_{j}}^{*} = \gamma_{0}+ \gamma_{1}\log P_{i_{j}}^{*}+\gamma_{2}\log Y_{i_{j}}^{*}+ R_{i_{j}}U_{i_{j},1}+\varepsilon_{i_{j},4}$ for $j=1,...,n$, where $U_{i,1}$ are the IV/Wald residuals from the full dataset.
\end{enumerate}
Adding elements of $(\varepsilon_{i_{j},1},\varepsilon_{i_{j},2},\varepsilon_{i_{j},3},\varepsilon_{i_{j},4})$ guarantees all variables are continuously distributed. Interacting $R_{i_{j}}$ and $U_{i_{j},k}$, $k \in \{1,2\}$, ensures that the conditional moment restrictions hold at $\gamma$ and allows for heteroskedasticity of an unknown form. Multiplying $U_{i_{j},1}$ and $U_{i_{j},2}$ by a common $R_{i_{j}}$ preserves the first- and second-stage residual covariance.
\subsubsection{Other Procedures}\label{sec:otherprocedures}
I implement semiparametric Bayes (SB) by running $S =50,000$ iterations of the Metropolis-Within-Gibbs sampler, discarding the first half as burn-in, and keeping every fifth thereafter (resulting in $5,001$ draws from the posterior). I compare with:
\begin{enumerate}
\item \textit{Two-Stage Least Squares (TSLS).}  $\hat{\gamma}_{TSLS} = (\sum_{i=1}^{n}z_{i}X_{i}')^{-1}\sum_{i=1}^{n}z_{i}\log Q_{i}$, where $X_{i} = (1, \log P_{i}, \log Y_{i})'$ and $z_{i} = (1,\log y_{i}, \log a_{i})'$. The estimator for the asymptotic variance is $\hat{V}_{TSLS} = n\left(\sum_{i=1}^{n}z_{i}X_{i}'\right)^{-1}\left(\sum_{i=1}^{n}z_{i}z_{i}'\hat{U}_{i}^{2}\right)\left(\sum_{i=1}^{n}X_{i}z_{i}'\right)^{-1},$ where $\hat{U}_{i} = \log Q_{i} - \hat{\gamma}'X_{i}$. Standard errors for DWL are based on the delta method.
    \item \textit{Bayesian Bootstrap (BB).} BB is a weighted bootstrap: for each repetition, sample $V_{1},...,V_{n} \overset{iid}{\sim} Expo(1)$ and compute $\gamma_{BB} = (\sum_{i=1}^{n}V_{i}z_{i}X_{i}')^{-1}\sum_{i=1}^{n}V_{i}z_{i}\log Q_{i}$. The output of $S$ bootstrap repetitions $\{\gamma^{[s]}_{BB}\}_{s=1}^{S}$ is a sample from the BB posterior for $\gamma$, and $\{DWL(\gamma^{[s]}_{BB},p^{0},p^{1},y)\}_{s=1}^{S}$ is a sample from the BB DWL posterior. Bayesian bootstrapping the Wald estimand conforms to the IV example in \cite{chamberlain2003nonparametric}. The BB posterior median is the point estimator for DWL, the DWL credible interval is the BB 95\% equitailed probability interval, and I set $S = 5,001$.
    \item \textit{Plug-In Efficient (PE).} The PE estimator is computed as follows: 1. calculate $\hat{\gamma}_{TSLS}$, and estimate $E[\log P | Z=z_{i}]$ and $E[U^{2}|Z=z_{i}]$ using the fitted values from series regressions of $\log P$ and $\hat{U}^{2}$, where $\hat{U} = \log Q - \tilde{\gamma}'Z$, on $Z$, respectively, and 2. compute $\hat{\gamma}_{PE} = (\sum_{i=1}^{n}\hat{\omega}(z_{i},\tilde{\gamma})\hat{M}(z_{i})X_{i}')^{-1}\sum_{i=1}^{n}\hat{\omega}(z_{i},\tilde{\gamma})\hat{M}(z_{i})\log Q_{i}$, where $\hat{M}(z) = (1,\hat{E}[\log P |Z=z],\log Y)'$ and $\hat{\omega}(z,\gamma) = 1/\hat{E}[U^{2}|Z=z]$ with $\hat{E}[U^{2}|Z]$ being the series regression estimator of $E[U^{2}|Z]$.\footnote{Due to divide-by-zero issues, $\hat{E}[U^{2}|Z]$ is computed as $\max\{\hat{E}[U^{2}|Z],10^{-3}\}$ for the PE estimator.} The plug-in estimator of the \cite{chamberlain1987asymptotic} asymptotic variance is computed using the formula $\hat{V}_{PE} = (n^{-1}\sum_{i=1}^{n}\hat{\omega}(z_{i},\tilde{\gamma})\hat{M}(z_{i})\hat{M}(z_{i})')^{-1}$ and standard errors for the DWL are based on the delta method. I use the polynomial basis $\{(\log y)^{j}(\log a)^{k}\}_{j,k=0}^{J}$ with ad hoc choices $J = 2$ for $n=1,000$ and $J=3$ for $n=2,500$.
    \item \textit{Bayesian ETEL (BETEL).} The BETEL posterior satisfies 
    \begin{align*}
    \pi(\gamma|W^{(n)}) \propto \prod_{i=1}^{n}\hat{p}_{i}(\gamma)\pi_{\Gamma}(\gamma)\mathbf{1}\{\gamma \in \text{int}(Co(\gamma))\},    
    \end{align*} 
    where 
    \begin{align*}
     \hat{p}_{i}(\gamma) =\frac{\exp(\hat{\lambda}(\gamma)'g(W_{i},z_{i},\gamma)h(z_{i}))}{\sum_{j=1}^{n}\exp(\hat{\lambda}(\gamma)'g(W_{j},z_{j},\gamma)h(z_{j}))},   
    \end{align*} 
    for each $i=1,...,n$, $\hat{\lambda}(\gamma)=\argmin_{\gamma \in \mathbb{R}^{d_{h}}}n^{-1} \sum_{i=1}^{n}\exp(\lambda'g(W_{i},z_{i},\gamma)h(z_{i}))$, $h(z)$ is a $d_{h}\times 1$ vector of transformations of $z$, $\pi_{\Gamma}$ is a prior for $\gamma$, and $Co(\gamma)$ is the convex hull of $\{g(W_{i},z_{i},\gamma)h(z_{i})\}_{i=1}^{n}$. I set $h(z)$ to the natural cubic spline basis from p. 744 of \cite{chib2022bayesian} with their recommended choice of $\lfloor 2\cdot n^{1/6} \rfloor$ knots, and use the \cite{chib2018bayesian,chib2022bayesian} default Student's $t$ prior for $\gamma$ in both cases. I use the posterior median as the point estimator and the 95\% equitailed probability interval to quantify uncertainty about DWL. Results are based on $50,000$ MH draws, discarding the first half as burn-in, and saving every fifth thereafter.
    \end{enumerate}
\subsubsection{Simulation Results}
Table \ref{tab:simulation} reports the simulation results based on $500$ independent datasets. The MAE and Bias columns report the median absolute error (MAE) and median bias (Bias), respectively, of the point estimators.\footnote{I report MAE instead of mean-squared error (and median bias instead of mean bias) because TSLS has no moments in the just-identified case with continuous data.} SB has noticeably smaller MAE when $n=1,000$. When $n=2,500$, SB has the smallest MAE among procedures that estimate the \cite{chamberlain1987asymptotic} optimal IVs (i.e., SB, PE, and BETEL), and is comparable to TSLS and BB. Regarding uncertainty quantification, when $n=1,000$, the procedures based on the optimal IVs exhibit more severe (frequentist) undercoverage than BB and TSLS. However, when $n=2,500$, SB achieves frequentist coverage close to the nominal level of 95\%, and, across all of the procedures, SB credible intervals have the shortest length. Although specific to one DGP, these results are evidence that SB has solid finite-sample frequentist performance in an empirically relevant setting (i.e., inference for counterfactuals).
\begin{table}
\centering
\caption{Simulation Results}\label{tab:simulation}
\resizebox{\textwidth}{!}{  \begin{threeparttable}
\begin{tabular}{l rc ccccc cccc }
\toprule
 &  & & \multicolumn{4}{c}{\underline{Low Income}} & &  \multicolumn{4}{c}{\underline{High Income}} \\
$n$ & Procedure & & MAE & Bias & CP (\%) & Length & & MAE & Bias & CP (\%) & Length \\
\midrule
\multirow{5}{*}{$1,000$}
  & SB & & $5.70$ & $0.34$ & $92.6$ & $30.82$ & & $6.90$ & $0.13$ & $91.6$ & $36.18$   \\
  & TSLS & & $6.22$ & $-0.11$ & $93.8$ & $34.29$ & & $7.47$ & $0.02$ & $93.6$ &  $40.44$   \\
  & BB &  & $6.17$ & $-0.05$ & $94.0$ & $34.31$ & & $7.41$ & $-0.02$ & $93.4$ & $40.47$    \\
   & PE & & $6.37$ & $-0.38$ & $91.6$ & $32.28$ & & $7.36$ & $-0.47$ & $91.8$ & $38.14$   \\
  & BETEL & & $6.31$ & $-0.26$ & $91.6$ & $33.15$ & & $7.41$ & $-0.25$ & $92.4$ & $39.14$   \\
\midrule
\multirow{5}{*}{$2,500$}
  & SB & & $3.61$ & $-0.12$ & $96.4$ & $20.23$ & & $4.39$ & $0.01$ & $95.6$ &  $23.73$  \\
  & TSLS & & $3.61$ & $-0.14$ & $95.6$ & $21.89$ & & $4.35$ & $-0.02$ & $95.4$ & $25.78$    \\
  & BB & & $3.60$ & $-0.19$ & $96.0$ & $21.94$ & & $4.36$ & $-0.01$ & $95.6$ &  $25.76$   \\
   & PE & & $3.83$ & $-0.14$ & $91.6$ & $20.60$ & & $4.55$ & $-0.18$ & $92.8$ & $24.36$   \\
  & BETEL & & $3.72$ & $0.07$ & $90.2$ & $21.38$ & & $4.44$ & $0.23$ & $89.2$ & $25.27$  \\
\bottomrule
\end{tabular}
\begin{tablenotes}
      \small
      \item \textit{Notes.} MAE is the median absolute difference between the point estimate and the true parameter value, Bias is the difference between the median of the point estimators and the truth, CP is the frequency with which the interval estimator contains the truth, and Length is the median length.
    \end{tablenotes}
  \end{threeparttable}
  }
\end{table}

\section{Proofs of Theorems and Corollaries}\label{ap:proof}
\subsection{Proof of Theorem \ref{thm:representation}}
I state technical lemmas. Proofs are in \cite{walker2026supplement}.
\begin{lemma}\label{lem:covariance2}
Suppose Assumptions \ref{as:covariance}.1 and \ref{as:concentration}.3 hold. Then, for $P_{0,Z}^{\infty}$-almost every fixed realization $\{z_{i}\}_{i \geq 1}$ of $\{Z_{i}\}_{i \geq 1}$,
\begin{align*}
\sup_{(\gamma,p_{W|Z}) \in \Gamma \times \tilde{\mathcal{P}}_{n,W|Z}}\max_{1 \leq i \leq n}||\Sigma^{-1}(z_{i},\gamma)||_{op} = O(1),   
\end{align*}
where $\{\tilde{\mathcal{P}}_{n,W|Z}\}_{n \geq 1}$ are the sets introduced in Assumption \ref{as:concentration}.
\end{lemma}
\begin{lemma}\label{lem:pointwiseconsistency}
Suppose that Assumptions \ref{as:dgp2}, \ref{as:correctspec}, \ref{as:moments}.1, \ref{as:covariance}.1, \ref{as:criterion}, \ref{as:concentration}.1, and \ref{as:concentration}.3 hold. Then, for $P_{0,Z}^{\infty}$-almost every fixed realization $\{z_{i}\}_{i \geq 1}$ of $\{Z_{i}\}_{i \geq 1}$,
\begin{align*}
 \limsup_{n\rightarrow \infty}\sup_{(\tilde{\gamma},p_{W|Z}) \in \Gamma \times \mathcal{P}_{n,W|Z}}||q_{n}(\tilde{\gamma},p_{W|Z})-\gamma_{0}||_{2} = 0
\end{align*}
where $\{\tilde{\mathcal{P}}_{n,W|Z}\}_{n \geq 1}$ are the sets defined in Assumption \ref{as:concentration}.
\end{lemma}
\begin{lemma}\label{lem:UCM}
If Assumptions \ref{as:dgp2}, \ref{as:correctspec}, \ref{as:moments}, \ref{as:covariance}.1, \ref{as:criterion}, \ref{as:localident}, and \ref{as:concentration} hold, then, for $P_{0,Z}^{\infty}$-almost every fixed realization $\{z_{i}\}_{i \geq 1}$ of $\{Z_{i}\}_{i\geq 1}$ and $n$ large, the function $\tilde{\gamma} \mapsto q_{n}(\tilde{\gamma},p_{W|Z})$ is a uniform contraction mapping over $\tilde{\mathcal{P}}_{n,W|Z}$. That is, there exists a constant $\rho \in [0,1)$ and $N \geq 1$ such that $||q_{n}(\tilde{\gamma}_{1},p_{W|Z})-q_{n}(\tilde{\gamma}_{2},p_{W|Z})||_{2} \leq \rho ||\tilde{\gamma}_{1}-\tilde{\gamma}_{2}||_{2}$ for all $\tilde{\gamma}_{1},\tilde{\gamma}_{2} \in \Gamma$, all $p_{W|Z} \in \tilde{\mathcal{P}}_{n,W|Z}$, and $n \geq N$. 
\end{lemma}
\begin{lemma}\label{lem:consistency}
If Assumptions \ref{as:uemfp}, \ref{as:dgp2}, \ref{as:correctspec}, \ref{as:moments}, \ref{as:covariance}.1, \ref{as:criterion}, \ref{as:localident}, and \ref{as:concentration} hold, then, for $P_{0,Z}^{\infty}$-almost every fixed realization $\{z_{i}\}_{i \geq 1}$ of $\{Z_{i}\}_{i \geq 1}$, 
\begin{align*}
\limsup_{n\rightarrow \infty}\sup_{p_{W|Z} \in \tilde{\mathcal{P}}_{n,W|Z}}||\gamma_{n}-\gamma_{0}||_{2} = 0,    
\end{align*}
where $\{\tilde{\mathcal{P}}_{n,W|Z}\}_{n \geq 1}$ are the sets defined in Assumption \ref{as:concentration}.
\end{lemma}
\begin{lemma}\label{lem:jacobiancross}
Suppose that Assumptions \ref{as:uemfp}, \ref{as:dgp2}, \ref{as:correctspec}, \ref{as:moments}, \ref{as:covariance}, \ref{as:criterion}, \ref{as:localident}, and \ref{as:concentration} hold. Then, for $P_{0,Z}^{\infty}$-almost every realization $\{z_{i}\}_{i \geq 1}$ of $\{Z_{i}\}_{i \geq 1}$,
\begin{align*}
\sup_{p_{W|Z}\in \tilde{\mathcal{P}}_{n,W|Z}}\left | \left | M(\cdot,\gamma_{n})-M_{0}(\cdot,\gamma_{0}) \right | \right|_{n,2} \lesssim o(n^{-1/4}) + \sup_{p_{W|Z} \in \tilde{\mathcal{P}}_{n,W|Z}}||\gamma_{n}-\gamma_{0}||_{2}
\end{align*}
and
\begin{align*}
\sup_{p_{W|Z}\in \tilde{\mathcal{P}}_{n,W|Z}}\left | \left | \Sigma(\cdot,\gamma_{n})-\Sigma_{0}(\cdot,\gamma_{0}) \right | \right|_{n,\infty} \lesssim o(n^{-1/4}) + \sup_{p_{W|Z} \in \tilde{\mathcal{P}}_{n,W|Z}}||\gamma_{n}-\gamma_{0}||_{2}
\end{align*}
where $\{\tilde{\mathcal{P}}_{n,W|Z}\}_{n \geq 1}$ are the sets defined in Assumption \ref{as:concentration}.
\end{lemma}
\begin{lemma}\label{lem:gram}
Let $V_{n}(\gamma,\bar{\gamma},p_{W|Z}) = n^{-1}\sum_{i=1}^{n}M(z_{i},\gamma)'\Sigma^{-1}(z_{i},\gamma)M(z_{i},\bar{\gamma})$. If Assumptions \ref{as:uemfp}, \ref{as:dgp2}, \ref{as:correctspec}, \ref{as:moments}, \ref{as:covariance}, \ref{as:criterion}, \ref{as:localident}, and \ref{as:concentration} hold, then, for $P_{0,Z}^{\infty}$-almost every fixed realization $\{z_{i}\}_{i \geq 1}$ of $\{Z_{i}\}_{i \geq 1}$ and for every positive sequence $\{r_{n}\}_{n \geq 1}$ with $r_{n} \rightarrow 0$ as $n\rightarrow \infty$, the collection of matrices $\{V_{n}(\gamma,\bar{\gamma},p_{W|Z}): (\gamma,\bar{\gamma},p_{W|Z}) \in B(0,r_{n})\times B(0,r_{n})\times \tilde{\mathcal{P}}_{n,W|Z}\}$ forms a class of nonsingular matrices for $n$ large, where $\{\tilde{\mathcal{P}}_{n,W|Z}\}_{n \geq 1}$ are defined in Assumption \ref{as:concentration} and $\{B(\gamma_{0},r_{n})\}_{n \geq 1}$ are open balls in $(\mathbb{R}^{d_{\gamma}},||\cdot||_{2})$ centered at $\gamma_{0}$ and with radii $\{r_{n}\}_{n \geq 1}$. Moreover, $\sup_{(\gamma,\bar{\gamma},p_{W|Z}) \in B(\gamma_{0},r_{n})\times B(\gamma_{0},r_{n})\times \tilde{\mathcal{P}}_{n,W|Z}}||V_{n}^{-1}(\gamma,\bar{\gamma},p_{W|Z})||_{op} = O(1)$.
\end{lemma}
\begin{lemma}\label{lem:contraction}
Suppose that Assumptions \ref{as:uemfp}, \ref{as:dgp2}, \ref{as:correctspec}, \ref{as:moments}, \ref{as:covariance}, \ref{as:criterion}, \ref{as:localident}, and \ref{as:concentration}.1--\ref{as:concentration}.5 hold. Then, for $P_{0,Z}^{\infty}$-almost every fixed realization $\{z_{i}\}_{i \geq 1}$ of $\{Z_{i}\}_{i \geq 1}$, the following holds:
\begin{align*}
\limsup_{n\rightarrow \infty}\sup_{p_{W|Z} \in \tilde{\mathcal{P}}_{n,W|Z}}||n^{\frac{1}{4}}(\gamma_{n}-\gamma_{0})||_{2} = 0,    
\end{align*}
where $\{\tilde{\mathcal{P}}_{n,W|Z}\}_{n \geq 1}$ are the sets defined in Assumption \ref{as:concentration}.
\end{lemma}
\begin{proof}[Proof of Theorem \ref{thm:representation}]
 The proof has five steps. Step 1 shows that the lemma holds if four conditions are met. Steps 2--5 are verify these conditions in order to prove the theorem.
\paragraph{Step 1.} Let $\{z_{i}\}_{i\geq 1}$ be an arbitrary fixed realization of $\{Z_{i}\}_{i \geq 1}$. Since $\gamma_{n}$ satisfies the first order conditions
\begin{align*}
    \frac{1}{n}\sum_{i=1}^{n}M(z_{i},\gamma_{n})'\Sigma(z_{i},\gamma_{n})^{-1}m(z_{i},\gamma_{n}) = 0
\end{align*}
and $\gamma_{n} \in \Gamma_{0}$ for $n$ large and for all $p_{W|Z} \in \tilde{\mathcal{P}}_{n,W|Z}$, I expand $m(z_{i},\gamma_{n})$ around $\gamma_{0}$ (using element by element mean value expansions) to obtain
\begin{align*}
0=\frac{1}{n}\sum_{i=1}^{n}M(z_{i},\gamma_{n})'\Sigma^{-1}(z_{i},\gamma_{n})m(z_{i},\gamma_{0})+\frac{1}{n}\sum_{i=1}^{n}M(z_{i},\gamma_{n})'\Sigma^{-1}(z_{i},\gamma_{n})M(z_{i},\gamma_{n}^{\dagger})(\gamma_{n}-\gamma_{0}),
\end{align*}
where $\gamma_{n}^{\dagger}$ is on the line segment joining $\gamma_{n}$ and $\gamma_{0}$ (and takes different values in each row of $M(z_{i},\gamma_{n}^{\dagger})$). Lemma \ref{lem:consistency} and \ref{lem:gram} then implies that
\begin{align*}
\sqrt{n}(\gamma_{n}-\gamma_{0}) = -V_{n}^{-1}\frac{1}{\sqrt{n}}\sum_{i=1}^{n}M(z_{i},\gamma_{n})'\Sigma^{-1}(z_{i},\gamma_{n})m(z_{i},\gamma_{0}). 
\end{align*}
for large $n$, where $V_{n}=n^{-1}\sum_{i=1}^{n}M(z_{i},\gamma_{n})'\Sigma^{-1}(z_{i},\gamma_{n})M(z_{i},\gamma_{n}^{\dagger})$. By definition of $\tilde{\chi}_{0}$,
\begin{align*}
    \frac{1}{\sqrt{n}}\sum_{i=1}^{n}E_{P_{W|Z}}[\tilde{\chi}_{0}(W,Z)|Z=z_{i}] = -V_{0}^{-1}\frac{1}{\sqrt{n}}\sum_{i=1}^{n}M_{0}(z_{i},\gamma_{0})'\Sigma_{0}^{-1}(z_{i},\gamma_{0})m(z_{i},\gamma_{0}).
\end{align*}
Consequently, the triangle inequality, the matrix identity $A^{-1}-B^{-1}= B^{-1}(B-A)A^{-1}$, and submultiplicativity of operator norms reveals
\begin{align*}
&\left | \left | \sqrt{n}(\gamma_{n}-\gamma_{0}) - \frac{1}{\sqrt{n}}\sum_{i=1}^{n}E_{P_{W|Z}}[\tilde{\chi}_{0}(W,Z)|Z=z_{i}] \right | \right |_{2} \\
&\leq \left| \left|V_{n}^{-1} \right| \right|_{op} \cdot \bigg ( \left|\left| \frac{1}{\sqrt{n}}\sum_{i=1}^{n}\left\{(M(z_{i},\gamma_{n})'\Sigma^{-1}(z_{i},\gamma_{n})- M_{0}(z_{i},\gamma_{0})'\Sigma_{0}^{-1}(z_{i},\gamma_{0}))m(z_{i},\gamma_{0})\right\}\right| \right|_{2} \\
&\quad \quad \quad \quad \quad \quad \quad \quad + n^{\frac{1}{4}}||V_{n}-V_{0}||_{op} \cdot ||V_{0}^{-1}||_{op} \cdot n^{\frac{1}{4}}\left| \left|\frac{1}{n}\sum_{i=1}^{n}M_{0}(z_{i},\gamma_{0})'\Sigma_{0}^{-1}(z_{i},\gamma_{0})m(z_{i},\gamma_{0}) \right| \right|_{2}\bigg)
\end{align*}
for all $p_{W|Z} \in \tilde{\mathcal{P}}_{n,W|Z}$. This indicates that I need to show
\begin{align}
&\sup_{p_{W|Z} \in \tilde{\mathcal{P}}_{n,W|Z}}\left|\left| \frac{1}{\sqrt{n}}\sum_{i=1}^{n}(M(z_{i},\gamma_{n})'\Sigma(z_{i},\gamma_{n})^{-1} -M_{0}(z_{i},\gamma_{0})'\Sigma_{0}(z_{i},\gamma_{0})^{-1})m(z_{i},\gamma_{0})\right| \right|_{2}=o(1) \label{eq:tech1} \\
& \sup_{p_{W|Z} \in \tilde{\mathcal{P}}_{n,W|Z}}n^{\frac{1}{4}}\left|\left|\frac{1}{n}\sum_{i=1}^{n}M_{0}(z_{i},\gamma_{0})'\Sigma_{0}^{-1}(z_{i},\gamma_{0})m(z_{i},\gamma_{0}) \right| \right|_{2} = o(1) \label{eq:tech11} \\
& \sup_{p_{W|Z} \in \tilde{\mathcal{P}}_{n,W|Z}}n^{\frac{1}{4}}\left | \left |V_{n} - V_{0} \right| \right|_{op}= o(1) \label{eq:tech2} \\
& \sup_{p_{W|Z} \in \tilde{\mathcal{P}}_{n,W|Z}}\left|\left|V_{n}^{-1}\right|\right|_{op} = O(1) \label{eq:tech22}.
\end{align}
\paragraph{Step 2.} I verify (\ref{eq:tech1}). Adding/subtracting $M_{0}(z_{i},\gamma_{0})$ and $\Sigma_{0}^{-1}(z_{i},\gamma_{0})$ appropriately, using that $m_{0}(z_{i},\gamma_{0}) = 0$ (Assumption \ref{as:correctspec}), and applying the triangle inequality, it suffices to show
\begin{align}\label{eq:tech1a}
\sup_{p_{W|Z} \in \tilde{\mathcal{P}}_{n,W|Z}}\frac{1}{\sqrt{n}}\sum_{i=1}^{n}\left | \left | (M(z_{i},\gamma_{n})-M_{0}(z_{i},\gamma_{0}))'\Sigma^{-1}(z_{i},\gamma_{n})\left(m(z_{i},\gamma_{0})-m_{0}(z_{i},\gamma_{0})\right) \right | \right |_{2}  =o(1)
\end{align}
and
\begin{align}\label{eq:tech1b}
\sup_{p_{W|Z} \in \tilde{\mathcal{P}}_{n,W|Z}}\frac{1}{\sqrt{n}}\sum_{i=1}^{n}\left | \left  |M_{0}(z_{i},\gamma_{0})'\left(\Sigma^{-1}(z_{i},\gamma_{n})-\Sigma_{0}^{-1}(z_{i},\gamma_{0})\right)\left(m(z_{i},\gamma_{0})-m_{0}(z_{i},\gamma_{0})\right) \right  | \right |_{2} =o(1).
\end{align}
For (\ref{eq:tech1a}), I use that $||Ax||_{2} \leq ||A||_{op} ||x||_{2}$ for a matrix $A$ and vector $x$ to obtain
\begin{align*}
&\left | \left | (M(z_{i},\gamma_{n})-M_{0}(z_{i},\gamma_{0}))'\Sigma(z_{i},\gamma_{n})^{-1}\left(m(z_{i},\gamma_{0})-m_{0}(z_{i},\gamma_{0})\right) \right | \right |_{2} \\
&\quad \leq ||M(z_{i},\gamma_{n})-M_{0}(z_{i},\gamma_{0})||_{op} \cdot || \Sigma^{-1}(z_{i},\gamma_{n})(m(z_{i},\gamma_{0})-m_{0}(z_{i},\gamma_{0}))||_{2} \\
&\quad \leq ||M(z_{i},\gamma_{n})-M_{0}(z_{i},\gamma_{0})||_{op} \cdot ||\Sigma^{-1}(z_{i},\gamma_{n})||_{op} \cdot ||m(z_{i},\gamma_{0})-m_{0}(z_{i},\gamma_{0})||_{2}
\end{align*}
for $i=1,...,n$ and for all $ n\geq 1$. Consequently, by H\"{o}lder's inequality,
\begin{align*}
&\sup_{p_{W|Z} \in \tilde{\mathcal{P}}_{n,W|Z}}\frac{1}{\sqrt{n}}\sum_{i=1}^{n}\left | \left | (M(z_{i},\gamma_{n})-M_{0}(z_{i},\gamma_{0}))'\Sigma^{-1}(z_{i},\gamma_{n})\left(m(z_{i},\gamma_{0})-m_{0}(z_{i},\gamma_{0})\right) \right | \right |_{2} \\
&\quad \leq \sup_{p_{W|Z} \in \tilde{\mathcal{P}}_{n,W|Z}}||\Sigma^{-1}(\cdot,\gamma_{n})||_{n,\infty} \times n^{\frac{1}{4}}\sup_{p_{W|Z} \in \tilde{\mathcal{P}}_{n,W|Z}}||M(\cdot,\gamma_{n})-M_{0}(\cdot,\gamma_{0})||_{n,2} \\
&\quad \quad \times n^{\frac{1}{4}}\sup_{p_{W|Z} \in \tilde{\mathcal{P}}_{n,W|Z}}||m(\cdot,\gamma_{0})-m_{0}(\cdot,\gamma_{0})||_{n,2} \\
&\quad = o(1)
\end{align*}
with the last equality holding by Lemma \ref{lem:covariance2}, Assumption \ref{as:concentration}.3(a), Lemma \ref{lem:jacobiancross}, and Lemma \ref{lem:contraction}. For (\ref{eq:tech1b}), I apply the inequality $||Ax||_{2} \leq ||A||_{op} ||x||_{2}$ for a matrix $A$ and vector $x$ as well as the identity $B^{-1}-A^{-1} = A^{-1}(A-B)B^{-1}$ for square matrices $A$ and $B$ to obtain
\begin{align*}
    &\left | \left  |M_{0}(z_{i},\gamma_{0})'\left(\Sigma^{-1}(z_{i},\gamma_{n})-\Sigma_{0}^{-1}(z_{i},\gamma_{0})\right)\left(m(z_{i},\gamma_{0})-m_{0}(z_{i},\gamma_{0})\right) \right  | \right |_{2}\\
    &\quad \leq ||M_{0}(z_{i},\gamma_{0})||_{op}\cdot ||\Sigma^{-1}(z_{i},\gamma_{n})||_{op} \cdot ||\Sigma(z_{i},\gamma_{n})-\Sigma_{0}(z_{i},\gamma_{0})||_{op} \\
    &\quad \quad \times ||\Sigma_{0}^{-1}(z_{i},\gamma_{0})||_{op} \cdot ||m(z_{i},\gamma_{0})-m_{0}(z_{i},\gamma_{0})||_{2}.
\end{align*}
Assumption \ref{as:moments}.3 and \ref{as:covariance} implies $\sup_{z \in \mathcal{Z}}||\Sigma_{0}^{-1}(z,\gamma_{0})||_{op}< \infty$ and $\sup_{z \in \mathcal{Z}}||M_{0}(z,\gamma_{0})||_{op} < \infty$. Consequently, using this result and $||Ax||_{2} \leq ||A||_{op} ||x||_{2}$ for a matrix $A$ and vector $x$,
\begin{align*}
&\left | \left  |M_{0}(z_{i},\gamma_{0})'\left(\Sigma^{-1}(z_{i},\gamma_{n})-\Sigma_{0}^{-1}(z_{i},\gamma_{0})\right)\left(m(z_{i},\gamma_{0})-m_{0}(z_{i},\gamma_{0})\right) \right  | \right |_{2} \\
&\quad \leq C \cdot  ||\Sigma^{-1}(z_{i},\gamma_{n})||_{op}\cdot ||\Sigma(z_{i},\gamma_{n})-\Sigma_{0}(z_{i},\gamma_{0})||_{op}||m(z_{i},\gamma_{0})-m_{0}(z_{i},\gamma_{0})||_{2}.   
\end{align*}
for $i=1,...,n$, where $C = \sup_{z \in \mathcal{Z}}||\Sigma_{0}^{-1}(z,\gamma_{0})||_{op} \cdot \sup_{z \in \mathcal{Z}}||M_{0}(z,\gamma_{0})||_{op} < \infty$. An application of H\"{o}lder's inequality then implies
\begin{align*}
&\sup_{p_{W|Z} \in  \tilde{\mathcal{P}}_{n,W|Z}}\frac{1}{\sqrt{n}}\sum_{i=1}^{n}\left | \left  |M_{0}(z_{i},\gamma_{0})'\left(\Sigma^{-1}(z_{i},\gamma_{n})-\Sigma_{0}^{-1}(z_{i},\gamma_{0})\right)\left(m(z_{i},\gamma_{0})-m_{0}(z_{i},\gamma_{0})\right) \right  | \right |_{2} \\
&\leq C\sup_{p_{W|Z} \in \tilde{\mathcal{P}}_{n,W|Z}}||\Sigma^{-1}(\cdot,\gamma_{n})||_{n,\infty} \times n^{\frac{1}{4}}\sup_{p_{W|Z} \in \tilde{\mathcal{P}}_{n,W|Z}}||\Sigma(\cdot,\gamma_{n})-\Sigma_{0}(\cdot,\gamma_{0})||_{n,\infty}\\
&\quad \times n^{\frac{1}{4}}\sup_{p_{W|Z} \in \tilde{\mathcal{P}}_{n,W|Z}}||m(\cdot,\gamma_{0})-m_{0}(\cdot,\gamma_{0})||_{n,2} \\
&= o(1),
\end{align*}
with the last equality holding by Lemma \ref{lem:covariance2}, Assumption \ref{as:concentration}.3(a), Lemma \ref{lem:jacobiancross}, and Lemma \ref{lem:contraction}.
\paragraph{Step 3.} I verify (\ref{eq:tech11}). Since $m_{0}(z_{i},\gamma_{0}) = 0$ for all $i$ by Assumption \ref{as:correctspec}, $\sup_{z \in \mathcal{Z}}||\Sigma_{0}^{-1}(z,\gamma_{0})||_{op} < \infty$ by Assumption \ref{as:covariance}, and $\sup_{z \in \mathcal{Z}}||M_{0}(z,\gamma_{0})||_{op}< \infty$ by Assumption \ref{as:moments}, it follows from the triangle inequality, H\"{o}lder's inequality, and $||Ax||_{2} \leq ||A||_{op}\cdot ||x||_{2}$ that
\begin{align*}
\sup_{p_{W|Z} \in \tilde{\mathcal{P}}_{n,W|Z}}\left|\left|\frac{1}{n}\sum_{i=1}^{n}M_{0}(z_{i},\gamma_{0})'\Sigma_{0}^{-1}(z_{i},\gamma_{0})m(z_{i},\gamma_{0})\right| \right|_{2} &\leq C\sup_{p_{W|Z} \in \tilde{\mathcal{P}}_{n,W|Z}}||m(\cdot,\gamma_{0})-m_{0}(\cdot,\gamma_{0})||_{n,2} \\
&=o(n^{-\frac{1}{4}}),
\end{align*}
where $C$ is the same as Step 2, and $o(n^{-1/4})$ by Assumption \ref{as:concentration}.3(a).

\paragraph{Step 4.} For (\ref{eq:tech2}), I first show that $\sup_{p_{W|Z} \in \tilde{\mathcal{P}}_{n,W|Z}}||V_{n}-V_{n,0}||_{op} = o(n^{-\frac{1}{4}})$, where $V_{n,0}=n^{-1}\sum_{i=1}^{n}M_{0}(z_{i},\gamma_{0})'\Sigma_{0}^{-1}(z_{i},\gamma_{0})M_{0}(z_{i},\gamma_{0})$. Adding/subtracting $M_{0}(z_{i},\gamma_{0})$ and $\Sigma_{0}^{-1}(z_{i},\gamma_{0})$,
\begin{align*}
 V_{n}-V_{n,0}&=\frac{1}{n}\sum_{i=1}^{n}(M(z_{i},\gamma_{n})-M_{0}(z_{i},\gamma_{0}))'\Sigma^{-1}(z_{i},\gamma_{n})M(z_{i},\gamma_{n}^{\dagger}) \\
 &\quad +\frac{1}{n}\sum_{i=1}^{n}M_{0}(z_{i},\gamma_{0})'\Sigma^{-1}(z_{i},\gamma_{n})(M(z_{i},\gamma_{n}^{\dagger})-M_{0}(z_{i},\gamma_{0})) \\
 &\quad + \frac{1}{n}\sum_{i=1}^{n}M_{0}(z_{i},\gamma_{0})'(\Sigma^{-1}(z_{i},\gamma_{n})-\Sigma_{0}^{-1}(z_{i},\gamma_{0}))M_{0}(z_{i},\gamma_{0}).
\end{align*}
The triangle inequality then reveals that it suffices to show three results
\begin{align*}
 &\sup_{p_{W|Z} \in  \tilde{\mathcal{P}}_{n,W|Z}} \frac{1}{n}\sum_{i=1}^{n} \left| \left |M_{0}(z_{i},\gamma_{0})'(\Sigma(z_{i},\gamma_{n})^{-1}-\Sigma_{0}(z_{i},\gamma_{0})^{-1})M_{0}(z_{i},\gamma_{0}) \right | \right|_{op} = o(n^{-\frac{1}{4}}) \\
  &\sup_{p_{W|Z} \in  \tilde{\mathcal{P}}_{n,W|Z}} \frac{1}{n}\sum_{i=1}^{n} \left| \left |M_{0}(z_{i},\gamma_{0})'\Sigma^{-1}(z_{i},\gamma_{n})(M(z_{i},\gamma_{n}^{\dagger})-M_{0}(z_{i},\gamma_{0})) \right | \right|_{op} = o(n^{-\frac{1}{4}}) \\
  &\sup_{p_{W|Z}\in  \tilde{\mathcal{P}}_{n,W|Z}}\frac{1}{n}\sum_{i=1}^{n} \left| \left |(M(z_{i},\gamma_{n})-M_{0}(z_{i},\gamma_{0}))'\Sigma(z_{i},\gamma_{n})^{-1}M(z_{i},\gamma_{n}^{\dagger})\right | \right |_{op} = o(n^{-\frac{1}{4}})
\end{align*}
These are verified by applying symmetric arguments to Step 1 of Lemma \ref{lem:gram} (see Remark 1 in \cite{walker2026supplement}). Consequently, $\sup_{p_{W|Z} \in \tilde{\mathcal{P}}_{n,W|Z}}\left|\left|V_{n} -V_{n,0}\right|\right|_{op} = o(n^{-\frac{1}{4}})$. It remains to show that $||V_{n,0}-V_{0}||_{op} = o(n^{-1/4})$. Under Assumptions \ref{as:dgp2}, \ref{as:moments}.3, and \ref{as:covariance}.1, the Law of the Iterated Logarithm (applied to each element of $V_{n,0}- V_{0}$) implies that $||V_{n,0}-V_{0}||_{F} = o(n^{-1/4})$, where $||\cdot||_{F}$ is Frobenius norm, and $||V_{n,0}-V_{0}||_{op} = o(n^{-1/4})$ using $||\cdot||_{op} \leq ||\cdot||_{F}$.
\paragraph{Step 5.} Verification of (\ref{eq:tech22}) follows immediately from the second claim of Lemma \ref{lem:gram} because $\gamma_{n},\gamma_{n}^{\dagger} \in B(\gamma_{0},r_{n})$ for $n$ large and $r_{n}\rightarrow 0$ arbitarily slowly as $n \rightarrow \infty$.
\end{proof}
\subsection{Proof of Theorem \ref{thm:bvm}}
I start with an intermediate lemma (a proof is in \cite{walker2026supplement}).
\begin{lemma}\label{lem:likelihoodexpansion}
 Let $S_{n,i}(p_{W|Z}) = \tilde{\chi}_{0}(W_{i},z_{i})- E_{P_{W|Z}}[\tilde{\chi}_{0}(W,Z)|Z=z_{i}]$ for $i=1,...,n$. If Assumptions \ref{as:dgp2}, \ref{as:correctspec}, \ref{as:moments}, \ref{as:covariance}, \ref{as:criterion}, \ref{as:localident}, and \ref{as:concentration} hold, then, for $P_{0,Z}^{\infty}$-almost every fixed realization $\{z_{i}\}_{i \geq 1}$ of $\{Z_{i}\}_{i \geq 1}$,
\begin{align*}
    \sup_{p_{W|Z} \in \tilde{\mathcal{P}}_{n,W|Z}}\left|\log L_{n}(p_{W|Z})- \log L_{n}(p_{t,n,W|Z})- t'\frac{1}{\sqrt{n}}\sum_{i=1}^{n}S_{n,i}(p_{W|Z})-\frac{1}{2}t'V_{0}^{-1}t\right| \overset{P_{0,W|Z}^{(n)}}{\longrightarrow} 0,
\end{align*}
where $\{\tilde{\mathcal{P}}_{n,W|Z}\}_{n \geq 1}$ are the sets defined in Assumption \ref{as:concentration} and $p_{t,n,W|Z} \propto p_{W|Z}\exp(-t'\tilde{\chi}_{0}/\sqrt{n})$.
\end{lemma}
\begin{proof}[Proof of Theorem \ref{thm:bvm}]
 The proof has two steps. Step 1 establishes a preliminary convergence result. Step 2 uses this result, Theorem \ref{thm:representation}, and Lemma \ref{lem:likelihoodexpansion} to conclude the BvM.
\paragraph{Step 1.} Let $d_{TV}$ be the total variation distance. I show that, for $P_{0,Z}^{\infty}$-almost every fixed realization $\{z_{i}\}_{i \geq 1}$ of $\{Z_{i}\}_{i \geq 1}$, $d_{TV}(\Pi(p_{W|Z} \in \cdot|W^{(n)}), \Pi(p_{W|Z}\in \cdot|W^{(n)},\tilde{\mathcal{P}}_{n,W|Z})) \rightarrow 0$ in $P_{0,W|Z}^{(n)}$-probability as $n\rightarrow \infty$. The law of total probability and $A\cap \tilde{\mathcal{P}}_{n,W|Z}^{c} \subseteq \tilde{\mathcal{P}}_{n,W|Z}^{c}$ implies that
\begin{align*}
    &\sup_{A}\left|\Pi(p_{W|Z} \in A|W^{(n)})-\Pi(p_{W|Z} \in A\cap \tilde{\mathcal{P}}_{n,W|Z}|W^{(n)})\right| \leq \Pi\left(p_{W|Z} \in \tilde{\mathcal{P}}_{n,W|Z}^{c}\middle | W^{(n)}\right).
\end{align*}
Similarly, the definition of conditional probability and complement rule implies
\begin{align*}
    \sup_{A}\left | \Pi(p_{W|Z} \in A\cap \tilde{\mathcal{P}}_{n,W|Z}|W^{(n)})-\Pi\left(p_{W|Z} \in A \middle | W^{(n)}, \tilde{\mathcal{P}}_{n,W|Z}\right) \right | \leq \Pi\left(p_{W|Z} \in \tilde{\mathcal{P}}_{n,W|Z}^{c}\middle | W^{(n)}\right).
    \end{align*}
Consequently, the triangle inequality and Assumption \ref{as:concentration} leads to the conclusion that
\begin{align*}
d_{TV}\left(\Pi(p_{W|Z} \in \cdot|W^{(n)}), \Pi(p_{W|Z}\in \cdot|W^{(n)},\tilde{\mathcal{P}}_{n,W|Z})\right) \leq  \Pi\left(p_{W|Z} \in \tilde{\mathcal{P}}_{n,W|Z}^{c}\middle | W^{(n)}\right) \overset{P_{0,W|Z}^{(n)}}{\longrightarrow} 0     
\end{align*}
as $n\rightarrow \infty$, for $P_{0,Z}^{\infty}$-almost every fixed realization $\{z_{i}\}_{i \geq 1}$ of $\{Z_{i}\}_{i \geq 1}$. 
\paragraph{Step 2.} Since $d_{BL} \leq d_{TV}$ and $\Pi(\sqrt{n}(\gamma_{n}-\hat{\gamma}_{n}) \in A|W^{(n)}) = \Pi(p_{W|Z} \in \gamma_{n}^{-1}(A_{n}) |W^{(n)})$ for all events $A$ (where $A_{n} = n^{-1/2}A + \hat{\gamma}_{n}$), Step 1 implies that, for $P_{0,Z}^{\infty}$-almost every fixed realization $\{z_{i}\}_{i \geq 1}$ of $\{Z_{i}\}_{i \geq 1}$, 
\begin{align*}
d_{BL}(\Pi(\sqrt{n}(\gamma_{n}-\hat{\gamma}_{n}) \in \cdot|W^{(n)}), \Pi(\sqrt{n}(\gamma_{n}-\hat{\gamma}_{n}) \in \cdot|W^{(n)},\tilde{\mathcal{P}}_{n,W|Z})) \overset{P_{0,W|Z}^{(n)}}{\longrightarrow} 0
\end{align*} 
as $n\rightarrow \infty$. Consequently, it suffices to show, for $P_{0,Z}^{\infty}$-almost every fixed realization $\{z_{i}\}_{i \geq 1}$ of $\{Z_{i}\}_{i \geq 1}$,
\begin{align}\label{eq:restrictedconvergence}
    d_{BL}\left(\Pi\left(\sqrt{n}(\gamma_{n}-\hat{\gamma}_{n}) \in \cdot\middle |W^{(n)}, \tilde{\mathcal{P}}_{n,W|Z} \right), \mathcal{N}(0,V_{0}^{-1})\right) \overset{P_{0,W|Z}^{(n)}}{\longrightarrow} 0
\end{align}
as $n\rightarrow \infty$, by the triangle inequality. By Lemma 1 of \cite{castillo2015bernstein}, (\ref{eq:restrictedconvergence}) holds if, for $P_{0,Z}^{\infty}$-almost every fixed realization $\{z_{i}\}_{i \geq 1}$ of $\{Z_{i}\}_{i \geq 1}$,
\begin{align*}
    \left| E\left[\exp\left(t'\sqrt{n}(\gamma_{n}-\hat{\gamma}_{n})\right)\middle |W^{(n)},\tilde{\mathcal{P}}_{n,W|Z}\right] - \exp\left(\frac{1}{2}t'V_{0}^{-1}t\right) \right| \overset{P_{0,W|Z}^{(n)}}{\longrightarrow} 0
\end{align*}
as $n\rightarrow \infty$ for any $t$ in a neighborhood of $0$. Adding/subtracting $t'\sqrt{n}\gamma_{0}$ and $\log L_{n}(p_{t,n,W|Z})$, I know that, for $P_{0,Z}^{\infty}$-almost every fixed realization $\{z_{i}\}_{i \geq 1}$ of $\{Z_{i}\}_{i \geq 1}$,
\begin{align*}
&t'\sqrt{n}(\gamma_{n}-\hat{\gamma}_{n}) + \log L_{n}(p_{W|Z}) \\
&\quad = t'\sqrt{n}(\gamma_{n}-\gamma_{0})-t'\sqrt{n}(\hat{\gamma}_{n}-\gamma_{0})+\log L_{n}(p_{W|Z}) - \log L_{n}(p_{t,n,W|Z}) + \log L_{n}(p_{t,n,W|Z}) \\
&\quad = - t'\frac{1}{\sqrt{n}}\sum_{i=1}^{n}S_{n,i}(p_{W|Z}) + t'\frac{1}{\sqrt{n}}\sum_{i=1}^{n}S_{n,i}(p_{W|Z}) + \frac{1}{2}t'V_{0}^{-1}t+ \log L_{n}(p_{t,n,W|Z}) + o_{P_{0,W|Z}^{(n)}}(1) \\
&\quad = \frac{1}{2}t'V_{0}^{-1}t + \log L_{n}(p_{t,n,W|Z}) + o_{P_{0,W|Z}^{(n)}}(1)
\end{align*}
uniformly along $\tilde{\mathcal{P}}_{n,W|Z}$, with the second equality using Lemma \ref{lem:likelihoodexpansion}, Theorem \ref{thm:representation}, and the definition of $\hat{\gamma}_{n}$ and $p_{t,n,W|Z}$. Hence, for $P_{0,Z}^{\infty}$-almost every fixed realization $\{z_{i}\}_{i \geq 1}$ of $\{Z_{i}\}_{i \geq 1}$,
\begin{align*}
     &E\left[\exp\left(t'\sqrt{n}(\gamma_{n}-\hat{\gamma}_{n})\right)\middle |W^{(n)},\tilde{\mathcal{P}}_{n,W|Z}\right] \\
     &\quad = (1+o_{P_{0,W|Z}^{(n)}}(1))\exp\left(\frac{1}{2}t'V_{0}^{-1}t\right)\frac{\int_{\tilde{\mathcal{P}}_{n,W|Z}}\exp\left(\log L_{n}(p_{t,n,W|Z})\right) d\Pi(p_{W|Z})}{\int_{\tilde{\mathcal{P}}_{n,W|Z}}\exp\left(\log L_{n}(p_{W|Z})\right) d\Pi(p_{W|Z})} \\
     &\quad = (1+o_{P_{0,W|Z}^{(n)}}(1))\exp\left(\frac{1}{2}t'V_{0}^{-1}t\right)
\end{align*}
as $n\rightarrow \infty$, where the last equality applies (\ref{eq:invariancegeneral}). This completes the proof.
\end{proof}
\subsection{Proofs of Corollaries}
\begin{proof}[Proof of Corollary \ref{cor:quantile}.]
Theorem 2 and the delta method implies that $\Pi(f(\gamma_{n}) \in \cdot | W^{(n)})$ satisfies $d_{BL}(\Pi\left(\sqrt{n}(f(\gamma_{n})-f(\hat{\gamma}_{n})) \in \cdot |W^{(n)}), \mathcal{N}\left(0,\Omega_{0,f}\right)\right) \rightarrow  0$ in $P_{0,W|Z}^{(n)}$-probability for $P_{0,Z}^{\infty}$-almost every fixed realization $\{z_{i}\}_{i \geq 1}$ of $\{Z_{i}\}_{i \geq 1}$ as $n\rightarrow \infty$,
where $\Omega_{0,f} = \frac{\partial}{\partial \gamma'}f(\gamma_{0})V_{0}^{-1}\frac{\partial}{\partial \gamma'}f(\gamma_{0})'$. Lemma 2 in \cite{castillo2015bernstein} then implies that, for $P_{0,Z}^{\infty}$-almost every fixed realization $\{z_{i}\}_{i \geq 1}$ of $\{Z_{i}\}_{i \geq 1}$, $\Pi(\sqrt{n}(f(\gamma_{n})-f(\hat{\gamma}_{n})) \leq x  |W^{(n)}) = P_{\mathcal{N}(0,\Omega_{0,f})}(X \leq x) + o_{P_{0,W|Z}^{(n)}}(1)$ holds for every $x \in \mathbb{R}$ as $n\rightarrow \infty$. Applying Lemma 21.2 of \cite{vaart_1998} and the equivariance of quantiles to monotone transformations, $\sqrt{n}(c_{n,f}(q)-f(\hat{\gamma}_{n})) = z_{f}(q) + o_{P_{0,W|Z}^{(n)}}(1)$ for $P_{0,Z}^{\infty}$-almost every fixed realization $\{z_{i}\}_{i \geq 1}$ of $\{Z_{i}\}_{i \geq 1}$, where $z_{f}(q)$ is the $q$-quantile of $\mathcal{N}\left(0,\Omega_{0,f}\right) $. The claim then follows because $z_{f}(q) = \Phi^{-1}(q) \sqrt{\Omega_{0,f}}$ for any $q \in (0,1)$.
\end{proof}
\begin{proof}[Proof of Corollary \ref{cor:unconditional}.]
Let $A_{n}(\xi) = \{d_{BL}(\Pi(\sqrt{n}(\gamma_{n}-\hat{\gamma}_{n}) \in \cdot |W^{(n)}),\mathcal{N}(0,V_{0}^{-1})) \geq \xi\}$ for $\xi>0$. Theorem 2 implies $P_{0,W|Z}^{(n)}(A_{n}(\xi)) \rightarrow 0 $ as $n\rightarrow \infty$ for $P_{0,Z}^{\infty}$-almost every fixed realization $\{z_{i}\}_{i \geq 1}$ of $\{Z_{i}\}_{i \geq 1}$. The Bounded Convergence Theorem then yields that $E_{P_{0,Z}^{\infty}}[P_{0,W|Z}^{(n)}(A_{n}(\xi))] \rightarrow 0$ as $n\rightarrow \infty$. Applying the law of iterated expectations and using independence, $E_{P_{0,Z}^{\infty}}[P_{0,W|Z}^{(n)}(A_{n}(\xi))] = E_{P_{0,Z}^{n}}[P_{0,W|Z}^{(n)}(A_{n}(\xi))]=P_{0,WZ}^{n}(A_{n}(\xi))$, so $P_{0,WZ}^{n}(A_{n}(\xi)) \rightarrow 0$ as $n\rightarrow \infty$. Since $\xi > 0$ was arbitrary, the proof is complete.
\end{proof}
\begin{remark}\label{rem:uncon}
An inspection of the proof of Corollary \ref{cor:quantile} reveals that if Corollary \ref{cor:unconditional} is invoked, then the asymptotic expansion of $c_{n,f}(q)$ holds unconditionally too (i.e., replacing $P_{0,W|Z}^{(n)}$ with $P_{0,WZ}^{n}$). Hence, the equitailed probability intervals achieve unconditional asymptotic coverage equal to $1-\alpha$ too (and have shortest length).
\end{remark}
\end{document}


\maketitle
\begin{abstract}
This Online Appendix contains more discussion related to Section 3 and proofs of propositions. The accompanying Supplementary Materials contains lemma proofs.
\end{abstract}
\section{Additional Discussion Related to Section 3}
\subsection{Deadweight Loss Prior Sensitivity} Table \ref{tab:empiricalsensitivity} reports the deadweight loss (DWL) estimates for $\alpha \in \{3/2,5/2,7/2,9/2\}$. Changes to $\alpha$ have a modest effect on the estimates.
\begin{table}[H]
\centering
\caption{Deadweight Loss Estimates for Different Values of $\alpha$}\label{tab:empiricalsensitivity}
\resizebox{\textwidth}{!}{ 
\begin{threeparttable}
\begin{tabular}{l cc cc cc}
\toprule
  & \multicolumn{2}{c}{Low Income} 
  & \multicolumn{2}{c}{Middle Income} 
  & \multicolumn{2}{c}{High Income} \\
\cmidrule(lr){2-3} \cmidrule(lr){4-5} \cmidrule(lr){6-7}
  & Point Est. & Cred. Int. 
  & Point Est. & Cred. Int. 
  & Point Est. & Cred. Int.  \\
\midrule
$\alpha = 3/2$        & $31.81$ & $(23.83,39.42)$  & $34.59 $ & $(25.88,42.88)$ & $ 36.89 $ & $(27.52,45.73)$ \\
$\alpha = 5/2$        & $32.02$ & $(23.86,39.82)$  & $34.83 $ & $(25.98,43.25)$ & $ 37.15 $ & $(27.75,45.12)$ \\
$\alpha = 7/2$        & $32.63$ & $(24.81,40.42)$  & $35.55 $ & $(26.96,44.01)$ & $ 37.94 $ & $(28.79,46.97)$ \\
$\alpha = 9/2$        & $32.17$ & $(24.39,40.25)$  & $35.03 $ & $(26.50,43.87)$ & $ 37.38 $ & $(28.22,45.86)$ \\
\bottomrule
\end{tabular}
\begin{tablenotes}
\small
\item \textit{Notes:} Results are based on $100,000$ iterations of the Metropolis-within-Gibbs sampler, discarding the first half as burn-in and saving every fifth draw thereafter. Point estimates are posterior medians and interval estimates are 95\% equitailed probability intervals.
\end{tablenotes}
\end{threeparttable}
}
\end{table}
\subsection{Deadweight Loss Parameter}
I derive the DWL parameter. Following equation (1.1) in \cite{hausman1995nonparametric}, $S(p^{0},y,\gamma)$ solves the differential equation
\begin{align*}
    \frac{\partial S(t,y,\gamma)}{\partial t} = -\exp(\gamma_{0})p(t)^{\gamma_{1}}(y-S(t,y,\gamma))^{\gamma_{2}}\frac{\partial p(t)}{\partial t}, \quad S(1,y,\gamma) = 0,
\end{align*}
where $t \mapsto p(t)$, $t \in [0,1]$, is a continuously differentiable price path with $p(0)=p^{0}$ and $p(1) = p^{1}$. Consequently, one can show that
\begin{align*}
    S(p^{0},y,\gamma) = -\left(\frac{(1-\gamma_{2})}{(1+\gamma_{1})y^{\gamma_{2}}}\left( p^{0}q(p^{0},y,\gamma)-p^{1}q(p^{1},y,\gamma)\right) + y^{1-\gamma_{2}}\right)^{\frac{1}{1-\gamma_{2}}}+ y.
\end{align*}
This implies the DWL from a price change from $p^{0}$ to $p^{1}$ at income $y$ is given by
\begin{align*}
DWL(p^{0},p^{1},y,\gamma) &= y-\left(\frac{(1-\gamma_{2})}{(1+\gamma_{1})y^{\gamma_{2}}}\left( p^{0}q(p^{0},y,\gamma)-p^{1}q(p^{1},y,\gamma)\right) + y^{1-\gamma_{2}}\right)^{\frac{1}{1-\gamma_{2}}} \\
&\quad - (p^{1}-p^{0})\exp(\gamma_{0})(p^{1})^{\gamma_{1}}y^{\gamma_{2}}.
\end{align*}
Deadweight loss as a fraction of income is $DWL(p^{0},p^{1},y,\gamma)/y$, while deadweight loss as a proportion of tax paid is $DWL(p^{0},p^{1},y,\gamma) /[(p^{1}-p^{0})q(p^{1},y,\gamma)]$.
\section{Restatement of Assumptions}
I restate the assumptions for the readers convenience.
\begin{assumption}\label{as:measurability}
For each $n \geq 1$ and almost every fixed realization $\{z_{i}\}_{i \geq 1}$ of $\{Z_{i}\}_{i \geq 1}$, the mapping $(w^{(n)},p_{W|Z}) \mapsto \prod_{i=1}^{n}p_{W|Z}(w_{i}|z_{i})$ is measurable on $(\mathcal{W}^{n} \times \mathcal{P}_{W|Z}, \mathscr{W}^{\otimes n} \otimes \mathscr{P}_{W|Z})$, where $\mathscr{W}$ is a $\sigma$-algebra over $\mathcal{W}$, and $\mathscr{W}^{\otimes n}$ is the corresponding product $\sigma$-algebra.
\end{assumption}
\begin{assumption}\label{as:evidence}
Let $p_{0,W|Z}$ be the true conditional density. For each $ n \geq 1$ and almost every fixed realization $\{z_{i}\}_{i \geq 1}$ of $\{Z_{i}\}_{i \geq 1}$, $P_{0,W|Z}^{(n)}(\int \prod_{i=1}^{n}p_{W|Z}(W_{i}|z_{i})d\Pi(p_{W|Z}) \in (0,\infty)) = 1$, where $P_{0,W|Z}^{(n)} = \bigotimes_{i=1}^{n}P_{0,W|z_{i}}$ (and $P_{0,W|z_{i}}$ is the distribution attached to $p_{0,W|Z}(\cdot|z_{i})$).
\end{assumption}
\begin{assumption}\label{as:uemfp}
For almost every fixed realization $\{z_{i}\}_{i \geq 1}$ of $\{Z_{i}\}_{i \geq 1}$ and for each $n \geq 1$, the limit $\gamma_{n}$ exists for each $p_{W|Z} \in \mathcal{P}_{W|Z}$ and $p_{W|Z} \mapsto \gamma_{n}$ is  $\mathscr{P}_{W|Z}$-measurable.
\end{assumption}
\begin{assumption}\label{as:dgp2}
1. $\{(W_{i}',Z_{i}')'\}_{i \geq 1}$ is an i.i.d sequence with common distribution $P_{0,WZ} = P_{0,W|Z} \otimes P_{0,Z}$, where $P_{0,W|Z}$ is the conditional distribution of $W$ given $Z$, and $P_{0,Z}$ is the marginal distribution of $Z$; 2. $P_{0,W|Z}$ has a density $p_{0,W|Z}$.
\end{assumption}
\begin{assumption}\label{as:correctspec} 1. $\Gamma$ is a compact subset of $\mathbb{R}^{d_{\gamma}}$, $d_{\gamma} < \infty$; 2. there is a unique $\gamma_{0} \in \text{int}(\Gamma)$ such that $m_{0}(z,\gamma_{0})= 0$ for every $z \in \mathcal{Z}$, where $m_{0}(z,\gamma)$ denotes $m(z,\gamma)$ at $P_{0,W|Z}$.
\end{assumption}
\begin{assumption}\label{as:moments}
1. $\sup_{(z,\gamma) \in \mathcal{Z}\times \Gamma}||m_{0}(z,\gamma)||_{2} < \infty$; 2. $m_{0}(z,\gamma)$ is twice continuously differentiable in $\gamma$ in a neighborhood $\Gamma_{0}$ of $\gamma_{0}$ for each $z \in \mathcal{Z}$, $\sup_{(z,\gamma) \in \mathcal{Z}\times \Gamma_{0}}||M_{0}(z,\gamma)||_{op}< \infty $, and $\sup_{(z,\gamma) \in \mathcal{Z}\times \Gamma_{0}}||\partial vec(M_{0}(z,\gamma))/\partial \gamma'||_{op} < \infty$, where $M_{0}(z,\gamma)$ is the Jacobian of $m_{0}(z,\gamma)$ in $\gamma$.
\end{assumption}
\begin{assumption}\label{as:covariance}
Let $\Sigma_{0}(z,\gamma)$ denote $\Sigma(z,\gamma)$ at $P_{0,W|Z}$. 1. there exists $0 < \underline{\lambda} < \overline{\lambda}< \infty$ such that $\underline{\lambda} \leq \inf_{(z,\gamma) \in \mathcal{Z}\times \Gamma}\lambda_{min}(\Sigma_{0}(z,\gamma))\leq \sup_{(z,\gamma) \in \mathcal{Z} \times \Gamma} \lambda_{max}(\Sigma_{0}(z,\gamma)) \leq \overline{\lambda}$; 2. $\Sigma_{0}(z,\gamma)$ is continuously differentiable in $\gamma$ in a neighborhood $\Gamma_{0}$ of $\gamma_{0}$ for each $z \in \mathcal{Z}$ with Jacobian satisfying $\sup_{(z,\gamma) \in \mathcal{Z}\times \Gamma_{0}}||\partial vec(\Sigma_{0}(z,\gamma))/\partial \gamma'||_{op} < \infty$.
\end{assumption}
\begin{assumption}\label{as:criterion}
1. the class $\{m_{0}(\cdot,\gamma)'\Sigma_{0}^{-1}(\cdot,\tilde{\gamma})m_{0}(\cdot,\gamma):(\gamma,\tilde{\gamma}) \in \Gamma \times \Gamma\}$ is $P_{0,Z}$-Glivenko-Cantelli; 2. for every $\xi>0$, $\inf_{\gamma \in \mathbb{R}^{d_{\gamma}}:||\gamma-\gamma_{0}||_{2} \geq \xi} E_{P_{0,Z}}[||m_{0}(Z,\gamma)||_{2}^{2}]> 0$.
\end{assumption}
\begin{assumption}\label{as:localident}
Let $V_{0}(Z,\tilde{\gamma}) = M_{0}(Z,\gamma_{0})'\Sigma_{0}^{-1}(Z,\tilde{\gamma})M_{0}(Z,\gamma_{0})$. 1. $\{V_{0}(\cdot,\tilde{\gamma}): \tilde{\gamma} \in \Gamma\}$ is $P_{0,Z}$-Glivenko-Cantelli; 2. $E_{P_{0,Z}}[M_{0}(Z,\gamma_{0})'M_{0}(Z,\gamma_{0})]$ is positive definite.
\end{assumption}
\begin{assumption}\label{as:concentration}
For $P_{0,Z}^{\infty}$-almost every fixed realization $\{z_{i}\}_{i \geq 1}$ of $\{Z_{i}\}_{i \geq 1}$, there exists a sequence of sets $\{\tilde{\mathcal{P}}_{n,W|Z}\}_{n\geq 1}$ such that
\begin{enumerate}
    \item $\Pi(p_{W|Z} \in \tilde{\mathcal{P}}_{n,W|Z} |W^{(n)}) \longrightarrow 1$ in $P_{0,W|Z}^{(n)}$-probability as $n\rightarrow \infty$
    \item For each $p_{W|Z} \in \tilde{\mathcal{P}}_{n,W|Z}$,
    \begin{enumerate}
     \item $m(z_{i},\gamma)$, $i=1,...,n$, is twice continuously differentiable in $\gamma$ over $\Gamma_{0}$
     \item $\Sigma(z_{i},\gamma)$, $i=1,...,n$, is once continuously differentiable in $\gamma$ over $\Gamma_{0}$,
    \end{enumerate}
    where $\Gamma_{0}$ is the neighborhood of $\gamma_{0}$ from Assumption \ref{as:moments}.
    \item The following holds:
    \begin{enumerate}
\item $\sup_{(\gamma,p_{W|Z}) \in \Gamma \times \tilde{\mathcal{P}}_{n,W|Z}}||m(\cdot,\gamma)-m_{0}(\cdot,\gamma)||_{n,2} = o(n^{-\frac{1}{4}})$ \item $\sup_{(\gamma,p_{W|Z}) \in \Gamma_{0} \times \tilde{\mathcal{P}}_{n,W|Z}}||M(\cdot,\gamma)-M_{0}(\cdot,\gamma)||_{n,2} = o(n^{-\frac{1}{4}})$ \item $\sup_{(\gamma,p_{W|Z}) \in \Gamma \times \tilde{\mathcal{P}}_{n,W|Z}}||\Sigma(\cdot,\gamma)-\Sigma_{0}(\cdot,\gamma)||_{n,\infty} = o(n^{-\frac{1}{4}})$ \item $\sup_{(\gamma,p_{W|Z}) \in \Gamma_{0} \times \tilde{\mathcal{P}}_{n,W|Z}}||\partial \text{vec}(M(\cdot,\gamma)')/\partial \gamma'||_{n,2} = O(1)$ \item $\sup_{(\gamma,p_{W|Z}) \in \Gamma_{0} \times \tilde{\mathcal{P}}_{n,W|Z}}||\partial\text{vec}(\Sigma(\cdot,\gamma))/\partial \gamma'||_{n,\infty}= O(1)$ \item $\sup_{p_{W|Z} \in  \tilde{\mathcal{P}}_{n,W|Z}}\max_{1 \leq i \leq n}E_{P_{W|Z}}[\exp(|t'\tilde{\chi}_{0}(W,Z)|)|Z=z_{i}] = O(1)$
\end{enumerate}
for any $t$ in a neighborhood of zero.
\end{enumerate}
\end{assumption}
\begin{assumption}\label{as:boundedinternal}
1. $\mathcal{W}\subseteq \mathbb{R}^{d_{w}}$ is compact; 2. $\mathcal{Z} = [0,1]^{d_{z}}$; 3. $p_{0,W|Z}$ is continuous and positive on $\mathcal{W} \times \mathcal{Z}$; 4. $P_{0,Z}$ admits a Lebesgue density $p_{0,Z}$ that is bounded away from zero and infinity.
\end{assumption}
\begin{assumption}\label{as:differentiablemoments}
1. $(w,z,\gamma) \mapsto g_{j}(w,z,\gamma)$ is continuous in all arguments over $\mathcal{W}\times \mathcal{Z} \times \Gamma$ for $j=1,...,d_{g}$; 2. $g_{j}(w,z,\gamma)$, $j=1,...,d_{g}$, is differentiable in $\gamma$ over a neighborhood $\Gamma_{0}$ of $\gamma_{0}$ with first and second derivatives continuous in all arguments over $\mathcal{W} \times \mathcal{Z} \times \Gamma_{0}$.
\end{assumption}
\begin{assumption}\label{as:parametricfamily}
1. $\Theta \subseteq \mathbb{R}^{d_{\theta}}$ is compact; 2. $(w,z,\theta) \mapsto f_{\theta,j}(w_{j}|w_{1},...,w_{j-1},z)$ is continuous in all arguments for all $j=1,...,d_{w}$; 3. $f_{\theta,j}(w_{j}|w_{1},...,w_{j-1},z) > 0$ for all $(w,z,\theta) \in \mathcal{W} \times \mathcal{Z} \times \Theta$ and all $j=1,...,d_{w}$; 4. there is a constant $L>0$ such that $\sup_{(w,z) \in \mathcal{W} \times \mathcal{Z}}|\log f_{\theta_{1}}(w|z) - \log f_{\theta_{2}}(w|z)| \leq L ||\theta_{1}-\theta_{2}||_{2}$ for all $\theta_{1},\theta_{2} \in \Theta$.
\end{assumption}
\begin{assumption}\label{as:exprior}
$\Pi = \Pi_{\Theta} \otimes \Pi_{\mathcal{B}}$, where 1. $\Pi_{\Theta}$ is a probability distribution over $\Theta$ with a continuous Lebesgue density $\pi_{\Theta}$, and 2. $\Pi_{\mathcal{B}}$ is the law $\tilde{\Pi}_{\mathcal{B}}$ of a centered Gaussian process $B$ in $(C([0,1]^{d},\mathbb{R}),||\cdot||_{\infty})$ conditioned on $\{||B||_{\infty}\leq \bar{B}\}$ for some large $\bar{B} \in (0,\infty)$ with $B \sim \tilde{\Pi}_{\mathcal{B}}$ taking values in $(C^{a}([0,1]^{d},\mathbb{R}),||\cdot||_{a})$ for every $0<a<\alpha$ and $\alpha > d_{z}/2$.
\end{assumption}
\begin{assumption}\label{as:param}
The following holds for almost every realization $\{z_{i}\}_{i \geq 1}$ of $\{Z_{i}\}_{i \geq 1}$ and every $n \geq 1$, 1. $\Gamma$ is a compact subset of $(\mathbb{R}^{d_{\gamma}},||\cdot||_{2})$, 2. $p_{W|Z} \mapsto Q_{n}(\gamma,\tilde{\gamma},p_{W|Z})$ is $\mathscr{P}_{W|Z}$-measurable for each $(\gamma,\tilde{\gamma}) \in \Gamma \times \Gamma$, 3. $\gamma \mapsto Q_{n}(\gamma,\tilde{\gamma},p_{W|Z})$ is continuous for every $\tilde{\gamma} \in \Gamma$ and every $p_{W|Z} \in \mathcal{P}_{W|Z}$, and 4. $q_{n}(\tilde{\gamma},p_{W|Z})= \argmin_{\gamma \in \Gamma} Q_{n}(\gamma,\tilde{\gamma},p_{W|Z})$ is unique and satisfies $q_{n}(\tilde{\gamma},p_{W|Z}) \in \text{int}(\Gamma)$ for every $\tilde{\gamma} \in \Gamma$ and every $p_{W|Z} \in \mathcal{P}_{W|Z}$.
\end{assumption}
\begin{assumption}\label{as:fp}
The following holds for almost every realization $\{z_{i}\}_{ i\geq 1}$ of $\{Z_{i}\}_{ i\geq 1}$ and every $n \geq 1$: 1. $\gamma \mapsto m(z_{i},\gamma)$, $i=1,...,n$, is twice continuously differentiable for each $p_{W|Z} \in \mathcal{P}_{W|Z}$, 2. $\gamma \mapsto \Sigma(z_{i},\gamma)$,$i=1,...,n$, is continuously differentiable for each $p_{W|Z} \in \mathcal{P}_{W|Z}$, 3. $\inf_{\gamma\in \Gamma}\min_{1 \leq i \leq n}\lambda_{min}(\Sigma(z_{i},\gamma)) > 0$ for each $p_{W|Z} \in \mathcal{P}_{W|Z}$, and 4. $(x_{1},x_{2}) \mapsto Q_{n}(x_{1},x_{2},p_{W|Z})$ satisfies
\begin{align*}
\sup_{\tilde{\gamma} \in \Gamma}\left | \left |  \frac{\partial^{2}}{\partial x_{1} \partial x_{2}'}Q_{n}(q_{n}(\tilde{\gamma},p_{W|Z}),\tilde{\gamma},p_{W|Z})\right| \right|_{op} < \inf_{\tilde{\gamma}\in \Gamma} \lambda_{min}\left(\frac{\partial^{2}}{\partial x_{1} \partial x_{1}'}Q_{n}(q_{n}(\tilde{\gamma},p_{W|Z}),\tilde{\gamma},p_{W|Z})\right)  
\end{align*}
for every $p_{W|Z} \in \mathcal{P}_{W|Z}$.
\end{assumption}
\section{Proofs of Propositions}
\subsection{Proofs of Propositions 1--3}
\subsubsection{Proof of Proposition 1}
I start with three technical lemmas that specialize Lemma 10, Lemma 1, and Theorem 4 of \cite{ghosal2007convergence} to my paper. For notation, define Kullback-Leibler neighborhoods $$\mathcal{U}_{n}(p_{0,W|Z},\delta) = \{(\theta,B): KL_{n}(\theta,B)< \delta^{2}, \ KLV_{n}(\theta,B)< \delta^{2} \},$$ where
\begin{align*}
 KL_{n}(\theta,B) =\frac{1}{n}\sum_{i=1}^{n} E_{P_{0,W|Z}}\left[\log \left(\frac{ p_{0,W|Z}(W|Z)}{p_{\theta,B}(W|Z)}\right)\middle | Z=z_{i}\right]
\end{align*}
and
\begin{align*}
 KLV_{n}(\theta,B) = \frac{1}{n}\sum_{i=1}^{n}Var_{P_{0,W|Z}}\left[\log \left(\frac{p_{0,W|Z}(W|Z)}{p_{\theta,B}(W|Z)}\right)\middle | Z=z_{i} \right].
\end{align*}
are the average Kullback-Leibler divergence and variation, respectively.
\setcounter{lemma}{8}
\begin{lemma}\label{lem:KLevent} Suppose Assumption \ref{as:dgp2} holds. The following holds for $P_{0,Z}^{\infty}$-almost every fixed realization $\{z_{i}\}_{i \geq 1}$ of $\{Z_{i}\}_{i \geq 1}$: for every $\delta > 0$, every probability measure $\bar{\Pi}_{n}$ supported on $\mathcal{U}_{n}(p_{0,W|Z},\delta)$, and for every constant $C>0$,
\begin{align*}
P_{0,W|Z}^{(n)}\left(\int \frac{L_{n}(p_{\theta,B})}{L_{n}(p_{0,W|Z})}d \bar{\Pi}_{n}(\theta,B) \leq \exp\left(-(1+C)n \delta^{2}\right) \right) \leq \frac{1}{C^{2}n\delta^{2}}.    
\end{align*}
\end{lemma}
\begin{lemma}\label{lem:posteriorsieve}
Suppose Assumption \ref{as:dgp2} holds and for $P_{0,Z}^{\infty}$-almost every fixed realization $\{z_{i}\}_{i \geq 1}$ of $\{Z_{i}\}_{i \geq 1}$, there is a sequence $\{\delta_{n}\}_{n \geq 1}$ with $n\delta_{n}^{2}\rightarrow \infty$ and sets $\{\bar{\mathcal{P}}_{n,W|Z}\}_{n \geq 1}$ such that, for some constants $\bar{C}> (2+C)$ with $C>0$ small, $\check{C}> 0$, and $\breve{C}>0$, the prior satisfies $\Pi(\bar{\mathcal{P}}_{n,W|Z}^{c}) \leq \check{C} \exp(-\bar{C}n\delta_{n}^{2})$ and $\Pi\left(\mathcal{U}_{n}(p_{0,W|Z},\delta_{n})\right) \geq \breve{C} \exp(-n\delta_{n}^{2})$. Then, for $P_{0,Z}^{\infty}$-almost every fixed realization $\{z_{i}\}_{i \geq 1}$ of $\{Z_{i}\}_{i \geq 1}$, $\Pi(\bar{\mathcal{P}}_{n,W|Z}^{c}|W^{(n)}) \rightarrow 0$ in $P_{0,W|Z}^{(n)}$-probability as $n\rightarrow \infty$. 
\end{lemma}
\begin{lemma}\label{lem:concentration}
Suppose that Assumption \ref{as:dgp2} holds and for $P_{0,Z}^{\infty}$-almost every fixed realization $\{z_{i}\}_{i \geq 1}$ of $\{Z_{i}\}_{i \geq 1}$, there is a sequence $\{\delta_{n}\}_{n \geq 1}$ with $\delta_{n}\rightarrow 0$ and $n \delta_{n}^{2}\rightarrow \infty$ and sets $\{\bar{\mathcal{P}}_{n,W|Z}\}_{n \geq 1}$ such that, for some constants $K> 0$, $\tilde{C} >0$, and $\bar{C} > (2+C)$ with $C >0$ small, the following holds
\begin{align}
&\Pi\left(\mathcal{U}_{n}(p_{0,W|Z},\delta_{n})\right) \gtrsim \exp(-n\delta_{n}^{2}) \label{eq:neighbor} \\
 &\Pi\left(\bar{\mathcal{P}}_{n,W|Z}^{c}\right) \lesssim \exp\left(-\bar{C}n\delta_{n}^{2}\right)  \label{eq:mass}  \\
&\log N(\delta_{n}/72, \bar{\mathcal{P}}_{n,W|Z},h_{n}) \lesssim \tilde{C}n\delta_{n}^{2}. \label{eq:entropy}
\end{align}
Then, for almost-every fixed realization $\{z_{i}\}_{i \geq 1}$ of $\{Z_{i}\}_{i \geq 1}$, the posterior satisfies
\begin{align*}
    E_{P_{0,W|Z}^{(n)}}\left[\Pi(h_{n,2}(p_{\theta,B},p_{0,W|Z})\geq D\delta_{n}|W^{(n)})\right]\longrightarrow 0
\end{align*}
as $n\rightarrow \infty$ for $D>0$ large.
\end{lemma}
\begin{proposition}\label{prop:RMSHrate}
Suppose Assumptions \ref{as:dgp2}, \ref{as:boundedinternal}, \ref{as:parametricfamily}, and \ref{as:exprior} holds. If 1. $\sup_{\theta \in \Theta}||b_{0,\theta}||_{\infty} \leq \bar{B}$, 2. $\{b_{0,\theta}: \theta \in \Theta\} \subseteq \overline{\mathcal{H}}$, and 3. $\sup_{\theta \in \Theta}\varphi_{b_{0,\theta}}(\delta_{n}) \leq n \delta_{n}^{2}$ for some $\{\delta_{n}\}_{n \geq 1}$ satisfying $\delta_{n} \rightarrow 0$ and $n\delta_{n}^{2} \rightarrow \infty$ as $n\rightarrow \infty$, then, for $P_{0,Z}^{\infty}$-almost every fixed realization $\{z_{i}\}_{i \geq 1}$ of $\{Z_{i}\}_{i \geq 1}$, $\Pi(h_{n,2}(p_{\theta,B},p_{0,W|Z}) < D \delta_{n}| W^{(n)}) \rightarrow 1$ in $P_{0,W|Z}^{(n)}$-probability as $n\rightarrow \infty$ for some $D>0$.
\end{proposition}
\begin{proof}
 The proof has three steps that amount to verifying Lemma \ref{lem:concentration}. Step 1 shows (\ref{eq:neighbor}). Step 2 constructs sets $\{\bar{\mathcal{P}}_{n,W|Z}\}_{n \geq 1}$ and shows condition (\ref{eq:mass}). Step 3 checks (\ref{eq:entropy}).
\paragraph{Step 1.}
The first step is to show (\ref{eq:neighbor}). Applying Lemma 3.1 of \cite{vaart2008rates}, it suffices to show that
\begin{align*}
\Pi\left(\sup_{(w,z) \in \mathcal{W}\times [0,1]^{d_{z}}}|\log f_{\theta}(w|z)+B(F_{\theta,z}(w),z)-\log p_{0,W|Z}(w|z)| < 2\delta_{n} \right) \geq \exp(-n\delta_{n}^{2}).   
\end{align*}
Applying the law of total probability and that $\Pi = \Pi_{\mathcal{B}} \otimes \Pi_{\Theta}$ and Assumption \ref{as:exprior}, I know that
\begin{align*}
    &\Pi\left(\sup_{(w,z) \in \mathcal{W} \times [0,1]^{d_{z}}}\left|\log f_{\theta}(w|z)+B(F_{\theta,z}(w),z)-\log p_{0,W|Z}(w|z)\right| < 2\delta \right) \\
    &\quad = \int_{\Theta}\Pi_{\mathcal{B}}\left(\sup_{(w,z) \in \mathcal{W} \times [0,1]^{d_{z}}}\left|\log f_{\theta}(w|z)+B(F_{\theta,z}(w),z)-\log p_{0,W|Z}(w|z)\right| < 2\delta \right)\pi_{\Theta}(\theta) d\theta
\end{align*}
Since $w = F_{\theta,z}^{-1}(F_{\theta,z}(w))$ for $\theta$ and $z$ fixed, I can use the definition of $b_{0,\theta}$ and that $F_{\theta,z}(\mathcal{W}) =[0,1]^{d_{w}}$ to conclude that
\begin{align*}
\sup_{w \in \mathcal{W}}\left|\log f_{\theta}(w|z)+B(F_{\theta,z}(w),z)-\log p_{0,W|Z}(w|z)\right| &= \sup_{w \in \mathcal{W}}\left|B(F_{\theta,z}(w),z) - b_{0,\theta}(F_{\theta,z}(w),z)\right| \\
&\leq \sup_{(u,z) \in [0,1]^{d_{w}}\times [0,1]^{d_{z}}}\left|B(u,z)-b_{0,\theta}(u,z)\right|
\end{align*}
for all $z \in [0,1]^{d_{z}}$. Taking the supremum over $z \in [0,1]^{d_{z}}$, it follows that
\begin{align*}
&\int_{\Theta}\Pi_{\mathcal{B}}\left(\sup_{(w,z) \in \mathcal{W} \times [0,1]^{d_{z}}}\left|\log f_{\theta}(w|z)+B(F_{\theta,z}(w),z)-\log p_{0,W|Z}(w|z)\right| < 2\delta \right)\pi_{\Theta}(\theta) d\theta \\
&\quad \geq \int_{\Theta}\Pi_{\mathcal{B}}\left(\sup_{(u,z) \in [0,1]^{d_{w}} \times [0,1]^{d_{z}}}\left|B(u,z)-b_{0,\theta}(u,z)\right| < 2\delta \right)\pi_{\Theta}(\theta) d\theta.
\end{align*}
Since $\sup_{(u,z) \in [0,1]^{d_{w}}\times [0,1]^{d_{z}}}|B(u,z)| \leq \bar{B}$ implies that
\begin{align*}
\Pi_{\mathcal{B}}\left(||B-b_{0,\theta}||_{\infty} < 2\delta \right)= \frac{\tilde{\Pi}_{\mathcal{B}}\left(||B-b_{0,\theta}||_{\infty}< 2\delta \right)}{\tilde{\Pi}_{\mathcal{B}}(||B||_{\infty} \leq \bar{B})}    
\end{align*}
for $\delta>0$ small, I apply Lemma 5.3 of \cite{van2008reproducing} to conclude that
\begin{align*}
\int_{\Theta}\Pi_{\mathcal{B}}\left(||B-b_{0,\theta}||_{\infty} < 2\delta \right)\pi_{\Theta}(\theta) d\theta \geq \frac{1}{\tilde{\Pi}_{\mathcal{B}}(||B||_{\infty} \leq \bar{B})}\int_{\Theta}\exp(-\varphi_{b_{0,\theta}}(\delta)) \pi_{\Theta}(\theta) d \theta    
\end{align*}
for $\delta> 0$ small. Consequently, the inequality $\sup_{\theta \in \Theta}\varphi_{b_{0,\theta}}(\delta_{n}) \leq n \delta_{n}^{2}$ implies
\begin{align*}
&\Pi\left(\sup_{(w,z) \in \mathcal{W}\times [0,1]^{d_{z}}}|\log f_{\theta}(w|z)+B(F_{\theta,z}(w),z)-\log p_{0,W|Z}(w|z)| < 2\delta_{n} \right) \\
&\quad \geq \frac{\exp(-n\delta_{n}^{2})}{\tilde{\Pi}_{\mathcal{B}}(||B||_{\infty} \leq \bar{B})}. 
\end{align*}
\paragraph{Step 2.} I now construct $\bar{\mathcal{P}}_{n,W|Z}$.  Let $\mathcal{B}^{(1)}$ be the unit ball in $(C(T,\mathbb{R}),||\cdot||_{\infty})$, let $\mathcal{H}^{(1)}$ be the unit ball in $(\mathcal{H},\langle \cdot,\cdot \rangle_{\mathcal{H}})$, let $\mathcal{B}_{n} = \delta_{n}\mathcal{B}^{(1)} + D_{n} \mathcal{H}^{(1)}$ for $D_{n} = -2\Phi^{-1}(\tilde{\Pi}_{\mathcal{B}}(||B||_{\infty} \leq \bar{B})\exp(-Cn\delta_{n}^{2}))$ for some $C>1$ large, and let $\bar{\mathcal{P}}_{n,W|Z} = \{p_{\theta,B}: (\theta,B) \in \Theta \times \mathcal{B}_{n}\}$. Theorem 5.1 of \cite{van2008reproducing} implies that
\begin{align*}
    \tilde{\Pi}_{\mathcal{B}}(B \in \mathcal{B}_{n}^{c}) \leq 1-\Phi\left(\Phi^{-1}\left(\tilde{\Pi}_{\mathcal{B}}(||B||_{\infty} \leq \delta_{n})\right)+ D_{n}\right),
\end{align*}
Since $\tilde{\Pi}_{\mathcal{B}}(||B||_{\infty} \leq \delta_{n}) \geq \exp(-n\delta_{n}^{2})$ by Theorem 5.1 of \cite{van2008reproducing} and $\varphi_{b_{0,\theta}}(\delta_{n}) \leq n\delta_{n}^{2}$ for all $\theta \in \Theta$, it follows that
\begin{align*}
    \Phi^{-1}\left(\tilde{\Pi}_{\mathcal{B}}(||B||_{\infty} \leq \delta_{n})\right)
    &\geq \Phi^{-1}\left(\exp(-n\delta_{n}^{2})\right) \\ &\geq\Phi^{-1}\left( \tilde{\Pi}_{\mathcal{B}}(||B||_{\infty} \leq \bar{B}) \exp(-Cn\delta_{n}^{2})\right) \\ &\geq - \frac{D_{n}}{2},
\end{align*}
where the second inequality uses that $C> 1$ and $0<\tilde{\Pi}_{\mathcal{B}}(||B||_{\infty} \leq \bar{B}) \leq 1$.
Consequently, I conclude that $ \tilde{\Pi}_{\mathcal{B}}(B \in \mathcal{B}_{n}^{c}) \leq 1-\Phi\left(D_{n}/2\right)   $ and then $\tilde{\Pi}_{\mathcal{B}}(B \in \mathcal{B}_{n}^{c}) \leq \tilde{\Pi}_{\mathcal{B}}\left(||B||_{\infty} \leq \bar{B}\right)\exp\left(-Cn\delta_{n}^{2}\right)$ using the symmetry of the standard normal distribution cumulative distribution function. Since $\tilde{\Pi}_{\mathcal{B}}\left(||B||_{\infty} \leq \bar{B}\right) \leq 1$, $\Pi_{\mathcal{B}}\left(B \in \mathcal{B}_{n}^{c}\right) \leq \exp\left(-Cn\delta_{n}^{2}\right)$. This verifies (\ref{eq:mass}).
\paragraph{Step 3.} The final step is to show (\ref{eq:entropy}). Using the triangle inequality
\begin{align*}
h(p_{\theta_{1},B_{1}}(\cdot|z_{i}),p_{\theta_{2},B_{2}}(\cdot|z_{i})) &\leq h(p_{\theta_{1},B_{1}}(\cdot|z_{i}),p_{\theta_{1},B_{2}}(\cdot|z_{i}))+ h(p_{\theta_{1},B_{2}}(\cdot|z_{i}),p_{\theta_{2},B_{2}}(\cdot|z_{i}))
\end{align*}
for $i=1,...,n$. Using Lemma 3.1 of \cite{vaart2008rates} and Assumption \ref{as:parametricfamily}, there is a constant $L>0$ such that
\begin{align*}
&h(p_{\theta_{1},B_{2}}(\cdot|z_{i}),p_{\theta_{2},B_{2}}(\cdot|z_{i})) \\
&\quad \leq \sup_{(w,z) \in \mathcal{W}\times [0,1]^{d_{z}}}|\log f_{\theta_{1}}(w|z) - \log f_{\theta_{2}}(w|z)|\\
&\quad \quad \times \exp\left(\frac{1}{2}\sup_{(w,z) \in \mathcal{W}\times [0,1]^{d_{z}}}|\log f_{\theta_{1}}(w|z) - \log f_{\theta_{2}}(w|z)|\right) \\
&\quad \leq L \cdot ||\theta_{1}-\theta_{2}||_{2} \exp\left(\frac{1}{2}L \cdot ||\theta_{1}-\theta_{2}||_{2}\right) \\
&\quad \leq L \cdot ||\theta_{1}-\theta_{2}||_{2} \exp\left(\frac{1}{2} L \cdot \text{diam} (\Theta) \right)
\end{align*}
for all $i=1,...,n$, where $\text{diam}(\Theta) = \sup_{\theta_{1},\theta_{2} \in \Theta}||\theta_{1}-\theta_{2}||_{2} < \infty$ (by Assumption \ref{as:parametricfamily}.1). Moreover, Lemma 3.1 of \cite{vaart2008rates} implies that
\begin{align*}
    h(p_{\theta_{1},B_{1}}(\cdot|z),p_{\theta_{1},B_{2}}(\cdot|z)) \leq ||B_{1}-B_{2}||_{\infty}\exp\left(\frac{1}{2}||B_{1}-B_{2}||_{\infty}\right) 
\end{align*}
for $i=1,...,n$. Consequently,
\begin{align*}
h_{n,2}(p_{\theta_{1},B_{1}},p_{\theta_{2},B_{2}}) &\leq   ||B_{1}-B_{2}||_{\infty}\exp\left(\frac{1}{2}||B_{1}-B_{2}||_{\infty}\right) \\
&\quad +  L \cdot ||\theta_{1}-\theta_{2}||_{2} \exp\left(\frac{1}{2} L \cdot \text{diam} (\Theta) \right),
\end{align*}
This bound implies that the $\delta_{n}/72$ entropy of $\bar{\mathcal{P}}_{n,W|Z}$ relative to $h_{n}$ is bounded by the $\tilde{c}\delta_{n}$ entropy of $\Theta \times \mathcal{B}_{n}$ with respect to the product norm $||(\theta,B)|| = ||\theta||_{2} + ||B||_{\infty}$, where $\tilde{c}>0$ is some constant. Consequently, I bound $\log N(c\delta_{n},\Theta \times \mathcal{B}_{n},||\cdot||)$ for any constant $c>0$. For any constant $c>0$, 
\begin{align*}
    \log N(c\delta_{n},\Theta \times \mathcal{B}_{n},||\cdot||) &\leq \log N\left(\frac{c\delta_{n}}{2},\Theta,||\cdot||_{2} \right) + \log N\left(\frac{c\delta_{n}}{2},\mathcal{B}_{n},||\cdot||_{\infty}\right) \\
    &\leq d_{\theta} \log\left(\frac{2 \cdot \text{diam}(\Theta)}{c\delta_{n}}\right) + \log N\left(\frac{c\delta_{n}}{2},\mathcal{B}_{n},||\cdot||_{\infty}\right),
\end{align*}
where the first inequality applies a standard inequality relating metric entropy of Cartesian products to the metric entropy of their components, and the second inequality uses the expression for the $||\cdot||_{2}$-metric entropy of a compact subset $\Theta$ of $\mathbb{R}^{d_{\theta}}$. Following the arguments of Theorem 2.1 in \cite{vaart2008rates}, it can also be shown that
\begin{align*}
\log N\left(\frac{c\delta_{n}}{2},\mathcal{B}_{n},||\cdot||_{\infty}\right) \lesssim n \delta_{n}^{2}
\end{align*}
for any $c>0$ if $\sup_{\theta \in \Theta}\varphi_{b_{0,\theta}}(\delta_{n}) \leq n \delta_{n}^{2}$. Consequently, $\log N(c\delta_{n},\Theta \times \mathcal{B}_{n},||\cdot||) \lesssim n\delta_{n}^{2}$ for any $c>0$, which implies (\ref{eq:entropy}).
\end{proof}
\subsubsection{Proof of Proposition 2}
The next lemma and proof of Proposition 2 relies on the theory of wavelets. Let $\{\psi_{jk}: j \geq 0, 0 \leq k \leq 2^{jd_{z}-1}\}$ be the $d_{z}$-dimensional tensor product of the boundary-adapted wavelet basis for $L^{2}([0,1])$ from \cite{COHEN199354} (`CDV wavelets' from hereafter).\footnote{There are other ways to express the CDV wavelet basis, however this notation is consistent with the Bayesian nonparametrics literature (e.g., \cite{castillo2014bayesian}, \cite{castillo2015bernstein}, \cite{ray2020semiparametric}).} The following properties of CDV wavelets are used in the lemma (and the proof of Proposition 2):
\begin{itemize}
\item  $\{\psi_{j,k}\}$ forms an orthonormal basis of $L^{2}([0,1]^{d_{z}})$.
\item $||\psi_{jk}||_{\infty} \lesssim 2^{jd_{z}/2}$ for all $j,k$.
\item $\sum_{k=0}^{2^{jd_{z}-1}}||\psi_{jk}||_{\infty} \lesssim 2^{jd_{z}}$ for each $j \geq 1$. \item $\langle \psi_{jk},1 \rangle_{L^{2}([0,1]^{d_{z}})} = 0$ for $j$ large.
\item $\{\psi_{jk}\}$ characterizes Besov spaces $\mathfrak{B}_{\infty,\infty,a}([0,1]^{d_{z}},\mathbb{R})$ for $a \leq \alpha$ with $\alpha > 0$ fixed in that $f \in \mathfrak{B}_{\infty,\infty,a}([0,1]^{d_{z}})$ iff $$||f||_{\infty,\infty,a}:= \sup_{j \geq 0, 0 \leq k \leq 2^{jd_{z}-1}}2^{j(d_{z}/2+a)}|\langle f, \psi_{jk} \rangle_{L^{2}([0,1]^{d_{z}})}|< \infty.$$
\end{itemize}
This last property is useful for H\"{o}lder functions because $\mathfrak{B}_{\infty,\infty,a}([0,1]^{d_{z}},\mathbb{R}) = C^{a}([0,1]^{d_{z}},\mathbb{R})$ for noninteger $a$ with equivalent norms, and, for integer $a$, $C^{a}([0,1]^{d_{z}},\mathbb{R}) \subseteq \mathfrak{B}_{\infty,\infty,a}([0,1]^{d_{z}},\mathbb{R})$ and there exists a universal constant $C >0 $ such that $||f||_{\infty,\infty,a} \leq C ||f||_{a}$ (i.e., it is a continuous embedding). I also make use of the following notation:    
\begin{itemize}
    \item $\tilde{\psi}_{jk}(z_{i}) = \psi_{jk}(z_{i}) - \hat{\tau}_{jk}'\psi^{J}_{-jk}(z_{i}),  
$ where $\psi_{-jk}^{J}(z_{i})$ is the subvector of $\psi^{J}(z_{i}) = (\psi_{jk}(z_{i}): 0 \leq j \leq J, 0 \leq k \leq 2^{jd_{z}-1})'$ that excludes $\psi_{jk}(z_{i})$, and $\hat{\tau}_{jk}$ is the sample ordinary least squares coefficients from the linear predictor of $\psi_{jk}$ on $\psi_{-jk}^{J}$.
\item $\Psi_{J}$ is the $n \times d(J)$ matrix with rows $\psi^{J}(z_{i})'$ and $d(J)$ is the number of elements of $\psi^{J}(z)$.
\end{itemize}
\begin{lemma}\label{lem:graminverse}
Suppose Assumptions \ref{as:dgp2} and \ref{as:boundedinternal} hold, and $2^{J_{n}} = (n / \log n)^{1/(2 \alpha + d_{z})}$ with $\alpha > d_{z}/2$, then there are $c_{1}, c_{2}> 0$ such that, for $P_{0,Z}^{\infty}$-almost every fixed realization $\{z_{i}\}_{i \geq 1}$ of $\{Z_{i}\}_{i \geq 1}$, $ \lambda_{min}\left(n^{-1}\Psi_{J_{n}}'\Psi_{J_{n}}\right) \geq c_{1}$ and $\min_{1 \leq j \leq J_{n}}\min_{0 \leq k \leq 2^{jd_{z}}-1}n^{-1}\sum_{i=1}^{n}\tilde{\psi}_{jk}^{2}(z_{i}) \geq c_{2}$ for large $n$.
\end{lemma}
\begin{proposition}\label{prop:interpolation}
Suppose Assumptions \ref{as:dgp2}, \ref{as:boundedinternal}, \ref{as:parametricfamily}, and \ref{as:exprior} hold, the conditions of Proposition \ref{prop:RMSHrate} hold, $p_{0,W|Z}$ satisfies $\int_{\mathcal{W}} ||\log p_{0,W|Z}(w|\cdot)||_{\alpha_{0}}dw < \infty$ for some $\alpha_{0} > d_{z}/2$, and $\{f_{\theta}: \theta \in \Theta\}$ satisfies $\sup_{\theta \in \Theta} \int_{\mathcal{W}}|| \log f_{\theta}(w|\cdot) ||_{a}dw < \infty$ and $\sup_{\theta \in \Theta}\int_{\mathcal{W}}||F_{\theta,\cdot}(w)||_{a}^{\lfloor a \rfloor}dw < \infty$  for each $a \in (0,\alpha)$. Then, for $P_{0,Z}^{\infty}$-almost every fixed realization $\{z_{i}\}_{i \geq 1}$ of $\{Z_{i}\}_{i \geq 1}$, there exists $\{\tilde{\mathcal{P}}_{1,n,W|Z}\}_{n \geq 1}$ and $\{\tilde{\delta}_{n}\}_{n \geq 1}$ such that $\Pi(\tilde{\mathcal{P}}_{n,1,W|Z}|W^{(n)}) \rightarrow 1$ in $P_{0,W|Z}^{(n)}$-probability as $n\rightarrow \infty$ and $\sup_{p_{\theta,B} \in \tilde{\mathcal{P}}_{1,n,W|Z}}h_{n,\infty}(p_{\theta,B},p_{0,W|Z})\leq \tilde{D}\tilde{\delta}_{n}$ for a constant $\tilde{D}>0$, where $\tilde{\delta}_{n}$ depends on $\alpha,\alpha_{0}$, $d_{z}$, and the RMSH rate $\delta_{n}$.
\end{proposition}
\begin{proof}
The proof has five steps. Step 1 constructs the sets $\{\tilde{\mathcal{P}}_{1,n,W|Z}\}_{n \geq 1}$ and shows that their posterior probability approaches one as the sample size grows. Step 2 establishes that the empirical supremum $L^{1}$-distance between the sample least squares projections of $p_{\theta,B}$ and $p_{0,W|Z}$ onto a finite-dimensional linear sieve formed using the wavelet system is bounded from above by the RMSH distance and a multiplicative constant that depends on the sieve dimension. Step 3 and Step 4 derive bounds on the approximation errors. Step 5 concludes.
    \paragraph{Step 1.} I show that $\Pi(\tilde{\mathcal{P}}_{1,n,W|Z}^{c}|W^{(n)})\rightarrow 0$ in $P_{0,W|Z}^{(n)}$-probability as $n\rightarrow \infty$, where
    \begin{align*}
        \tilde{\mathcal{P}}_{1,n,W|Z} = \{(\theta,B): h_{n,2}(p_{\theta,B},p_{0,W|Z}) \leq D \delta_{n}\} \cap \{(\theta,B): ||B||_{a} \leq D\sqrt{n}\delta_{n}, \ ||B||_{\infty} \leq \bar{B}\}
    \end{align*}
    for some $a \in (0,\alpha)$ and some $D>0$ large. Given Proposition 1 and that $||B||_{\infty} \leq \bar{B}$ for $\Pi$-almost every $B$, I need to show that $\Pi((\theta,B): ||B||_{a} > D\sqrt{n}\delta_{n}|W^{(n)})\rightarrow 0$ in $P_{0,W|Z}^{(n)}$-probability as $n\rightarrow \infty$. Since $\tilde{\Pi}_{\mathcal{B}}$ is concentrated on functions with elements in $C^{a}([0,1]^{d})$, it can be viewed as the law of a Gaussian process on $C^{a}([0,1]^{d})$, and, as a result, Proposition A.2.1 of \cite{vaart2023weak} implies that the following holds: $\tilde{\Pi}_{\mathcal{B}}(||B||_{a} > D \sqrt{n}\delta_{n}) \leq \exp(-C n\delta_{n}^{2})$ for some constant $C > 0$, and, as a result, $\Pi_{\mathcal{B}}(||B||_{a} > D \sqrt{n}\delta_{n}) \leq  \tilde{\Pi}_{\mathcal{B}}^{-1}(||B||_{\infty} \leq \bar{B})\exp(-C n\delta_{n}^{2})$. Applying Lemma \ref{lem:posteriorsieve} (with $\Theta \times \{B:||B||_{a} > D\sqrt{n}\delta_{n}\}$ in place of $\bar{\mathcal{P}}_{n,W|Z}$) yields the posterior concentration $\Pi((\theta,B): ||B||_{a} > D\sqrt{n}\delta_{n}|W^{(n)})\rightarrow 0$ in $P_{0,W|Z}^{(n)}$-probability as $n\rightarrow \infty$ for $D>0$ large.
    \paragraph{Step 2.} Let $J = (n/\log n)^{1/(2\alpha + d_{z})}$ and suppose that $n$ is sufficiently large so that the bounds in Lemma \ref{lem:graminverse} hold. Let $\hat{p}^{J}_{\theta,B}(w,z) = \sum_{j=0}^{J}\sum_{k=0}^{2^{jd_{z}}-1}\psi_{jk}(z)\hat{\beta}_{jk}^{J}(w)$ and $\hat{p}^{J}_{0,W|Z}(w,z) = \sum_{j=0}^{J}\sum_{k=0}^{2^{jd_{z}}-1}\psi_{jk}(z)\hat{\beta}_{0,jk}^{J}(w)$ be the \textit{sample} least squares projections of $p_{\theta,B}(w|\cdot)$ and $p_{0,W|Z}(w|\cdot)$ onto the linear space generated by $\{\psi_{jk}: 0 \leq j \leq J, 0 \leq k \leq 2^{jd_{z}}-1\}$, respectively. Applying the triangle inequality and $\sum_{k=0}^{2^{jd_{z}}-1}||\psi_{jk}||_{\infty} \lesssim 2^{jd_{z}/2}$, I know that
    \begin{align}\label{eq:sievebound1}
        \int_{\mathcal{W}}|\hat{p}^{J}_{\theta,B}(w,z_{i})-\hat{p}_{0}^{J}(w,z_{i})|dw  &\leq \sum_{j=0}^{J}\max_{0 \leq k \leq 2^{jd_{z}}-1}\int_{\mathcal{W}}\left|\hat{\beta}_{jk}^{J}(w)-\hat{\beta}_{0,jk}^{J}(w)\right|dw\sum_{k=0}^{2^{jd_{z}}-1}||\psi_{jk}||_{\infty} \nonumber \\
        &\lesssim \sum_{j=0}^{J}2^{\frac{jd_{z}}{2}}\max_{0 \leq k \leq 2^{jd_{z}}-1}\int_{\mathcal{W}}\left|\hat{\beta}_{jk}^{J}(w)-\hat{\beta}_{0,jk}^{J}(w)\right|dw
    \end{align}
    for any $J \geq 1$. Consequently, I seek a uniform bound (in $j,k$) for $\int_{\mathcal{W}}|\hat{\beta}_{jk}^{J}(w)-\hat{\beta}_{0,jk}^{J}(w)|dw$. Using residual regression (cf. Chapter 17.3 of \cite{goldberger1991course}), I know that
    \begin{align*}
                \hat{\beta}_{jk}^{J}(w)-\hat{\beta}_{0,jk}^{J}(w) = \frac{\frac{1}{n}\sum_{i=1}^{n}\tilde{\psi}_{jk}(z_{i})\left(p_{\theta,B}(w|z_{i})-p_{0,W|Z}(w|z_{i})\right)}{\frac{1}{n}\sum_{i=1}^{n}\tilde{\psi}_{jk}^{2}(z_{i})}.
    \end{align*}
The triangle inequality, H\"{o}lder's inequality, and the definition of the minimum then implies that
    \begin{align*}
    &\int_{\mathcal{W}}|\hat{\beta}_{jk}^{J}(w)-\hat{\beta}_{0,jk}^{J}(w)|dw \\
    &\quad = \left(\frac{1}{n}\sum_{i=1}^{n}\tilde{\psi}_{jk}^{2}(z_{i})\right)^{-1}\int_{\mathcal{W}}\left|\frac{1}{n}\sum_{i=1}^{n}\tilde{\psi}_{jk}(z_{i})(p_{\theta,B}(w|z_{i})-p_{0,W|Z}(w|z_{i}))\right|dw \\
    & \quad \leq \left(\frac{1}{n}\sum_{i=1}^{n}\tilde{\psi}_{jk}^{2}(z_{i})\right)^{-1}\frac{1}{n}\sum_{i=1}^{n}\left|\tilde{\psi}_{jk}(z_{i})\right | \cdot \int_{\mathcal{W}}\left |p_{\theta,B}(w|z_{i})-p_{0,W|Z}(w|z_{i})\right|dw \\
    &\quad \leq \left(\frac{1}{n}\sum_{i=1}^{n}\tilde{\psi}_{jk}^{2}(z_{i})\right)^{-\frac{1}{2}}\sqrt{\frac{1}{n}\sum_{i=1}^{n}\left(\int_{\mathcal{W}}|p_{\theta,B}(w|z_{i})-p_{0,W|Z}(w|z_{i})|dw\right)^{2}} \\
    &\quad \leq 2 \left(\min_{0 \leq j \leq J}\min_{0 \leq k \leq 2^{jd_{z}-1}}\frac{1}{n}\sum_{i=1}^{n}\tilde{\psi}_{jk}^{2}(z_{i})\right)^{-\frac{1}{2}}h_{n,2}(p_{\theta,B},p_{0,W|Z}),
    \end{align*}
    where the last inequality uses the equivalence between the total variation distance and the Hellinger metric. Substituting this into the upper bound (\ref{eq:sievebound1}) and applying Lemma \ref{lem:graminverse}, there is a constant $\tilde{c}> 0$ such that
    \begin{align*}
   \sup_{(\theta,B) \in \tilde{\mathcal{P}}_{1,n,W|Z}}\max_{1 \leq i \leq n}\int_{\mathcal{W}}|\hat{p}^{J}_{\theta,B}(w,z_{i})-\hat{p}_{0,W|Z}^{J}(w,z_{i})|dw\leq \tilde{c}2^{\frac{Jd_{z}}{2}}\delta_{n}.
    \end{align*}
    \paragraph{Step 3.} Applying the triangle inequality, I know that
    \begin{align*}
        |p_{0,W|Z}(w|z_{i}) - \hat{p}^{J}_{0,W|Z}(w,z_{i})| &\leq \underbrace{\left|p_{0,W|Z}(w|z_{i}) -p^{J}_{0,W|Z}(w,z_{i})\right|}_{T_{1}(w,z_{i})} \\
        &\quad + \underbrace{\left|p^{J}_{0,W|Z}(w,z_{i})-\hat{p}^{J}_{0,W|Z}(w,z_{i})\right|}_{T_{2}(w,z_{i})},
    \end{align*}
    where $p_{0,W|Z}^{J}(w,z) = \sum_{j=0}^{J}\sum_{k=0}^{2^{jd_{z}}-1}\psi_{jk}(z)\beta_{0,jk}(w)$ is \textit{population} least squares projection of $p_{0,W|Z}(w|\cdot)$ onto $\text{span}\{\psi_{lk}: 0 \leq j \leq J, 0 \leq k \leq 2^{jd_{z}}-1\}$. Since 1. $P_{0,Z}$ has a Lebesgue density on $[0,1]^{d_{z}}$ that is bounded away from zero and infinity, and 2. $p_{0,W|Z}$ is bounded away from zero and infinity, there is a constant $\tilde{C} > 0$ independent of $w$ such that
    \begin{align*}
    \sup_{z \in [0,1]^{d_{z}}}T_{1}(w,z) \leq \tilde{C} 2^{-\alpha_{0}J}||\log p_{0,W|Z}(w|\cdot)||_{\alpha_{0}},    
    \end{align*} 
by standard uniform approximation theory for tensor products of CDV wavelets (e.g., \cite{CHEN20075549}) and the chain rule. Consequently, the condition 
    \begin{align*}
    \int_{\mathcal{W}}||\log p_{0,W|Z}(w|\cdot)||_{\alpha_{0}}dw < \infty    
    \end{align*} 
    implies 
    \begin{align*}
    \max_{1 \leq i \leq n}\int_{\mathcal{W}}T_{1}(w,z_{i})dw \leq \tilde{\tilde{C}} 2^{-\alpha_{0}J}     
    \end{align*}
    for some constant $\tilde{\tilde{C}}>0$. For $T_{2}(w,z_{i})$, I first apply the Cauchy-Schwarz inequality to obtain that
\begin{align*}
T_{2}(w,z_{i}) \leq \sqrt{\sum_{j=0}^{J}\sum_{k=0}^{2^{jd_{z}}-1}\psi_{jk}^{2}(z_{i})}\sqrt{\sum_{j=0}^{J}\sum_{k=0}^{2^{jd_{z}-1}}|\hat{\beta}_{0,jk}(w)-\beta_{0,jk}(w)|^{2}}
\end{align*}
Since $||\psi_{jk}||_{\infty} \lesssim 2^{jd_{z}/2}$ for all $j,k$ and $\sum_{k=0}^{2^{jd_{z}}-1}||\psi_{jk}||_{\infty} \lesssim 2^{jd_{z}/2}$, I can conclude that $T_{2}(w,z_{i}) \lesssim 2^{\frac{Jd_{z}}{2}}||\hat{\beta}_{0}^{J}(w)-\beta_{0}^{J}(w)||_{2} $, where $\hat{\beta}_{0}^{J}(w)$ is the vector of projection coefficients $(\hat{\beta}_{0,jk}^{J}(w): 0 \leq j \leq J, \ 0 \leq k \leq 2^{jd_{z}-1})'$ and $\beta^{J}_{0}(w)$ is defined similarly. Let $\varepsilon(w)$ be the $n\times 1$ vector with $i$th element $\varepsilon_{i}(w) = p_{0,W|Z}(w|z_{i})- p^{J}_{0,W|Z}(w,z_{i})$, $i=1,...,n$. Since
    \begin{align*}
        \hat{\beta}_{0}^{J}(w) = \left(\frac{1}{n}\Psi_{J}'\Psi_{J}\right)^{-1}\left(\frac{1}{n}\Psi_{J}'Y(w)\right)
    \end{align*}
    with $Y(w)$ being the $n\times 1$ vector with element $p_{0,W|Z}(w|z_{i})$ for $i=1,...,n$, standard least squares algebra reveals that
    \begin{align*}
        \hat{\beta}_{0}^{J}(w) = \beta_{0}^{J}(w)+\left(\frac{1}{n}\Psi_{J}'\Psi_{J}\right)^{-1}\frac{1}{n}\Psi_{J}'\varepsilon(w),
    \end{align*}
  Let $\hat{\varepsilon}_{J}(w) = \Psi_{J}(\Psi_{J}'\Psi_{J})^{-1}\Psi_{J}'\varepsilon(w)$. Using the variational characterization of eigenvalues, the fact that $\varepsilon(w) = \hat{\varepsilon}_{J}(w) + (\varepsilon(w)-\hat{\varepsilon}_{J}(w))$ with $\varepsilon(w)-\hat{\varepsilon}_{J}(w)$ orthogonal to $\hat{\varepsilon}_{J}(w)$ in the empirical $L^{2}$ inner product, and the fact that projection is norm reducing, I obtain
   \begin{align*}
       &||\hat{\beta}_{0}^{J}(w)-\beta_{0}^{J}(w)||_{2}^{2} \\
       &\quad= \left(\frac{1}{n}\varepsilon(w)'\Psi_{J}\right)\left(\frac{1}{n}\Psi_{J}'\Psi_{J}\right)^{-\frac{1}{2}}\left(\frac{1}{n}\Psi_{J}'\Psi_{J}\right)^{-1}\left(\frac{1}{n}\Psi_{J}'\Psi_{J}\right)^{-\frac{1}{2}}\left(\frac{1}{n}\Psi_{J}'\varepsilon(w)\right) \\
       &\quad\leq \left(\lambda_{min}\left(\frac{1}{n}\Psi_{J}'\Psi_{J}\right)\right)^{-1} \left| \left| \left(\frac{1}{n}\Psi_{J}'\Psi_{J}\right)^{-\frac{1}{2}}\frac{1}{n}\Psi_{J}'\varepsilon(w) \right| \right|_{2}^{2} \\
       &\quad= \left(\lambda_{min}\left(\frac{1}{n}\Psi_{J}'\Psi_{J}\right)\right)^{-1} \frac{1}{n}\varepsilon(w)'\hat{\varepsilon}_{J}(w) \\
       &\quad\leq   \left(\lambda_{min}\left(\frac{1}{n}\Psi_{J}'\Psi_{J}\right)\right)^{-1} \frac{1}{n}\hat{\varepsilon}_{J}(w)'\hat{\varepsilon}_{J}(w) \\
    &\quad\leq \left(\lambda_{min}\left(\frac{1}{n}\Psi_{J}'\Psi_{J}\right)\right)^{-1} \frac{1}{n}\varepsilon(w)'\varepsilon(w).
   \end{align*}
Consequently, I obtain
  \begin{align*}
   T_{2}(w,z_{i}) &\lesssim  \left(\lambda_{min}\left(\frac{1}{n}\Psi_{J}'\Psi_{J}\right)\right)^{-\frac{1}{2}} 2^{\frac{Jd_{z}}{2}}  \sqrt{\frac{1}{n}\varepsilon(w)'\varepsilon(w)}. 
   \end{align*}
   for all $i=1,...,n$. Since $n^{-1}\varepsilon(w)'\varepsilon(w) = n^{-1}\sum_{i=1}^{n}|T_{1}(w,z_{i})|^{2}$, I know that 
   \begin{align*}
    \sqrt{\frac{1}{n}\varepsilon(w)'\varepsilon(w)} \leq \tilde{C} 2^{-\alpha_{0}J}||\log p_{0,W|Z}(w|\cdot)||_{\alpha_{0}}.   
   \end{align*} 
   Consequently, I obtain that there is a constant $C^{\dagger} >0$ such that
   \begin{align*}
   \int_{\mathcal{W}}T_{2}(w,z_{i})dw  \leq C^{\dagger} 2^{J(\frac{d_{z}}{2}-\alpha_{0})}  \left(\lambda_{min}\left(\frac{1}{n}\Psi_{J}'\Psi_{J}\right)\right)^{-\frac{1}{2}}
   \end{align*}
   for all $i=1,...,n$. Putting everything together and applying Lemma \ref{lem:graminverse}, I obtain that
   \begin{align*}
   \max_{1 \leq i \leq n} \int_{\mathcal{W}}|p_{0,W|Z}(w|z_{i})-\hat{p}_{0,W|Z}^{J}(w,z_{i})|dw \leq \acute{C}(2^{J(\frac{d_{z}}{2}-\alpha_{0})}+2^{-\alpha_{0}J}).
  \end{align*}
  for some constant $\acute{C} > 0$.
  \paragraph{Step 4.} Like Step 3, I apply the triangle inequality to conclude that
  \begin{align*}
      \max_{1 \leq i \leq n}\int_{\mathcal{W}} |p_{\theta,B}(w|z_{i})-\hat{p}^{J}_{\theta,B}(w,z_{i})|dw &\leq \max_{1 \leq i \leq n}\int_{\mathcal{W}}|p_{\theta,B}(w|z_{i})-p^{J}_{\theta,B}(w,z_{i})|dw \\
      &\quad + \max_{1 \leq i \leq n} \int_{\mathcal{W}}|p^{J}_{\theta,B}(w,z_{i})-\hat{p}^{J}_{\theta,B}(w,z_{i})|dw,
  \end{align*}
  where $p^{J}_{\theta,B}$ is the population least squares projection of $p_{\theta,B}$ onto the $\text{span}\{\psi_{lk}: 0 \leq j \leq J, 0 \leq k \leq 2^{jd_{z}}-1\}$. Following a symmetric argument to the previous step, I can show that $  |p_{\theta,B}(w|z)-p_{\theta,B}^{J}(w,z)| \leq \breve{C} 2^{-a J}||\log p_{\theta,B}(w|\cdot)||_{a}$ for all $z \in \mathcal{Z}$, where $\breve{C} > 0$ is a constant independent of $w$. Consequently,
  \begin{align*}
      \max_{1 \leq i \leq n}\int_{\mathcal{W}}|p_{\theta,B}(w|z)-p^{J}_{\theta,B}(w,z)|dw \leq \breve{C}2^{-aJ}\int_{\mathcal{W}}||\log p_{\theta,B}(w|\cdot)||_{a}dw
  \end{align*}
  Since $\log p_{\theta,B}(w|\cdot) = \log f_{\theta}(w|\cdot) + B(F_{\theta,\cdot}(w),\cdot)-\log I(\cdot,B)$ with $$I(z,B) = \int_{[0,1]^{d_{w}}}\exp(B(u,z))du,$$ the triangle inequality reveals
  \begin{align*}
  ||\log p_{\theta,B}(w|\cdot) ||_{a} \leq ||\log f_{\theta}(w|\cdot)||_{a}+ ||\log I(\cdot,B)||_{a} + ||B(F_{\theta,\cdot}(w),\cdot)||_{a}.    
  \end{align*}
By assumption $\sup_{\theta \in \Theta}\int_{\mathcal{W}}||\log f_{\theta}(w|\cdot)||_{a}dw < \infty$, and, as a result, the independence of this quantity from $B$ implies
  \begin{align*}
      \sup_{(\theta,B) \in \tilde{\mathcal{P}}_{1,n,W|Z}}\int_{\mathcal{W}}||\log f_{\theta}(w|\cdot)||_{a}dw < \infty.
  \end{align*}
Turning to the second term, since $||B||_{\infty} \leq \bar{B}$ implies $I(\cdot,B)$ is bounded away from zero and infinity, compositions of H\"{o}lder functions and smooth functions (i.e. $x \mapsto \log x$) reveals
\begin{align*}
||\log I(\cdot,B)||_{a} \leq \bar{C}_{1}(|| I(\cdot,B)||_{a}+1)    
\end{align*} with $\bar{C}_{1} >0$ being a constant depending on $a$, $\bar{B}$, $d_{w}$, and $d_{z}$. Similarly, using that $||B||_{\infty} \leq \bar{B}$ implies $\exp(B)$ is bounded away from zero and infinity, interchanging integration and differentiation, and using compositions of H\"{o}lder functions and smooth functions (i.e. $x \mapsto \exp(x)$), I obtain
\begin{align*}
||I(\cdot,B)||_{a} \leq \bar{C}_{2}(||B||_{a}+1)   
\end{align*}
with $\bar{C}_{2} > 0$ also being a constant depending on $a$, $\bar{B}$, $d_{w}$, and $d_{z}$ (note the H\"{o}lder norm on the RHS is actually that of $C^{a}([0,1]^{d},\mathbb{R})$ rather than $C^{a}([0,1]^{d_{z}},\mathbb{R})$). Consequently,
\begin{align*}
    ||\log I(\cdot,B)||_{a} \leq \bar{C}_{1}\bar{C}_{2}(||B||_{a} + 1) + \bar{C}_{1} 
\end{align*}
Consequently, I can use that $||B||_{a} \leq D\sqrt{n}\delta_{n}$ to conclude that 
  \begin{align*}
  \sup_{(\theta,B) \in \tilde{\mathcal{P}}_{1,n,W|Z}}\int_{\mathcal{W}} ||\log I(\cdot,B)||_{a}dw  \leq \bar{C} (\sqrt{n}\delta_{n} +1)   
  \end{align*}
  where $\bar{C} > 0$ is a constant that depends on $\bar{B},\alpha$, $d_{z}$, $d_{w}$, the Lebesgue measure of $\mathcal{W}$, and $D$. For the third term, I note that, for $w$ and $\theta$ fixed, $B(F_{\theta,\cdot}(w),\cdot)$ can be written as a composition $B \circ G_{\theta,w}: [0,1]^{d_{z}} \rightarrow \mathbb{R}$, where $G_{\theta,w}(z) = (F_{\theta,z}(w),z)$. Consequently, repeated applications of the chain rule reveal that
  \begin{align*}
      ||B(F_{\theta,\cdot}(w),\cdot)||_{a} &\leq \grave{C}||B||_{a}(1+||F_{\theta,\cdot}(w)||^{\lfloor a \rfloor}_{a})
  \end{align*}
  where $\grave{C} > 0$ is a constant that depends on $a$, $d_{w}$, and $d_{z}$. This implies that
  \begin{align*}
    \sup_{(\theta,B) \in \tilde{\mathcal{P}}_{1,n,W|Z}}\int_{\mathcal{W}}||B(F_{\theta,\cdot}(w),\cdot)||_{a} dw \leq  \grave{C}\sqrt{n}\delta_{n} + \grave{C}\sqrt{n}\delta_{n}\sup_{\theta \in \Theta}\int_{\mathcal{W}}||F_{\theta,\cdot}(w)||^{\lfloor a \rfloor}_{a}dw,
  \end{align*}
  with the integral in the second term being finite by the assumptions of the proposition. Putting all this together,
  \begin{align*}
     \sup_{(\theta,B) \in \tilde{\mathcal{P}}_{1,n,W|Z}} \max_{1 \leq i \leq n}\int_{\mathcal{W}}|p_{\theta,B}(w|z_{i})-p_{\theta,B}^{J}(w,z_{i})|dw
      \leq \ddot{C}_{1}(\sqrt{n}\delta_{n}2^{-aJ}+2^{-aJ})
  \end{align*}
  for some constant $\ddot{C}_{1}>0$. Given this result, a symmetric argument to Step 3 can be used to conclude that
  \begin{align*}
   \sup_{(\theta,B) \in \tilde{\mathcal{P}}_{1,n,W|Z}} \max_{1 \leq i \leq n} \int_{\mathcal{W}}|p^{J}_{\theta,B}(w,z_{i})-\hat{p}^{J}_{\theta,B}(w,z_{i})|dw \leq \ddot{C}_{2}(\sqrt{n}\delta_{n}2^{J(\frac{d_{z}}{2}-a)}+ 2^{J(\frac{d_{z}}{2}-a)}).   
  \end{align*}
  for some constant $\ddot{C}_{2} > 0$.
   \paragraph{Step 5.} Combining all of these steps via the triangle inequality, I can conclude that
   \begin{align*}
       &\sup_{(\theta,B) \in \tilde{\mathcal{P}}_{1,n,W|Z}}\max_{1 \leq i \leq n}\int_{\mathcal{W}}\left|p_{\theta,B}(w|z_{i})-p_{0,W|Z}(w|z_{i})\right|dw\\
       &\quad \lesssim 2^{\frac{Jd_{z}}{2}}\delta_{n}  +2^{J(\frac{d_{z}}{2}-\alpha_{0})}+\sqrt{n}\delta_{n}2^{J(\frac{d_{z}}{2}-a)} \\
       &\quad \quad +2^{J(\frac{d_{z}}{2}-a)}+2^{-\alpha_{0}J} + \sqrt{n}\delta_{n}2^{-\alpha J} +  2^{-aJ} 
   \end{align*}
  Define
   \begin{align*}
       \tilde{\delta}_{n} = 2^{\frac{Jd_{z}}{2}}\delta_{n}  +2^{J(\frac{d_{z}}{2}-\alpha_{0})}+\sqrt{n}\delta_{n}2^{J(\frac{d_{z}}{2}-a)} +2^{J(\frac{d_{z}}{2}-a)}+2^{-\alpha_{0}J} + \sqrt{n}\delta_{n}2^{-\alpha J} + 2^{-aJ} 
   \end{align*}
   and use the equivalence between Hellinger distance and $L^{1}$ distance to conclude the proof.
   \end{proof}
   \subsubsection{Proof of Proposition 3}
\begin{proposition}\label{prop:invariance}
Suppose Assumptions \ref{as:dgp2}, \ref{as:correctspec}, \ref{as:boundedinternal}, \ref{as:differentiablemoments}, \ref{as:parametricfamily}, \ref{as:exprior}, and the conditions of Proposition \ref{prop:RMSHrate} hold. Further, suppose there are sequences $\{\tilde{g}_{n,\theta,j}^{*}: \theta \in \Theta\}_{n \geq 1}$, $j=1,...,d_{\gamma},$ in $ \mathcal{H}$ and $\{\zeta_{n}\}_{n\geq 1}$ in $\mathbb{R}_{++}$ with $\zeta_{n}\rightarrow 0$ such that 1. $\sup_{\theta \in \Theta}\max_{1 \leq j \leq d_{\gamma}}||\tilde{g}_{n,\theta,j}^{*}-\tilde{g}_{0,\theta,j}||_{\infty} \lesssim \zeta_{n}$, 2. $\sup_{\theta \in \Theta}\max_{1 \leq j \leq d_{\gamma}}||\tilde{g}_{n,\theta,j}^{*}||_{\mathcal{H}} \lesssim \sqrt{n}\zeta_{n}$, 3. $\sqrt{n}\zeta_{n}\delta_{n}\rightarrow 0$ as $n\rightarrow \infty$, where $\delta_{n}$ is the RMSH rate, and 4. $\{\Delta_{n,\theta}^{*}: \theta \in \Theta\}$ with $\Delta_{n,\theta}^{*}(w,z)=\tilde{g}_{n,\theta}^{*}(F_{\theta,z}(w),z)-\tilde{g}_{0,\theta}(F_{\theta,z}(w),z)$ for $(w,z) \in \mathcal{W} \times [0,1]^{d_{z}}$ satisfies, for $P_{0,Z}^{\infty}$-almost every fixed realization $\{z_{i}\}_{i\geq 1}$ of $\{Z_{i}\}_{i \geq 1}$,
\begin{align*}
 \sup_{\theta \in \Theta}\left|\left|\frac{1}{\sqrt{n}}\sum_{i=1}^{n}\{\Delta_{n,\theta}^{*}(W_{i},z_{i})-E_{P_{0,W|Z}}[\Delta_{n,\theta}^{*}(W,Z)|Z=z_{i}]\}\right|\right|_{2} \overset{P_{0,W|Z}^{(n)}}{\longrightarrow} 0   
\end{align*}
as $n\rightarrow \infty$. Then, for $P_{0,Z}^{\infty}$-almost every fixed realization $\{z_{i}\}_{i\geq 1}$ of $\{Z_{i}\}_{i \geq 1}$, there is $\{\tilde{\mathcal{P}}_{n,W|Z}\}_{n \geq 1}$ such that $\Pi(p_{W|Z} \in \tilde{\mathcal{P}}_{n,W|Z}|W^{(n)}) \rightarrow 1$ in $P_{0,W|Z}^{(n)}$-probability and (4) from the Main Text holds.\footnote{Recall that (4) from the Main Text is, for $P_{0,Z}^{\infty}$-almost every fixed realization $\{z_{i}\}_{i \geq 1}$ of $\{Z_{i}\}_{i \geq 1}$, \begin{align*}
\frac{\int_{\tilde{\mathcal{P}}_{n,W|Z}}L_{n}(p_{t,n,W|Z})d\Pi(p_{W|Z})}{\int_{\tilde{\mathcal{P}}_{n,W|Z}}L_{n}(p_{W|Z}) d\Pi(p_{W|Z})} = 1 + o_{P_{0,W|Z}^{(n)}}(1) 
\end{align*}
for $t$ in a neighborhood of zero, $p_{t,n,W|Z} \propto p_{W|Z}\exp\left(-t'\tilde{\chi}_{0}/\sqrt{n}\right)$, and, for \textit{this} proposition, $p_{W|Z} = p_{\theta,B}$.}
\end{proposition}
   \begin{proof}
The goal is to show that, for $P_{0,Z}^{\infty}$-almost every fixed realization $\{z_{i}\}_{i \geq 1}$ of $\{Z_{i}\}_{i \geq 1}$, there are sets $\{\tilde{\mathcal{P}}_{n,W|Z}\}_{n \geq 1}$ for which $\Pi(\tilde{\mathcal{P}}_{n,W|Z}|W^{(n)} ) 1 + o_{P_{0,W|Z}^{(n)}}(1)$ and
\begin{align*}
\frac{\int_{\tilde{\mathcal{P}}_{n,W|Z}}L_{n}(p_{t,n,\theta,B})d\Pi(p_{W|Z})}{\int_{\tilde{\mathcal{P}}_{n,W|Z}}L_{n}(p_{W|Z}) d\Pi(p_{W|Z})} = 1 + o_{P_{0,W|Z}^{(n)}}(1), 
\end{align*}
where $p_{t,n,\theta,B}(w|z) \propto p_{\theta,B}(w|z)\exp(-t'\tilde{\chi}_{0,\theta}(F_{\theta,z}(w),z)/\sqrt{n})$ for each $(w,z) \in \mathcal{W}\times [0,1]^{d_{z}}$ and $\tilde{\chi}_{0,\theta} = -V_{0}^{-1}\tilde{g}_{0,\theta}$.\footnote{Since $\tilde{\chi}_{0,\theta}(F_{\theta,z}(w),z) = \tilde{\chi}_{0}$, the conditional density $p_{t,n,\theta,B}$ is exactly that of (4) in the Main Text.}
To that end, Step 1 defines sets $\{\tilde{\mathcal{P}}_{n,W|Z}\}_{ n\geq 1}$ and shows that the posterior concentrates on them. Step 2 shows that $L_{n}(p_{t,n,\theta,B})$ can be replaced with $L_{n}(p_{t,n,\theta,B}^{*})$, where $p_{t,n,\theta,B}^{*}(w|z) \propto p_{\theta,B}(w|z)\exp(-t'\tilde{\chi}_{n,\theta}^{*}(F_{\theta,z}(w),z)/\sqrt{n})$ for each $(w,z) \in \mathcal{W}\times [0,1]^{d_{z}}$, with $\tilde{\chi}_{n,\theta}^{*} = -V_{0}^{-1}\tilde{g}_{n,\theta}^{*}$. Step 3 applies the Cameron-Martin theorem to show that the ratio of integrated likelihoods is asymptotically equivalent to a ratio of posterior probabilities. Step 4 shows that this ratio converges in probability to one.
\paragraph{Step 1.} Since the elements of $\tilde{g}_{n,\theta}^{*}$ are contained in $\mathcal{H}$, the elements of $\tilde{\chi}_{n,\theta}^{*}$ are also contained in $\mathcal{H}$ because they are finite linear combinations of $\tilde{g}_{n,\theta}^{*}$. Let $H_{n,\theta}^{*}$ be a $d_{\gamma} \times d_{\gamma}$ positive definite matrix with elements $H_{jk,n,\theta}^{*} = \langle \tilde{\chi}_{n,\theta,j}^{*},\tilde{\chi}_{n,\theta,k}^{*} \rangle_{\mathcal{H}}$. Let $G_{n,t,\theta}^{*}(B)$ be a random variable that satisfies $G_{n,t,\theta}^{*}(B) \sim \mathcal{N}(0,t'H_{n,\theta}^{*}t)$ under $\tilde{\Pi}_{\mathcal{B}}$ for any $t$ in a neighborhood of zero. I define a sequence of sets $\{\tilde{\mathcal{P}}_{n,2,W|Z}\}_{n \geq 1}$ that satisfy 
\begin{align*}
\tilde{\mathcal{P}}_{n,2,W|Z} = \left\{(\theta,B): |G_{n,t,\theta}^{*}(B)|\leq D \sqrt{n}\delta_{n}\sqrt{t'H_{n,\theta}^{*}t}\right\}       
\end{align*}
for each $n \geq 1$. Since $\Pi_{\mathcal{B}}(A) = \tilde{\Pi}_{\mathcal{B}}(A \cap \{||B||_{\infty} \leq \bar{B}\})/\tilde{\Pi}_{\mathcal{B}}(||B||_{\infty} \leq \bar{B})$, it follows from the law of total probability that
\begin{align*}
&\Pi\left((\theta,B): \frac{1}{\sqrt{n}}\left|G_{n,t,\theta}^{*}(B)\right| \geq D\delta_{n}\sqrt{t'H_{n,\theta}^{*}t}\right)\\
&\quad =\int_{\Theta}\Pi_{\mathcal{B}}\left(B:\frac{1}{\sqrt{n}}\left|G_{n,t,\theta}^{*}(B)\right| \geq D\delta_{n}\sqrt{t'H_{n,\theta}^{*}t}\right)\pi_{\Theta}(\theta)d\theta\\
&\quad \quad \leq \frac{\int_{\Theta}\tilde{\Pi}_{\mathcal{B}}\left(B:\frac{1}{\sqrt{n}}\left|G_{n,t,\theta}^{*}(B)\right| \geq D\delta_{n}\sqrt{t'H_{n,\theta}^{*}t}\right)\pi_{\Theta}(\theta)d\theta}{\tilde{\Pi}_{\mathcal{B}}(B:||B||_{\infty} \leq \bar{B})}
\end{align*}
Since $G_{n,t,\theta}^{*}(B) \sim \mathcal{N}(0,t'H_{n,\theta}^{*}t)$ under $\tilde{\Pi}_{\mathcal{B}}$ for $\theta$ fixed, it follows that
\begin{align*}
\tilde{\Pi}_{\mathcal{B}}\left(B:\frac{1}{\sqrt{n}}\left|G_{n,t,\theta}^{*}(B)\right| \geq D\delta_{n}\sqrt{t'H_{n,\theta}^{*}t}\right) = 1- \Phi(D\sqrt{n}\delta_{n}).    
\end{align*}
for all $\theta \in \Theta$. Consequently,
\begin{align*}
    \Pi\left((\theta,B): \frac{1}{\sqrt{n}}\left|G_{n,t,\theta}^{*}(B)\right| \geq D\delta_{n}\sqrt{t'H_{n,\theta}^{*}t}\right) \leq \exp\left(-\frac{D^{2}n\delta_{n}^{2}}{2}\right)
\end{align*}
by the Gaussian tail bound. Lemma \ref{lem:posteriorsieve} then implies $\Pi(\tilde{\mathcal{P}}_{n,2,W|Z}^{c}|W^{(n)}) \rightarrow 0$ in $P_{0,W|Z}^{(n)}$-probability as $n\rightarrow \infty$. Combining this with Step 1 of Proposition 2, it follows that $\tilde{\mathcal{P}}_{n,W|Z} = \tilde{\mathcal{P}}_{n,1,W|Z} \cap \tilde{\mathcal{P}}_{n,2,W|Z},$ where $\tilde{\mathcal{P}}_{n,1,W|Z}$ is defined in Proposition 2, satisfies $$\Pi(\tilde{\mathcal{P}}_{n,W|Z}^{c}|W^{(n)}) \overset{P_{0,W|Z}^{(n)}}{\longrightarrow} 0$$ as $n\rightarrow \infty$.
\paragraph{Step 2.} I show that
\begin{align*}
\sup_{(\theta,B) \in \tilde{\mathcal{P}}_{n,W|Z}}|\log L_{n}(p_{t,n,\theta,B}) - \log L_{n}(p_{t,n,\theta,B}^{*})| \overset{P_{0,W|Z}^{(n)}}{\longrightarrow} 0    
\end{align*}
as $n\rightarrow \infty$. Direct calculations reveal that
\begin{align*}
 &\log L_{n}(p_{t,n,\theta,B})- \log L_{n}(p_{t,n,\theta,B}^{*}) \\
 &= t'\frac{1}{\sqrt{n}}\sum_{i=1}^{n}[\tilde{\chi}_{n,\theta}^{*}(F_{\theta,z_{i}}(W_{i}),z_{i})-\tilde{\chi}_{0,\theta}(F_{\theta,z_{i}}(W_{i}),z_{i})]\\
 &\quad +\sum_{i=1}^{n}\log E_{P_{\theta,B}}[\exp(-t'\tilde{\chi}_{n,\theta}^{*}(F_{\theta,Z}(W),Z)/\sqrt{n})|Z=z_{i}] \\
 &\quad \quad -\sum_{i=1}^{n}\log E_{P_{\theta,B}}[\exp(-t'\tilde{\chi}_{0,\theta}(F_{\theta,Z}(W),Z)/\sqrt{n})|Z=z_{i}].
\end{align*}
Since $E_{P_{0,W|Z}}[\tilde{\chi}_{0,\theta}(F_{\theta,Z}(W),Z)|Z=z_{i}] = E_{P_{0,W|Z}}[\tilde{\chi}_{0}(W,Z)|Z=z_{i}]  = 0$ for $i=1,...,n$ by the definition of $\tilde{\chi}_{0,\theta}$ and Assumption \ref{as:correctspec}, it follows that
\begin{align*}
 &t'\frac{1}{\sqrt{n}}\sum_{i=1}^{n}[\tilde{\chi}_{n,\theta}^{*}(F_{\theta,z_{i}}(W_{i}),z_{i})-\tilde{\chi}_{0,\theta}(F_{\theta,z_{i}}(W_{i}),z_{i})] \\
 &\quad = t'V_{0}^{-1}\frac{1}{\sqrt{n}}\sum_{i=1}^{n}\Delta_{n,\theta}(W_{i},z_{i}) + t'\frac{1}{\sqrt{n}}\sum_{i=1}^{n}E_{P_{0,W|Z}}[\tilde{\chi}_{n,\theta}^{*}(F_{\theta,Z}(W),Z)|Z=z_{i}]  \\
 &\quad =t'\frac{1}{\sqrt{n}}\sum_{i=1}^{n}E_{P_{0,W|Z}}[\tilde{\chi}_{n,\theta}^{*}(F_{\theta,Z}(W),Z)|Z=z_{i}] + o_{P_{0,W|Z}^{(n)}}(1)
\end{align*}
uniformly over $ \Theta$ by assumption. Moreover, a Taylor expansion around $t=0$, the fact that $\tilde{\chi}_{n,\theta}^{*} = V_{0}^{-1}\tilde{g}_{n,\theta}^{*}$ and $\tilde{\chi}_{0,\theta} = V_{0}^{-1}\tilde{g}_{0,\theta}$ and the assumption that $\sup_{\theta \in \Theta}||\tilde{g}_{n,\theta}^{*}-\tilde{g}_{0,\theta}||_{\infty} \rightarrow 0$ as $n\rightarrow \infty$ implies that
\begin{align*}
&\sum_{i=1}^{n}\log E_{P_{\theta,B}}[\exp(-t'\tilde{\chi}_{n,\theta}^{*}(F_{\theta,Z}(W),Z)/\sqrt{n}|Z=z_{i}]\\
&\quad \quad -\sum_{i=1}^{n}\log E_{P_{\theta,B}}[\exp(-t'\tilde{\chi}_{0,\theta}(F_{\theta,Z}(W),Z)/\sqrt{n}|Z=z_{i}] \\
&\quad = -t'\frac{1}{\sqrt{n}}\sum_{i=1}^{n}E_{P_{\theta,B}}[\tilde{\chi}_{n,\theta}^{*}(F_{\theta,Z}(W),Z)|Z=z_{i}] \\
&\quad \quad + t'\frac{1}{\sqrt{n}}\sum_{i=1}^{n}E_{P_{\theta,B}}[\tilde{\chi}_{0,\theta}(F_{\theta,Z}(W),Z)|Z=z_{i}] + o(1)
\end{align*}
uniformly over $\tilde{\mathcal{P}}_{n,W|Z}$. Consequently,
\begin{align*}
&\sup_{(\theta,B) \in \tilde{\mathcal{P}}_{n,W|Z}}\left|\log L_{n}(p_{t,n,\theta,B})- \log L_{n}(p_{t,n,\theta,B}^{*}) \right| \\
&\leq \sup_{(\theta,B) \in \tilde{\mathcal{P}}_{n,W|Z}}\left|\frac{1}{\sqrt{n}}\sum_{i=1}^{n}\int(\tilde{\chi}_{n,\theta}^{*}(F_{\theta,z_{i}}(w),z_{i})-\tilde{\chi}_{0,\theta}(F_{\theta,z_{i}}(w),z_{i}))(p_{\theta,B}(w|z_{i})-p_{0,W|Z}(w|z_{i})) d w \right| \\
&\quad + o_{P_{0,W|Z}^{(n)}}(1) \\
&\leq \sqrt{n}\sup_{\theta \in \Theta}||\tilde{\chi}_{n,\theta}^{*}-\tilde{\chi}_{0,\theta}||_{\infty}\sup_{(\theta,B) \in \tilde{\mathcal{P}}_{n,W|Z}}\frac{1}{n}\sum_{i=1}^{n}||p_{\theta,B}(\cdot|z_{i})-p_{0,W|Z}(\cdot|z_{i})||_{L^{1}} + o_{P_{0,W|Z}^{(n)}}(1)\\
&\leq \sqrt{n}\sup_{\theta \in \Theta}||\tilde{\chi}_{n,\theta}^{*}-\tilde{\chi}_{0,\theta}||_{\infty}\sup_{(\theta,B) \in \tilde{\mathcal{P}}_{n,W|Z}}\sqrt{\frac{1}{n}\sum_{i=1}^{n}||p_{\theta,B}(\cdot|z_{i})-p_{0,W|Z}(\cdot|z_{i})||_{L^{1}}^{2}}+ o_{P_{0,W|Z}^{(n)}}(1) \\
&\lesssim \sqrt{n}\zeta_{n}\delta_{n}+ o_{P_{0,W|Z}^{(n)}}(1),
\end{align*}
where the first inequality is the triangle inequality the results derived earlier, the third is H\"{o}lder's inequality, and the final inequality uses the equivalence between total variation and Hellinger distances as well as the assumption that $\sup_{\theta \in \Theta}||\tilde{g}_{n,\theta}^{*}-\tilde{g}_{0,\theta}||_{\infty} \leq \zeta_{n}$ (and that $\tilde{\chi}_{0,\theta} = -V_{0}^{-1}\tilde{g}_{0,\theta}$ and $\tilde{\chi}_{n,\theta}^{*} = -V_{0}^{-1}g_{n,\theta}^{*}$). Consequently, the assumption that $\sqrt{n}\zeta_{n}\delta_{n}\rightarrow 0$ as $n\rightarrow \infty$ implies
\begin{align*}
\sup_{(\theta,B) \in \tilde{\mathcal{P}}_{n,W|Z}}\left|\log L_{n}(p_{t,n,\theta,B})- \log L_{n}(p_{t,n,\theta,B}^{*}) \right| \overset{P_{0,W|Z}^{(n)}}{\longrightarrow} 0    
\end{align*}
as $n\rightarrow \infty$.
\paragraph{Step 3.} Using the result from Step 1 and 2, it follows that
\begin{align}\label{eq:step1result}
\frac{\int_{\mathcal{P}_{n,W|Z}}L_{n}(p_{t,n,\theta,B})d \Pi(\theta,B)}{\int_{\mathcal{P}_{n,W|Z}}L_{n}(p_{\theta,B})d \Pi(\theta,B)} = (1+o_{P_{0,W|Z}^{(n)}}(1))\frac{\int_{\mathcal{P}_{n,W|Z}}L_{n}(p_{t,n,\theta,B}^{*})d \Pi(\theta,B)}{\int_{\mathcal{P}_{n,W|Z}}L_{n}(p_{\theta,B})d \Pi(\theta,B)} 
\end{align}
Letting $c(u,z,b) =\exp(b(u,z))/\int_{[0,1]^{d_{w}}}\exp(b(\tilde{u},z))d\tilde{u}$, the ratio on the RHS of (\ref{eq:step1result}) satisfies
\begin{align*}
&\frac{\int_{\mathcal{P}_{n,W|Z}}L_{n}(p_{t,n,\theta,B}^{*})d \Pi(\theta,B)}{\int_{\mathcal{P}_{n,W|Z}}L_{n}(p_{\theta,B})d \Pi(\theta,B)}  \\
&\quad = \frac{\int_{\mathcal{P}_{n,W|Z}}\prod_{i=1}^{n}f_{\theta}(W_{i}|z_{i})C(F_{\theta,z_{i}}(W_{i}),z_{i},B-t'\tilde{\chi}_{n,\theta}^{*}/\sqrt{n})d \Pi(\theta,B)}{\int_{\mathcal{P}_{n,W|Z}}\prod_{i=1}^{n}f_{\theta}(W_{i}|z_{i})C(F_{\theta,z_{i}}(W_{i}),z_{i},B)d \Pi(\theta,B)},    
\end{align*}
Under Assumption \ref{as:exprior}, I then conclude that
\begin{align*}
&\frac{\int_{\tilde{\mathcal{P}}_{n,W|Z}}\prod_{i=1}^{n}f_{\theta}(W_{i}|z_{i})C(F_{\theta,z_{i}}(W_{i}),z_{i},B-t'\tilde{\chi}_{n,\theta}^{*}/\sqrt{n})d \Pi(\theta,B)}{\int_{\tilde{\mathcal{P}}_{n,W|Z}}\prod_{i=1}^{n}f_{\theta}(W_{i}|z_{i})C(F_{\theta,z_{i}}(W_{i}),z_{i},B)d \Pi(\theta,B)} \\
&\quad =    \frac{\int_{\Theta}\prod_{i=1}^{n}f_{\theta}(W_{i}|z_{i}) \int \delta_{\theta,B}(\tilde{\mathcal{P}}_{n,W|Z})\prod_{i=1}^{n}C(F_{\theta,z_{i}}(W_{i}),z_{i},B-t'\tilde{\chi}_{n,\theta}^{*}/\sqrt{n})d \tilde{\Pi}_{\mathcal{B}}(B)\pi_{\Theta}(\theta)d\theta}{\int_{\Theta}\prod_{i=1}^{n}f_{\theta}(W_{i}|z_{i}) \int \delta_{\theta,B}(\tilde{\mathcal{P}}_{n,W|Z})\prod_{i=1}^{n}C(F_{\theta,z_{i}}(W_{i}),z_{i},B)d \tilde{\Pi}_{\mathcal{B}}(B)\pi_{\Theta}(\theta)d\theta},
\end{align*}
where $\delta_{\theta,B}(\tilde{\mathcal{P}}_{n,W|Z}) = \mathbf{1}\{(\theta,B) \in \tilde{\mathcal{P}}_{n,W|Z}\}$. Note that the integrals are taken with respect to $\tilde{\Pi}_{\mathcal{B}}$ because $\tilde{\mathcal{P}}_{n,W|Z}$ imposes $||B||_{\infty} \leq \bar{B}$ and $1/\tilde{\Pi}_{\mathcal{B}}(||B||_{\infty} \leq \bar{B})$ is a common positive constant in the numerator and the denominator. Applying the change of variables $\upsilon = B - t'\tilde{\chi}_{n,\theta}^{*}/\sqrt{n}$,
\begin{align*}
    &\frac{\int_{\Theta}\prod_{i=1}^{n}f_{\theta}(W_{i}|z_{i}) \int \delta_{\theta,B}(\tilde{\mathcal{P}}_{n,W|Z})\prod_{i=1}^{n}C(F_{\theta,z_{i}}(W_{i}),z_{i},B-t'\tilde{\chi}_{n,\theta}^{*}/\sqrt{n})d \tilde{\Pi}_{\mathcal{B}}(B)\pi_{\Theta}(\theta)d\theta}{\int_{\Theta}\prod_{i=1}^{n}f_{\theta}(W_{i}|z_{i}) \int \delta_{\theta,B}(\tilde{\mathcal{P}}_{n,W|Z})\prod_{i=1}^{n}C(F_{\theta,z_{i}}(W_{i}),z_{i},B)d \tilde{\Pi}_{\mathcal{B}}(B)\pi_{\Theta}(\theta)d\theta} \\
    &\quad =\frac{\int_{\Theta}\prod_{i=1}^{n}f_{\theta}(W_{i}|z_{i}) \int \delta_{\theta,v+t'\tilde{\chi}_{n,\theta}^{*}/\sqrt{n}}(\tilde{\mathcal{P}}_{n,W|Z})\prod_{i=1}^{n}C(F_{\theta,z_{i}}(W_{i}),z_{i},\upsilon)d \tilde{\Pi}_{t,n,\theta,\mathcal{B}}(\upsilon)\pi_{\Theta}(\theta)d\theta}{\int_{\Theta}\prod_{i=1}^{n}f_{\theta}(W_{i}|z_{i}) \int \delta_{\theta,B}(\tilde{\mathcal{P}}_{n,W|Z})\prod_{i=1}^{n}C(F_{\theta,z_{i}}(W_{i}),z_{i},B)d \tilde{\Pi}_{\mathcal{B}}(B)\pi_{\Theta}(\theta)d\theta},
\end{align*}
where $\tilde{\Pi}_{t,n,\theta,\mathcal{B}}$ is the probability law of $B-t'\tilde{\chi}_{n,\theta}^{*}/\sqrt{n}$ under $\tilde{\Pi}_{\mathcal{B}}$ for $\theta$ fixed. Since $\{\tilde{\chi}_{n,\theta,j}^{*}: \theta \in \Theta, j=1,...,d_{\gamma}\} \subseteq \mathcal{H}$ for all $n \geq 1$, I know that $t'\tilde{\chi}_{n,\theta}^{*} \in \mathcal{H}$ for all $\theta \in \Theta$ (as a linear space), and, as a result, the Cameron-Martin theorem (Lemma 3.2 in \cite{van2008reproducing}) yields that $\tilde{\Pi}_{t,n,\theta,\mathcal{B}} \ll \tilde{\Pi}_{\mathcal{B}}$ with Radon-Nikodym derivative
\begin{align*}
\left(\frac{d\tilde{\Pi}_{t,n,\theta,\mathcal{B}}}{d\tilde{\Pi}_{\mathcal{B}}} \right)(B)= \exp\left(\frac{1}{\sqrt{n}}G_{n,t,\theta}^{*}(B)-\frac{1}{2n}t'H_{n,\theta}^{*}t\right),   
\end{align*}
where I recall from Step 1 that $H_{n,\theta}^{*}$ is a $d_{\gamma} \times d_{\gamma}$ positive definite matrix with elements $H_{jk,n,\theta}^{*} = \langle \tilde{\chi}_{n,\theta,j}^{*},\tilde{\chi}_{n,\theta,k}^{*} \rangle_{\mathcal{H}}$ and $G_{n,t,\theta}^{*}(B)$ is a random variable that satisfies $G_{n,t,\theta}^{*}(B) \sim \mathcal{N}(0,t'H_{n,\theta}^{*}t)$ under $\tilde{\Pi}_{\mathcal{B}}$. I argue that $\sup_{\theta \in \Theta} t'H_{n,\theta}^{*}t \leq C ||t||_{2}^{2}n \zeta_{n}^{2}$ for some constant. Applying the triangle inequality and the Cauchy-Schwarz inequality I obtain,
\begin{align*}
  t'H_{n,\theta}^{*}t &\leq \sum_{j=1}^{d_{\gamma}}\sum_{k=1}^{d_{\gamma}}|t_{j}|\cdot |t_{k}| \cdot | \langle \tilde{\chi}_{n,\theta,j}^{*}, \tilde{\chi}_{n,\theta,k}^{*} \rangle_{\mathcal{H}}| \\
    &\leq \sum_{j=1}^{d_{\gamma}}\sum_{k=1}^{d_{\gamma}} |t_{j}|\cdot |t_{k}| \cdot || \tilde{\chi}_{n,\theta,j}^{*}||_{\mathcal{H}}\cdot ||\tilde{\chi}_{n,\theta,k}^{*} ||_{\mathcal{H}}
    \end{align*}
Since $\tilde{\chi}_{n,\theta,j}^{*} = \sum_{l=1}^{d_{\gamma}}a_{0,jl}\tilde{g}_{n,\theta,l}^{*}$ with $a_{0,jl}$ being the $(j,l)$th entry of $V_{0}^{-1}$, it follows from the linearity of inner products and that
\begin{align*}
|| \tilde{\chi}_{n,\theta,j}^{*}||_{\mathcal{H}}^{2}   = \left | \sum_{l=1}^{d_{\gamma}}\sum_{l'=1}^{d_{\gamma}}a_{0,jl}a_{0,jl'}\langle \tilde{g}_{n,\theta,l}^{*},\tilde{g}_{n,\theta,l'}^{*} \rangle_{\mathcal{H}} \right| \leq C_{j} \left(\sup_{\theta \in \Theta}\max_{1 \leq l \leq d_{\gamma}}||g_{n,\theta,l}^{*}||_{\mathcal{H}}\right)^{2},
\end{align*}
where $C_{j} = \sum_{l=1}^{d_{\gamma}}a_{0,jl}^{2}$. Taking the maximum and applying H\"{o}lder's inequality, it follows that there exists $C >0$ such that
\begin{align*}
  t'H_{n,\theta}^{*}t \leq  C ||t||^{2}_{2}    \left(\sup_{\theta \in \Theta}\max_{1 \leq l \leq d_{\gamma}}||g_{n,\theta,l}^{*}||_{\mathcal{H}}\right)^{2} \lesssim ||t||_{2}^{2}n \zeta_{n}^{2}
\end{align*}
with the last inequality holding by assumption. Consequently, on the event $\tilde{\mathcal{P}}_{n,W|Z}$, I know that
\begin{align*}
    \sup_{\theta \in \Theta}\left|\frac{1}{\sqrt{n}}G_{n,t,\theta}^{*}(B) \right| \lesssim \delta_{n} \sqrt{t'H_{n,\theta}^{*}t} \lesssim ||t||_{2}\sqrt{n}\delta_{n}\zeta_{n}
\end{align*}
and
\begin{align*}
 \sup_{\theta \in \Theta}\frac{1}{2n}t'H_{n,\theta}^{*}t \lesssim \zeta_{n}^{2}.   
\end{align*}
This implies that, for any $t$ in a neighborhood of zero,
\begin{align*}
 \sup_{(\theta,B) \in \tilde{\mathcal{P}}_{n,W|Z}}\left|\log \left( \frac{d\tilde{\Pi}_{t,n,\theta,\mathcal{B}}}{d\tilde{\Pi}_{\mathcal{B}}} \right)(B)\right|   \leq \tilde{C}_{1}||t||_{2}\sqrt{n}\delta_{n}\zeta_{n}+\tilde{C}_{2}\zeta_{n}^{2}||t||_{2}^{2}
\end{align*}
for some universal constants $\tilde{C}_{1},\tilde{C}_{2}>0$. Using that $\sqrt{n}\delta_{n}\zeta_{n}\rightarrow 0$ and $\zeta_{n}\rightarrow 0$ as $n\rightarrow \infty$ (by assumption), I conclude that
\begin{align*}
    \sup_{(\theta,B) \in \tilde{\mathcal{P}}_{n,W|Z}}\left|\log \left( \frac{d\tilde{\Pi}_{t,n,\theta,\mathcal{B}}}{d\tilde{\Pi}_{\mathcal{B}}} \right)(B)\right| \longrightarrow 0
\end{align*}
as $n\rightarrow \infty$, and, as a result, I can conclude that
\begin{align*}
&\frac{\int_{\tilde{\mathcal{P}}_{n,W|Z}}L_{n}(p_{t,n,\theta,B})d \Pi(\theta,B)}{\int_{\tilde{\mathcal{P}}_{n,W|Z}}L_{n}(p_{\theta,B})d \Pi(\theta,B)} \\
&=(1+o_{P_{0,W|Z}^{(n)}}(1))\\
&\quad \times \frac{\int_{\Theta}\prod_{i=1}^{n}f_{\theta}(W_{i}|z_{i}) \int \delta_{\theta,v+t'\tilde{\chi}_{n,\theta}^{*}/\sqrt{n}}(\tilde{\mathcal{P}}_{n,W|Z})\prod_{i=1}^{n}C(F_{\theta,z_{i}}(W_{i}),z_{i},\upsilon)d \tilde{\Pi}_{t,n,\theta,\mathcal{B}}(\upsilon)\pi_{\Theta}(\theta)d\theta}{\int_{\Theta}\prod_{i=1}^{n}f_{\theta}(W_{i}|z_{i}) \int \delta_{\theta,B}(\tilde{\mathcal{P}}_{n,W|Z})\prod_{i=1}^{n}C(F_{\theta,z_{i}}(W_{i}),z_{i},B)d \tilde{\Pi}_{\mathcal{B}}(B)\pi_{\Theta}(\theta)d\theta} \\
&= (1+o_{P_{0,W|Z}^{(n)}}(1)) \\
&\quad \times \frac{\int_{\Theta}\prod_{i=1}^{n}f_{\theta}(W_{i}|z_{i}) \int \delta_{\theta,v+t'\tilde{\chi}_{n,\theta}^{*}/\sqrt{n}}(\tilde{\mathcal{P}}_{n,W|Z})\prod_{i=1}^{n}C(F_{\theta,z_{i}}(W_{i}),z_{i},\upsilon)d \tilde{\Pi}_{\mathcal{B}}(\upsilon)\pi_{\Theta}(\theta)d\theta}{\int_{\Theta}\prod_{i=1}^{n}f_{\theta}(W_{i}|z_{i}) \int \delta_{\theta,B}(\tilde{\mathcal{P}}_{n,W|Z})\prod_{i=1}^{n}C(F_{\theta,z_{i}}(W_{i}),z_{i},B)d \tilde{\Pi}_{\mathcal{B}}(B)\pi_{\Theta}(\theta)d\theta} \\
&= (1+o_{P_{0,W|Z}^{(n)}}(1))\frac{\int_{\{(\theta,B): (\theta,B+t'\tilde{\chi}_{n,\theta}^{*}/\sqrt{n}) \in \tilde{\mathcal{P}}_{n,W|Z}\} }L_{n}(p_{\theta,B})d \Pi_{\mathcal{B}}(B)\pi_{\Theta}(\theta)d\theta}{\int_{\tilde{\mathcal{P}}_{n,W|Z} }L_{n}(p_{\theta,B})d \Pi_{\mathcal{B}}(B)\pi_{\Theta}(\theta)d\theta} \\
    &= (1+o_{P_{0,W|Z}^{(n)}}(1))\frac{\Pi\left((\theta,B): (\theta,B + t'\tilde{\chi}_{n,\theta}^{*}/\sqrt{n}) \in \tilde{\mathcal{P}}_{n,W|Z} \middle|W^{(n)}\right)}{\Pi\left(\tilde{\mathcal{P}}_{n,W|Z}\middle|W^{(n)}\right)},
\end{align*}
where the last equality holds by dividing and multiplying by $\int L_{n}(p_{\theta,B})d \Pi_{\mathcal{B}}(B)\pi_{\Theta}(\theta)d\theta$ and applying the definition of the posterior.
\paragraph{Step 4.} It remains to show that
\begin{align*}
\frac{\Pi\left((\theta,B): (\theta,B + t'\tilde{\chi}_{n,\theta}^{*}/\sqrt{n}) \in \tilde{\mathcal{P}}_{n,W|Z} \middle|W^{(n)}\right)}{\Pi\left(\tilde{\mathcal{P}}_{n,W|Z}\middle|W^{(n)}\right)} = 1 + o_{P_{0,W|Z}^{(n)}}(1)
\end{align*}
Since I established that
\begin{align*}
\Pi\left(\tilde{\mathcal{P}}_{n,W|Z}\middle|W^{(n)}\right) = 1 + o_{P_{0,W|Z}^{(n)}}(1)
\end{align*}
in Step 1, I need to show that
\begin{align*}
\Pi\left((\theta,B): (\theta,B + t'\tilde{\chi}_{n,\theta}^{*}/\sqrt{n}) \in \tilde{\mathcal{P}}_{n,W|Z} \middle|W^{(n)}\right) = 1 + o_{P_{0,W|Z}^{(n)}}(1)
\end{align*}
to conclude the result. Using the definition of $\tilde{\mathcal{P}}_{n,W|Z}$, it follows that
\begin{align*}
&\Pi\left((\theta,B): (\theta,B + t'\tilde{\chi}_{n,\theta}^{*}/\sqrt{n}) \notin \tilde{\mathcal{P}}_{n,W|Z} \middle|W^{(n)}\right)\\
&\leq \Pi\left((\theta,B): (\theta,B + t'\tilde{\chi}_{n,\theta}^{*}/\sqrt{n}) \notin \tilde{\mathcal{P}}_{n,1,W|Z} \middle|W^{(n)}\right) \\
&\quad +\Pi\left((\theta,B): (\theta,B + t'\tilde{\chi}_{n,\theta}^{*}/\sqrt{n}) \notin \tilde{\mathcal{P}}_{n,2,W|Z} \middle|W^{(n)}\right) \\
&\leq \Pi\left(||B + t'\tilde{\chi}_{n,\theta}^{*}/\sqrt{n}||_{a}>D\sqrt{n}\delta_{n}\middle | W^{(n)}\right) +\Pi\left(h_{n,2}(p_{\theta,B+t'\tilde{\chi}_{n,\theta}^{*}/\sqrt{n}},p_{0,W|Z}) \geq D \delta_{n} \middle | W^{(n)}\right) \\
&\quad \quad +\Pi\left(\frac{1}{\sqrt{n}}|G_{n,t,\theta}^{*}(B+t'\tilde{\chi}_{n,\theta}^{*}/\sqrt{n})| > L\delta_{n}\sqrt{t'H_{n,\theta}^{*}t}\middle| W^{(n)}\right)
\end{align*}
I show these probabilities converges to zero in $P_{0,W|Z}^{(n)}$-probability in the following substeps.
\paragraph{Step 4A.} For the first probability, Lemma 8.1 of \cite{van2008reproducing} implies that I can view $B$ as a Gaussian random variable in $(C^{a}([0,1]^{d},\mathbb{R}),||\cdot||_{a})$ without changing the RKHS $(\mathcal{H},\langle \cdot,\cdot \rangle_{\mathcal{H}})$ because $C([0,1]^{d},\mathbb{R})$ is the completion of $C^{a}([0,1]^{d},\mathbb{R})$ with respect to the uniform norm $||\cdot||_{\infty}$ and $||\cdot||_{\infty} \leq ||\cdot||_{a}$. Since the RKHS norm is stronger than the Banach space norm, it follows $||t'\tilde{\chi}_{n,\theta}^{*}||_{a}/\sqrt{n} \lesssim \sup_{\theta \in \Theta}||t'\tilde{\chi}_{n,\theta}^{*}||_{\mathcal{H}}/\sqrt{n}$, and, as a result, I deduce that $\sup_{\theta \in \Theta}||t'\tilde{\chi}_{n,\theta}^{*}||_{a}/\sqrt{n}\rightarrow 0$ as $n\rightarrow \infty$ because of the assumption on the RKHS norm of $\tilde{g}_{n,\theta}^{*}$. This means that the event $ \{||B + t'\tilde{\chi}_{n,\theta}^{*}/\sqrt{n}||_{a}>D\sqrt{n}\delta_{n} \}$ hardly differs from  $\left\{||B ||_{a}>D\sqrt{n}\delta_{n} \right\}$, and, as a result, the same argument as encountered in Step 1 of Proposition 2 leads to $\Pi(||B + t'\tilde{\chi}_{n,\theta}^{*}/\sqrt{n}||_{a}>D\sqrt{n}\delta_{n}|W^{(n)})\rightarrow 0$ in $P_{0,W|Z}^{(n)}$-probability as $n\rightarrow \infty$.
\paragraph{Step 4B.} For the second probability, I use the triangle inequality to conclude that
\begin{align*}
h\left(p_{\theta,B+t'\tilde{\chi}_{n,\theta}^{*}/\sqrt{n}}(\cdot|z_{i}),p_{0,W|Z}(\cdot|z_{i})\right) &\leq  h\left(p_{\theta,B+t'\tilde{\chi}_{n,\theta}^{*}/\sqrt{n}}(\cdot|z_{i}),p_{\theta,B}(\cdot|z_{i})\right)\\
&\quad +  h\left(p_{\theta,B}(\cdot|z_{i}),p_{0,W|Z}(\cdot|z_{i})\right)
\end{align*}
Applying Lemma 3.1 of \cite{vaart2008rates} and adding/subtracting $t'\tilde{\chi}_{0,\theta}$ inside $||\cdot||_{\infty}$, I know that
\begin{align*}
    &h\left(p_{\theta,B+t'\tilde{\chi}_{n,\theta}^{*}/\sqrt{n}}(\cdot|z_{i}),p_{\theta,B}(\cdot|z_{i})\right) \\
    &\quad \leq \frac{1}{\sqrt{n}}||t'\tilde{\chi}_{n,\theta}^{*}||_{\infty}\exp\left(||t'\tilde{\chi}_{n,\theta}^{*}||_{\infty}/2\sqrt{n}\right) \\
    &\quad \leq \frac{1}{\sqrt{n}}||t'\tilde{\chi}_{n,\theta}^{*}-t'\tilde{\chi}_{0,\theta}||_{\infty}\exp\left(||t'\tilde{\chi}_{n,\theta}^{*}||_{\infty}/2\sqrt{n}\right) +\frac{1}{\sqrt{n}}||t'\tilde{\chi}_{0,\theta}||_{\infty}\exp\left(||t'\tilde{\chi}_{n,\theta}^{*}||_{\infty}/2\sqrt{n}\right) \\
    &\quad \leq \frac{1}{\sqrt{n}}||t'\tilde{\chi}_{n,\theta}^{*}-t'\tilde{\chi}_{0,\theta}||_{\infty}\exp\left(||t'\tilde{\chi}_{n,\theta}^{*}-t'\tilde{\chi}_{0,\theta}||_{\infty}/2\sqrt{n}\right)\exp\left(||t'\tilde{\chi}_{0,\theta}||_{\infty}/2\sqrt{n}\right) \\
    &\quad +\frac{1}{\sqrt{n}}||t'\tilde{\chi}_{0,\theta}||_{\infty}\exp\left(||t'\tilde{\chi}_{n,\theta}^{*}-t'\tilde{\chi}_{0,\theta}||_{\infty}/2\sqrt{n}\right)\exp\left(||t'\tilde{\chi}_{0,\theta}||_{\infty}/2\sqrt{n}\right) \\
    &\quad = o(1) + (1+o(1))\frac{1}{\sqrt{n}}||t'\tilde{\chi}_{0,\theta}||_{\infty}
\end{align*}
for $i=1,...,n$, where the last equality holds because of the assumption that $\{\tilde{g}_{n,\theta,j}\}_{n \geq 1}$ satisfies $\max_{1 \leq j \leq d_{\gamma}}||\tilde{g}_{n,\theta,j}-\tilde{g}_{0,\theta,j}||_{\infty} \rightarrow 0$ as $n\rightarrow \infty$. Since the assumptions imply $\{\tilde{\chi}_{0,\theta}: \theta \in \Theta\}$ forms a uniformly bounded class of functions, there exists a constant $C \in (0,\infty)$ such that
\begin{align*}
    &\Pi\left(h_{n,2}(p_{\theta,B+t'\tilde{\chi}_{n,\theta}^{*}/\sqrt{n}},p_{0,W|Z}) \geq D \delta_{n} \middle | W^{(n)}\right) \\
    &\quad \leq \Pi\left(h_{n}\left(p_{\theta,B},p_{0,W|Z}\right) \geq D\delta_{n} - C/\sqrt{n}\middle| W^{(n)}\right).
\end{align*}
for any $t$ in a neighborhood of zero. Proposition 1 then implies that
\begin{align*}
\Pi\left(h_{n}\left(p_{\theta,B},p_{0,W|Z}\right) \geq D\delta_{n} - C/\sqrt{n}\middle| W^{(n)}\right) = o_{P_{0,W|Z}^{(n)}}(1)
\end{align*}
as $n\rightarrow \infty$.
\paragraph{Step 4C.} For the third probability, I apply Lemma 13 of \cite{ray2020semiparametric} to conclude that 
\begin{align*}
G_{n,t,\theta}^{*}(B+t'\tilde{\chi}_{n,\theta}^{*}/\sqrt{n}) \sim \mathcal{N}\left(-\frac{t'H_{n,\theta}^{*}t}{\sqrt{n}},t'H_{n,\theta}^{*}t\right)    
\end{align*} 
under $\tilde{\Pi}_{\mathcal{B}}$. Since the mean $t'H_{n,\theta}^{*}t/\sqrt{n}$ is negligible relative to the standard deviation, a similar application of the Gaussian tail bound to that encountered Step 1 reveals that 
\begin{align*}
\Pi\left(\frac{1}{\sqrt{n}}|G_{n,t,\theta}^{*}(B+t'\tilde{\chi}_{n,\theta}^{*}/\sqrt{n})| \geq D\delta_{n}\sqrt{t'H_{n,\theta}^{*}t}\middle|W^{(n)}\right)=o_{P_{0,W|Z}^{(n)}}(1).    
\end{align*}
\end{proof}
\subsection{Proof of Proposition 4}
\begin{proposition}\label{prop:fixedpoint}
If Assumptions \ref{as:param} and \ref{as:fp} hold, then Assumption \ref{as:uemfp} holds.
\end{proposition}
\begin{proof}
The proof has two steps. Step 1 verifies that $q_{n}(\tilde{\gamma},p_{W|Z})$ exists, is unique, and is a measurable function of $p_{W|Z}$. Step 2 verifies that $\tilde{\gamma} \mapsto q_{n}(\tilde{\gamma},p_{W|Z})$ is a random contraction mapping and has a unique fixed point $\gamma_{n}$. 
\paragraph{Step 1.} An $q_{n}(\tilde{\gamma},p_{W|Z}) \in \Gamma$ such that $Q_{n}(q_{n}(\tilde{\gamma},p_{W|Z}),p_{W|Z}) = \inf_{\gamma \in \Gamma}Q_{n}(\gamma,p_{W|Z})$ exists for each $p_{W|Z}$ by Assumptions \ref{as:param}.1 and \ref{as:param}.3. Uniqueness of $q_{n}(\tilde{\gamma},p_{W|Z})$ for each $p_{W|Z}$ holds by Assumption \ref{as:param}.4. Notational adjustments to Lemma 2 of \cite{jennrich1969asymptotic} establishes measurability of $p_{W|Z} \mapsto q_{n}(\tilde{\gamma},p_{W|Z})$.
\paragraph{Step 2.} Let $F_{n}(\gamma,\tilde{\gamma},p_{W|Z}) = 2n^{-1}\sum_{i=1}^{n}M(z_{i},\gamma)'\Sigma(z_{i},\tilde{\gamma})^{-1}m(z_{i},\gamma)$ for every $(\gamma,\tilde{\gamma},p_{W|Z}) \in \Gamma \times \Gamma \times \mathcal{P}_{W|Z}$. Since $q_{n}(\tilde{\gamma},p_{W|Z})$ is an interior minimizer for every $p_{W|Z} \in \mathcal{P}_{W|Z}$, it follows that 
\begin{align*}
F_{n}(q_{n}(\tilde{\gamma},p_{W|Z}),\tilde{\gamma},p_{W|Z})  = 0    
\end{align*} 
for every $p_{W|Z} \in \mathcal{P}_{W|Z}$. The twice continuous differentiability of $\gamma \mapsto m(z,\gamma)$, the continuously differentiability of $\tilde{\gamma} \mapsto \Sigma(z,\tilde{\gamma})$, and the positive definiteness requirement for $\Sigma(z,\gamma)$ enables use of the implicit function theorem to conclude that $(\gamma,\tilde{\gamma}) \mapsto F_{n}(\gamma,\tilde{\gamma},p_{W|Z})$ is continuously differentiable for each $p_{W|Z} \in \mathcal{P}_{W|Z}$ with derivative equal to
\begin{align*}
    \frac{\partial }{\partial \tilde{\gamma}'}q_{n}(\tilde{\gamma},p_{W|Z}) = -\left(\frac{\partial^{2}}{\partial x_{1} \partial x_{1}'}Q_{n}(q_{n}(\tilde{\gamma},p_{W|Z}),\tilde{\gamma},p_{W|Z})\right)^{-1}\frac{\partial^{2}}{\partial x_{1} \partial x_{2}'}Q_{n}(q_{n}(\tilde{\gamma},p_{W|Z}),\tilde{\gamma},p_{W|Z}),
\end{align*}
where it is recalled that $x_{1}$ and $x_{2}$ refer to the first two arguments of $Q_{n}$ (i.e., $Q_{n}(x_{1},x_{2},p_{W|Z})$). Consequently, the submultiplicativity of the operator norm implies
\begin{align*}
 \left | \left | \frac{\partial }{\partial \tilde{\gamma}'}q_{n}(\tilde{\gamma},p_{W|Z})  \right| \right|_{op} &\leq \sup_{\tilde{\gamma} \in \Gamma}\left | \left|\left(\frac{\partial^{2}}{\partial x_{1} \partial x_{1}'}Q_{n}(q_{n}(\tilde{\gamma},p_{W|Z}),\tilde{\gamma},p_{W|Z})\right )^{-1}\right| \right|_{op} \\
 &\quad \times \sup_{\tilde{\gamma} \in \Gamma}\left| \left|  \frac{\partial^{2}}{\partial x_{1} \partial x_{2}'}Q_{n}(q_{n}(\tilde{\gamma},p_{W|Z}),\tilde{\gamma},p_{W|Z})\right| \right|_{op}
\end{align*}
Since a Hessian matrix is real and symmetric,
\begin{align*}
&\sup_{\tilde{\gamma} \in \Gamma}\left | \left|\left(\frac{\partial^{2}}{\partial x_{1} \partial x_{1}'}Q_{n}(q_{n}(\tilde{\gamma},p_{W|Z}),\tilde{\gamma},p_{W|Z})\right )^{-1}\right| \right|_{op} \\ &\quad \leq    \sup_{\tilde{\gamma} \in \Gamma}\lambda_{min}^{-1}\left(\frac{\partial^{2}}{\partial x_{1} \partial x_{1}'}Q_{n}(q_{n}(\tilde{\gamma},p_{W|Z}),\tilde{\gamma},p_{W|Z})\right ) \\
&\quad \leq \left(\inf_{\tilde{\gamma} \in \Gamma}\lambda_{min}\left(\frac{\partial^{2}}{\partial x_{1} \partial x_{1}'}Q_{n}(q_{n}(\tilde{\gamma},p_{W|Z}),\tilde{\gamma},p_{W|Z})\right ) \right)^{-1}
\end{align*}
Consequently, Assumption \ref{as:fp}.4 implies that 
\begin{align*}
\sup_{\tilde{\gamma} \in \Gamma}\left | \left | \frac{\partial q_{n}(\tilde{\gamma},p_{W|Z})}{\partial \tilde{\gamma}'}  \right| \right|_{op} < 1    
\end{align*} 
for each $p_{W|Z} \in \mathcal{P}_{W|Z}$. Coupled with the mean-value theorem, this implies that the continuous random operator $q(\cdot,p_{W|Z}):\Gamma \rightarrow \Gamma$ is a random contraction operator. Theorem 7 of \cite{bharucha1976fixed} then implies there is a random variable $\gamma_{n}$ which is the unique fixed point of $q_{n}(\cdot,p_{W|Z})$.
\end{proof}
\newpage
{\noindent\Large Supplementary Material (Not For Publication)}
\section{Proofs of Lemmas}
\subsection{Proofs of Lemmas Relevant to Theorem 1 and 2}
\setcounter{lemma}{0}

\begin{lemma}\label{lem:covariance2}
Suppose Assumptions \ref{as:covariance}.1 and \ref{as:concentration}.3 hold. Then, for $P_{0,Z}^{\infty}$-almost every fixed realization $\{z_{i}\}_{i \geq 1}$ of $\{Z_{i}\}_{i \geq 1}$,
\begin{align*}
\sup_{(\gamma,p_{W|Z}) \in \Gamma \times \tilde{\mathcal{P}}_{n,W|Z}}\max_{1 \leq i \leq n}||\Sigma^{-1}(z_{i},\gamma)||_{op} = O(1),   
\end{align*}
where $\{\tilde{\mathcal{P}}_{n,W|Z}\}_{n \geq 1}$ are the sets introduced in Assumption \ref{as:concentration}.
\end{lemma}
\begin{proof}
First, note that Assumption and \ref{as:covariance}.1 \ref{as:concentration}.3(c) imply that, for $P_{0,Z}^{\infty}$-almost every fixed realization $\{z_{i}\}_{i \geq 1}$ of $\{Z_{i}\}_{i \geq 1}$,
\begin{align*}
   \sup_{(\gamma,p_{W|Z}) \in \Gamma \times \tilde{\mathcal{P}}_{n,W|Z}}\max_{1 \leq i \leq n}\lambda_{max}(\Sigma(z_{i},\gamma)) \leq 2 \overline{\lambda}
\end{align*}
and
\begin{align*}
    \inf_{(\gamma,p_{W|Z}) \in \Gamma \times \tilde{\mathcal{P}}_{n,W|Z}}\min_{1 \leq i \leq n}\lambda_{min}(\Sigma(z_{i},\gamma)) \geq \frac{1}{2}\underline{\lambda}
\end{align*}
for $n$ large, by Weyl's inequality. Consequently, for $P_{0,Z}^{\infty}$-almost every fixed realization $\{z_{i}\}_{i \geq 1}$ of $\{Z_{i}\}_{i \geq 1}$, $\{\{\Sigma(z_{i},\gamma)\}_{i=1}^{n}:(\gamma,p_{W|Z}) \in \Gamma \times \tilde{\mathcal{P}}_{n,W|Z}\}$ forms a class of invertible matrices for $n$ large. Next, by Assumptions \ref{as:covariance}.1, \ref{as:concentration}.3(c), and the submultiplicativity of the operator norm, I know that for $P_{0,Z}^{\infty}$-almost every fixed realization $\{z_{i}\}_{i \geq 1}$ of $\{Z_{i}\}_{i \geq 1}$,
\begin{align}\label{eq:boundfornlarge}
\sup_{(\gamma,p_{W|Z}) \in \Gamma \times \tilde{\mathcal{P}}_{n,W|Z}}\max_{1 \leq i \leq n}||\Sigma_{0}^{-1}(z_{i},\gamma)(\Sigma(z_{i},\gamma)-\Sigma_{0}(z_{i},\gamma))||_{op} < \frac{1}{2}
\end{align}
for $n$ large. Consequently, writing 
\begin{align*}
\Sigma^{-1}(z_{i},\gamma) = \left(I + \Sigma_{0}^{-1}(z_{i},\gamma)(\Sigma(z_{i},\gamma)-\Sigma_{0}(z_{i},\gamma))\right)^{-1}\Sigma_{0}^{-1}(z_{i},\gamma),    
\end{align*} 
it follows that, for $P_{0,Z}^{\infty}$-almost every fixed realization $\{z_{i}\}_{i \geq 1}$ of $\{Z_{i}\}_{i \geq 1}$,
\begin{align*}
  &\max_{1 \leq i \leq n}  ||\Sigma^{-1}(z_{i},\gamma)||_{op} \\
  &\quad = \max_{1 \leq i \leq n}||(I + \Sigma_{0}^{-1}(z_{i},\gamma)(\Sigma(z_{i},\gamma)-\Sigma_{0}(z_{i},\gamma)))^{-1}\Sigma_{0}^{-1}(z_{i},\gamma)||_{op} \\
  &\quad \leq  \max_{1 \leq i \leq n} ||\Sigma_{0}^{-1}(z_{i},\gamma)||_{op}\cdot  \max_{1 \leq i \leq n}||(I + \Sigma_{0}^{-1}(z_{i},\gamma)(\Sigma(z_{i},\gamma)-\Sigma_{0}(z_{i},\gamma)))^{-1}||_{op} \\
  &\quad = \max_{1 \leq i \leq n} ||\Sigma_{0}^{-1}(z_{i},\gamma)||_{op}\cdot \max_{1 \leq i \leq n}\left | \left  |\sum_{k=0}^{\infty}(-1)^{k}(\Sigma_{0}^{-1}(z_{i},\gamma)(\Sigma(z_{i},\gamma)-\Sigma_{0}(z_{i},\gamma)))^{k}\right|\right|_{op} \\
  &\quad \leq  \max_{1 \leq i \leq n} ||\Sigma_{0}^{-1}(z_{i},\gamma)||_{op}\cdot \max_{1 \leq i \leq n}\sum_{k=0}^{\infty}||(\Sigma_{0}^{-1}(z_{i},\gamma)(\Sigma(z_{i},\gamma)-\Sigma_{0}(z_{i},\gamma))||^{k}_{op} \\
  &\quad \leq \max_{1 \leq i \leq n}||\Sigma_{0}^{-1}(z_{i},\gamma)||_{op} \sum_{k=0}^{\infty}\left(\frac{1}{2}\right)^{k} \\
  &\quad \leq 2 \max_{1 \leq i \leq n}||\Sigma_{0}^{-1}(z_{i},\gamma)||_{op},
\end{align*}
for $n$ large and for all $(\gamma,p_{W|Z}) \in \Gamma \times \tilde{\mathcal{P}}_{n,W|Z}$, where the first inequality uses the submultiplicativity of the operator norm, the second equality uses the Neumann series, the second inequality uses the triangle inequality (valid because the Neumann series is absolutely convergent), the third uses (\ref{eq:boundfornlarge}), and the fourth uses convergence of a geometric series. Consequently, by Assumption \ref{as:covariance}.1, the following holds for $P_{0,Z}^{\infty}$-almost every fixed realization $\{z_{i}\}_{i \geq 1}$ of $\{Z_{i}\}_{i \geq 1}$:
\begin{align*}
\sup_{(\gamma,p_{W|Z}) \in \Gamma \times \tilde{\mathcal{P}}_{n,W|Z}} \max_{1 \leq i \leq n}  ||\Sigma^{-1}(z_{i},\gamma)||_{op}    \leq 2 \max_{1 \leq i \leq n}||\Sigma_{0}^{-1}(z_{i},\gamma)||_{op} \leq \frac{2}{\underline{\lambda}} < \infty
\end{align*}
for $n$ large.
\end{proof}
\begin{lemma}\label{lem:pointwiseconsistency}
Suppose that Assumptions \ref{as:dgp2}, \ref{as:correctspec}, \ref{as:moments}.1, \ref{as:covariance}.1, \ref{as:criterion}, \ref{as:concentration}.1, and \ref{as:concentration}.3 hold. Then, for $P_{0,Z}^{\infty}$-almost every fixed realization $\{z_{i}\}_{i \geq 1}$ of $\{Z_{i}\}_{i \geq 1}$,
\begin{align*}
 \limsup_{n\rightarrow \infty}\sup_{(\tilde{\gamma},p_{W|Z}) \in \Gamma \times \mathcal{P}_{n,W|Z}}||q_{n}(\tilde{\gamma},p_{W|Z})-\gamma_{0}||_{2} = 0
\end{align*}
where $\{\tilde{\mathcal{P}}_{n,W|Z}\}_{n \geq 1}$ are the sets defined in Assumption \ref{as:concentration}.
\end{lemma}
\begin{proof}
 The proof has three steps. Step 1 obtains a preliminary uniform convergence result. Step 2 uses Step 1 to establish a convergence result for the population objective function. Finally, Step 3 uses the result from Step 2 to conclude the lemma.
\paragraph{Step 1.} Let $\{z_{i}\}_{i \geq 1}$ be a fixed realization of $\{Z_{i}\}_{i \geq 1}$. I show that 
\begin{align}\label{eq:objucon}
\limsup_{n\rightarrow \infty}\sup_{(\gamma,\tilde{\gamma},p_{W|Z}) \in \Gamma\times \Gamma \times \tilde{\mathcal{P}}_{n,W|Z}}|Q_{n}(\gamma,\tilde{\gamma},p_{W|Z})-Q_{P_{0,Z}}(\gamma,\tilde{\gamma},p_{0,W|Z})| = 0,    
\end{align} 
where $Q_{P_{0,Z}}(\gamma,\tilde{\gamma},p_{0,W|Z}) = E_{P_{0,Z}}[m_{0}(Z,\gamma)'\Sigma_{0}^{-1}(Z,\tilde{\gamma})m_{0}(Z,\gamma)]$. Since the triangle inequality implies that
\begin{align*}
|Q_{n}(\gamma,\tilde{\gamma},p_{W|Z})-Q_{P_{0,Z}}(\gamma,\tilde{\gamma},p_{0,W|Z})| &\leq |Q_{n}(\gamma,\tilde{\gamma},p_{W|Z})-Q_{n}(\gamma,\tilde{\gamma},p_{0,W|Z})| \\
&\quad + | Q_{n}(\gamma,\tilde{\gamma},p_{0,W|Z})-Q_{P_{0,Z}}(\gamma,\tilde{\gamma},p_{0,W|Z})|,  
\end{align*}
it suffices to show that
\begin{align*}
   \limsup_{n\rightarrow \infty}\sup_{(\gamma,\tilde{\gamma},p_{W|Z}) \in \Gamma\times \Gamma \times \tilde{\mathcal{P}}_{n,W|Z}} |Q_{n}(\gamma,\tilde{\gamma},p_{W|Z})-Q_{n}(\gamma,\tilde{\gamma},p_{0,W|Z})| = 0
\end{align*}
because Assumption \ref{as:criterion}.1 implies
\begin{align*}
   \limsup_{n\rightarrow \infty}\sup_{(\gamma,\tilde{\gamma}) \in \Gamma\times \Gamma } |Q_{n}(\gamma,\tilde{\gamma},p_{0,W|Z})-Q_{P_{0,Z}}(\gamma,\tilde{\gamma},p_{0,W|Z})| = 0    
\end{align*}
Adding and subtracting $m_{0}(z_{i},\gamma)$ and $\Sigma_{0}^{-1}(z_{i},\tilde{\gamma})$, applying the triangle inequality, the operator norm inequality, and H\"{o}lder's inequality, I obtain the upper bound
\begin{align}
    &\left|Q_{n}(\gamma,\tilde{\gamma},p_{W|Z})-Q_{n}(\gamma,\tilde{\gamma},p_{0,W|Z})\right| \nonumber\\
    &\quad  \leq ||m(\cdot,\gamma)-m_{0}(\cdot,\gamma)||_{n,2} \max_{1 \leq i \leq n}||\Sigma^{-1}(z_{i},\tilde{\gamma})||_{op} ||m(\cdot,\gamma)||_{n,2} \nonumber \\
    &\quad \quad  + ||m_{0}(\cdot,\gamma)||_{n,2} \max_{1 \leq i \leq n}||\Sigma^{-1}(z_{i},\tilde{\gamma})-\Sigma_{0}^{-1}(z_{i},\tilde{\gamma})||_{op} ||m(\cdot,\gamma)||_{n,2} \nonumber \\
    &\quad \quad  + ||m_{0}(\cdot,\gamma)||_{n,2} \max_{1 \leq i \leq n}||\Sigma_{0}^{-1}(z_{i},\tilde{\gamma})||_{op} ||m(\cdot,\gamma)-m_{0}(\cdot,\gamma)||_{n,2} \label{eq:bounddiff}
\end{align}
for all $n \geq 1$. Assumptions \ref{as:moments}.1 and \ref{as:concentration}.3(a)  implies 
\begin{align*}
\sup_{(\gamma,p_{W|Z}) \in \tilde{\mathcal{P}}_{n,W|Z}}||m(\cdot,\gamma)-m_{0}(\cdot,\gamma)||_{n,2} &= o(1), \\ \sup_{\gamma \in \Gamma}||m_{0}(\cdot,\gamma)||_{n,2} &= O(1), \\ \sup_{(\gamma,p_{W|Z}) \in \Gamma \times \tilde{\mathcal{P}}_{n,W|Z}} ||m(\cdot,\gamma)||_{n,2} &= O(1), 
\end{align*}
where the third is an implication of the first two (by the triangle inequality). Lemma 1 shows that
\begin{align*}
    \sup_{(\gamma,p_{W|Z}) \in \Gamma \times \tilde{\mathcal{P}}_{n,W|Z}}\max_{1 \leq i \leq n}||\Sigma^{-1}(z_{i},\gamma)||_{op} = O(1)
\end{align*}
while Assumption \ref{as:covariance}.1 implies that
\begin{align*}
    \sup_{\gamma \in \Gamma}\max_{1 \leq i \leq n}||\Sigma_{0}^{-1}(z_{i},\gamma)||_{op} = O(1).
\end{align*}
Consequently, I need to show that
\begin{align}\label{eq:uniforminverse}
\sup_{(\gamma,p_{W|Z}) \in \Gamma \times \tilde{\mathcal{P}}_{n,W|Z}}\max_{1 \leq i \leq n}||\Sigma^{-1}(z_{i},\gamma)-\Sigma_{0}^{-1}(z_{i},\gamma)||_{op} = o(1) 
\end{align}
to conclude from (\ref{eq:bounddiff}) that
\begin{align*}
\sup_{(\gamma,\tilde{\gamma},p_{W|Z}) \in \Gamma \times \Gamma \times \tilde{\mathcal{P}}_{n,W|Z}}|Q_{n}(\gamma,\tilde{\gamma},p_{W|Z})-Q_{n}(\gamma,\tilde{\gamma},p_{0,W|Z})| = o(1)    
\end{align*}
as $n\rightarrow \infty$. To show (\ref{eq:uniforminverse}), I use the matrix identity $B^{-1}-A^{-1} = A^{-1}(A-B)B^{-1}$ and the submultiplicativity of the operator norm to conclude that
\begin{align*}
&\sup_{(\gamma,p_{W|Z}) \in \Gamma \times \tilde{\mathcal{P}}_{n,W|Z}}\max_{1 \leq i \leq n}||\Sigma^{-1}(z_{i},\gamma)-\Sigma^{-1}_{0}(z_{i},\gamma)||_{op} \\
&\quad \leq  \sup_{(\gamma,p_{W|Z}) \in \Gamma \times \tilde{\mathcal{P}}_{n,W|Z}}\max_{1 \leq i \leq n}||\Sigma^{-1}(z_{i},\gamma)||_{op} \times \sup_{(\gamma,p_{W|Z}) \in \Gamma \times \tilde{\mathcal{P}}_{n,W|Z}}\max_{1 \leq i \leq n}||\Sigma^{-1}_{0}(z_{i},\gamma)||_{op}  \\
&\quad \quad \times \sup_{(\gamma,p_{W|Z}) \in \Gamma \times \tilde{\mathcal{P}}_{n,W|Z}}\max_{1 \leq i \leq n}||\Sigma(z_{i},\gamma)-\Sigma_{0}(z_{i},\gamma)||_{op} \\
&\quad =O(1)O(1)o(1) \\
&\quad =o(1)
\end{align*}
where the first equality is Lemma 1, Assumption \ref{as:covariance}, and Assumption \ref{as:concentration}.3(c), respectively. Consequently, I can conclude that
\begin{align*}
      \limsup_{n\rightarrow \infty}\sup_{(\gamma,\tilde{\gamma},p_{W|Z}) \in \Gamma\times \Gamma \times \tilde{\mathcal{P}}_{n,W|Z}} |Q_{n}(\gamma,\tilde{\gamma},p_{W|Z})-Q_{n}(\gamma,\tilde{\gamma},p_{0,W|Z})| = 0.
\end{align*}
The result holds for $P_{0,Z}^{\infty}$-almost every fixed realization $\{z_{i}\}_{i \geq 1}$ of $\{Z_{i}\}_{i \geq 1}$ because I fixed the realization $\{z_{i}\}_{i \geq 1}$ of $\{Z_{i}\}_{i \geq 1}$ arbitrarily.
\paragraph{Step 2.} Let $\{z_{i}\}_{i \geq 1}$ be a fixed realization of $\{Z_{i}\}_{i \geq 1}$. I show that 
\begin{align*}
\sup_{(\tilde{\gamma},p_{W|Z})\in \Gamma \times \tilde{\mathcal{P}}_{n,W|Z}}\left|Q_{P_{0,Z}}(q_{n}(\tilde{\gamma},p_{W|Z}),\tilde{\gamma},p_{0,W|Z}) - Q_{P_{0,Z}}(\gamma_{0},\tilde{\gamma},p_{0,W|Z})\right| = o(1).   
\end{align*}
Since the criterion function satisfies
\begin{align*}
 Q_{P_{0,Z}}(\gamma_{0},\tilde{\gamma},p_{0,W|Z}) \leq   Q_{P_{0,Z}}(q_{n}(\tilde{\gamma},p_{W|Z}),\tilde{\gamma},p_{0,W|Z})  
\end{align*}
for all $n$ (as Assumption \ref{as:correctspec} implies $Q_{P_{0,Z}}(\gamma_{0},\tilde{\gamma},p_{0,W|Z}) = 0$), it suffices to show
\begin{align*}
\sup_{(\tilde{\gamma},p_{W|Z})\in\Gamma \times  \tilde{\mathcal{P}}_{n,W|Z}}\left(Q_{P_{0,Z}}(q_{n}(\tilde{\gamma},p_{W|Z}),\tilde{\gamma},p_{0,W|Z}) - Q_{P_{0,Z}}(\gamma_{0},\tilde{\gamma},p_{0,W|Z})\right) = o(1) ,   
\end{align*}
To that end, I add/subtract $Q_{n}(q_{n}(\tilde{\gamma},p_{W|Z}),\tilde{\gamma},p_{W|Z})$ to obtain
\begin{align*}
    &Q_{P_{0,Z}}(q_{n}(\tilde{\gamma},p_{W|Z}),\tilde{\gamma},p_{0,W|Z}) - Q_{P_{0,Z}}(\gamma_{0},\tilde{\gamma},p_{0,W|Z}) \\
    &\quad = Q_{P_{0,Z}}(q_{n}(\tilde{\gamma},p_{W|Z}),\tilde{\gamma},p_{0,W|Z}) - Q_{n}(q_{n}(\tilde{\gamma},p_{W|Z}),\tilde{\gamma},p_{W|Z}) \\
    &\quad \quad +Q_{n}(q_{n}(\tilde{\gamma},p_{W|Z}),\tilde{\gamma},p_{W|Z}) - Q_{P_{0,Z}}(\gamma_{0},\tilde{\gamma},p_{0,W|Z}) \\
    & \quad \leq Q_{P_{0,Z}}(q_{n}(\tilde{\gamma},p_{W|Z}),\tilde{\gamma},p_{0,W|Z}) - Q_{n}(q_{n}(\tilde{\gamma},p_{W|Z}),\tilde{\gamma},p_{W|Z}) \\
    &\quad \quad +Q_{n}(\gamma_{0},\tilde{\gamma},p_{W|Z}) - Q_{P_{0,Z}}(\gamma_{0},\tilde{\gamma},p_{0,W|Z}),
\end{align*}
where the inequality uses that $Q_{n}(\gamma_{0},\tilde{\gamma},p_{W|Z}) \geq Q_{n}(q_{n}(\tilde{\gamma},p_{W|Z}),\tilde{\gamma},p_{W|Z})$, which holds because $q_{n}(\tilde{\gamma},p_{W|Z}) \in \Gamma$ minimizes $Q_{n}(\gamma,\tilde{\gamma},p_{W|Z})$. Using the definition of the supremum,
\begin{align*}
    &\sup_{(\tilde{\gamma},p_{W|Z}) \in \Gamma \times  \tilde{\mathcal{P}}_{n,W|Z}}\left(Q_{P_{0,Z}}(q_{n}(\tilde{\gamma},p_{W|Z}),\tilde{\gamma},p_{0,W|Z}) - Q_{P_{0,Z}}(\gamma_{0},\tilde{\gamma},p_{0,W|Z})\right) \\
    &\quad \leq 2 \sup_{(\gamma,\tilde{\gamma},p_{W|Z}) \in \Gamma \times \tilde{\mathcal{P}}_{n,W|Z}}\left|Q_{n}(\gamma,\tilde{\gamma},p_{W|Z})-Q_{P_{0,Z}}(\gamma,\tilde{\gamma},p_{0,W|Z})\right|
\end{align*}
Applying Step 1, it then follows that
\begin{align*}
\sup_{(\tilde{\gamma},p_{W|Z})\in\Gamma \times  \tilde{\mathcal{P}}_{n,W|Z}}\left(Q_{P_{0,Z}}(q_{n}(\tilde{\gamma},p_{W|Z}),\tilde{\gamma},p_{0,W|Z}) - Q_{P_{0,Z}}(\gamma_{0},\tilde{\gamma},p_{0,W|Z})\right) = o(1)   
\end{align*}
as $n\rightarrow \infty$. Since I fixed $\{z_{i}\}_{i \geq 1}$ arbitrarily, I can extend the argument so that it holds for $P_{0,Z}^{\infty}$-almost every fixed realization $\{z_{i}\}_{i \geq 1}$ of $\{Z_{i}\}_{i \geq 1}$.
\paragraph{Step 3.} Let $A$ be the set of $\{z_{i}\}_{i \geq 1}$ for which the claim of the lemma is not true, and assume by contradiction that $P_{0,Z}^{\infty}(A) > 0$. For $\{z_{i}\}_{i \geq 1} \in A$, there is a constant $\xi> 0$ such that 
\begin{align*}
\sup_{(\tilde{\gamma},p_{W|Z}) \in \Gamma \times \tilde{\mathcal{P}}_{n,W|Z}}||q_{n}(\tilde{\gamma},p_{W|Z})-\gamma_{0}||_{2} \geq 2 \xi.
\end{align*}
for $n$ large. Using the definition of the supremum, there exists $(\tilde{\gamma}_{n},p_{n,W|Z}) \in \Gamma \times  \tilde{\mathcal{P}}_{n,W|Z}$ such that $||q_{n}(\tilde{\gamma}_{n},p_{n,W|Z})-\gamma_{0}||_{2} \geq  \xi$ for $n$ large. Using that
\begin{align*}
 \inf_{\gamma \in \Gamma: ||\gamma-\gamma_{0}||_{2}\geq \xi} Q_{P_{0,Z}}(\gamma,\tilde{\gamma}_{n},p_{0,W|Z}) \leq   Q_{P_{0,Z}}(q_{n}(\tilde{\gamma}_{n},p_{n,W|Z}),\tilde{\gamma}_{n},p_{0,W|Z}) 
\end{align*}
by definition of the infimum, it follows that
\begin{align*}
&\inf_{\gamma \in \Gamma: ||\gamma-\gamma_{0}||_{2}\geq \xi} Q_{P_{0,Z}}(\gamma,\tilde{\gamma}_{n},p_{0,W|Z}) - Q_{P_{0,Z}}(\gamma_{0},\tilde{\gamma}_{n},p_{0,W|Z})\\
&\quad \leq Q_{P_{0,Z}}(q_{n}(\tilde{\gamma}_{n},p_{n,W|Z}),\tilde{\gamma}_{n},p_{0,W|Z}) - Q_{P_{0,Z}}(\gamma_{0},\tilde{\gamma}_{n},p_{0,W|Z}) 
\end{align*}
for $n$ large. Step 2 shows that, for $P_{0,Z}^{\infty}$-almost every realization $\{z_{i}\}_{i \geq 1}$ of $\{Z_{i}\}_{i \geq 1}$,
\begin{align*}
\sup_{(\tilde{\gamma},p_{W|Z})\in\Gamma \times  \tilde{\mathcal{P}}_{n,W|Z}}\left(Q_{P_{0,Z}}(q_{n}(\tilde{\gamma},p_{W|Z}),\tilde{\gamma},p_{0,W|Z}) - Q_{P_{0,Z}}(\gamma_{0},\tilde{\gamma},p_{0,W|Z})\right) = o(1).
\end{align*}
Consequently, the inequality
\begin{align*}
&Q_{P_{0,Z}}(q_{n}(\tilde{\gamma}_{n},p_{n,W|Z}),\tilde{\gamma}_{n},p_{0,W|Z}) - Q_{P_{0,Z}}(\gamma_{0},\tilde{\gamma}_{n},p_{0,W|Z}) \\
&\quad \leq \sup_{(\tilde{\gamma},p_{W|Z})\in\Gamma \times  \tilde{\mathcal{P}}_{n,W|Z}}\left(Q_{P_{0,Z}}(q_{n}(\tilde{\gamma},p_{W|Z}),\tilde{\gamma},p_{0,W|Z}) - Q_{P_{0,Z}}(\gamma_{0},\tilde{\gamma},p_{0,W|Z})\right)   
\end{align*}
implies
\begin{align*}
Q_{P_{0,Z}}(q_{n}(\tilde{\gamma}_{n},p_{n,W|Z}),\tilde{\gamma}_{n},p_{0,W|Z}) - Q_{P_{0,Z}}(\gamma_{0},\tilde{\gamma}_{n},p_{0,W|Z}) = o(1).    
\end{align*}
Meanwhile, Assumption \ref{as:correctspec}.2, Assumption \ref{as:covariance}.1, and Assumption \ref{as:criterion}.2 implies that
\begin{align*}
&\inf_{\gamma \in \Gamma: ||\gamma-\gamma_{0}||_{2}\geq \xi} Q_{P_{0,Z}}(\gamma,\tilde{\gamma}_{n},p_{0,W|Z}) - Q_{P_{0,Z}}(\gamma_{0},\tilde{\gamma}_{n},p_{0,W|Z})\\
&\quad = \inf_{\gamma \in \Gamma: ||\gamma-\gamma_{0}||_{2}\geq \xi} Q_{P_{0,Z}}(\gamma,\tilde{\gamma}_{n},p_{0,W|Z}) \\
&\quad \geq \inf_{(z,\gamma) \in \mathcal{Z} \times \Gamma}\lambda_{min}(\Sigma_{0}^{-1}(z,\gamma))\inf_{\gamma \in \Gamma: ||\gamma-\gamma_{0}||_{2}\geq \xi}E_{P_{0}}[||m_{0}(Z,\gamma)||_{2}^{2}] \\
&\quad \geq \frac{1}{\overline{\lambda}}\inf_{\gamma \in \Gamma: ||\gamma-\gamma_{0}||_{2}\geq \xi}E_{P_{0}}[||m_{0}(Z,\gamma)||_{2}^{2}]\\
&\quad > \varepsilon
\end{align*}
for some $\varepsilon >0$ (depending on $\xi$ and $\overline{\lambda}$). This leads to a contradiction. Consequently, $P_{0,Z}^{\infty}(A) = 0$ and the proof is complete.
\end{proof}
\begin{lemma}\label{lem:UCM}
If Assumptions \ref{as:dgp2}, \ref{as:correctspec}, \ref{as:moments}, \ref{as:covariance}.1, \ref{as:criterion}, \ref{as:localident}, and \ref{as:concentration} hold, then, for $P_{0,Z}^{\infty}$-almost every fixed realization $\{z_{i}\}_{i \geq 1}$ of $\{Z_{i}\}_{i\geq 1}$ and $n$ large, the function $\tilde{\gamma} \mapsto q_{n}(\tilde{\gamma},p_{W|Z})$ is a uniform contraction mapping over $\tilde{\mathcal{P}}_{n,W|Z}$. That is, there exists a constant $\rho \in [0,1)$ and $N \geq 1$ such that $||q_{n}(\tilde{\gamma}_{1},p_{W|Z})-q_{n}(\tilde{\gamma}_{2},p_{W|Z})||_{2} \leq \rho ||\tilde{\gamma}_{1}-\tilde{\gamma}_{2}||_{2}$ for all $\tilde{\gamma}_{1},\tilde{\gamma}_{2} \in \Gamma$, all $p_{W|Z} \in \tilde{\mathcal{P}}_{n,W|Z}$, and $n \geq N$. 
\end{lemma}
\begin{proof}
The proof has three steps. Step 1 establishes some intermediate technical results. Step 2 establishes the claim.
    \paragraph{Step 1.} I start by showing some intermediate technical results. Let
    \begin{align*}
    B_{n}(\tilde{\gamma},p_{W|Z}) =\frac{2}{n}\sum_{i=1}^{n} &\bigg\{\left[m(z_{i},q_{n}(\tilde{\gamma},p_{W|Z}))' \otimes M(z_{i},q_{n}(\tilde{\gamma},p_{W|Z}))'\right]\\ &\quad \times \left[\Sigma^{-1}(z_{i},\tilde{\gamma}) \otimes \Sigma^{-1}(z_{i},\tilde{\gamma}) \right ] \frac{\partial}{\partial \tilde{\gamma}'}\text{vec}(\Sigma(z_{i},\tilde{\gamma}))\bigg \}
\end{align*}
and
\begin{align*}
D_{n}(\tilde{\gamma},p_{W|Z}) &= \frac{1}{n}\sum_{i=1}^{n}M(z_{i},q_{n}(\tilde{\gamma},p_{W|Z}))'\Sigma(z_{i},\tilde{\gamma})^{-1}M(z_{i},q_{n}(\tilde{\gamma},p_{W|Z})) \\
&\quad + \frac{1}{n}\sum_{i=1}^{n}(m(z_{i},q_{n}(\tilde{\gamma},p_{W|Z}))'\Sigma^{-1}(z_{i},\tilde{\gamma})\otimes I)\frac{\partial}{\partial \gamma'}\text{vec}(M(z_{i},q_{n}(\tilde{\gamma},p_{W|Z}))').
\end{align*}
Steps 1A. and 1B-C. show 1. $\sup_{(\tilde{\gamma},p_{W|Z}) \in \Gamma_{0} \times \tilde{\mathcal{P}}_{n,W|Z}}\left| \left |B_{n}(\tilde{\gamma},p_{W|Z}) \right| \right|_{op} = o(1)$, and 2. $D_{n}(\tilde{\gamma},p_{W|Z})$ is invertible for large $n$ and $\sup_{(\tilde{\gamma},p_{W|Z}) \in \Gamma_{0} \times \tilde{\mathcal{P}}_{n,W|Z}}\left| \left|D_{n}(\tilde{\gamma},p_{W|Z})^{-1}\right | \right|_{op} = O(1)$, respectively.
\paragraph{Step 1A.} I show that $\sup_{(\tilde{\gamma},p_{W|Z}) \in \Gamma_{0} \times \tilde{\mathcal{P}}_{n,W|Z}}\left| \left |B_{n}(\tilde{\gamma},p_{W|Z}) \right| \right|_{op}  = o(1).$ Using the submultiplicativity of the operator norm and the fact that $||A\otimes B||_{op} = ||A||_{op} \cdot ||B||_{op}$ for matrices $A$ and $B$,
\begin{align*}
\left| \left |B_{n}(\tilde{\gamma},p_{W|Z}) \right| \right|_{op} &\leq \left(\max_{1 \leq i \leq n}||\Sigma^{-1}(z_{i},\tilde{\gamma})||_{op}\right)^{2} \max_{1 \leq i \leq n}\left|\left|\frac{\partial}{\partial \tilde{\gamma}'}\text{vec}(\Sigma(z_{i},\tilde{\gamma}))\right|\right|_{op} \\
&\quad \times \frac{1}{n}\sum_{i=1}^{n}\left\{||m(z_{i},q_{n}(\tilde{\gamma},p_{W|Z}))||_{op} \cdot ||M(z_{i},q_{n}(\tilde{\gamma},p_{W|Z}))||_{op}\right\}
\end{align*}
Lemma 1 implies that
\begin{align*}
\sup_{(\tilde{\gamma},p_{W|Z}) \in \Gamma_{0} \times \tilde{\mathcal{P}}_{n,W|Z}}\max_{1 \leq i \leq n}||\Sigma^{-1}(z_{i},\tilde{\gamma})||_{op} = O(1)
\end{align*}
while Assumption \ref{as:concentration}.3(e) implies that
\begin{align*}
    \sup_{(\tilde{\gamma},p_{W|Z}) \in \Gamma_{0} \times \tilde{\mathcal{P}}_{n,W|Z}}\max_{1 \leq i \leq n}\left|\left|\frac{\partial}{\partial \tilde{\gamma}'}\text{vec}(\Sigma(z_{i},\tilde{\gamma}))\right|\right|_{op} = O(1).
\end{align*}
Consequently, I need to show that
\begin{align*}
 \sup_{(\tilde{\gamma},p_{W|Z}) \in \Gamma_{0} \times \tilde{\mathcal{P}}_{n,W|Z}}\frac{1}{n}\sum_{i=1}^{n}\left\{||m(z_{i},q_{n}(\tilde{\gamma},p_{W|Z}))||_{op} \cdot ||M(z_{i},q_{n}(\tilde{\gamma},p_{W|Z}))||_{op}\right\} = o(1)  
\end{align*}
Applying H\"{o}lder's inequality, I conclude that
\begin{align*}
&\frac{1}{n}\sum_{i=1}^{n}\left\{||m(z_{i},q_{n}(\tilde{\gamma},p_{W|Z}))||_{op} \cdot ||M(z_{i},q_{n}(\tilde{\gamma},p_{W|Z}))||_{op}\right\} \\
&\quad \leq \sqrt{\frac{1}{n}\sum_{i=1}^{n}||m(z_{i},q_{n}(\tilde{\gamma},p_{W|Z}))||_{2}^{2}} \cdot \sqrt{\frac{1}{n}\sum_{i=1}^{n}||M(z_{i},q_{n}(\tilde{\gamma},p_{W|Z}))||_{op}^{2}}.
\end{align*}
The triangle inequality and a mean-value expansion reveals that
\begin{align*}
    \sqrt{\frac{1}{n}\sum_{i=1}^{n}||m(z_{i},q_{n}(\tilde{\gamma},p_{W|Z}))||_{2}^{2}} &\leq     \sqrt{\frac{1}{n}\sum_{i=1}^{n}||m(z_{i},q_{n}(\tilde{\gamma},p_{W|Z}))-m_{0}(z_{i},q_{n}(\tilde{\gamma},p_{W|Z}))||_{2}^{2}} \\
    &\quad + \sqrt{\frac{1}{n}\sum_{i=1}^{n}||m_{0}(z_{i},q_{n}(\tilde{\gamma},p_{W|Z}))||_{2}^{2}} \\
    &\leq \sqrt{\frac{1}{n}\sum_{i=1}^{n}||m(z_{i},q_{n}(\tilde{\gamma},p_{W|Z}))-m_{0}(z_{i},q_{n}(\tilde{\gamma},p_{W|Z}))||_{2}^{2}} \\
    &\quad + ||q_{n}(\tilde{\gamma},p_{W|Z})-\gamma_{0}||_{2}\sqrt{\frac{1}{n}\sum_{i=1}^{n}||M_{0}(z_{i},\tilde{\gamma}_{n}^{*})||_{op}^{2}},
\end{align*}
where $\tilde{\gamma}_{n}^{*}$ is on the line segment joining $q_{n}(\tilde{\gamma},p_{W|Z})$ and $\gamma_{0}$ (and takes different values for each row of $M_{0}(z_{i},\tilde{\gamma}_{n}^{*})$). Such an expansion is valid because, for large $n$, Lemma 2 implies that $q_{n}(\tilde{\gamma},p_{W|Z}) \in \Gamma_{0}$ for all $\tilde{\gamma} \in \Gamma$ and all $p_{W|Z} \in \tilde{\mathcal{P}}_{n,W|Z}$, and, as a result, I can apply Assumption \ref{as:moments}.2 to obtain the mean-value expansion. Lemma 2 and Assumption \ref{as:moments}.2 implies that
\begin{align*}
  \sup_{(\tilde{\gamma},p_{W|Z}) \in \Gamma_{0} \times \tilde{\mathcal{P}}_{n,W|Z}}\left(  ||q_{n}(\tilde{\gamma},p_{W|Z})-\gamma_{0}||_{2}\sqrt{\frac{1}{n}\sum_{i=1}^{n}||M_{0}(z_{i},\tilde{\gamma}_{n}^{*})||_{op}^{2}}\right) = o(1)
\end{align*}
while
\begin{align*}
    \sup_{(\tilde{\gamma},p_{W|Z}) \in \Gamma_{0} \times \tilde{\mathcal{P}}_{n,W|Z}} \sqrt{\frac{1}{n}\sum_{i=1}^{n}||m(z_{i},q_{n}(\tilde{\gamma},p_{W|Z}))-m_{0}(z_{i},q_{n}(\tilde{\gamma},p_{W|Z}))||_{2}^{2}} = o(1)
\end{align*}
by Assumption \ref{as:concentration}.3(a). Consequently,
\begin{align*}
 \sup_{(\tilde{\gamma},p_{W|Z}) \in \Gamma_{0} \times \tilde{\mathcal{P}}_{n,W|Z}} \sqrt{\frac{1}{n}\sum_{i=1}^{n}||m(z_{i},q_{n}(\tilde{\gamma},p_{W|Z}))||_{2}^{2}}  = o(1).   
\end{align*}
Lemma 2 implies that $\{q_{n}(\tilde{\gamma},p_{W|Z}): (\tilde{\gamma},p_{W|Z}) \in \Gamma_{0} \times \tilde{\mathcal{P}}_{n,W|Z}\} \subseteq \Gamma_{0}$ for large $n$, where $\Gamma_{0}$ is the neighborhood of $\gamma_{0}$ defined in Assumption \ref{as:moments}.2. Consequently, the triangle inequality and this observation imply that, for $n$ large,
\begin{align*}
&\sqrt{\frac{1}{n}\sum_{i=1}^{n}||M(z_{i},q_{n}(\tilde{\gamma},p_{W|Z}))||_{op}^{2}} \\
&\quad \leq \sqrt{\frac{1}{n}\sum_{i=1}^{n}||M(z_{i},q_{n}(\tilde{\gamma},p_{W|Z}))-M_{0}(z_{i},q_{n}(\tilde{\gamma},p_{W|Z}))||_{op}^{2}} + \sqrt{\frac{1}{n}\sum_{i=1}^{n}||M_{0}(z_{i},q_{n}(\tilde{\gamma},p_{W|Z}))||_{op}^{2}} \\
&\quad \leq  \sup_{(\gamma,p_{W|Z}) \in \Gamma_{0} \times \tilde{\mathcal{P}}_{n,W|Z}}\sqrt{\frac{1}{n}\sum_{i=1}^{n}||M(z_{i},\gamma)-M_{0}(z_{i},\gamma)||_{op}^{2}} + \sup_{\gamma \in \Gamma_{0}}\sqrt{\frac{1}{n}\sum_{i=1}^{n}||M_{0}(z_{i},\gamma)||_{op}^{2}},   
\end{align*}
for all $\tilde{\gamma} \in \Gamma_{0}$. Consequently, it follows that
\begin{align*}
\limsup_{n\rightarrow \infty}\sup_{(\tilde{\gamma},p_{W|Z}) \in \Gamma_{0} \times \tilde{\mathcal{P}}_{n,W|Z}}\sqrt{\frac{1}{n}\sum_{i=1}^{n}||M(z_{i},q_{n}(\tilde{\gamma},p_{W|Z}))||_{op}^{2}}  \leq \limsup_{n\rightarrow \infty}\sup_{\gamma \in \Gamma_{0}}\sqrt{\frac{1}{n}\sum_{i=1}^{n}||M_{0}(z_{i},\gamma)||_{op}^{2}}< \infty   
\end{align*}
using Assumption \ref{as:moments}.2 and Assumption \ref{as:concentration}.3(b). Putting these two results together,
\begin{align*}
 \sup_{(\tilde{\gamma},p_{W|Z}) \in \Gamma_{0} \times \tilde{\mathcal{P}}_{n,W|Z}}\frac{1}{n}\sum_{i=1}^{n}\left\{||m(z_{i},q_{n}(\tilde{\gamma},p_{W|Z}))||_{op} \cdot ||M(z_{i},q_{n}(\tilde{\gamma},p_{W|Z}))||_{op}\right\} = o(1). 
\end{align*}
This leads to the conclusion that $\sup_{(\tilde{\gamma},p_{W|Z}) \in \Gamma_{0} \times \tilde{\mathcal{P}}_{n,W|Z}}||B_{n}(\tilde{\gamma},p_{W|Z})||_{op} = o(1)$.
\paragraph{Step 1B.} I show that
\begin{align}\label{eq:convergenceofgram}
\sup_{(\tilde{\gamma},p_{W|Z}) \in \Gamma_{0} \times \tilde{\mathcal{P}}_{n,W|Z}}\left | \left |D_{n}(\tilde{\gamma},p_{W|Z}) - E_{P_{0,Z}}[M_{0}(Z,\gamma_{0})'\Sigma_{0}^{-1}(Z,\tilde{\gamma})M_{0}(Z,\gamma_{0})] \right | \right|_{op} = o(1)   
\end{align}
and then, in Step 1C, I show that this implies $\sup_{(\tilde{\gamma},p_{W|Z}) \in \Gamma_{0}\times \tilde{\mathcal{P}}_{n,W|Z}}||D_{n}^{-1}(\tilde{\gamma},p_{W|Z})||_{op} = O(1)$. Consider the first term in the expression for $D_{n}(\tilde{\gamma},p_{W|Z})$,
\begin{align*}
 \frac{1}{n}\sum_{i=1}^{n}M(z_{i},q_{n}(\tilde{\gamma},p_{W|Z}))'\Sigma(z_{i},\tilde{\gamma})^{-1}M(z_{i},q_{n}(\tilde{\gamma},p_{W|Z})).   
\end{align*}
The triangle inequality, H\"{o}lder's inequality, Lemma 1, and Assumption \ref{as:concentration}.3(b)-\ref{as:concentration}.3(c) implies that
\begin{align*}
&\frac{1}{n}\sum_{i=1}^{n}M(z_{i},q_{n}(\tilde{\gamma},p_{W|Z}))'\Sigma(z_{i},\tilde{\gamma})^{-1}M(z_{i},q_{n}(\tilde{\gamma},p_{W|Z})) \\
&\quad = \frac{1}{n}\sum_{i=1}^{n}M_{0}(z_{i},q_{n}(\tilde{\gamma},p_{W|Z}))'\Sigma_{0}^{-1}(z_{i},\tilde{\gamma})M_{0}(z_{i},q_{n}(\tilde{\gamma},p_{W|Z})) + o(1)
\end{align*}
uniformly over $\Gamma_{0} \times \tilde{\mathcal{P}}_{n,W|Z}$. Next, I show that
\begin{align}\label{eq:gramattruth}
 &\frac{1}{n}\sum_{i=1}^{n}M_{0}(z_{i},q_{n}(\tilde{\gamma},p_{W|Z}))'\Sigma_{0}^{-1}(z_{i},\tilde{\gamma})M_{0}(z_{i},q_{n}(\tilde{\gamma},p_{W|Z})) \nonumber \\
 &\quad = \frac{1}{n}\sum_{i=1}^{n}M_{0}(z_{i},\gamma_{0})'\Sigma_{0}^{-1}(z_{i},\tilde{\gamma})M_{0}(z_{i},\gamma_{0}) + o(1)   
\end{align}
uniformly over $\Gamma_{0} \times \tilde{\mathcal{P}}_{n,W|Z}$. Applying the triangle inequality and H\"{o}lder's inequality, I obtain the following inequality
\begin{align*}
&\left| \left|\frac{1}{n}\sum_{i=1}^{n}\left(M_{0}(z_{i},q_{n}(\tilde{\gamma},p_{W|Z}))-M_{0}(z_{i},\gamma_{0})\right)'\Sigma_{0}^{-1}(z_{i},\tilde{\gamma})M_{0}(z_{i},q_{n}(\tilde{\gamma},p_{W|Z})) \right| \right|_{op}   \\
&\quad \leq \max_{1 \leq i \leq n}||\Sigma_{0}^{-1}(z_{i},\tilde{\gamma})||_{op} \sqrt{\frac{1}{n}\sum_{i=1}^{n}\left|\left|M_{0}(z_{i},q_{n}(\tilde{\gamma},p_{W|Z}))-M_{0}(z_{i},\gamma_{0})\right|\right|_{op}^{2}}\\
&\quad \quad \times \sqrt{\frac{1}{n}\sum_{i=1}^{n}\left|\left|M_{0}(z_{i},q_{n}(\tilde{\gamma},p_{W|Z}))\right|\right|_{op}^{2}}.
\end{align*}
Assumption \ref{as:covariance}.1 implies that
\begin{align*}
   \sup_{\tilde{\gamma} \in \Gamma_{0}} \max_{1 \leq i \leq n}||\Sigma_{0}^{-1}(z_{i},\tilde{\gamma})||_{op} = O(1).
\end{align*}
Using Assumption \ref{as:moments}.2 and that Lemma 2 implies $\{q_{n}(\tilde{\gamma},p_{W|Z}): (\tilde{\gamma},p_{W|Z}) \in \Gamma_{0} \times \tilde{\mathcal{P}}_{n,W|Z}\} \subseteq \Gamma_{0}$ for large $n$, I can conclude that
\begin{align*}
\sup_{(\tilde{\gamma},p_{W|Z}) \in \Gamma_{0} \times \tilde{\mathcal{P}}_{n,W|Z}}\sqrt{\frac{1}{n}\sum_{i=1}^{n}\left|\left|M_{0}(z_{i},q_{n}(\tilde{\gamma},p_{W|Z}))\right|\right|_{op}^{2}} \leq \sup_{\gamma \in \Gamma_{0}}\sqrt{\frac{1}{n}\sum_{i=1}^{n}\left|\left|M_{0}(z_{i},\gamma)\right|\right|_{op}^{2}} = O(1).  
\end{align*}
for large $n$, and, as a result, I obtain
\begin{align*}
    \sup_{(\tilde{\gamma},p_{W|Z}) \in \Gamma_{0} \times \tilde{\mathcal{P}}_{n,W|Z}}\sqrt{\frac{1}{n}\sum_{i=1}^{n}\left|\left|M_{0}(z_{i},q_{n}(\tilde{\gamma},p_{W|Z}))\right|\right|_{op}^{2}} = O(1).
\end{align*}
Assumption \ref{as:moments}.2 and Lemma 2 implies that
\begin{align*}
 \sup_{(\tilde{\gamma},p_{W|Z}) \in \Gamma_{0} \times \tilde{\mathcal{P}}_{n,W|Z}}\sqrt{\frac{1}{n}\sum_{i=1}^{n}\left|\left|M_{0}(z_{i},q_{n}(\tilde{\gamma},p_{W|Z}))-M_{0}(z_{i},\gamma_{0})\right|\right|_{op}^{2}} = o(1).
\end{align*}
by the mean-value theorem. Consequently,
\begin{align*}
\sup_{(\tilde{\gamma},p_{W|Z}) \in \Gamma_{0} \times \tilde{\mathcal{P}}_{n,W|Z}}    \left| \left|\frac{1}{n}\sum_{i=1}^{n}\left(M_{0}(z_{i},q_{n}(\tilde{\gamma},p_{W|Z}))-M_{0}(z_{i},\gamma_{0})\right)'\Sigma_{0}^{-1}(z_{i},\tilde{\gamma})M_{0}(z_{i},q_{n}(\tilde{\gamma},p_{W|Z})) \right| \right|_{op} = o(1).
\end{align*}
A similar argument reveals that
\begin{align*}
  \sup_{(\tilde{\gamma},p_{W|Z}) \in \Gamma_{0} \times \tilde{\mathcal{P}}_{n,W|Z}}\left|\left|\frac{1}{n}\sum_{i=1}^{n}M_{0}(z_{i},\gamma_{0})'\Sigma_{0}^{-1}(z_{i},\tilde{\gamma})\left[M_{0}(z_{i},q_{n}(\tilde{\gamma},p_{W|Z}))-M_{0}(z_{i},\gamma_{0})\right] \right | \right|_{op} = o(1).
\end{align*}
As such, I am able to conclude (\ref{eq:gramattruth}). Applying Assumption \ref{as:localident}.1, I can further conclude that
\begin{align*}
\frac{1}{n}\sum_{i=1}^{n}M_{0}(z_{i},\gamma_{0})'\Sigma_{0}^{-1}(z_{i},\tilde{\gamma})M_{0}(z_{i},\gamma_{0}) = E_{P_{0,Z}}[M_{0}(Z,\gamma_{0})'\Sigma_{0}^{-1}(Z,\tilde{\gamma})M_{0}(Z,\gamma_{0})]+o(1)   
\end{align*}
uniformly over $\Gamma_{0}$. It remains to show that
\begin{align}\label{eq:convergetozero}
 \frac{1}{n}\sum_{i=1}^{n}(m(z_{i},q_{n}(\tilde{\gamma},p_{W|Z}))'\Sigma^{-1}(z_{i},\tilde{\gamma})\otimes I)\frac{\partial}{\partial \gamma'}\text{vec}(M(z_{i},q_{n}(\tilde{\gamma},p_{W|Z}))') = o(1)    
\end{align}
uniformly over $\Gamma \times \tilde{\mathcal{P}}_{n,W|Z}$ to conclude (\ref{eq:convergenceofgram}). I use the triangle inequality, submultiplicativity of operator norms, and H\"{o}lder's inequality to conclude that
\begin{align*}
    &\left| \left| \frac{1}{n}\sum_{i=1}^{n}(m(z_{i},q_{n}(\tilde{\gamma},p_{W|Z}))'\Sigma^{-1}(z_{i},\tilde{\gamma})\otimes I)\frac{\partial}{\partial \gamma'}\text{vec}(M(z_{i},q_{n}(\tilde{\gamma},p_{W|Z}))') \right| \right|_{op} \\
    &\quad \leq \frac{1}{n}\sum_{i=1}^{n} \left| \left| (m(z_{i},q_{n}(\tilde{\gamma},p_{W|Z}))'\Sigma^{-1}(z_{i},\tilde{\gamma})\otimes I)\frac{\partial}{\partial \gamma'}\text{vec}(M(z_{i},q_{n}(\tilde{\gamma},p_{W|Z}))') \right| \right|_{op} \\
    &\quad \leq \frac{1}{n}\sum_{i=1}^{n} \left| \left| (m(z_{i},q_{n}(\tilde{\gamma},p_{W|Z})) \right| \right|_{op}\left| \left|\Sigma^{-1}(z_{i},\tilde{\gamma})\right| \right|_{op} \left| \left|\frac{\partial}{\partial \gamma'}\text{vec}(M(z_{i},q_{n}(\tilde{\gamma},p_{W|Z}))') \right| \right|_{op} \\
    &\quad \leq \max_{1 \leq i \leq n}||\Sigma^{-1}(z_{i},\tilde{\gamma})||_{op}\frac{1}{n}\sum_{i=1}^{n} \left| \left| m(z_{i},q_{n}(\tilde{\gamma},p_{W|Z})) \right| \right|_{op}\cdot \left| \left|\frac{\partial}{\partial \gamma'}\text{vec}(M(z_{i},q_{n}(\tilde{\gamma},p_{W|Z}))') \right| \right|_{op} \\
    &\quad \leq \max_{1 \leq i \leq n}||\Sigma^{-1}(z_{i},\tilde{\gamma})||_{op} \sqrt{\frac{1}{n}\sum_{i=1}^{n}||m(z_{i},q_{n}(\tilde{\gamma},p_{W|Z}))||_{2}^{2}}\sqrt{\frac{1}{n}\sum_{i=1}^{n}\left| \left|\frac{\partial}{\partial \gamma'}\text{vec}(M(z_{i},q_{n}(\tilde{\gamma},p_{W|Z}))') \right| \right|_{F}^{2}}.
\end{align*}
Step 1A shows that
\begin{align*}
\sup_{(\gamma,p_{W|Z}) \in \Gamma_{0} \times \tilde{\mathcal{P}}_{n,W|Z}}\sqrt{\frac{1}{n}\sum_{i=1}^{n}||m(z_{i},q_{n}(\tilde{\gamma},p_{W|Z}))||_{2}^{2}}    =o(1).
\end{align*}
Lemma 1 implies that
\begin{align*}
\sup_{(\tilde{\gamma},p_{W|Z}) \in \Gamma_{0} \times \tilde{\mathcal{P}}_{n,W|Z}}\max_{1 \leq i \leq n}||\Sigma^{-1}(z_{i},\tilde{\gamma})||_{op}    = O(1).
\end{align*}
Since Lemma 2 implies $\{q_{n}(\tilde{\gamma},p_{W|Z}): (\tilde{\gamma},p_{W|Z}) \in \Gamma_{0} \times \tilde{\mathcal{P}}_{n,W|Z}\} \subseteq \Gamma_{0}$ for large $n$, Assumption \ref{as:concentration}.3(d) implies that
\begin{align*}
 \sup_{(\tilde{\gamma},p_{W|Z}) \in \Gamma_{0} \times \tilde{\mathcal{P}}_{n,W|Z}}\sqrt{\frac{1}{n}\sum_{i=1}^{n}\left| \left|\frac{\partial}{\partial \gamma'}\text{vec}(M(z_{i},q_{n}(\tilde{\gamma},p_{W|Z}))') \right| \right|_{F}^{2}}   = O(1)
\end{align*}
This leads to the conclusion that
\begin{align*}
 \frac{1}{n}\sum_{i=1}^{n}(m(z_{i},q_{n}(\tilde{\gamma},p_{W|Z}))'\Sigma^{-1}(z_{i},\tilde{\gamma})\otimes I)\frac{\partial}{\partial \gamma'}\text{vec}(M(z_{i},q_{n}(\tilde{\gamma},p_{W|Z}))') =o(1)
\end{align*}
uniformly over $\Gamma_{0} \times \tilde{\mathcal{P}}_{n,W|Z}$, and, as a result, I am able to conclude (\ref{eq:convergetozero}).
\paragraph{Step 1C.} Let $V_{0}(\tilde{\gamma}) = E_{P_{0,Z}}[M_{0}(Z,\gamma_{0})'\Sigma_{0}^{-1}(Z,\tilde{\gamma})M_{0}(Z,\gamma_{0})]$, and note that, by Assumptions \ref{as:covariance}.1 and \ref{as:localident}.2, $\{V_{0}(\tilde{\gamma}): \gamma \in \Gamma_{0}\}$ forms a bounded class of matrices. The convergence (\ref{eq:convergenceofgram}) implies that
\begin{align*}
    \limsup_{n\rightarrow \infty}\sup_{(\tilde{\gamma},p_{W|Z}) \in \Gamma_{0} \times \tilde{\mathcal{P}}_{n,W|Z}}|\det(D_{n}(\tilde{\gamma},p_{W|Z})) - \det( V_{0}(\tilde{\gamma}))| = 0
\end{align*}
Consequently, for every $\xi >0$, there exists $N_{\varepsilon} \geq 1$ such that
\begin{align*}
    \inf_{\tilde{\gamma} \in \Gamma_{0}}\det(V_{0}(\tilde{\gamma}) - \xi < \inf_{(\tilde{\gamma},p_{W|Z}) \in \Gamma_{0} \times \tilde{\mathcal{P}}_{n,W|Z}}\det(D_{n}(\tilde{\gamma},p_{W|Z}))
\end{align*}
and
\begin{align*}
 \sup_{\tilde{\gamma} \in \Gamma_{0}}\det(V_{0}(\tilde{\gamma})) + \xi > \sup_{(\tilde{\gamma},p_{W|Z}) \in \Gamma_{0} \times \tilde{\mathcal{P}}_{n,W|Z}}\det(D_{n}(\tilde{\gamma},p_{W|Z}))
 \end{align*}
 for all $n \geq N_{\xi}$. Since Assumptions \ref{as:covariance}.1 and \ref{as:localident}.2 imply that the LHS of the above inequalities is contained in $(0,\infty)$ for $\xi >0$ small and for $n$ large, the class $\{D_{n}(\tilde{\gamma},p_{W|Z}): (\tilde{\gamma},p_{W|Z}) \in \Gamma_{0} \times \tilde{\mathcal{P}}_{n,W|Z}\}$ is invertible for large $n$. Using the identity
 \begin{align*}
     D_{n}^{-1}(\tilde{\gamma},p_{W|Z}) = (I+ V_{0}^{-1}(\tilde{\gamma})(D_{n}(\tilde{\gamma},p_{W|Z})-V_{0}(\tilde{\gamma})))^{-1}V_{0}^{-1}(\tilde{\gamma})
 \end{align*}
 and the fact that Assumption \ref{as:covariance}.1, Assumption \ref{as:localident}.2, and (\ref{eq:convergenceofgram}) implies
 \begin{align}\label{eq:boundcontract}
\sup_{(\tilde{\gamma},p_{W|Z}) \in \Gamma_{0}\times \tilde{\mathcal{P}}_{n,W|Z}}||V_{0}^{-1}(\tilde{\gamma})(D_{n}(\tilde{\gamma},p_{W|Z})-V_{0}(\tilde{\gamma}))||_{op} < \frac{1}{2}     
 \end{align}
 for large $n$, it follows that, for large $n$,
 \begin{align*}
     ||D_{n}^{-1}(\tilde{\gamma},p_{W|Z})||_{op} &\leq ||(I+ V_{0}^{-1}(\tilde{\gamma})(D_{n}(\tilde{\gamma},p_{W|Z})-V_{0}(\tilde{\gamma})))^{-1}||_{op} \cdot || V_{0}^{-1}(\tilde{\gamma})||_{op} \\
     &= \left| \left | \sum_{k=0}^{\infty}(-1)^{k} V_{0}^{-1}(\tilde{\gamma})(D_{n}(\tilde{\gamma},p_{W|Z})-V_{0}(\tilde{\gamma}))^{k}\right | \right |_{op}\cdot || V_{0}^{-1}(\tilde{\gamma})||_{op} \\
     &\leq \sup_{\tilde{\gamma} \in \Gamma_{0}}||V_{0}^{-1}(\tilde{\gamma})||_{op} \sum_{k=0}^{\infty}\left | \left| V_{0}^{-1}(\tilde{\gamma})(D_{n}(\tilde{\gamma},p_{W|Z})-V_{0}(\tilde{\gamma}))\right | \right |_{op}^{k} \\
     &\leq  \sup_{\tilde{\gamma} \in \Gamma_{0}}||V_{0}^{-1}(\tilde{\gamma})||_{op} \sum_{k=0}^{\infty}2^{-k} \\
     &\leq 2\sup_{\tilde{\gamma} \in \Gamma_{0}}||V_{0}^{-1}(\tilde{\gamma})||_{op} < \infty
 \end{align*}
holds uniformly over $\Gamma_{0}\times \tilde{\mathcal{P}}_{n,W|Z}$, where the first inequality uses the submultiplicativity of the operator norm, the equality uses the Neumann series, the second inequality is the triangle inequality and the submultiplicativity of the operator norm, the third inequality uses (\ref{eq:boundcontract}), and the fourth inequality convergence of a geometric series and the fact that Assumptions \ref{as:covariance}.1 and \ref{as:localident}.2 imply $\sup_{\tilde{\gamma} \in \Gamma_{0}}||V_{0}^{-1}(\tilde{\gamma})||_{op} < \infty$. Consequently, I conclude that
\begin{align*}
    \sup_{(\tilde{\gamma},p_{W|Z}) \in \Gamma_{0}\times \tilde{\mathcal{P}}_{n,W|Z}}||D_{n}(\tilde{\gamma},p_{W|Z})||_{op} = O(1).
\end{align*}

\paragraph{Step 2.} I use the results from Step 1 to establish that $q_{n}(\cdot,p_{W|Z}): \Gamma_{0} \rightarrow \mathbb{R}^{d_{\gamma}}$ is a uniform contraction mapping over $\tilde{\mathcal{P}}_{n,W|Z}$. Applying the Implicit Function Theorem, $\tilde{\gamma} \mapsto q_{n}(\tilde{\gamma},p_{W|Z})$ is continuously differentiable over $\Gamma_{0}$ for large $n$ with derivative given by
\begin{align*}
    \frac{\partial }{\partial \tilde{\gamma}'}q_{n}(\tilde{\gamma},p_{W|Z}) = -D_{n}(\tilde{\gamma},p_{W|Z})^{-1}B_{n}(\tilde{\gamma},p_{W|Z}),
\end{align*}
Consequently, by the Mean-Value Theorem, 
\begin{align*}
    &\sup_{p_{W|Z} \in \tilde{\mathcal{P}}_{n,W|Z}}||q_{n}(\tilde{\gamma}_{1},p_{W|Z})-q_{n}(\tilde{\gamma}_{2},p_{W|Z})||_{2} \\
    &\quad \leq \sup_{(\tilde{\gamma},p_{W|Z}) \in \Gamma_{0}\times \tilde{\mathcal{P}}_{n,W|Z}}\left| \left | D_{n}(\tilde{\gamma},p_{W|Z})^{-1}B_{n}(\tilde{\gamma},p_{W|Z}) \right| \right|_{op} \cdot ||\tilde{\gamma}_{1}-\tilde{\gamma}_{2}||_{2}.
\end{align*}
for any $\tilde{\gamma}_{1},\tilde{\gamma}_{2} \in \Gamma_{0}$. Using the submultiplicativity of the operator norm to conclude that
\begin{align*}
\sup_{(\tilde{\gamma},p_{W|Z}) \in \Gamma_{0} \times \tilde{\mathcal{P}}_{n,W|Z}} \left | \left | \frac{\partial }{\partial \tilde{\gamma}'}q_{n}(\tilde{\gamma},p_{W|Z})  \right| \right|_{op} &\leq \sup_{(\tilde{\gamma},p_{W|Z}) \in \Gamma \times \tilde{\mathcal{P}}_{n,W|Z}}\left | \left|D_{n}(\tilde{\gamma},p_{W|Z})^{-1}\right| \right|_{op}\\
&\quad \times \sup_{(\tilde{\gamma},p_{W|Z}) \in \Gamma \times \tilde{\mathcal{P}}_{n,W|Z}}\left| \left| B_{n}(\tilde{\gamma},p_{W|Z})\right| \right|_{op} \\
&\quad = O(1)o(1) \\ 
&\quad = o(1)
\end{align*}
by Step 1. Consequently, there is a constant $\rho \in [0,1)$ such that
\begin{align*}
     \sup_{(\tilde{\gamma},p_{W|Z}) \in \Gamma_{0}\times \tilde{\mathcal{P}}_{n,W|Z}}\left| \left | D_{n}(\tilde{\gamma},p_{W|Z})^{-1}B_{n}(\tilde{\gamma},p_{W|Z}) \right| \right|_{op} \leq \rho
\end{align*}
for large $n$, establishing that $q_{n}(\cdot,p_{W|Z}): \Gamma_{0} \rightarrow \mathbb{R}^{d_{\gamma}}$ is a uniform (over $\tilde{\mathcal{P}}_{n,W|Z}$) contraction for large $n$.
\end{proof}
\begin{lemma}\label{lem:consistency}
If Assumptions \ref{as:uemfp}, \ref{as:dgp2}, \ref{as:correctspec}, \ref{as:moments}, \ref{as:covariance}.1, \ref{as:criterion}, \ref{as:localident}, and \ref{as:concentration} hold, then, for $P_{0,Z}^{\infty}$-almost every fixed realization $\{z_{i}\}_{i \geq 1}$ of $\{Z_{i}\}_{i \geq 1}$, 
\begin{align*}
\limsup_{n\rightarrow \infty}\sup_{p_{W|Z} \in \tilde{\mathcal{P}}_{n,W|Z}}||\gamma_{n}-\gamma_{0}||_{2} = 0,    
\end{align*}
where $\{\tilde{\mathcal{P}}_{n,W|Z}\}_{n \geq 1}$ are the sets defined in Assumption \ref{as:concentration}.
\end{lemma}
\begin{proof}
 Let $\{z_{i}\}_{i \geq 1}$ be an arbitrary fixed realization of $\{Z_{i}\}_{i \geq 1}$. I show that for any $\xi>0$, there exists a $N_{\xi} \in \mathbb{Z}_{++} $ such that 
 \begin{align*}
 \sup_{p_{W|Z} \in \tilde{\mathcal{P}}_{n,W|Z}}||\gamma_{n}-\gamma_{0}||_{2}< \xi
 \end{align*}
 for all $n \geq N_{\xi}$. Let $\xi > 0$ be arbitrary. Recall that $\gamma_{n,r}$ is defined by the recursive relation $\gamma_{n,r}=q_{n}(\gamma_{n,r-1},p_{W|Z})$ for $r \geq 1$ and some initial condition $\gamma_{n,0} \in \Gamma$. I first claim that there is a sequence $\{r_{n}\}_{n \geq 1}$ and $N_{\xi,1} \in \mathbb{Z}_{++}$ such that 
 \begin{align}\label{eq:L41}
 \sup_{p_{W|Z} \in \tilde{\mathcal{P}}_{n,W|Z}}||\gamma_{n} - \gamma_{n,r_{n}}||_{2} < \frac{\xi}{2}   
 \end{align} 
 for all $n \geq N_{\xi,1}$. Without loss of generality I can set $\gamma_{n,0} = q_{n}(\tilde{\gamma},p_{W|Z})$ for some $\tilde{\gamma} \in \Gamma$, since, for any given alternative $\tilde{\gamma}_{n,0} \in \Gamma$, one can always redefine $\gamma_{n,0} = q_{n}(\tilde{\gamma}_{n,0},p_{W|Z})$. Lemma 2 thus implies that $\{\gamma_{n,r}(p_{W|Z}): p_{W|Z} \in \tilde{\mathcal{P}}_{n,W|Z}, r \geq 0\} \subseteq \Gamma_{0}$ for $n$ large. Moreover, since Lemma 3 implies that $\{q_{n}(\cdot,p_{W|Z}): \Gamma_{0} \rightarrow \mathbb{R}^{d_{\gamma}}, p_{W|Z} \in \tilde{\mathcal{P}}_{n,W|Z}\}$ defines a collection of contraction mappings for $n$ large with a common (across $\tilde{\mathcal{P}}_{n,W|Z}$) constant $\rho \in [0,1)$, I can use the a priori estimate of the contraction mapping to obtain
 \begin{align*}
 \sup_{p_{W|Z} \in \tilde{\mathcal{P}}_{n,W|Z}}||\gamma_{n}-\gamma_{n,r}||_{2} \leq \frac{\rho^{r}}{1-\rho}\cdot\sup_{p_{W|Z} \in \tilde{\mathcal{P}}_{n,W|Z}}||\gamma_{n,1}-\gamma_{n,0}||_{2}    
 \end{align*}
 Since $\rho \in [0,1)$ and $\sup_{p_{W|Z} \in \tilde{\mathcal{P}}_{n,W|Z}}||\gamma_{n,1}-\gamma_{n,0}||_{2}  = o(1)$ by Lemma 2, I can find $\{r_{n}\}_{n \geq 1}$ with $r_{n}\rightarrow \infty$ as $n\rightarrow \infty$ such that (\ref{eq:L41}) holds. Next, since $q_{n}(\gamma_{n,r_{n}-1},p_{W|Z}) \in \{q_{n}(\tilde{\gamma},p_{W|Z}): \tilde{\gamma} \in \Gamma\}$, Lemma 2 implies there exists $N_{\xi,2} \in \mathbb{Z}_{++}$ such that 
 \begin{align*}
 \sup_{p_{W|Z} \in \tilde{\mathcal{P}}_{n,W|Z}}||\gamma_{n,r_{n}} - \gamma_{0}||_{2} < \frac{\xi}{2}    
 \end{align*}
 for all $n \geq N_{\xi,2}$. Defining $N_{\xi} = \max\{N_{\xi,1},N_{\xi,2}\}$ and applying the triangle inequality,
 \begin{align*}
 \sup_{p_{W|Z} \in \tilde{\mathcal{P}}_{n,W|Z}}||\gamma_{n}-\gamma_{0}||_{2} < \xi    
 \end{align*} 
 for all $n \geq N_{\xi}$. Since $\xi> 0$ was chosen arbitrarily, the claim follows.
\end{proof} 
\begin{lemma}\label{lem:jacobiancross}
Suppose that Assumptions \ref{as:uemfp}, \ref{as:dgp2}, \ref{as:correctspec}, \ref{as:moments}, \ref{as:covariance}, \ref{as:criterion}, \ref{as:localident}, and \ref{as:concentration} hold. Then, for $P_{0,Z}^{\infty}$-almost every realization $\{z_{i}\}_{i \geq 1}$ of $\{Z_{i}\}_{i \geq 1}$,
\begin{align*}
\sup_{p_{W|Z}\in \tilde{\mathcal{P}}_{n,W|Z}}\left | \left | M(\cdot,\gamma_{n})-M_{0}(\cdot,\gamma_{0}) \right | \right|_{n,2} \lesssim o(n^{-1/4}) + \sup_{p_{W|Z} \in \tilde{\mathcal{P}}_{n,W|Z}}||\gamma_{n}-\gamma_{0}||_{2}
\end{align*}
and
\begin{align*}
\sup_{p_{W|Z}\in \tilde{\mathcal{P}}_{n,W|Z}}\left | \left | \Sigma(\cdot,\gamma_{n})-\Sigma_{0}(\cdot,\gamma_{0}) \right | \right|_{n,\infty} \lesssim o(n^{-1/4}) + \sup_{p_{W|Z} \in \tilde{\mathcal{P}}_{n,W|Z}}||\gamma_{n}-\gamma_{0}||_{2}
\end{align*}
where $\{\tilde{\mathcal{P}}_{n,W|Z}\}_{n \geq 1}$ are the sets defined in Assumption \ref{as:concentration}.
\end{lemma}
\begin{proof}
 The proof has two steps with each corresponding to the statements in the lemma. That is, Step 1 verifies the first claim and Step 2 verifies the second.
\paragraph{Step 1.} Let $\{z_{i}\}_{i \geq 1}$ be a fixed realization of $\{Z_{i}\}_{i \geq 1}$. Since $\gamma_{n} \in \Gamma_{0}$ for all $p_{W|Z} \in \tilde{\mathcal{P}}_{n,W|Z}$ and $n$ large by Lemma 4, I can apply the triangle inequality and the definition of the supremum to conclude that
\begin{align*}
&\sup_{p_{W|Z}\in \tilde{\mathcal{P}}_{n,W|Z}}\left | \left | M(\cdot,\gamma_{n})-M_{0}(\cdot,\gamma_{0}) \right | \right|_{n,2} \\
&\quad \leq \sup_{p_{W|Z}\in \tilde{\mathcal{P}}_{n,W|Z}}\left | \left | M(\cdot,\gamma_{n})-M_{0}(\cdot,\gamma_{n}) \right | \right|_{n,2}+\sup_{p_{W|Z}\in \tilde{\mathcal{P}}_{n,W|Z}}\left | \left | M_{0}(\cdot,\gamma_{n})-M_{0}(\cdot,\gamma_{0}) \right | \right|_{n,2} \\
&\quad \leq \sup_{(\gamma,p_{W|Z})\in \Gamma_{0}\times \tilde{\mathcal{P}}_{n,W|Z}}\left | \left | M(\cdot,\gamma)-M_{0}(\cdot,\gamma) \right | \right|_{n,2}+\sup_{p_{W|Z}\in \tilde{\mathcal{P}}_{n,W|Z}}\left | \left | M_{0}(\cdot,\gamma_{n})-M_{0}(\cdot,\gamma_{0}) \right | \right|_{n,2}.
    \end{align*}
for large $n$. Applying Assumption \ref{as:concentration}.3(b), I know that
\begin{align*}
    \sup_{(\gamma,p_{W|Z})\in \Gamma_{0}\times \tilde{\mathcal{P}}_{n,W|Z}}\left | \left | M(\cdot,\gamma)-M_{0}(\cdot,\gamma) \right | \right|_{n,2} =o(n^{-\frac{1}{4}})
\end{align*}
Since $\{\gamma_{n}(p_{W|Z}): p_{W|Z} \in \tilde{\mathcal{P}}_{n,W|Z}\} \subseteq \Gamma_{0}$ for large $n$ by Lemma 4, Assumption \ref{as:moments}.2 and the mean-value theorem implies there exists $\overline{M}_{0} \in [0,\infty)$ such that
\begin{align*}
\sup_{p_{W|Z}\in \tilde{\mathcal{P}}_{n,W|Z}}\left | \left | M_{0}(\cdot,\gamma_{n})-M_{0}(\cdot,\gamma_{0}) \right | \right|_{n,2}  \leq \overline{M}_{0}\sup_{p_{W|Z}\in \tilde{\mathcal{P}}_{n,W|Z}}||\gamma_{n}-\gamma_{0}||_{2} 
\end{align*}
as $n\rightarrow \infty$. Together, these results imply that the following holds
\begin{align*}
\sup_{p_{W|Z}\in \tilde{\mathcal{P}}_{n,W|Z}}\left | \left | M(\cdot,\gamma_{n})-M_{0}(\cdot,\gamma_{0}) \right | \right|_{n,2}\lesssim o(n^{-1/4}) + \sup_{p_{W|Z} \in \tilde{\mathcal{P}}_{n,W|Z}}||\gamma_{n}-\gamma_{0}||_{2} 
\end{align*}
Since the argument holds for an arbitrary fixed realization $\{z_{i}\}_{i \geq 1}$, I can conclude that it holds for $P_{0,Z}^{\infty}$-almost every fixed realization $\{z_{i}\}_{i \geq 1}$ of $\{Z_{i}\}_{i \geq 1}$.
\paragraph{Step 2.} Let $\{z_{i}\}_{i \geq 1}$ be a fixed realization of $\{Z_{i}\}_{i \geq 1}$. Applying the triangle inequality and the definition of the supremum, I know that
\begin{align*}
    &\sup_{p_{W|Z} \in \tilde{\mathcal{P}}_{n,W|Z}}||\Sigma(\cdot,\gamma_{n})-\Sigma_{0}(\cdot,\gamma_{0})||_{n,\infty}  \\
    &\quad \leq \sup_{p_{W|Z} \in \tilde{\mathcal{P}}_{n,W|Z}}||\Sigma(\cdot,\gamma_{n})-\Sigma_{0}(\cdot,\gamma_{n})||_{n,\infty} + \sup_{p_{W|Z} \in \tilde{\mathcal{P}}_{n,W|Z}}||\Sigma_{0}(\cdot,\gamma_{n})-\Sigma_{0}(\cdot,\gamma_{0})||_{n,\infty}  \\
    &\quad \leq \sup_{(\gamma,p_{W|Z}) \in \Gamma \times \tilde{\mathcal{P}}_{n,W|Z}}||\Sigma(\cdot,\gamma)-\Sigma_{0}(\cdot,\gamma)||_{n,\infty} + \sup_{p_{W|Z} \in \tilde{\mathcal{P}}_{n,W|Z}}||\Sigma_{0}(\cdot,\gamma_{n})-\Sigma_{0}(\cdot,\gamma_{0})||_{n,\infty}.
\end{align*}
Assumption \ref{as:concentration}.3(c) yields
\begin{align*}
    \sup_{(\gamma,p_{W|Z}) \in \Gamma \times \tilde{\mathcal{P}}_{n,W|Z}}||\Sigma(\cdot,\gamma)-\Sigma_{0}(\cdot,\gamma)||_{n,\infty} = o(n^{-\frac{1}{4}})
\end{align*}
Moreover, Lemma 4 implies that $\gamma_{n} \in \Gamma_{0}$ for large $n$. Consequently, Assumption \ref{as:covariance}.2 implies that there exists a constant $\overline{\Sigma}_{0} \in [0,\infty)$ such that
\begin{align*}
\sup_{p_{W|Z} \in \tilde{\mathcal{P}}_{n,W|Z}}||\Sigma_{0}(\cdot,\gamma_{n})-\Sigma_{0}(\cdot,\gamma_{0})||_{n,\infty} \leq \overline{\Sigma}_{0} \sup_{p_{W|Z}\in \tilde{\mathcal{P}}_{n,W|Z}}||\gamma_{n}-\gamma_{0}||_{2}.
\end{align*}
as $n\rightarrow \infty$. Putting these results together,
\begin{align*}
 \sup_{p_{W|Z} \in \tilde{\mathcal{P}}_{n,W|Z}}||\Sigma(\cdot,\gamma_{n})-\Sigma_{0}(\cdot,\gamma_{0})||_{n,\infty} \lesssim o(n^{-1/4}) + \sup_{p_{W|Z} \in \tilde{\mathcal{P}}_{n,W|Z}}||\gamma_{n}-\gamma_{0}||_{2}.
\end{align*}
Since the argument holds for an arbitrary fixed realization $\{z_{i}\}_{i \geq 1}$, I can conclude that it holds for $P_{0,Z}^{\infty}$-almost every realization $\{z_{i}\}_{i \geq 1}$ of $\{Z_{i}\}_{i \geq 1}$.
\end{proof}
\begin{lemma}\label{lem:gram}
Let $V_{n}(\gamma,\bar{\gamma},p_{W|Z}) = n^{-1}\sum_{i=1}^{n}M(z_{i},\gamma)'\Sigma^{-1}(z_{i},\gamma)M(z_{i},\bar{\gamma})$. If Assumptions \ref{as:uemfp}, \ref{as:dgp2}, \ref{as:correctspec}, \ref{as:moments}, \ref{as:covariance}, \ref{as:criterion}, \ref{as:localident}, and \ref{as:concentration} hold, then, for $P_{0,Z}^{\infty}$-almost every fixed realization $\{z_{i}\}_{i \geq 1}$ of $\{Z_{i}\}_{i \geq 1}$ and for every positive sequence $\{r_{n}\}_{n \geq 1}$ with $r_{n} \rightarrow 0$ as $n\rightarrow \infty$, the collection of matrices $\{V_{n}(\gamma,\bar{\gamma},p_{W|Z}): (\gamma,\bar{\gamma},p_{W|Z}) \in B(0,r_{n})\times B(0,r_{n})\times \tilde{\mathcal{P}}_{n,W|Z}\}$ forms a class of nonsingular matrices for $n$ large, where $\{\tilde{\mathcal{P}}_{n,W|Z}\}_{n \geq 1}$ are defined in Assumption \ref{as:concentration} and $\{B(\gamma_{0},r_{n})\}_{n \geq 1}$ are open balls in $(\mathbb{R}^{d_{\gamma}},||\cdot||_{2})$ centered at $\gamma_{0}$ and with radii $\{r_{n}\}_{n \geq 1}$. Moreover, $\sup_{(\gamma,\bar{\gamma},p_{W|Z}) \in B(\gamma_{0},r_{n})\times B(\gamma_{0},r_{n})\times \tilde{\mathcal{P}}_{n,W|Z}}||V_{n}^{-1}(\gamma,\bar{\gamma},p_{W|Z})||_{op} = O(1)$.
\end{lemma}
\begin{proof}
The proof has three steps. Step 1 establishes an intermediate uniform convergence property. Step 2 establishes the first claim. Step 3 establishes the second claim.
\paragraph{Step 1.} As an intermediate technical result, I show that
\begin{align*}
    \limsup_{n\rightarrow \infty}\sup_{(\gamma,\bar{\gamma},p_{W|Z}) \in B(\gamma_{0},r_{n})\times B(\gamma_{0},r_{n}) \times \tilde{\mathcal{P}}_{n,W|Z}} \left| \left |V_{n}(\gamma,\bar{\gamma},p_{W|Z})-V_{0}\right| \right|_{F} = 0,
\end{align*}
where $||\cdot||_{F}$ denotes the Frobenius norm. Since Assumption \ref{as:localident}.1 implies that $$||V_{n}(\gamma_{0},\gamma_{0},p_{0,W|Z}) - V_{0}||_{F} = o(1),$$ it suffices to show that
\begin{align*}
   \limsup_{n\rightarrow \infty}\sup_{(\gamma,\bar{\gamma},p_{W|Z}) \in B(\gamma_{0},r_{n})\times B(\gamma_{0},r_{n}) \times \tilde{\mathcal{P}}_{n,W|Z}} \left| \left |V_{n}(\gamma,\bar{\gamma},p_{W|Z})-V_{n}(\gamma_{0},\gamma_{0},p_{0,W|Z})\right| \right|_{F} = 0. 
\end{align*}
Adding and subtracting $M_{0}(z_{i},\gamma_{0})$ and $\Sigma_{0}(z_{i},\gamma_{0})$ appropriately, I obtain
\begin{align*}
 &V_{n}(\gamma,\bar{\gamma},p_{W|Z})-V_{n}(\gamma_{0},\gamma_{0},p_{0,W|Z}) \\
 &\quad =\frac{1}{n}\sum_{i=1}^{n}M(z_{i},\gamma)'(\Sigma^{-1}(z_{i},\gamma)-\Sigma^{-1}_{0}(z_{i},\gamma_{0}))M(z_{i},\bar{\gamma}) \\
 &\quad \quad +\frac{1}{n}\sum_{i=1}^{n}M(z_{i},\gamma)'\Sigma^{-1}_{0}(z_{i},\gamma_{0})(M(z_{i},\bar{\gamma})-M_{0}(z_{i},\gamma_{0})) \\
 &\quad \quad + \frac{1}{n}\sum_{i=1}^{n}(M(z_{i},\gamma)-M_{0}(z_{i},\gamma_{0}))'\Sigma^{-1}_{0}(z_{i},\gamma_{0})^{-1}M_{0}(z_{i},\gamma_{0})
\end{align*}
for all $(\gamma,\bar{\gamma},p_{W|Z}) \in B(\gamma_{0},r_{n}) \times B(\gamma_{0},r_{n})  \times \tilde{\mathcal{P}}_{n,W|Z}$. The triangle inequality then reveals that it suffices to show three results
\begin{align}
 &\frac{1}{n}\sum_{i=1}^{n} \left| \left |M(z_{i},\gamma)'(\Sigma^{-1}(z_{i},\gamma)-\Sigma^{-1}_{0}(z_{i},\gamma_{0}))M(z_{i},\bar{\gamma}) \right | \right|_{F} = o(1) \label{eq:c1} \\
  & \frac{1}{n}\sum_{i=1}^{n} \left| \left |M(z_{i},\gamma)'\Sigma^{-1}_{0}(z_{i},\gamma_{0})(M(z_{i},\bar{\gamma})-M_{0}(z_{i},\gamma_{0})) \right | \right|_{F} = o(1) \label{eq:c2} \\
  &\frac{1}{n}\sum_{i=1}^{n} \left| \left |(M(z_{i},\gamma)-M_{0}(z_{i},\gamma_{0}))'\Sigma^{-1}_{0}(z_{i},\gamma_{0})^{-1}M_{0}(z_{i},\gamma_{0})\right | \right |_{F} = o(1) \label{eq:c3}
\end{align}
uniformly in $(\gamma,\bar{\gamma},p_{W|Z}) \in B(\gamma_{0},r_{n})  \times B(\gamma_{0},r_{n})  \times\tilde{\mathcal{P}}_{n,W|Z}$. I verify these using three sub-steps.
\paragraph{Step 1A.} For (\ref{eq:c1}), I use the submultiplicative property of the Frobenius norm and apply H\"{o}lder's inequality to conclude that
\begin{align*}
&\sup_{(\gamma,\bar{\gamma},p_{W|Z}) \in B(\gamma_{0},r_{n})\times B(\gamma_{0},r_{n})   \times \tilde{\mathcal{P}}_{n,W|Z}} \frac{1}{n}\sum_{i=1}^{n} \left| \left |M(z_{i},\gamma)'(\Sigma^{-1}(z_{i},\gamma)-\Sigma^{-1}_{0}(z_{i},\gamma_{0}))M(z_{i},\bar{\gamma}) \right | \right|_{F} \\
&\quad \leq \sup_{(\gamma,p_{W|Z}) \in B(\gamma_{0},r_{n})  \times \tilde{\mathcal{P}}_{n,W|Z}}\max_{1 \leq i \leq n}||\Sigma^{-1}(z_{i},\gamma)-\Sigma^{-1}_{0}(z_{i},\gamma_{0})||_{F}\\
&\quad \quad \times \left(\sup_{(\gamma,p_{W|Z}) \in B(\gamma_{0},r_{n})  \times \tilde{\mathcal{P}}_{n,W|Z}}\sqrt{\frac{1}{n}\sum_{i=1}^{n} \left| \left |M(z_{i},\gamma) \right| \right|_{F}^{2}}\right)^{2}.   
\end{align*}
Using the matrix identity $A^{-1}-B^{-1} = B^{-1}(B-A)A^{-1}$ and the submultiplicativity of the Frobenius norm, I know that
\begin{align*}
    &\sup_{(\gamma,p_{W|Z}) \in B(\gamma_{0},r_{n})  \times \tilde{\mathcal{P}}_{n,W|Z}}\max_{1 \leq i \leq n}||\Sigma^{-1}(z_{i},\gamma)-\Sigma^{-1}_{0}(z_{i},\gamma_{0})||_{F} \\
    &\quad \leq \sup_{(\gamma,p_{W|Z}) \in \Gamma  \times \tilde{\mathcal{P}}_{n,W|Z}}\max_{1 \leq i \leq n}||\Sigma^{-1}(z_{i},\gamma)||_{F}    \max_{1 \leq i \leq n}||\Sigma^{-1}_{0}(z_{i},\gamma_{0})||_{F} \\
    &\quad \quad  \times \sup_{(\gamma,p_{W|Z}) \in B(\gamma_{0},r_{n})  \times \tilde{\mathcal{P}}_{n,W|Z}}\max_{1 \leq i \leq n}||\Sigma(z_{i},\gamma)-\Sigma_{0}(z_{i},\gamma_{0})||_{F}.
\end{align*}
Part 2 of Lemma 5 and $||\cdot||_{F} \leq \sqrt{d_{g}}||\cdot||_{op}$ implies that
\begin{align}\label{eq:remark21}
\sup_{(\gamma,p_{W|Z}) \in B(\gamma_{0},r_{n})  \times \tilde{\mathcal{P}}_{n,W|Z}}\max_{1 \leq i \leq n}||\Sigma(z_{i},\gamma)-\Sigma_{0}(z_{i},\gamma_{0})||_{F} \lesssim \sup_{\gamma \in B(\gamma_{0},r_{n})}||\gamma-\gamma_{0}||_{2} + o(n^{-\frac{1}{4}}) = o(1).
\end{align}
for $r_{n}\rightarrow 0$ as $n\rightarrow \infty$. The symmetry of $\Sigma^{-1}_{0}(z_{i},\gamma)$ and Assumption \ref{as:covariance} implies that
\begin{align*}
    \limsup_{n\rightarrow \infty}\max_{1 \leq i \leq n}||\Sigma^{-1}_{0}(z_{i},\gamma_{0})||_{F}  &= \limsup_{n\rightarrow \infty}\max_{1 \leq i \leq n}\sqrt{\sum_{j=1}^{d_{g}}\lambda_{j}^{2}(\Sigma_{0}^{-1}(z_{i},\gamma_{0}))} \\
    &\leq \frac{\sqrt{d_{g}}}{\liminf_{n\rightarrow \infty}\min_{1 \leq i \leq n}\lambda_{min}(\Sigma_{0}(z_{i},\gamma_{0}))} \\
    &\leq \frac{\sqrt{d_{g}}}{\underline{\lambda}} < \infty.
\end{align*}
Lemma 1 and the inequality $||A||_{F} \leq \sqrt{d_{g}}||A||_{op}$ for any $A \in \mathbb{R}^{d_{g}\times d_{g}}$ implies that
\begin{align*}
 &\sup_{(\gamma,p_{W|Z}) \in \Gamma  \times \tilde{\mathcal{P}}_{n,W|Z}}\max_{1 \leq i \leq n}||\Sigma^{-1}(z_{i},\gamma)||_{F} \\
 &\quad \leq    \sqrt{d_{g}}\sup_{(\gamma,p_{W|Z}) \in \Gamma  \times \tilde{\mathcal{P}}_{n,W|Z}}\max_{1 \leq i \leq n}||\Sigma^{-1}(z_{i},\gamma)||_{op} = O(1). 
\end{align*}
Consequently,
\begin{align*}
    \sup_{(\gamma,p_{W|Z}) \in B(\gamma_{0},r_{n})  \times \tilde{\mathcal{P}}_{n,W|Z}}\max_{1 \leq i \leq n}||\Sigma^{-1}(z_{i},\gamma)-\Sigma^{-1}_{0}(z_{i},\gamma_{0})||_{F} = o(1).
\end{align*}
Applying the triangle inequality and Part 1 of Lemma 5, there is a constant $C>0$ such that 
\begin{align*}
 &\sup_{(\gamma,p_{W|Z}) \in B(\gamma_{0},r_{n})  \times \tilde{\mathcal{P}}_{n,W|Z}}\sqrt{\frac{1}{n}\sum_{i=1}^{n} \left| \left |M(z_{i},\gamma) \right| \right|_{F}^{2}}\\
 &\quad \leq  \sup_{(\gamma,p_{W|Z}) \in B(\gamma_{0},r_{n})  \times \tilde{\mathcal{P}}_{n,W|Z}}\sqrt{\frac{1}{n}\sum_{i=1}^{n} \left| \left |M(z_{i},\gamma)-M_{0}(z_{i},\gamma_{0}) \right| \right|_{F}^{2}} +   \sqrt{\frac{1}{n}\sum_{i=1}^{n} \left| \left |M_{0}(z_{i},\gamma_{0}) \right| \right|_{F}^{2}} \\
 &\quad \leq C\sup_{\gamma \in B(\gamma_{0},r_{n}) } \left| \left |\gamma-\gamma_{0} \right| \right|_{2} + o(n^{-\frac{1}{4}}) + \sqrt{\frac{1}{n}\sum_{i=1}^{n} \left| \left |M_{0}(z_{i},\gamma_{0}) \right| \right|_{F}^{2}},
\end{align*}
Consequently, I use that $r_{n}\rightarrow 0$ as $n\rightarrow \infty$ and Assumption \ref{as:moments} to conclude that
\begin{align*}
    \sup_{(\gamma,p_{W|Z}) \in B(\gamma_{0},r_{n})  \times \tilde{\mathcal{P}}_{n,W|Z}}\sqrt{\frac{1}{n}\sum_{i=1}^{n} \left| \left |M(z_{i},\gamma) \right| \right|_{F}^{2}} = O(1).
\end{align*}
\paragraph{Step 1B.} For (\ref{eq:c2}), I use the submultiplicativity of the Frobenius norm and apply H\"{o}lder's inequality to obtain that
\begin{align*}
&\sup_{(\gamma,\bar{\gamma},p_{W|Z}) \in B(\gamma_{0},r_{n})\times B(\gamma_{0},r_{n})\times  \tilde{\mathcal{P}}_{n,W|Z}} \frac{1}{n}\sum_{i=1}^{n} \left| \left |M(z_{i},\gamma)'\Sigma^{-1}_{0}(z_{i},\gamma_{0})(M(z_{i},\bar{\gamma})-M_{0}(z_{i},\gamma_{0})) \right | \right|_{F} \\
&\quad \leq  \sup_{(\gamma,p_{W|Z}) \in B(\gamma_{0},r_{n})\times  \tilde{\mathcal{P}}_{n,W|Z}}\sqrt{\frac{1}{n}\sum_{i=1}^{n} \left| \left |M(z_{i},\gamma) \right | \right |_{F}^{2}}\max_{1 \leq i \leq n}||\Sigma^{-1}_{0}(z_{i},\gamma_{0})||_{F} \\
&\quad \quad \times \sup_{(\gamma,p_{W|Z})\in B(\gamma_{0},r_{n})\times \tilde{\mathcal{P}}_{n,W|Z}}\sqrt{\frac{1}{n}\sum_{i=1}^{n}\left | \left | M(z_{i},\gamma)-M_{0}(z_{i},\gamma_{0})) \right | \right|_{F}^{2}}.
\end{align*}
In dealing with (\ref{eq:c1}), I showed that 
\begin{align*}
    \sup_{(\gamma,p_{W|Z}) \in B(\gamma_{0},r_{n})\times  \tilde{\mathcal{P}}_{n,W|Z}}\sqrt{\frac{1}{n}\sum_{i=1}^{n} \left| \left |M(z_{i},\gamma) \right | \right |_{F}^{2}} = O(1).
\end{align*}
Assumption \ref{as:covariance} and the equivalence between the Frobenius and operator norm implies that 
\begin{align*}
    \max_{1 \leq i \leq n}||\Sigma^{-1}_{0}(z_{i},\gamma_{0})||_{F} = O(1).
\end{align*}
Part 1 of Lemma 5 shows that 
\begin{align}\label{eq:remark22}
    \sup_{(\gamma,p_{W|Z})\in \tilde{\mathcal{P}}_{n,W|Z}\times B(\gamma_{0},r_{n})}\sqrt{\frac{1}{n}\sum_{i=1}^{n}\left | \left | M(z_{i},\gamma)-M_{0}(z_{i},\gamma_{0}) \right | \right|_{F}^{2}} \lesssim \sup_{\gamma\in  B(\gamma_{0},r_{n})}||\gamma-\gamma_{0}||_{2} + o(n^{-\frac{1}{4}}).
\end{align}
Since $r_{n}\rightarrow 0$ as $n\rightarrow  \infty$, I conclude that
\begin{align*}
    \sup_{(\gamma,\bar{\gamma},p_{W|Z}) \in B(\gamma_{0},r_{n})\times B(\gamma_{0},r_{n})\times  \tilde{\mathcal{P}}_{n,W|Z}} \frac{1}{n}\sum_{i=1}^{n} \left| \left |M(z_{i},\gamma)'\Sigma^{-1}_{0}(z_{i},\gamma_{0})(M(z_{i},\gamma)-M_{0}(z_{i},\gamma_{0})) \right | \right|_{F}=o(1).
\end{align*}
\paragraph{Step 1C.} For (\ref{eq:c3}), a similar argument to that applied (\ref{eq:c2}) to yields that
\begin{align}
&\sup_{(\gamma,p_{W|Z}) \in B(\gamma_{0},r_{n})  \times \tilde{\mathcal{P}}_{n,W|Z}}\frac{1}{n}\sum_{i=1}^{n} \left| \left |(M(z_{i},\gamma)-M_{0}(z_{i},\gamma_{0}))'\Sigma^{-1}_{0}(z_{i},\gamma_{0})^{-1}M_{0}(z_{i},\gamma_{0})\right | \right |_{F} \nonumber \\
&\quad \leq  \max_{1 \leq i \leq n}||\Sigma^{-1}_{0}(z_{i},\gamma_{0})||_{F}\sqrt{\frac{1}{n}\sum_{i=1}^{n}||M_{0}(z_{i},\gamma_{0})||_{F}^{2}} \nonumber \\
&\quad \quad \times \sup_{(\gamma,p_{W|Z}) \in B(\gamma_{0},r_{n}) \times\tilde{\mathcal{P}}_{n,W|Z}}\sqrt{\frac{1}{n}\sum_{i=1}^{n}||M(z_{i},\gamma)-M_{0}(z_{i},\gamma_{0})||_{F}^{2}} \nonumber\\
&\quad \lesssim \max_{1 \leq i \leq n}||\Sigma^{-1}_{0}(z_{i},\gamma_{0})||_{F}\sqrt{\frac{1}{n}\sum_{i=1}^{n}||M_{0}(z_{i},\gamma_{0})||_{F}^{2}}\left(\sup_{\gamma \in B(\gamma_{0},r_{n})}||\gamma-\gamma_{0}||_{2}+o(n^{-\frac{1}{4}})\right) \label{eq:remark23} \\
&\quad = o(1). \nonumber
\end{align}
\paragraph{Step 2.} Since $V_{0}$ is a bounded matrix (by Assumptions \ref{as:covariance}.1 and \ref{as:localident}.2), Step 1 implies that
\begin{align*}
\limsup_{n\rightarrow \infty}\sup_{(\gamma,\bar{\gamma},p_{W|Z}) \in B(\gamma_{0},r_{n})\times B(\gamma_{0},r_{n}) \times \tilde{\mathcal{P}}_{n,W|Z}}|\det(V_{n}(\gamma,\bar{\gamma},p_{W|Z}))-\det(V_{0})| = 0
\end{align*}
Consequently, for every $\xi > 0$, there exists $N_{\xi} \geq 1$ such that 
\begin{align*}
    \det(V_{0})-\xi < \inf_{(\gamma,\bar{\gamma},p_{W|Z}) \in B(\gamma_{0},r_{n})\times B(\gamma_{0},r_{n}) \times \tilde{\mathcal{P}}_{n,W|Z}}\det(V_{n}(\gamma,\bar{\gamma},p_{W|Z})) 
\end{align*}
and
\begin{align*}
    \sup_{(\gamma,\bar{\gamma},p_{W|Z}) \in B(\gamma_{0},r_{n})\times B(\gamma_{0},r_{n}) \times \tilde{\mathcal{P}}_{n,W|Z}}\det(V_{n}(\gamma,\bar{\gamma},p_{W|Z})) < \det(V_{0})+ \xi
\end{align*}
for all $n \geq N_{\xi}$. Setting $\xi = \det(V_{0})/2 > 0$ (by Assumption \ref{as:covariance}.1 and \ref{as:localident}.2), it follows that
\begin{align*}
\{\det V_{n}(\gamma,\bar{\gamma},p_{W|Z}): (\gamma,\bar{\gamma},p_{W|Z}) \in B(0,r_{n})\times B(0,r_{n})\times \tilde{\mathcal{P}}_{n,W|Z}\} \subseteq \left[\frac{\det(V_{0})}{2}, \frac{3\det(V_{0})}{2}\right]
\end{align*}
for $n$ large. This establishes that $\{V_{n}(\gamma,\bar{\gamma},p_{W|Z}): (\gamma,\bar{\gamma},p_{W|Z}) \in B(0,r_{n})\times B(0,r_{n})\times \tilde{\mathcal{P}}_{n,W|Z}\}$ forms a class of nonsingular matrices for $n$ large.
\paragraph{Step 3.} Step 1 implies $\sup_{(\gamma,\bar{\gamma},p_{W|Z}) \in B(\gamma_{0},r_{n})\times B(\gamma_{0},r_{n}) \times \tilde{\mathcal{P}}_{n,W|Z}}||V_{n}(\gamma,\bar{\gamma},p_{W|Z})-V_{0}||_{op}= o(1)$ and Assumptions \ref{as:covariance}.1 and \ref{as:localident}.2 imply that $||V_{0}^{-1}||_{op} < \infty$. Consequently, the submultiplicativity of operator norms implies $$\sup_{(\gamma,\bar{\gamma},p_{W|Z}) \in B(\gamma_{0},r_{n})\times B(\gamma_{0},r_{n}) \times \tilde{\mathcal{P}}_{n,W|Z}}||V_{0}^{-1}(V_{n}(\gamma,\bar{\gamma},p_{W|Z})-V_{0})||_{op} < 1/2$$ for $n$ large. Consequently, I obtain, for $n$ large,
\begin{align*}
    ||V_{n}^{-1}(\gamma,\bar{\gamma},p_{W|Z})||_{op} &= ||(I+V_{0}^{-1}(V_{n}(\gamma,\bar{\gamma},p_{W|Z})-V_{0}))^{-1}V_{0}^{-1}||_{op} \\
    &\leq ||(I+V_{0}^{-1}(V_{n}(\gamma,\bar{\gamma},p_{W|Z})-V_{0}))^{-1}||_{op} \cdot || V_{0}^{-1}||_{op} \\
    &= \left|\left|\sum_{k=0}^{\infty} (-1)^{k}(V_{0}^{-1}(V_{n}(\gamma,\bar{\gamma},p_{W|Z})-V_{0}))^{k}\right| \right|_{op} \cdot ||V_{0}^{-1}||_{op} \\
    &\leq ||V_{0}^{-1}||_{op} \sum_{k=0}^{\infty}||V_{0}^{-1}(V_{n}(\gamma,\bar{\gamma},p_{W|Z})-V_{0})||_{op}^{k} \\
    &\leq ||V_{0}^{-1}||_{op}\sum_{k=0}^{\infty}\left(\frac{1}{2}\right)^{k} \\
    &\leq 2 ||V_{0}^{-1}||_{op},
\end{align*}
where the first equality identity $V_{n}^{-1}(\gamma,\bar{\gamma},p_{W|Z})=(I+V_{0}^{-1}(V_{n}(\gamma,\bar{\gamma},p_{W|Z})-V_{0}))^{-1}V_{0}^{-1}$, the first inequality uses submultiplicativity of the operator norm, the second equality is the Neumann series, the second inequality is the triangle inequality, the third inequality is the uniform bound on $||V_{0}^{-1}(V_{n}(\gamma,\bar{\gamma},p_{W|Z})-V_{0})||_{op}$, and the final inequality uses the formula for geometric series. Consequently, $$\sup_{(\gamma,\bar{\gamma},p_{W|Z}) \in B(\gamma_{0},r_{n})\times B(\gamma_{0},r_{n}) \times \tilde{\mathcal{P}}_{n,W|Z}}||V_{n}^{-1}(\gamma,\bar{\gamma},p_{W|Z})||_{op} < \infty$$ for $n$ large.
\end{proof}
\begin{remark}
If $r_{n} = o(n^{-1/4})$, then an inspection of the proof reveals that RHS of (\ref{eq:c1}), (\ref{eq:c2}), and (\ref{eq:c3}) can be strengthened to $o(n^{-1/4})$ (in particular, (\ref{eq:remark21}), (\ref{eq:remark22}), and (\ref{eq:remark23}) become $o(n^{-1/4})$). Since Lemma 7 establishes that $\sup_{p_{W|Z} \in \tilde{\mathcal{P}}_{n,W|Z}}||\gamma_{n} -\gamma_{0}||_{2} = o(n^{-1/4})$, one can assume that $\gamma_{n},\gamma_{n}^{\dagger} \in B(\gamma_{0},r_{n})$ with $r_{n} = o(n^{-1/4})$ for large $n$. Consequently, the three claims in Step 4 of the proof of Theorem 1 follow.
\end{remark}
\begin{lemma}\label{lem:contraction}
Suppose that Assumptions \ref{as:uemfp}, \ref{as:dgp2}, \ref{as:correctspec}, \ref{as:moments}, \ref{as:covariance}, \ref{as:criterion}, \ref{as:localident}, and \ref{as:concentration}.1--\ref{as:concentration}.5 hold. Then, for $P_{0,Z}^{\infty}$-almost every fixed realization $\{z_{i}\}_{i \geq 1}$ of $\{Z_{i}\}_{i \geq 1}$, the following holds:
\begin{align*}
\limsup_{n\rightarrow \infty}\sup_{p_{W|Z} \in \tilde{\mathcal{P}}_{n,W|Z}}||n^{\frac{1}{4}}(\gamma_{n}-\gamma_{0})||_{2} = 0,    
\end{align*}
where $\{\tilde{\mathcal{P}}_{n,W|Z}\}_{n \geq 1}$ are the sets defined in Assumption \ref{as:concentration}.
\end{lemma}
\begin{proof}
 The proof has three steps. Step 1 verifies the conclusion of the lemma after asserting two intermediate technical results. Step 2 and Step 3 are dedicated to the verification these two results.
\paragraph{Step 1.} Let $\{z_{i}\}_{i \geq 1}$ be an arbitrary fixed realization of $\{Z_{i}\}_{i\geq 1}$. First note that $\gamma_{n}$ satisfies first order conditions
\begin{align*}
    \frac{1}{n}\sum_{i=1}^{n}M(z_{i},\gamma_{n})'\Sigma^{-1}(z_{i},\gamma_{n})m(z_{i},\gamma_{n}) = 0.
\end{align*}
Since $\gamma_{n} \in \Gamma_{0}$ for large $n$, I expand $m(z_{i},\gamma_{n})$ around $\gamma_{0}$ (using element by element mean value expansions) to obtain
\begin{align*}
m(z_{i},\gamma_{n}) = m(z_{i},\gamma_{0}) + M(z_{i},\gamma_{n}^{\dagger}) (\gamma_{n}-\gamma_{0}),    
\end{align*}
for all $i=1,...,n$, where $\{\gamma_{n}^{\dagger}\}_{n \geq 1}$ is a sequence such that $\gamma_{n}^{\dagger}$ on the line segment joining $\gamma_{n}$ and $\gamma_{0}$  (and takes different values in each row of $M(z_{i},\gamma_{n}^{\dagger})$). Substituting this into the first order conditions, I conclude that
\begin{align*}
    \frac{1}{n}\sum_{i=1}^{n}M(z_{i},\gamma_{n})'\Sigma^{-1}(z_{i},\gamma_{n})m(z_{i},\gamma_{0}) + \frac{1}{n}\sum_{i=1}^{n}M(z_{i},\gamma_{n})'\Sigma^{-1}(z_{i},\gamma_{n})M(z_{i},\gamma_{n}^{\dagger})(\gamma_{n}-\gamma_{0})=0
\end{align*}
for $n$ large. Consequently, for $n$ large, an application of Lemma 6 (valid because Lemma 4 implies $\gamma_{n},\gamma_{n}^{\dagger} \in B(\gamma_{0},r_{n})$ for $n$ large, where $r_{n} \rightarrow 0$ arbitrarily slowly as $n\rightarrow \infty$) reveals that
\begin{align*}
\gamma_{n}-\gamma_{0} = -V_{n}^{-1}(\gamma_{n},\gamma_{n}^{\dagger},p_{W|Z}) \frac{1}{n}\sum_{i=1}^{n}M(z_{i},\gamma_{n})'\Sigma^{-1}(z_{i},\gamma_{n})m(z_{i},\gamma_{0}),   
\end{align*}
which implies that
\begin{align*}
    \sup_{p_{W|Z} \in \tilde{\mathcal{P}}_{n,W|Z}}||n^{\frac{1}{4}}(\gamma_{n}-\gamma_{0})||_{2} &\leq \sup_{p_{W|Z} \in \tilde{\mathcal{P}}_{n,W|Z}} \left| \left | V_{n}^{-1}(\gamma_{n},\gamma_{n}^{\dagger},p_{W|Z})\right | \right |_{op} \\
    &\quad \times \sup_{p_{W|Z} \in \tilde{\mathcal{P}}_{n,W|Z}}n^{\frac{1}{4}}\left| \left| \frac{1}{n}\sum_{i=1}^{n}M(z_{i},\gamma_{n})'\Sigma^{-1}(z_{i},\gamma_{n})m(z_{i},\gamma_{0})\right| \right|_{2}
\end{align*}
using $||Ax||_{2} \leq ||A||_{op}\cdot ||x||_{2}$ for conformable matrix $A$ and vector $x$. It thus suffices show that
\begin{align}\label{eq:gramnorm}
\sup_{p_{W|Z} \in \tilde{\mathcal{P}}_{n,W|Z}} \left| \left | V_{n}^{-1}(\gamma_{n},\gamma_{n}^{\dagger},p_{W|Z})\right | \right |_{op} = O(1)
\end{align}
and
\begin{align}\label{eq:optimalIVplug}
 \sup_{p_{W|Z} \in \tilde{\mathcal{P}}_{n,W|Z}} n^{\frac{1}{4}}\left | \left | \frac{1}{n}\sum_{i=1}^{n}M(z_{i},\gamma_{n})'\Sigma^{-1}(z_{i},\gamma_{n})m(z_{i},\gamma_{0})\right| \right|_{2} = o(1).
\end{align}
Indeed, these results complete the proof because they imply that
\begin{align*}
    &\sup_{p_{W|Z} \in \tilde{\mathcal{P}}_{n,W|Z}}||n^{\frac{1}{4}}(\gamma_{n}-\gamma_{0})||_{2} = o(1)
\end{align*}
and, because $\{z_{i}\}_{i \geq 1}$ was fixed arbitrarily, it holds for $P_{0,Z}^{\infty}$-almost every fixed realization $\{z_{i}\}_{i \geq 1}$.
\paragraph{Step 2.} I show (\ref{eq:gramnorm}). Since $\gamma_{n},\gamma_{n}^{\dagger} \in B(\gamma_{0},r_{n})$ for $n$ large, where $r_{n} \rightarrow 0$ arbitrarily slowly as $n\rightarrow \infty$, a direct application of Lemma 6 establishes (\ref{eq:gramnorm}).
\paragraph{Step 3.} Applying the triangle inequality and $||Ax||_{2} \leq ||A||_{op} ||x||_{2}$ for conformable matrix $A$ and vector $x$,
\begin{align*}
    &\sup_{p_{W|Z} \in \tilde{\mathcal{P}}_{n,W|Z}}n^{\frac{1}{4}}\left | \left | \frac{1}{n}\sum_{i=1}^{n}M(z_{i},\gamma_{n})'\Sigma^{-1}(z_{i},\gamma_{n})m(z_{i},\gamma_{0}) \right | \right |_{2} \\
    &\quad \leq \sup_{p_{W|Z} \in \tilde{\mathcal{P}}_{n,W|Z}}n^{\frac{1}{4}}\frac{1}{n}\sum_{i=1}^{n}||M(z_{i},\gamma_{n})||_{op}\cdot ||\Sigma^{-1}(z_{i},\gamma_{n})||_{op}\cdot ||m(z_{i},\gamma_{0})||_{2}
    \end{align*}
for all $n \geq 1$. Next, an application of H\"{o}lder's inequality and the submultiplicativity of the supremum yields that
\begin{align*}
        &\sup_{p_{W|Z} \in \tilde{\mathcal{P}}_{n,W|Z}}n^{\frac{1}{4}}\frac{1}{n}\sum_{i=1}^{n}||M(z_{i},\gamma_{n})||_{op}\cdot ||\Sigma^{-1}(z_{i},\gamma_{n})||_{op}\cdot ||m(z_{i},\gamma_{0})||_{2}\\
        &\quad  \leq \sup_{p_{W|Z} \in \tilde{\mathcal{P}}_{n,W|Z}}\max_{1 \leq i \leq n}||\Sigma^{-1}(z_{i},\gamma_{n})||_{op}\sup_{p_{W|Z} \in \tilde{\mathcal{P}}_{n,W|Z}}\frac{1}{n}\sum_{i=1}^{n}||M(z_{i},\gamma_{n})||_{op}\cdot ||m(z_{i},\gamma_{0})||_{2} \\
        & \quad \leq \sup_{p_{W|Z} \in \tilde{\mathcal{P}}_{n,W|Z}}\max_{1 \leq i \leq n}||\Sigma^{-1}(z_{i},\gamma_{n})||_{op}\sup_{p_{W|Z} \in \tilde{\mathcal{P}}_{n,W|Z}}\sqrt{\frac{1}{n}\sum_{i=1}^{n}||M(z_{i},\gamma_{n})||_{op}^{2}}\\
        &\quad \quad \times \sup_{p_{W|Z} \in \tilde{\mathcal{P}}_{n,W|Z}}n^{\frac{1}{4}}\sqrt{\frac{1}{n}\sum_{i=1}^{n}||m(z_{i},\gamma_{0})||_{2}^{2}}.
\end{align*}
for all $n \geq 1$. Lemma 1 states that
\begin{align*}
\sup_{(\gamma,p_{W|Z}) \in \Gamma \times \tilde{\mathcal{P}}_{n,W|Z}}\max_{1 \leq i \leq n}||\Sigma^{-1}(z_{i},\gamma)||_{op} =O(1).  
\end{align*}
Next, since $\{\gamma_{n}(p_{W|Z}): p_{W|Z} \in \tilde{\mathcal{P}}_{n,W|Z}\} \subseteq \Gamma_{0}$ for large $n$ by Lemma 4, Assumption \ref{as:concentration}.3(b) and the triangle inequality implies that
\begin{align*}
\sup_{p_{W|Z} \in \tilde{\mathcal{P}}_{n,W|Z}}\sqrt{\frac{1}{n}\sum_{i=1}^{n}||M(z_{i},\gamma_{n})||_{op}^{2}} \leq     o(1)+ \sup_{\gamma \in \Gamma_{0}}\sqrt{\frac{1}{n}\sum_{i=1}^{n}||M_{0}(z_{i},\gamma)||_{op}^{2}}.
\end{align*}
Moreover, Assumption \ref{as:moments}.2 implies that
\begin{align*}
\sup_{\gamma \in \Gamma_{0}}\sqrt{\frac{1}{n}\sum_{i=1}^{n}||M_{0}(z_{i},\gamma)||_{op}^{2}} = O(1),    
\end{align*}
and, as a result, I can conclude that
\begin{align*}
   \sup_{p_{W|Z} \in \tilde{\mathcal{P}}_{n,W|Z}} \sqrt{\frac{1}{n}\sum_{i=1}^{n}||M(z_{i},\gamma_{n})||_{op}^{2}} = O(1)
\end{align*}
Finally, Assumption \ref{as:correctspec} and Assumption \ref{as:concentration}.3(a) yields that
\begin{align*}
    \sup_{p_{W|Z} \in \tilde{\mathcal{P}}_{n,W|Z}}n^{\frac{1}{4}}\sqrt{\frac{1}{n}\sum_{i=1}^{n}||m(z_{i},\gamma_{0})||_{2}^{2}} =o(1).
\end{align*}
Putting these results together, I conclude that
\begin{align*}
    \sup_{p_{W|Z} \in \tilde{\mathcal{P}}_{n,W|Z}}n^{\frac{1}{4}}\left | \left | \frac{1}{n}\sum_{i=1}^{n}M(z_{i},\gamma_{n})'\Sigma^{-1}(z_{i},\gamma_{n})m(z_{i},\gamma_{0}) \right | \right |_{2} = o(1).
\end{align*}
This completes the proof.
\end{proof}
\begin{lemma}\label{lem:likelihoodexpansion}
 Let $S_{n,i}(p_{W|Z}) = \tilde{\chi}_{0}(W_{i},z_{i})- E_{P_{W|Z}}[\tilde{\chi}_{0}(W,Z)|Z=z_{i}]$ for $i=1,...,n$. If Assumptions \ref{as:dgp2}, \ref{as:correctspec}, \ref{as:moments}, \ref{as:covariance}, \ref{as:criterion}, \ref{as:localident}, and \ref{as:concentration} hold, then, for $P_{0,Z}^{\infty}$-almost every fixed realization $\{z_{i}\}_{i \geq 1}$ of $\{Z_{i}\}_{i \geq 1}$,
\begin{align*}
    \sup_{p_{W|Z} \in \tilde{\mathcal{P}}_{n,W|Z}}\left|\log L_{n}(p_{W|Z})- \log L_{n}(p_{t,n,W|Z})- t'\frac{1}{\sqrt{n}}\sum_{i=1}^{n}S_{n,i}(p_{W|Z})-\frac{1}{2}t'V_{0}^{-1}t\right| \overset{P_{0,W|Z}^{(n)}}{\longrightarrow} 0,
\end{align*}
where $\{\tilde{\mathcal{P}}_{n,W|Z}\}_{n \geq 1}$ are the sets defined in Assumption \ref{as:concentration} and $p_{t,n,W|Z} \propto p_{W|Z}\exp(-t'\tilde{\chi}_{0}/\sqrt{n})$.
\end{lemma}
\begin{proof}
The proof has two steps. Step 1 establishes a uniform convergence result. Step 2 modifies the argument from Theorem 4.1 of \cite{castillo2015bernstein} to obtain the expansion.
\paragraph{Step 1.} I start by showing that 
\begin{align*}
\sup_{p_{W|Z} \in \tilde{\mathcal{P}}_{n,W|Z}}\max_{1 \leq i \leq n}||E_{P_{W|Z}}[\tilde{\chi}_{0}(W,Z)\tilde{\chi}_{0}(W,Z)'|Z=z_{i}]-E_{P_{0,W|Z}}[\tilde{\chi}_{0}(W,Z)\tilde{\chi}_{0}(W,Z)'|Z=z_{i}]||_{F} = o(1).
\end{align*}
Using the definition of $\tilde{\chi}_{0}(W,Z)$ and the submultiplicativity of the Frobenius norm, I know that
\begin{align*}
&\left| \left | E_{P_{W|Z}}[\tilde{\chi}_{0}(W,Z)\tilde{\chi}_{0}(W,Z)'|Z=z_{i}]- E_{P_{0,W|Z}}[\tilde{\chi}_{0}(W,Z)\tilde{\chi}_{0}(W,Z)'|Z=z_{i}]\right | \right|_{F}\\
&\quad = \left | \left | V_{0}^{-1}M_{0}(z_{i},\gamma_{0})'\Sigma_{0}^{-1}(z_{i},\gamma_{0})\left(\Sigma(z_{i},\gamma_{0})-\Sigma_{0}(z_{i},\gamma_{0})\right) \Sigma_{0}^{-1}(z_{i},\gamma_{0})M_{0}(z_{i},\gamma_{0})V_{0}^{-1}\right| \right |_{F} \\
&\quad \leq ||V_{0}^{-1}||_{F}^{2}||M_{0}(z_{i},\gamma_{0})||_{F}^{2}||\Sigma_{0}^{-1}(z_{i},\gamma_{0})||_{F}^{2}||\Sigma(z_{i},\gamma_{0})-\Sigma_{0}(z_{i},\gamma_{0})||_{F}
\end{align*}
for $i=1,...,n$. Since $V_{0}$ is real, symmetric, and positive definite by Assumptions \ref{as:covariance}.1 and \ref{as:localident}.2, I know that there exists a constant $\underline{c}>0$ such that
\begin{align*}
||V_{0}^{-1}||_{F}^{2} = \sum_{j=1}^{d_{\gamma}}\frac{1}{\lambda_{j}^{2}(V_{0})}   \leq  \frac{d_{\gamma}}{\lambda_{min}^{2}(V_{0})} \leq \frac{d_{\gamma}}{\underline{c}^{2}}
\end{align*}
Moreover, $\sup_{z \in \mathcal{Z}}||M_{0}(z,\gamma_{0})||_{F}^{2} < \infty$  and $\sup_{z \in \mathcal{Z}}||\Sigma_{0}^{-1}(z,\gamma_{0})||^{2}_{F}< \infty$ by Assumption \ref{as:moments}.2 and Assumption \ref{as:covariance}. Consequently, there is a constant $C> 0$ such that
\begin{align*}
&\sup_{p_{W|Z} \in \tilde{\mathcal{P}}_{n,W|Z}}\max_{1 \leq i \leq n}||E_{P_{W|Z}}[\tilde{\chi}_{0}(W,Z)\tilde{\chi}_{0}(W,Z)'|Z=z_{i}]-E_{P_{0,W|Z}}[\tilde{\chi}_{0}(W,Z)\tilde{\chi}_{0}(W,Z)'|Z=z_{i}]||_{F} \\
&\quad \leq C  \sup_{p_{W|Z} \in \tilde{\mathcal{P}}_{n,W|Z}}\max_{1 \leq i \leq n}||\Sigma(z_{i},\gamma_{0})-\Sigma_{0}(z_{i},\gamma_{0})||_{F}  \\
&\quad \leq C\sqrt{d_{g}} \sup_{p_{W|Z} \in \tilde{\mathcal{P}}_{n,W|Z}}||\Sigma(\cdot,\gamma_{0})-\Sigma_{0}(\cdot,\gamma_{0})||_{n,\infty},
\end{align*}
where the last inequality uses that $||\Sigma(z,\gamma)||_{F} \leq \sqrt{d_{g}}||\Sigma(z,\gamma)||_{op}$. Assumption \ref{as:concentration}.3(c) then implies
\begin{align*}
\max_{1 \leq i \leq n}||E_{P_{W|Z}}[\tilde{\chi}_{0}(W,Z)\tilde{\chi}_{0}(W,Z)'|Z=z_{i}]-E_{P_{0,W|Z}}[\tilde{\chi}_{0}(W,Z)\tilde{\chi}_{0}(W,Z)'|Z=z_{i}]||_{F} = o(1).
\end{align*}
for all $p_{W|Z} \in \tilde{\mathcal{P}}_{n,W|Z}$.
\paragraph{Step 2.} Using the definition of $p_{t,n,W|Z}$ and Markov's inequality, a sufficient condition for
\begin{align*}
 &P_{0,W|Z}^{(n)}\left(\sup_{p_{W|Z} \in \tilde{\mathcal{P}}_{n,W|Z}}\left|\log\left(\frac{L_{n}(p_{W|Z})}{L_{n}(p_{t,n,W|Z})}\right)- t'\frac{1}{\sqrt{n}}\sum_{i=1}^{n}S_{n,i}(p_{W|Z})-\frac{1}{2}t'V_{0}^{-1}t \right|>\xi \right) = o(1) 
 \end{align*}
 for any $\xi>0$ is
\begin{align*}
    \left|\sum_{i=1}^{n}\log E_{P_{W|Z}}\left[\exp\left(-\frac{t'\tilde{\chi}_{0}(W,Z)}{\sqrt{n}}\right)\middle | Z=z_{i}\right] + t'\frac{1}{\sqrt{n}}\sum_{i=1}^{n}E_{P_{W|Z}}\left[\tilde{\chi}_{0}(W,Z)\middle| Z=z_{i}\right] - \frac{1}{2}t'V_{0}^{-1}t \right| = o(1)
\end{align*}
uniformly over $\tilde{\mathcal{P}}_{n,W|Z}$. Using the series representation of the exponential function,
\begin{align*}
    &\sum_{i=1}^{n}\log E_{P_{W|Z}}\left[\exp\left(-\frac{t'\tilde{\chi}_{0}(W,Z)}{\sqrt{n}}\right)\middle |Z=z_{i}\right] \\
   &= \sum_{i=1}^{n}\log \Bigg(
1 - t'\frac{1}{\sqrt{n}}E_{P_{W|Z}}[\tilde{\chi}_{0}(W,Z)\mid Z=z_{i}] \\
&\quad \quad  \quad  \quad \quad \quad + \frac{1}{2}t'\frac{1}{n}E_{P_{W|Z}}[\tilde{\chi}_{0}(W,Z)\tilde{\chi}_{0}(W,Z)'\mid Z=z_{i}]t 
+ r_{n,i}(t,p_{W|Z})
\Bigg)
\end{align*}
for all $n \geq 1$, where
\begin{align*}
    r_{n,i}(t,p_{W|Z}) = E_{P_{W|Z}}\left[\sum_{l=3}^{\infty}\frac{(-1)^{l}(t'\tilde{\chi}_{0}(W,Z))^{l}}{l!n^{\frac{l}{2}}}\middle|Z=z_{i}\right] 
\end{align*}
for all $i=1,...,n$. Applying the triangle inequality, I obtain
\begin{align*}
\left|r_{n,i}(t,p_{W|Z})\right| \leq \frac{1}{n^{\frac{3}{2}}} E_{P_{W|Z}}\left[\sum_{l=0}^{\infty}\frac{|t'\tilde{\chi}_{0}(W,Z)|^{l}}{l!}\middle | Z=z_{i} \right] = \frac{1}{n^{\frac{3}{2}}}E_{P_{W|Z}}\left[\exp(|t'\tilde{\chi}_{0}(W,Z)|)\middle | Z=z_{i}\right]
\end{align*}
for $i=1,...,n$ with the equality using the series representation of the exponential. Consequently,
\begin{align*}
   \sup_{p_{W|Z}\in \tilde{\mathcal{P}}_{n,W|Z}} \max_{1 \leq i \leq n }\left|r_{n,i}(t,p_{W|Z})\right| \leq \frac{1}{n^{\frac{3}{2}}}\sup_{p_{W|Z}\in \tilde{\mathcal{P}}_{n,W|Z}} \max_{1 \leq i \leq n }E_{P_{W|Z}}\left[\exp(|t'\tilde{\chi}_{0}(W,Z)|)\middle | Z=z_{i}\right].
\end{align*} 
As such, Assumption \ref{as:concentration}.3(f) implies that $\max_{1 \leq i \leq n}\left|r_{n,i}(t,p_{W|Z}) \right| = O(n^{-\frac{3}{2}})$ uniformly over $\tilde{\mathcal{P}}_{n,W|Z}$. Combining this result with that from Step 1, it follows that
\begin{align*}
&\sum_{i=1}^{n}\log E_{P_{W|Z}}\left[\exp\left(-\frac{t'\tilde{\chi}_{0}(W,Z)}{\sqrt{n}}\right)\middle | Z=z_{i}\right]  \\
&= \sum_{i=1}^{n}\log \left(1 - t'\frac{1}{\sqrt{n}}E_{P_{W|Z}}\left[\tilde{\chi}_{0}(W,Z)\middle| Z=z_{i}\right]+ \frac{1}{2}t'\frac{1}{n}E_{P_{0,W|Z}}\left[\tilde{\chi}_{0}(W,Z)\tilde{\chi}_{0}(W,Z)'\middle | Z=z_{i}\right]t \right. \\ &\quad \quad \quad \quad\quad \quad\quad \quad\left.+O\left(\frac{1}{n^{\frac{3}{2}}}\right) + o\left(\frac{1}{n}\right)\right)
\end{align*}
uniformly over $ \tilde{\mathcal{P}}_{n,W|Z}$ as $n\rightarrow \infty$. A Taylor expansion of $x \mapsto \log(1+x)$ around $x= 0$ then implies that
\begin{align*}
&\sum_{i=1}^{n}\log E_{P_{W|Z}}\left[\exp\left(-\frac{t'\tilde{\chi}_{0}(W,Z)}{\sqrt{n}}\right)\middle | Z=z_{i}\right] \\
&\quad =-t'\frac{1}{\sqrt{n}}\sum_{i=1}^{n}E_{P_{W|Z}}\left[\tilde{\chi}_{0}(W,Z)\middle | Z=z_{i}\right]+ \frac{1}{2}t'\frac{1}{n}\sum_{i=1}^{n}E_{P_{0,W|Z}}\left[\tilde{\chi}_{0}(W,Z)\tilde{\chi}_{0}(W,Z)'\middle | Z=z_{i}\right]t + o(1) \\
&\quad = -t'\frac{1}{\sqrt{n}}\sum_{i=1}^{n}E_{P_{W|Z}}\left[\tilde{\chi}_{0}(W,Z)\middle | Z=z_{i}\right]\\
&\quad \quad +\frac{1}{2}t'V_{0}^{-1}\left(\frac{1}{n}\sum_{i=1}^{n}M_{0}(z_{i},\gamma_{0})'\Sigma_{0}^{-1}(z_{i},\gamma_{0})M_{0}(z_{i},\gamma_{0})\right)V_{0}^{-1}t + o(1) \\
&\quad = -t'\frac{1}{\sqrt{n}}\sum_{i=1}^{n}E_{P_{W|Z}}\left[\tilde{\chi}_{0}(W,Z)\middle | Z=z_{i}\right]+\frac{1}{2}t'V_{0}^{-1}t + o(1)
\end{align*}
uniformly over $\tilde{\mathcal{P}}_{n,W|Z}$. The second equality uses the definition of $\tilde{\chi}_{0}$ and the linearity of expectation while the third equality uses the strong law of law numbers to obtain
\begin{align*}
\frac{1}{n}\sum_{i=1}^{n}M_{0}(z_{i},\gamma_{0})'\Sigma_{0}^{-1}(z_{i},\gamma_{0})M_{0}(z_{i},\gamma_{0}) = V_{0}+ o(1)   
\end{align*}
as $n\rightarrow \infty$. Putting this all together, I conclude that
\begin{align*}
 &\sum_{i=1}^{n}\log E_{P_{W|Z}}\left[\exp\left(-\frac{t'\tilde{\chi}_{0}(W,Z)}{\sqrt{n}}\right)\middle | Z=z_{i}\right] \\
 &\quad =   -t'\frac{1}{\sqrt{n}}\sum_{i=1}^{n}E_{P_{W|Z}}\left[\tilde{\chi}_{0}(W,Z)\middle | Z=z_{i}\right]+\frac{1}{2}t'V_{0}^{-1}t + o(1) 
\end{align*}
uniformly over $\{\tilde{\mathcal{P}}_{n,W|Z}\}_{n \geq 1}$ as $n\rightarrow \infty$. Consequently, 
\begin{align*}
    \left|\sum_{i=1}^{n}\log E_{P_{W|Z}}\left[\exp\left(-\frac{t'\tilde{\chi}_{0}(W,Z)}{\sqrt{n}}\right)\middle | Z=z_{i}\right] + t'\frac{1}{\sqrt{n}}\sum_{i=1}^{n}E_{P_{W|Z}}\left[\tilde{\chi}_{0}(W,Z)\middle| Z=z_{i}\right] - \frac{1}{2}t'V_{0}^{-1}t \right| = o(1)
\end{align*}
uniformly over $\{\tilde{\mathcal{P}}_{n,W|Z}\}_{n \geq 1}$ as $n\rightarrow \infty$. This completes the proof.
\end{proof}
\subsection{Proofs of Lemmas Related to Propositions 1--3 and Mat\'{e}rn Example}
\begin{proof}[Proof of Lemma \ref{lem:KLevent}.]
Let $C>0$ be an arbitrary constant and let $R_{n}(\theta,B) = \log (L_{n}(p_{\theta,B})/L_{n}(p_{0,W|Z}))$. Applying Jensen's inequality, I know that
\begin{align*}
\log \int \frac{L_{n}(p_{\theta,B})}{L_{n}(p_{0,W|Z})}d \bar{\Pi}_{n}(\theta,B)  &\geq  \int R_{n}(\theta,B) d\bar{\Pi}_{n}(\theta,B) \\
&=\int \left\{R_{n}(\theta,B)-E_{P_{0,W|Z}^{(n)}}\left[R_{n}(\theta,B)\right]\right\} d\bar{\Pi}_{n}(\theta,B) \\
&\quad + \int E_{P_{0,W|Z}^{(n)}}\left[R_{n}(\theta,B)\right] d\bar{\Pi}_{n}(\theta,B)
\end{align*}
Since 
\begin{align*}
E_{P_{0,W|Z}^{(n)}}\left[R_{n}(\theta,B)\right] = -\sum_{i=1}^{n}KL(p_{0,W|Z}(\cdot|z_{i}),p_{\theta,B}(\cdot|z_{i})) \geq -n \delta^{2}
\end{align*}
for $\bar{\Pi}_{n}$-almost every $(\theta,B)$, it follows that
\begin{align*}
 \int E_{P_{0,W|Z}^{(n)}}\left[R_{n}(\theta,B)\right] d\bar{\Pi}_{n}(\theta,B) \geq -n \delta^{2}.   
\end{align*}
As such, I know that
\begin{align*}
&P_{0,W|Z}^{(n)}\left(\int \frac{L_{n}(p_{\theta,B})}{L_{n}(p_{0,W|Z})}d \bar{\Pi}_{n}(\theta,B) \leq \exp\left(-(1+C)n \delta^{2}\right) \right) \\
&\quad \leq P_{0,W|Z}^{(n)}\left(\int \left\{R_{n}(\theta,B)-E_{P_{0,W|Z}^{(n)}}\left[R_{n}(\theta,B)\right]\right\} d\bar{\Pi}_{n}(\theta,B) \leq -Cn\delta^{2}\right) \\
&\quad \leq P_{0,W|Z}^{(n)}\left(\left|\int \left\{R_{n}(\theta,B)-E_{P_{0,W|Z}^{(n)}}\left[R_{n}(\theta,B)\right]\right\} d\bar{\Pi}_{n}(\theta,B)\right| \geq Cn\delta^{2}\right)
\end{align*}
Applying Chebyshev's inequality, I know that
\begin{align*}
    &P_{0,W|Z}^{(n)}\left(\left|\int \left\{R_{n}(\theta,B)-E_{P_{0,W|Z}^{(n)}}\left[R_{n}(\theta,B)\right]\right\} d\bar{\Pi}_{n}(\theta,B) \right| \geq Cn \delta^{2}\right) \\
    &\quad \leq \frac{1}{C^{2}n^{2}\delta^{4}}E_{P_{0,W|Z}^{(n)}}\left[\left|\int \left\{R_{n}(\theta,B)-E_{P_{0,W|Z}^{(n)}}\left[R_{n}(\theta,B)\right]\right\} d\bar{\Pi}_{n}(\theta,B)\right|^{2} \right]
\end{align*}
Applying Jensen's inequality and then Fubini's theorem,
\begin{align*}
&E_{P_{0,W|Z}^{(n)}}\left[\left|\int \left\{R_{n}(\theta,B)-E_{P_{0,W|Z}^{(n)}}\left[R_{n}(\theta,B)\right]\right\} d\bar{\Pi}_{n}(\theta,B)\right|^{2} \right] \\
&\quad \leq \int Var_{P_{0,W|Z}^{(n)}}\left[R_{n}(\theta,B)\right] d \bar{\Pi}_{n}(\theta,B) \\
&\quad \leq n \delta^{2},
\end{align*}
where the last inequality uses that
\begin{align*}
Var_{P_{0,W|Z}^{(n)}}\left[R_{n}(\theta,B)\right] = \sum_{i=1}^{n}KLV\left(p_{0,W|Z}(\cdot|z_{i}),p_{\theta,B}(\cdot|z_{i})\right) \leq n \delta^{2}
\end{align*}
for $\bar{\Pi}_{n}$-almost every $(\theta,B)$. As such,
\begin{align*}
 P_{0,W|Z}^{(n)}\left(\left |\int \left\{R_{n}(\theta,B)-E_{P_{0,W|Z}^{(n)}}\left[R_{n}(\theta,B)\right]\right\} d\bar{\Pi}_{n}(\theta,B) \right|\geq Cn \delta^{2}\right) \leq \frac{1}{C^{2}n\delta^{2}}. 
\end{align*}
Since this argument holds for $C > 0$ fixed arbitrarily, it extends to any $C > 0$.
\end{proof}
\begin{proof}[Proof of Lemma \ref{lem:posteriorsieve}]
Let $\bar{\Pi}_{n}$ be the renormalized restriction of $\Pi$ to $\mathcal{U}_{n}(p_{0,W|Z},\delta_{n})$ and let $A_{n}$ be the event
\begin{align*}
\int \frac{L_{n}(p_{\theta,B})}{L_{n}(p_{0,W|Z})}d \bar{\Pi}_{n}(\theta,B) \geq \exp\left(-(1+C)n \delta^{2}_{n}\right).
\end{align*}
for some small constant $C>0$. By Lemma \ref{lem:KLevent}, $P_{0,W|Z}^{(n)}(A_{n})\rightarrow 1$ as $n\rightarrow \infty$, so it suffices to show
\begin{align*}
    E_{P_{0,W|Z}^{(n)}}\left[\Pi(\bar{\mathcal{P}}_{n,W|Z}^{c}|W^{(n)})\mathbf{1}\{A_{n}\}\right]\longrightarrow 0
\end{align*}
as $n\rightarrow \infty$. Using the definition of the posterior and $A_{n}$ and that $\Pi(\mathcal{U}_{n}(p_{0,W|Z},\delta_{n})) \geq \breve{C}\exp(-n\delta_{n}^{2})$, it follows that
\begin{align*}
     E_{P_{0,W|Z}^{(n)}}\left[\Pi(\bar{\mathcal{P}}_{n,W|Z}^{c}|W^{(n)})\mathbf{1}\{A_{n}\}\right] \leq \breve{C}\exp((2+C)n\delta_{n}^{2})E_{P_{0,W|Z}^{(n)}}\left[\int_{\bar{\mathcal{P}}_{n,W|Z}^{c}}\frac{L_{n}(p_{\theta,B})}{L_{n}(p_{0,W|Z})}d \Pi(\theta,B)\right]
\end{align*}
Applying Fubini's theorem, I know that
\begin{align*}
E_{P_{0,W|Z}^{(n)}}\left[\int_{\bar{\mathcal{P}}_{n,W|Z}^{c}}\frac{L_{n}(p_{\theta,B})}{L_{n}(p_{0,W|Z})}d \Pi(\theta,B)\right] = \int_{\bar{\mathcal{P}}_{n,W|Z}^{c}}E_{P_{0,W|Z}^{(n)}}\left(\frac{L_{n}(p_{\theta,B})}{L_{n}(p_{0,W|Z})}\right)d \Pi(\theta,B) = \Pi(\bar{\mathcal{P}}_{n,W|Z}^{c}).
\end{align*}
Consequently, the assumption $\Pi(\bar{\mathcal{P}}_{n,W|Z}^{c}) \leq \check{C} \exp(-\bar{C}n\delta_{n}^{2})$ implies that
\begin{align*}
E_{P_{0,W|Z}^{(n)}}\left[\Pi(\bar{\mathcal{P}}_{n,W|Z}^{c}|W^{(n)})\mathbf{1}\{A_{n}\}\right] \leq \breve{C}\check{C}\exp(((2+C)-\bar{C})n\delta_{n}^{2}) \rightarrow 0   
\end{align*}
as $n\rightarrow \infty$ because $n\delta_{n}^{2}\rightarrow \infty$ as $n\rightarrow \infty$ and $\bar{C} > 2+C$.
\end{proof}
\begin{proof}[Proof of Lemma \ref{lem:concentration}.]
The proof has two steps. Step 1 shows it is enough to study $\Pi(\bar{\mathcal{P}}_{n,W|Z}\cap \{h_{n,2}(p_{\theta,B},p_{0,W|Z}) \geq D \delta_{n}\}|W^{(n)})$ on an event $A_{n}$ satisfying $P_{0,W|Z}^{(n)}(A_{n}) \rightarrow 1$ as $n\rightarrow \infty$. Step 2 uses this and the entropy restriction (\ref{eq:entropy}) to conclude the lemma. 
\paragraph{Step 1.} Conditions (\ref{eq:neighbor}) and (\ref{eq:mass}) imply that it suffices to show that
\begin{align*}
    \Pi(\bar{\mathcal{P}}_{n,W|Z} \cap \{h_{n,2}(p_{\theta,B},p_{0,W|Z}) \geq D \delta_{n}\}|W^{(n)}) \overset{P_{0,W|Z}^{(n)}}{\longrightarrow} 0
\end{align*}
as $n\rightarrow \infty$ because Lemma \ref{lem:posteriorsieve} implies that if these two conditions hold then
\begin{align*}
\Pi(h_{n,2}(p_{\theta,B},p_{0,W|Z}) \geq D \delta_{n}|W^{(n)}) =\Pi(\bar{\mathcal{P}}_{n,W|Z} \cap \{h_{n,2}(p_{\theta,B},p_{0,W|Z}) \geq D \delta_{n}\}|W^{(n)}) + o_{P_{0,W|Z}^{(n)}}(1). 
\end{align*}
To this end, let $\bar{\Pi}_{n}$ be the renormalized restriction of $\Pi$ to $\mathcal{U}_{n}(p_{0,W|Z},\delta_{n})$ and let $A_{n}$ be the event
\begin{align*}
\int \frac{L_{n}(p_{\theta,B})}{L_{n}(p_{0,W|Z})}d \bar{\Pi}_{n}(\theta,B) \geq \exp\left(-(1+C)n \delta^{2}_{n}\right).
\end{align*}
for some small constant $C>0$. Lemma \ref{lem:KLevent} implies that $P_{0,W|Z}^{(n)}(A_{n}^{c}) \rightarrow 0$ as $n\rightarrow \infty$, and, as a result, 
\begin{align*}
&E_{P_{0,W|Z}^{(n)}}\left[\Pi(\{h_{n,2}(p_{\theta,B},p_{0,W|Z})\geq D\delta_{n}\} \cap \bar{\mathcal{P}}_{n,W|Z}|W^{(n)})\right] \\
&\quad \leq E_{P_{0,W|Z}^{(n)}}\left[\Pi(\{h_{n,2}(p_{\theta,B},p_{0,W|Z})\geq D\delta_{n}\} \cap \bar{\mathcal{P}}_{n,W|Z}|W^{(n)})\mathbf{1}\{A_{n}\}\right] + o(1),
\end{align*}
meaning that it suffices to show that 
\begin{align}\label{eq:suffRMSHcontract}
    E_{P_{0,W|Z}^{(n)}}\left[\Pi(\bar{\mathcal{P}}_{n,W|Z} \cap \{h_{n,2}(p_{\theta,B},p_{0,W|Z})\geq D\delta_{n}\}|W^{(n)})\mathbf{1}\{A_{n}\}\right] \longrightarrow 0
\end{align}
as $n\rightarrow \infty$ to complete the proof.
\paragraph{Step 2.} I show (\ref{eq:suffRMSHcontract}). Let $\{\phi_{n}\}_{n \geq 1}$ be a sequence of tests with $0 \leq \phi_{n}(W^{(n)}) \leq 1$ such that, for any $\delta> 0$ and any $p_{\theta,B}$ with $h_{n,2}(p_{\theta,B},p_{0,W|Z})>\delta$, the following holds
\begin{align}
    E_{P_{0,W|Z}^{(n)}}[\phi_{n}(W^{(n)})] \leq \exp\left(-\frac{1}{2}n\delta^{2}\right) \label{eq:typeIbasic}
\end{align}
and
\begin{align}
\sup_{(\tilde{\theta},\tilde{B}): h_{n,2}(p_{\tilde{\theta},\tilde{B}},p_{\theta,B}) \leq \delta/18}E_{P_{\theta,B}^{(n)}}[1-\phi_{n}(W^{(n)})] \leq \exp\left(-\frac{1}{2}n\delta^{2}\right), \label{eq:typeIIbasic}
\end{align}
where $P_{\theta,B}^{(n)}$ being the probability measure associated with the product density $\prod_{i=1}^{n}p_{\theta,B}(\cdot|z_{i})$. Lemma 2 of \cite{ghosal2007convergence} establishes the existence of such tests. If
\begin{align}\label{eq:bigentropy}
    \sup_{\delta > \delta_{n}}\log N\left(\frac{\delta}{36},\{p_{\theta,B} \in \bar{\mathcal{P}}_{n,W|Z}: h_{n,2}(p_{\theta,B},p_{0,W|Z}) < \delta\},h_{n}\right) \leq  \tilde{C}n\delta_{n}^{2}
\end{align}
for some constant $\tilde{C}>0$, then Lemma 9 of \cite{ghosal2007convergence} states that (\ref{eq:typeIbasic}) and (\ref{eq:typeIIbasic}) implies 
\begin{align}\label{eq:testnull}
 E_{P_{0,W|Z}^{(n)}}[\phi_{n}(W^{(n)})] \leq \exp(\tilde{C}n\delta_{n}^{2})\frac{\exp(-\frac{1}{2}D^{2}n\delta_{n}^{2})}{1-\exp(-\frac{1}{2}D^{2}n\delta_{n}^{2})}   
\end{align}
and
\begin{align} \label{eq:testalternative}
    \sup_{p_{\theta,B} \in \bar{\mathcal{P}}_{n,W|Z} \cap \{(\theta,B): h_{n,2}(p_{\theta,B},p_{0,W|Z}) \geq D\delta_{n}\}}E_{P_{\theta,B}^{(n)}}[1-\phi_{n}(W^{(n)})] \leq \exp\left(-\frac{1}{2}D^{2}n\delta_{n}^{2}\right) 
\end{align}
Since $\{p_{\theta,B} \in \bar{\mathcal{P}}_{n,W|Z}: h_{n,2}(p_{\theta,B},p_{0,W|Z}) < D \delta_{n}\} \subseteq \bar{\mathcal{P}}_{n,W|Z}$, the approximate monotonicity of covering numbers under set inclusion (see 4.2.10 in \cite{Vershynin_2018}) leads to the conclusion that $\log N(\delta_{n}/72, \bar{\mathcal{P}}_{n,W|Z},h_{n}) \leq \tilde{C}n\delta_{n}^{2}$ implies (\ref{eq:bigentropy}). Consequently, condition (\ref{eq:entropy}) guarantees (\ref{eq:testnull}) and (\ref{eq:testalternative}). Recognizing that $1 = 1 - \phi_{n}(W^{(n)}) + \phi_{n}(W^{(n)})$ and that probabilities are bounded by one, I can then conclude that
\begin{align*}
     &E_{P_{0,W|Z}^{(n)}}\left[\Pi(\bar{\mathcal{P}}_{n,W|Z} \cap \{h_{n,2}(p_{\theta,B},p_{0,W|Z})\geq D\delta_{n}\} |W^{(n)})\mathbf{1}\{A_{n}\}\right] \\
     &\quad \leq E_{P_{0,W|Z}^{(n)}}[\phi_{n}(W^{(n)})] \\
     &\quad \quad + \exp((2+C)n\delta_{n}^{2})\\
     &\quad \quad \quad \times E_{P_{0,W|Z}^{(n)}}\left[\int_{\bar{\mathcal{P}}_{n,W|Z} \cap \{(\theta,B): h_{n,2}(p_{\theta,B},p_{0,W|Z}) \geq D\delta_{n}\}}(1-\phi_{n}(W^{(n)}))\frac{L_{n}(p_{\theta,B})}{L_{n}(p_{0,W|Z})}d \Pi(\theta,B)  \right]. 
\end{align*}
Applying Fubini's theorem, it follows that
\begin{align*}
&E_{P_{0,W|Z}^{(n)}}\left[\int_{\bar{\mathcal{P}}_{n,W|Z} \cap \{(\theta,B): h_{n,2}(p_{\theta,B},p_{0,W|Z}) \geq D\delta_{n}\}}(1-\phi_{n}(W^{(n)}))\frac{L_{n}(p_{\theta,B})}{L_{n}(p_{0,W|Z})}d \Pi(\theta,B) \right] \\
&\quad \leq \sup_{p_{W|Z} \in  \bar{\mathcal{P}}_{n,W|Z} \cap \{(\theta,B): h_{n,2}(p_{\theta,B},p_{0,W|Z}) \geq D\delta_{n}\}}E_{P_{\theta,B}^{(n)}}\left[1-\phi_{n}(W^{(n)})\right]
\end{align*}
So by setting $D>0$ large, it follows from (\ref{eq:testnull}), (\ref{eq:testalternative}), and $n\delta_{n}^{2}\rightarrow \infty$ that
\begin{align*}
 E_{P_{0,W|Z}^{(n)}}\left[\Pi(\bar{\mathcal{P}}_{n,W|Z} \cap \{h_{n,2}(p_{\theta,B},p_{0,W|Z})\geq D\delta_{n}\} |W^{(n)})\mathbf{1}\{A_{n}\}\right] \longrightarrow 0   
\end{align*}
as $n\rightarrow \infty$.
\end{proof}
\begin{proof}[Proof of Lemma \ref{lem:graminverse}]
The proof has two steps. The first establishes that an initial concentration result using a matrix Bernstein-type inequality. Step 2 uses this result to establish the inequalities.
\paragraph{Step 1.} Let $X_{i} = \psi^{J_{n}}(Z_{i})\psi^{J_{n}}(Z_{i})' - E_{P_{0,Z}}[\psi^{J_{n}}(Z)\psi^{J_{n}}(Z)']$ for $i=1,...,n$. Observing that $\{X_{i}\}_{i =1}^{n}$ forms an i.i.d sequence, $n^{-1}\Psi_{J_{n}}'\Psi_{J_{n}}- E_{P_{0,Z}}[\psi^{J_{n}}(Z)\psi^{J_{n}}(Z)'] = n^{-1}\sum_{i=1}^{n}X_{i}$, $X_{i}$ is symmetric, and $E_{P_{0,Z}}[X_{1}] = 0$, I can apply Theorem 4.1 of \cite{CHEN2015447} to conclude that for any $\xi> 0$,
\begin{align*}
    P_{0,Z}^{\infty}\left(\left| \left|  \frac{1}{n}\Psi_{J_{n}}'\Psi_{J_{n}} - E_{P_{0,Z}}[\psi^{J_{n}}(Z)\psi^{J_{n}}(Z)'] \right | \right|_{op} \geq \xi \right) &\leq 2 d(J_{n}) \exp\left(-\frac{n^{2}\xi^{2}/2}{s_{n}^{2}+R_{n}n\xi/3}\right),
    \end{align*}
    where $s_{n}^{2} = n ||E_{P_{0,Z}}[X^2_{1}]||_{op}$, and $\max_{1 \leq i \leq n}||X_{i}||_{op} \leq R_{n}$. Using the properties of CDV wavelets and that $p_{0,Z}$ is bounded away from zero and infinity (by Assumption \ref{as:boundedinternal}), one can show $d(J_{n}) = O(2^{J_{n}d_{z}})$, $s_{n}^{2} = O(n2^{J_{n}d_{z}})$, and $R_{n} = O(2^{J_{n}d_{z}})$. Consequently, for every $\xi > 0$, there are constants $\tilde{c}_{1} >0$ and $\tilde{c}_{2} >0$ (with $\tilde{c}_{2}$ depending on $\xi$) such that
    \begin{align*}
        P_{0,Z}^{\infty}\left(\left| \left|  \frac{1}{n}\Psi_{J_{n}}'\Psi_{J_{n}} - E_{P_{0,Z}}[\psi^{J_{n}}(Z)\psi^{J_{n}}(Z)'] \right | \right|_{op} \geq \xi \right) \leq \tilde{c}_{1} 2^{J_{n}d_{z}}\exp\left(-\frac{\tilde{c}_{2}n}{2^{J_{n}d_{z}}}\right)=:a_{n}
    \end{align*}
    for $n$ large. Since $\sum_{n = 1}^{\infty}a_{n} < \infty$ for $2^{J_{n}} = (n/\log n)^{1/(2\alpha+d_{z})}$, it follows that, for $P_{0,Z}^{\infty}$-almost every fixed realization $\{z_{i}\}_{i \geq 1}$ of $\{Z_{i}\}_{i \geq 1}$,
\begin{align}\label{eq:sieveoperator}
    \left | \left | \frac{1}{n}\Psi_{J_{n}}'\Psi_{J_{n}} - E_{P_{0,Z}}[\psi^{J_{n}}(Z)\psi^{J_{n}}(Z)'] \right | \right|_{op} \longrightarrow 0
\end{align} 
as $n\rightarrow \infty$ by the Borel-Cantelli lemma.
\paragraph{Step 2.} By (\ref{eq:sieveoperator}), I know that, for $P_{0,Z}^{\infty}$-almost every fixed realization $\{z_{i}\}_{i \geq 1}$ of $\{Z_{i}\}_{i \geq 1}$,
\begin{align*}
\lambda_{min}\left(\frac{1}{n}\Psi_{J_{n}}'\Psi_{J_{n}}\right)  \geq \frac{1}{2}\lambda_{min}\left(E_{P_{0,Z}}[\psi^{J_{n}}(Z)\psi^{J_{n}}(Z)']\right)
\end{align*}
for $n$ large. Using that $\{\psi_{jk}\}_{j,k}$ is an orthonormal basis for $L^{2}([0,1]^{d_{z}})$, I know that for $v \neq 0$,
\begin{align*}
 v'E_{P_{0,Z}}[\psi^{J_{n}}(Z)\psi^{J_{n}}(Z)']v &\geq  \int (v'\psi^{J_{n}}(z))^{2}p_{0,Z}(z)dz \\
 &\geq \inf_{z \in \mathcal{Z}}p_{0,Z}(z) \int (v'\psi^{J_{n}}(z))^{2}dz \\
 &\geq \inf_{z \in \mathcal{Z}}p_{0,Z}(z)||v||^{2}_{2}\lambda_{min}(I)
\end{align*}
Consequently,
\begin{align*}
    \lambda_{min}\left(\frac{1}{n}\Psi_{J_{n}}'\Psi_{J_{n}}\right)  \geq \frac{1}{2} \inf_{z \in \mathcal{Z}}p_{0,Z}(z)
\end{align*}
for $n$ large. Set $c_{1} = \inf_{z \in \mathcal{Z}}p_{0,Z}(z)/2 >0$ to obtain the first claim (with positivity following the boundedness of $p_{0,Z}$ away from zero and infinity). As such, it remains to establish the second inequality. Observe that
\begin{align*}
    \tilde{\psi}_{jk}(z_{i}) = \psi_{jk}(z_{i})-\hat{\tau}_{n,jk}'\psi_{-jk}^{J_{n}}(z_{i}) = \hat{v}_{n,jk}'\psi^{J}(z_{i})
\end{align*}
for each $i=1,...,n$, where $\hat{v}_{n,jk}$ is a vector that comprises an appropriate ordering of the elements of $(1,-\hat{\tau}_{n,jk}')'$. Squaring and summing over $i$,
\begin{align*}
    \frac{1}{n}\sum_{i=1}^{n}\tilde{\psi}_{jk}^{2}(z_{i}) = \hat{v}_{n,jk}'\left(\frac{1}{n}\Psi_{J_{n}}'\Psi_{J_{n}}\right)\hat{v}_{n,jk} \geq ||\hat{v}_{n,jk}||_{2}^{2}\lambda_{min}\left(\frac{1}{n}\Psi_{J_{n}}'\Psi_{J_{n}}\right) \geq c_{1}||\hat{v}_{n,jk}||_{2}^{2} 
\end{align*}
for $n$ large. Using $||\hat{v}_{n,jk}||_{2}^{2} = ||\hat{\tau}_{n,jk}||_{2}^{2} + 1 \geq 1
$ for all $j,k$, I obtain
\begin{align*}
 \min_{0 \leq j \leq J_{n}}\min_{0 \leq k \leq 2^{jd_{z}}-1}\frac{1}{n}\sum_{i=1}^{n}\tilde{\psi}_{jk}(z_{i}) ^{2}  \geq c_{1}
\end{align*}
for $n$ large. Set $c_{2} = c_{1}$ to conclude the proof.
\end{proof}
\setcounter{lemma}{12}
\begin{lemma}
Suppose $\{w_{\theta}: \theta \in \Theta\} \subseteq C^{\beta}([0,1]^{d},\mathbb{R}) \cap S^{\beta}([0,1]^{d},\mathbb{R})$ for some $\beta > 0$, $\sup_{\theta \in \Theta}||w_{\theta}||_{\beta} < \infty$, and $\sup_{\theta \in \Theta}||w_{\theta}||_{2,2,\beta} < \infty$. Then, for $\delta \in (0,1)$, the following is true:
\begin{enumerate}
    \item If $\beta \leq \alpha$, then
\begin{align*}
    \sup_{\theta \in \Theta}\inf_{f \in \mathcal{H}:||f-w_{\theta}||_{\infty}<\delta }\frac{1}{2}||f||_{\mathcal{H}}^{2} \lesssim \delta^{-(2\alpha + d-2 \beta)/\beta}
\end{align*}
\item If $\alpha < \beta < \alpha + d/2$, then
\begin{align*}
    \sup_{\theta \in \Theta}\inf_{f \in \mathcal{H}:||f-w_{\theta}||_{\infty}<\delta }\frac{1}{2}||f||_{\mathcal{H}}^{2} \lesssim \delta^{-\frac{d}{\beta}}
\end{align*}
\item If $\beta \geq \alpha + d/2$, then $\{w_{\theta}: \theta \in \Theta\} \subseteq \mathcal{H}$ and
\begin{align*}
 \sup_{\theta \in \Theta}\inf_{f\in \mathcal{H}:||f-w_{\theta}||_{\infty}<\delta }\frac{1}{2}||f||_{\mathcal{H}}^{2} < \infty
\end{align*}
\end{enumerate}
where $(\mathcal{H}, \langle \cdot, \cdot \rangle_{\mathcal{H}})$ is the RKHS of the Mat\'{e}rn process.

\end{lemma}
\begin{proof}
The proof follows a similar argument of Lemma 4 in \cite{van2011information} with appropriate modifications that account for uniformity over $\Theta$.
\paragraph{Step 1.} I start by constructing a higher order kernel $K$. It is the same as that from Lemma 4 of \cite{van2011information}, and, for the sake of clarity, I outline its construction. Let $k: \mathbb{R} \rightarrow \mathbb{R}$ be a function with a real, symmetric Fourier transform $\hat{k}$, which equals $1/2\pi$ in a neighborhood of zero and has compact support. Using $\hat{k}(\lambda)= (2\pi)^{-1}\int \exp(i\lambda u)k(u)du$, one can show $\int k(t) dt = 1$ and $\int(it)^{l}k(t)dt = 0$ for all integers $l \geq 1$, where, in this lemma $i$, is the imaginary unit. Defining $K(t) = k(t_{1})\cdots k(t_{d})$ for $t = (t_{1},...,t_{d})'$, it follows that $K$ integrates to one and has finite and vanishing moments of all orders greater than zero, thereby defining a higher order kernel.
\paragraph{Step 2.} For $h>0$, let $K_{h}(t) = h^{-d}K(t/h)$ for each $t=(t_{1},...,t_{d})'$ (with $t/h = (t_{1}/h,...,t_{d}/h)'$), and $f_{\theta,h} = K_{h} * w_{\theta}$ (i.e., $f_{\theta,h}(t) = \int K_{h}(t-\tilde{t})w_{\theta}(\tilde{t})d\tilde{t}$). I show that $\sup_{\theta \in \Theta}||w_{\theta}-f_{\theta,h}||_{\infty} \leq \tilde{C}(\beta,K,\Theta) h^{\beta}$, where $\tilde{C}(\beta,K,\Theta)$ is a constant that depends on $\beta$, $K$, and $\Theta$. Since $K_{h}$ is a higher order kernel and $w_{\theta} \in C^{\beta}([0,1]^{d},\mathbb{R})$, I can apply Proposition 4.3.33(ii) of \cite{gine2016mathematical} to conclude $||w_{\theta}-f_{\theta,h}||_{\infty} \leq C(\beta,K)||D^{\lfloor \beta \rfloor}w_{\theta}||_{\infty} h^{\beta}$, where $C(\beta,K)$ is a constant that depends on $\beta$ and $K$ only. Since $||D^{\lfloor \beta \rfloor} w_{\theta}||_{\infty} \leq ||w_{\theta}||_{\beta}$ by the definition of $||\cdot||_{\beta}$ and $\sup_{\theta \in \Theta} ||w_{\theta}||_{\beta} < \infty$ by assumption, set $\tilde{C}(\beta,K,\Theta) = C(\beta,K)\sup_{\theta \in \Theta}||w_{\theta}||_{\beta}$.
\paragraph{Step 3.} I now use the result from Step 2 to conclude the lemma with three substeps that conform to the parts of the lemma.
\paragraph{Step 3A.} I start with $\beta \leq \alpha$. I can write $f_{\theta,h}(t) = \int \exp(i\lambda't)\hat{f}_{\theta,h}(\lambda)(1+||\lambda||_{2}^{2})^{\alpha+d/2}d \mu(\lambda)$, where $\hat{f}_{\theta,h}(\lambda) = \hat{K}(h\lambda)\hat{w}_{\theta}(\lambda)$ is the Fourier transform of $f_{\theta,h}$, $\hat{K}$ is the Fourier transform of $K$, $\hat{w}_{\theta}$ is the Fourier transform of $w_{\theta}$, and $d\mu(\lambda) = (1+||\lambda||_{2}^{2})^{-\alpha-d/2}d\lambda$. Consequently, using the definition of $||\cdot||_{\mathcal{H}}$ for the Mat\'{e}rn process (cf. (12) in \cite{van2011information}), I obtain
\begin{align*}
    ||f_{\theta,h}||_{\mathcal{H}}^{2} \leq \int |\hat{K}(h\lambda)\hat{w}_{\theta}(\lambda)|^{2}(1+||\lambda||_{2}^{2})^{\alpha+ \frac{d}{2}}d\lambda
\end{align*}
By dividing and multiplying by $(1+||\lambda||_{2}^{2})^{\beta}$ inside the integral, applying H\"{o}lder's inequality, and using the definition of $||\cdot||_{2,2,\beta}$, I obtain that
\begin{align*}
\int |\hat{K}(h\lambda)\hat{w}_{\theta}(\lambda)|^{2}(1+||\lambda||_{2}^{2})^{\alpha+ \frac{d}{2}}d\lambda    &=\int |\hat{K}(h\lambda)|^{2}(1+||\lambda||_{2}^{2})^{\alpha+ \frac{d}{2}-\beta}\cdot |\hat{w}_{\theta}(\lambda)|^{2}(1+||\lambda||_{2}^{2})^{\beta}d\lambda \\
&\leq \sup_{\lambda}\left(|\hat{K}(h\lambda)|^{2}(1+||\lambda||_{2}^{2})^{\alpha+ \frac{d}{2}-\beta}\right)||w_{\theta}||_{2,2,\beta}^{2} \\
&\leq \sup_{\lambda}\left(|\hat{K}(h\lambda)|^{2}(1+||\lambda||_{2}^{2})^{\alpha+ \frac{d}{2}-\beta}\right)\sup_{\theta \in \Theta}||w_{\theta}||_{2,2,\beta}^{2}
\end{align*}
for all $\theta \in \Theta$. Dividing and multiplying by $(1+||h\lambda||_{2}^{2})^{\alpha + \frac{d}{2}-\beta}$ within the supremum and writing $\tilde{\lambda} = h\lambda$, it follows that
\begin{align*}
    ||f_{\theta,h}||_{\mathcal{H}}^{2} \leq  \bar{C}(\alpha,\beta,d,h)\tilde{C}(\alpha,\beta,d,\Theta),
\end{align*}
where 
\begin{align*}
\bar{C}(\alpha,\beta,d,h) = \sup_{\lambda}\left(\frac{1+||\lambda||_{2}^{2}}{1+||h\lambda||_{2}^{2}}\right)^{\alpha+\frac{d}{2}-\beta}    \leq \left(\frac{1}{h}\right)^{\alpha+d/2-\beta}
\end{align*}
for $0 < h \leq 1$, and
\begin{align*}
\tilde{C}(\alpha,\beta,d,\Theta)= \sup_{\tilde{\lambda}}\left[(1+||\tilde{\lambda}||_{2}^{2})^{\alpha+\frac{d}{2}-\beta}|\hat{K}(\tilde{\lambda})|^{2}\right]\sup_{\theta \in \Theta}||w_{\theta}||_{2,2,\beta}^{2} < \infty
\end{align*}
because $K$ is a higher-order kernel. Consequently, by setting $h= \delta^{1/\beta}$ and noting that this implies $||f_{\theta,h}-w_{\theta}||_{\infty} < \delta$, I obtain that there is a constant $\acute{C}(\alpha,\beta,d,\Theta)> 0 $ such that
\begin{align*}
\inf_{f \in \mathcal{H}: ||f-w_{\theta}||_{\infty}< \delta }\frac{1}{2}||f||_{\mathcal{H}}^{2} \leq \acute{C}(\alpha,\beta,d,\Theta)\delta^{-(2\alpha + d-2 \beta)/\beta} \quad \forall \ \theta \in \Theta,   
\end{align*}
which implies Part 1 of the lemma.
\paragraph{Step 3B.} Suppose that $\alpha < \beta < \alpha + d/2$. Since $\sup_{\theta \in \Theta}||w_{\theta}||_{2,2,\alpha} < \infty$ if $\alpha < \beta$ and $\sup_{\theta \in \Theta}||w_{\theta}||_{2,2,\beta}<\infty$, a similar argument to Step 3A reveals that
\begin{align*}
    ||f_{\theta,h}||_{\mathcal{H}}^{2} \leq \bar{C}(d,h)\tilde{C}(d,\alpha,\Theta)\leq \acute{C}(\alpha,d,\Theta)h^{-d},
\end{align*}
for $0<h \leq 1$, where
\begin{align*}
\bar{C}(d,h) = \sup_{\lambda}\left(\frac{1+||\lambda||_{2}^{2}}{1+||h\lambda||_{2}^{2}}\right)^{\frac{d}{2}}.    
\end{align*}
and
\begin{align*}
   \tilde{C}(d,\alpha,\Theta) \sup_{\lambda}\left[(1+||\lambda||_{2}^{2})^{\frac{d}{2}}|\hat{K}(\lambda)|^{2}\right]\sup_{\theta \in \Theta}||w_{\theta}||_{2,2,\alpha}^{2}
\end{align*}
Setting $h = \delta^{1/\beta}$, it similarly follows that
\begin{align*}
    \sup_{\theta \in \Theta}\inf_{f \in \mathcal{H}: ||f-w_{\theta}||_{\infty}< \delta }\frac{1}{2}||f||_{\mathcal{H}}^{2} \leq \acute{C}(\alpha,d,\Theta)\delta^{-\frac{d}{\beta}}.
\end{align*}
\paragraph{Step 3C.} Suppose that $\beta \geq \alpha + d/2$. To show that $\{w_{\theta}: \theta \in \Theta\} \subseteq \mathcal{H}$, it suffices, by Lemma 11.35 of \cite{ghosal2017fundamentals}, that the following holds:
\begin{align*}
    \int |\hat{w}_{\theta}(\lambda)|^{2}(1+||\lambda||_{2}^{2})^{\alpha + \frac{d}{2}}d \lambda < \infty
\end{align*}
for all $\theta \in \Theta$. Since $\beta \geq \alpha + d/2$ implies that
\begin{align*}
    \sup_{\lambda}\left(\frac{(1+||\lambda||_{2}^{2})^{\alpha+\frac{d}{2}}}{(1+||\lambda||_{2}^{2})^{\beta}}\right) \leq 1,
\end{align*}
it follows, by H\"{o}lder's inequality, that
\begin{align*}
   \int |\hat{w}_{\theta}(\lambda)|^{2}(1+||\lambda||_{2}^{2})^{\alpha + \frac{d}{2}} d \lambda \leq \int |\hat{w}_{\theta}(\lambda)|^{2}(1+||\lambda||_{2}^{2})^{\beta} d \lambda \leq \sup_{\theta \in \Theta} ||w_{\theta}||_{2,2,\beta}^{2} < \infty
\end{align*}
for all $\theta \in \Theta$ by the assumption on $\{w_{\theta}: \theta \in \Theta\}$. This establishes that $\{w_{\theta}: \theta \in \Theta\} \subseteq \mathcal{H}$, and, since, for any $\delta > 0$, $\{w_{\theta}: \theta \in \Theta\} \subseteq \{f \in \mathcal{H}: ||f-w_{\theta}||_{\infty} < \delta\}$, it follows that
\begin{align*}
  \sup_{\theta \in \Theta}\inf_{f:||f-w_{\theta}||_{\infty}<\delta }\frac{1}{2}||f||_{\mathcal{H}}^{2} \leq \sup_{\theta \in \Theta}||w_{\theta}||_{2,2,\beta}^{2} < \infty.  
\end{align*}
\end{proof}
\paragraph{Discussion of the Mat\'{e}rn Example.} Given a formal statement of Lemma 13, I elaborate on the rates for the Mat\'{e}rn GP:
\begin{enumerate}
\item \textit{RMSH Rate:} Lemma 3 of \cite{van2011information} states that $\tilde{\Pi}_{\mathcal{B}}$ satisfies $-\log \tilde{\Pi}_{\mathcal{B}}(||B||_{\infty} < \delta) \leq C \delta^{-d/\alpha}$ for any $\delta >0$, where $C>0$ is a constant that does not depend on $\delta$. Using this and Lemma 13, a sequence $\delta_{n}$ that satisfies $\sup_{\theta \in \Theta}\varphi_{b_{0,\theta}}(\delta_{n}) \lesssim n \delta_{n}^{2}$ can be computed as follows:
\begin{enumerate}
    \item Solve $\delta_{n}^{-(2\alpha + d - 2\alpha_{0})/\alpha_{0}} \lesssim n \delta_{n}^{2}$ and $\delta_{n}^{-\frac{d}{\alpha}} \lesssim n \delta_{n}^{2}$ if $\alpha_{0} \leq \alpha$.
    \item Solve $\delta_{n}^{-\frac{d}{\alpha}} \lesssim n \delta_{n}^{2}$ if $\alpha_{0} > \alpha$.
    \end{enumerate}
Case (a) is satisfied for $\delta_{n} = n^{-\alpha_{0}/(2\alpha+d)} = n^{-\min\{\alpha,\alpha_{0}\}/(2\alpha + d)}$, while Case (b) is satisfied for $\delta_{n} = n^{-\alpha/(2\alpha+d)} = n^{-\min\{\alpha,\alpha_{0}\}/(2\alpha+d)}$. Rearranging $\min\{\alpha_{0},\alpha\}/(2\alpha + d) > 1/4$ yields $\alpha_{0} > \alpha/2 + d/4$ and $\alpha > d/2$ as conditions for $o(n^{-1/4})$.
\item \textit{Maximum Hellinger Rate:} The leading term for the sequence $\tilde{\delta}_{n}$ derived in Proposition 2 is $2^{J(\frac{d_{z}}{2}-\alpha_{0})}+\sqrt{n}\delta_{n}2^{J(\frac{d_{z}}{2}-a)}$. Consequently, plugging in $\delta_{n} = n^{-\min\{\alpha,\alpha_{0}\}/(2\alpha + d)}$ and $2^{J} = (n/\log n)^{1/(2\alpha + d_{z})}$ leads to
\begin{align*}
\sqrt{n}\delta_{n}2^{J_{n}(\frac{d_{z}}{2}-a)} + 2^{J_{n}(\frac{d_{z}}{2}-\alpha_{0})} = n^{\frac{(\alpha - a) + d_{z}}{2\alpha + d_{z}}-\frac{\min\{\alpha,\alpha_{0}\}}{(2\alpha+d)}}(\log n)^{c_{1}} + n^{\frac{(d_{z}-2\alpha_{0})}{2(2\alpha+d_{z})}}(\log n)^{c_{2}},
\end{align*}
where $c_{1},c_{2} > 0$. Consequently, the condition $\tilde{\delta}_{n} = o(n^{-1/4})$ holds if 
\begin{align*}
\frac{\min\{\alpha,\alpha_{0}\}}{2\alpha + d} - \frac{(\alpha-a)+d_{z}}{2\alpha + d_{z}} > \frac{1}{4}, \quad     \frac{2\alpha_{0}-d_{z}}{2(2\alpha+d_{z})} > 1/4,
\end{align*}
which, when rearranged, lead to the conditions in the maim text. Figure \ref{fig:smooth} plots feasible $(\alpha,\alpha_{0})$ pairs for different values of $(d_{w},d_{z})$, highlighting the $o(n^{-1/4})$ condition is met if there is enough smoothness, and the prior and true smoothness are not too different.
\item \textit{Prior Invariance:} Applying Lemma 13, $\zeta_{n,j} = n^{-\beta_{0,j}/(2\alpha + d)}$ if $\beta_{0,j} \leq \alpha$ and $\zeta_{n,j} = n^{-\beta_{0,j}/(2\beta_{0,j} + d)}$ if $\beta_{0,j} \in (\alpha,\alpha+d/2)$. Consequently, if $\underline{\beta}:=\min_{1 \leq j \leq d_{\gamma}}\beta_{0,j} \leq \alpha$, then $\max_{1 \leq j \leq d_{\gamma}}\zeta_{n,j} = o(n^{-1/4})$ if $\underline{\beta}/(2\alpha + d) > 1/4$, which is the same as $ \underline{\beta} > \alpha/2 + d/4$. If $\underline{\beta} \in (\alpha,\alpha + d/2)$, then $\max_{1\leq j \leq d_{\gamma}}\zeta_{n,j} = o(n^{-1/4})$ iff $\beta_{0,j} > d/2$ for all $j=1,...,d_{\gamma}$, which is true if $\alpha > d/2$. If $\underline{\beta} \geq \alpha + d/2$, then the conditions of Proposition 3 hold automatically.
\end{enumerate}
\begin{figure}[H]
    \centering
        \caption{Feasible $(\alpha,\alpha_{0})$ Pairs for Empirical Uniform Convergence}
    \includegraphics[width=0.7\linewidth]{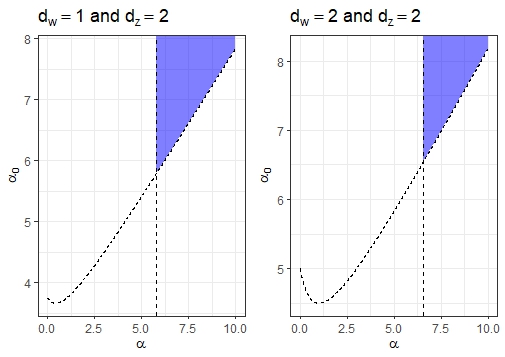}
    \label{fig:smooth}
\end{figure}
\bibliography{conditional}
\bibliographystyle{ecca}